\PassOptionsToPackage{svgnames}{xcolor}

\documentclass[a4paper,UKenglish,cleveref, autoref, thm-restate]{lipics-v2021}

\title{Designing Asynchronous Multiparty Protocols with Crash-Stop Failures}

\author{Adam D. Barwell}{%
University of St. Andrews and University of Oxford, UK
}{%
adb23@st-andrews.ac.uk
}{%
https://orcid.org/0000-0003-1236-7160
}{}
\author{Ping Hou}{University of Oxford, UK}{%
ping.hou@cs.ox.ac.uk
}{%
https://orcid.org/0000-0001-6899-9971
}{}
\author{Nobuko Yoshida}{University of Oxford, UK}{%
nobuko.yoshida@cs.ox.ac.uk
}{%
https://orcid.org/0000-0002-3925-8557
}{}
\author{Fangyi Zhou}
{Imperial College London and University of Oxford, UK}
{%
fangyi.zhou15@imperial.ac.uk
}{%
https://orcid.org/0000-0002-8973-0821
}{}

\authorrunning{A.D. Barwell, P. Hou, N. Yoshida, F. Zhou}

\Copyright{Adam D. Barwell, Ping Hou, Nobuko Yoshida, and Fangyi Zhou}



\category{} %

\supplement{ECOOP 2023 Artifact Evaluation approved artifact, also available at \url{https://github.com/adbarwell/ECOOP23-Artefact}}
\funding{
  Work supported by:
  EPSRC EP/T006544/2, EP/K011715/1,
EP/K034413/1, EP/L00058X/1, EP/N027833/2, EP/N028201/1, EP/T014709/2, EP/V000462/1,
EP/X015955/1, NCSS/EPSRC VeTSS,  and Horizon EU TaRDIS 101093006.
}%

\acknowledgements{We thank the anonymous reviewers for their useful comments and suggestions.
We thank Jia Qing Lim for his contribution to the \Effpi extension.  
We thank Alceste Scalas for useful discussions and advice in the 
development of this paper and for his assistance with \Effpi.}

\nolinenumbers %
\hideLIPIcs
\EventEditors{Karim Ali and Guido Salvaneschi}
\EventNoEds{2}
\EventLongTitle{37th European Conference on Object-Oriented Programming (ECOOP 2023)}
\EventShortTitle{ECOOP 2023}
\EventAcronym{ECOOP}
\EventYear{2023}
\EventDate{July 17--21, 2023}
\EventLocation{Seattle, Washington, United States}
\EventLogo{}
\SeriesVolume{263}
\ArticleNo{30}

\usepackage{booktabs}   %
\usepackage{subcaption} %
\usepackage{amsthm}
\usepackage{etoolbox}%

\usepackage[dvipsnames]{xcolor}%

\usepackage[nomargin,inline,index,%
  status=draft %
]{fixme} %
\fxusetheme{colorsig}
\FXRegisterAuthor{fxAS}{anfxAS}{AS}%
\FXRegisterAuthor{fxAB}{anfxAB}{AB}%
\FXRegisterAuthor{fxNY}{anfxNY}{NY}%
\FXRegisterAuthor{fxFZ}{anfxFZ}{FZ}
\FXRegisterAuthor{fxPH}{anfxPH}{PH}

\makeatletter
\DeclareFontFamily{OMX}{MnSymbolE}{}
\DeclareSymbolFont{MnLargeSymbols}{OMX}{MnSymbolE}{m}{n}
\SetSymbolFont{MnLargeSymbols}{bold}{OMX}{MnSymbolE}{b}{n}
\DeclareFontShape{OMX}{MnSymbolE}{m}{n}{
    <-6>  MnSymbolE5
   <6-7>  MnSymbolE6
   <7-8>  MnSymbolE7
   <8-9>  MnSymbolE8
   <9-10> MnSymbolE9
  <10-12> MnSymbolE10
  <12->   MnSymbolE12
}{}
\DeclareFontShape{OMX}{MnSymbolE}{b}{n}{
    <-6>  MnSymbolE-Bold5
   <6-7>  MnSymbolE-Bold6
   <7-8>  MnSymbolE-Bold7
   <8-9>  MnSymbolE-Bold8
   <9-10> MnSymbolE-Bold9
  <10-12> MnSymbolE-Bold10
  <12->   MnSymbolE-Bold12
}{}

\let\llangle\@undefined
\let\rrangle\@undefined
\DeclareMathDelimiter{\llangle}{\mathopen}%
                     {MnLargeSymbols}{'164}{MnLargeSymbols}{'164}
\DeclareMathDelimiter{\rrangle}{\mathclose}%
                     {MnLargeSymbols}{'171}{MnLargeSymbols}{'171}
\makeatother

\usepackage[inline]{enumitem}%

\usepackage{nicefrac}%

\usepackage{proof}%

\usepackage[most]{tcolorbox}%
\newtcolorbox{myframe}[1][]{
  enhanced,
  arc=0pt,
  outer arc=0pt,
  colback=white,
  boxrule=0.5pt,
  boxsep=0mm,
  left=1mm,
  right=1mm,
  top=0.5mm,
  bottom=0.5mm,
  #1
}

\lstdefinelanguage{Scribble}{%
  basicstyle=\footnotesize\ttfamily,
  stringstyle=\color{Blue},
  showstringspaces=false,
  keywords={nested,new,calls,and,as,at,by,catches,choice,continue,do,from,global,import,instantiates,interruptible,local,module,or,par,protocol,rec,role,sig,throws,to,type,with,int,aux,reliable,crash},
  morestring=[b]",
  morestring=[b]',
  morecomment=[l][\color{greencomments}]{//},
}

\lstdefinelanguage{nuScr}{%
  basicstyle=\footnotesize\ttfamily,
  stringstyle=\color{Blue},
  showstringspaces=false,
  keywords={
    nested,new,calls,and,as,at,by,catches,choice,continue,do,from,global,import,instantiates,interruptible,local,module,or,par,protocol,rec,role,sig,throws,to,type,with,int,aux,
    safe
  },
  morestring=[b]",
  morestring=[b]',
  morecomment=[l][\color{greencomments}]{//},
  morecomment=[s][\color{magenta}]{(*}{*)},
}

\lstdefinelanguage{effpi}{
  keywords=[1]{
    case,class,sealed,abstract,object,extends,type,def,val,if,else,new,var,match
  },
  keywords=[2]{
    InChan,OutChan,RecVar,Rec,Out,In,InErr,Loop,
  },
  keywords=[3]{
    rec,send,receive,receiveErr,eval,par,Channel
  },
  keywordstyle=[1]{\color{blue}},
  keywordstyle=[2]{\color{ImperialIris}}, %
  keywordstyle=[3]{\color{OliveGreen}},
  otherkeywords={=>,.type,<:,>>:},
  morecomment=[l][\color{darkgray}]{//},
}

\usepackage{stmaryrd}
\usepackage{thmtools}%
\usepackage{arydshln}

\usepackage{xifthen}%
\usepackage{xspace}%

\usepackage{mathtools}%

\usepackage{balance}%

\usepackage{hyperref}
\hypersetup{hidelinks}
\usepackage[capitalise]{cleveref}

\definecolor{ImperialBlue}{HTML}{003E74}
\definecolor{ImperialDarkGreen}{HTML}{02893B}
\definecolor{ImperialTangerine}{HTML}{EC7300}
\definecolor{ImperialIris}{HTML}{751E66}

\definecolor{RYB1}{RGB}{141, 211, 199}
\definecolor{RYB2}{RGB}{255, 255, 179}
\definecolor{RYB3}{RGB}{190, 186, 218}
\definecolor{RYB4}{RGB}{251, 128, 114}
\definecolor{RYB5}{RGB}{128, 177, 211}
\definecolor{RYB6}{RGB}{253, 180, 98}
\definecolor{RYB7}{RGB}{179, 222, 105}

\usetikzlibrary{automata, positioning, arrows, calc, fit}
\tikzset{
  >=stealth,
  node distance=2cm,
  every state/.style={thick, fill=gray!10},
  initial text=$ $,
}
\usepackage{pgfplots}
\pgfplotscreateplotcyclelist{stackexchange}{
  {RYB1!50!black,fill=RYB1},
  {RYB4!50!black,fill=RYB4},
  {RYB3!50!black,fill=RYB3},
  {RYB2!50!black,fill=RYB2},
  {RYB5!50!black,fill=RYB5},
  {RYB6!50!black,fill=RYB6},
  {RYB7!50!black,fill=RYB7},
}
\pgfplotscreateplotcyclelist{stackexchangePlusOne}{
  {RYB4!50!black,fill=RYB4},
  {RYB3!50!black,fill=RYB3},
  {RYB2!50!black,fill=RYB2},
  {RYB5!50!black,fill=RYB5},
  {RYB6!50!black,fill=RYB6},
  {RYB7!50!black,fill=RYB7},
}
\pgfplotsset{
  compat=1.8,
  /pgfplots/bar cycle list/.style={/pgfplots/cycle list={%
    {brown!60!black,fill=brown!30!white,mark=none},
    {red,fill=red!30!white,mark=none},
    {blue,fill=blue!30!white,mark=none},
    {black,fill=gray,mark=none},
    }
  },
}
\newcolumntype{L}{>{$}l<{$}}
\newcolumntype{C}{>{$}c<{$}}
\newcolumntype{P}[1]{>{\centering\arraybackslash$}p{#1}<{$}}
\usepackage{multirow}

\usepackage[export]{adjustbox}

\Crefname{section}{\S\!}{\S\!}%
\Crefname{subsection}{\S\!}{\S\!}%
\Crefname{subsubsection}{\S\!}{\S\!}%
\Crefname{appendix}{Appendix \S\!}{Appendix \S\!}
\Crefname{definition}{Def.\@}{Defs.\@}%
\Crefname{figure}{Fig.\@}{Figs.\@}%
\Crefname{example}{Ex.\@}{Exs.\@}%
\Crefname{corollary}{Cor.\@}{Cors.\@}%
\Crefname{theorem}{Thm.\@}{Thms.\@}%
\Crefname{proposition}{Prop.\@}{Props.\@}%
\Crefname{lemma}{Lem.\@}{Lems.\@}
\Crefname{equation}{Eq.\@}{Eqs.\@}

\crefname{section}{\S\!}{\S\!}%
\crefname{subsection}{\S\!}{\S\!}%
\crefname{subsubsection}{\S\!}{\S\!}%
\crefname{appendix}{Appendix \S\!}{Appendix \S\!}
\crefname{definition}{Def.\@}{Defs.\@}%
\crefname{figure}{Fig.\@}{Figs.\@}%
\crefname{example}{Ex.\@}{Exs.\@}%
\crefname{corollary}{Cor.\@}{Cors.\@}%
\crefname{theorem}{Thm.\@}{Thms.\@}%
\crefname{proposition}{Prop.\@}{Props.\@}%
\crefname{lemma}{Lem.\@}{Lems.\@}
\crefname{equation}{Eq.\@}{Eqs.\@}

\usepackage{cite} %

\usepackage{magicvariables}

\newif\ifdraft%
  \draftfalse%
  \drafttrue%

\newcommand{\ifempty}[3]{%
  \ifthenelse{\isempty{#1}}{#2}{#3}%
}%

\newtoggle{techreport}%

\newtoggle{cruft}%

\newcommand{\effpi}{\texttt{Effpi}\xspace}%

\newcommand{\makeFunc}[2]{%
  \expandafter\newcommand\csname #1\endcsname[1]{%
    \operatorname{#2}\!\left({##1}\right)%
  }%
}%
\newcommand{\dom}[1]{{\color{black}\operatorname{dom}\!\left({#1}\right)}}%
\makeFunc{fsucc}{suc}
\makeFunc{finv}{inv}
\newcommand{\fv}[1]{\operatorname{fv}\!\left({#1}\right)}%
\newcommand{\unfoldOne}[1]{%
  {\color{black}\operatorname{unf}\!\left({#1}\right)}}%
\newcommand{\notImpliedBy}{\mathrel{{\kern 1em}{\not{\kern -1em}\impliedby}}}%

\newcommand{\coloncolonequals}{\Coloneqq}%
\newcommand{\bnfdef}{\coloncolonequals}%
\newcommand{\bnfsep}{\mathbin{\;\big|\;}}%

\newcommand{\Scribble}[0]{\textsc{Scribble}\xspace}
\newcommand{\Effpi}[0]{\textsc{Effpi}\xspace}
\newcommand{\theTool}[0]{\textsc{Teatrino}\xspace}
\newcommand{\Scala}[0]{\textsc{Scala}\xspace}
\newcommand{\Rust}[0]{\textsc{Rust}\xspace}
\newcommand{\Go}[0]{\textsc{Go}\xspace}

\newcommand{\Choral}[0]{\textsc{Choral}\xspace}

\def\aka{a.k.a.\@\xspace}%
\def\cf{cf.\@\xspace}%
\def\etc{\emph{etc.}\@\xspace}%
\def\eg{e.g.\@\xspace}%
\def\ie{i.e.\@\xspace}%
\def\wrt{w.r.t.\@\xspace}%
\def\etal{\emph{et al.}\@\xspace}%

\newcommand{\exampleName}[1]{\ensuremath{\mathsf{#1}}\xspace}

\definecolor{ruleColor}{rgb}{0.1, 0.3, 0.1}%
\newcommand{\inferrule}[1]{{\color{ruleColor}\text{\upshape\textsc{\scriptsize[#1]}}}}%
\newcommand{\inference}[3][]{\infer[\ifempty{#1}{}{\inferrule{#1}}]{#3}{#2}}%
\newcommand{\cinference}[3][]{\infer=[\ifempty{#1}{}{\inferrule{#1}}]{#3}{#2}}%

\newcommand{\setenum}[1]{\mathord{{\color{black}\left\{#1\right\}}}}%
\newcommand{\setcomp}[2]{\mathord{%
  {\color{black}\left\{{#1} \,\middle|\, {#2}\right\}}}}%

\newcommand{\setNat}{\mathord{\mathbb{N}}}%

\newcommand{\predP}[1][]{\ifempty{#1}{\varphi}{\varphi_{#1}}}%
\newcommand{\predPApp}[2][]{\ifempty{#1}{\predP}{\predP[{#1}]}\!\left({#2}\right)}%

\newcommand{\bind}[2]{\nicefrac{#2}{#1}}%
\newcommand{\substenum}[1]{\mathord{\left\{{#1}\right\}}}%
\newcommand{\subst}[2]{\substenum{\bind{#1}{#2}}}%

\definecolor{hlColor}{rgb}{0.65, 1.0, 0.65}%

\newcommand{\highlight}[2][\highlightColour]{\mathchoice%
  {\setlength{\fboxsep}{0pt}\colorbox{#1}{$\displaystyle#2$}}%
  {\setlength{\fboxsep}{0pt}\colorbox{#1}{$\textstyle#2$}}%
  {\setlength{\fboxsep}{0pt}\colorbox{#1}{$\scriptstyle#2$}}%
  {\setlength{\fboxsep}{0pt}\colorbox{#1}{$\scriptscriptstyle#2$}}}%

\newcommand{\runtime}[2][\runtimeColour]{\mathchoice%
  {\setlength{\fboxsep}{0pt}\colorbox{#1}{$\displaystyle#2$}}%
  {\setlength{\fboxsep}{0pt}\colorbox{#1}{$\textstyle#2$}}%
  {\setlength{\fboxsep}{0pt}\colorbox{#1}{$\scriptstyle#2$}}%
  {\setlength{\fboxsep}{0pt}\colorbox{#1}{$\scriptscriptstyle#2$}}}%

\newcommand{\lbbar}{\{\kern-0.2em|}
\newcommand{\rbbar}{|\kern-0.2em\}}

\newcommand{\tyGroundSet}{\stFmt{\mathcal{B}}}
\newcommand{\tyGround}[1][]{\stFmt{\ifempty{#1}{B}{B_{#1}}}}%
\newcommand{\tyGroundi}[1][]{\stFmt{\ifempty{#1}{B'}{B'_{#1}}}}%
\newcommand{\tyGroundii}[1][]{\stFmt{\ifempty{#1}{B''}{B''_{#1}}}}%
\newcommand{\tyGroundiii}[1][]{\stFmt{\ifempty{#1}{B'''}{B'''_{#1}}}}%
\newcommand{\tyBool}{\stFmtC{bool}}%
\newcommand{\tyUnit}{\stFmtC{unit}}%
\newcommand{\tyInt}{\stFmtC{int}}%
\newcommand{\tyReal}{\stFmtC{real}}%
\newcommand{\tyString}{\stFmtC{str}}%
{\centerline{\bf --- Begin Copied From Previous Paper ---} \hrule}%
{\hrule \centerline{\bf --- End Copied From Previous Paper ---}}%

{\centerline{\bf --- Begin Discussion\ifempty{#1}{}{: {#1}} ---}
  \hrule\vspace{1mm}}%
{\hrule\vspace{1mm}\centerline{\bf --- End Discussion ---}}%

\definecolor{roleColor}{rgb}{0.5, 0.0, 0.0}%
\newcommand{\roleCol}[1]{{\color{roleColor}#1}}%
\newcommand{\roleSet}{\roleCol{\mathcal{R}}}%
\newcommand{\roleFmt}[1]{\ensuremath{{\boldsymbol{\roleCol{\mathtt{#1}}}}}\xspace}%
\newcommand{\roleSetOf}[1]{\roleFmt{\{{#1}\}}}%
\newcommand{\roleCrashedSym}{\roleCol{\lightning}}
\newcommand{\roleMaybeCrashedSym}{\roleCol{\dagger}}

\newcommand{\roleP}[1][]{%
  \ifempty{#1}{{\color{roleColor}\roleFmt{p}}}{{\color{roleColor}\roleFmt{p}_{#1}}}%
}%
\newcommand{\rolePCrashed}[1][]{%
  \ifempty{#1}{{\color{roleColor}\roleFmt{p^{\roleCrashedSym}}}}{{\color{roleColor}\roleFmt{p^{\roleCrashedSym}}_{#1}}}%
}%
\newcommand{\rolePMaybeCrashed}[1][]{%
  \ifempty{#1}{{\color{roleColor}\roleFmt{p^{\roleMaybeCrashedSym}}}}{{\color{roleColor}\roleFmt{p^{\roleMaybeCrashedSym}}_{#1}}}%
}%
\newcommand{\roleQ}[1][]{%
  \ifempty{#1}{{\color{roleColor}\roleFmt{q}}}{{\color{roleColor}\roleFmt{q}_{#1}}}%
}%
\newcommand{\roleQCrashed}[1][]{%
  \ifempty{#1}{{\color{roleColor}\roleFmt{q^{\roleCrashedSym}}}}{{\color{roleColor}\roleFmt{q^{\roleCrashedSym}}_{#1}}}%
}%
\newcommand{\roleQMaybeCrashed}[1][]{%
  \ifempty{#1}{{\color{roleColor}\roleFmt{q^{\roleMaybeCrashedSym}}}}{{\color{roleColor}\roleFmt{q^{\roleMaybeCrashedSym}}_{#1}}}%
}%
\newcommand{\roleR}[1][]{%
  \ifempty{#1}{{\color{roleColor}\roleFmt{r}}}{{\color{roleColor}\roleFmt{r}_{\!#1}}}%
}%
\newcommand{\roleRMaybeCrashed}[1][]{%
  \ifempty{#1}{{\color{roleColor}\roleFmt{r^{\roleMaybeCrashedSym}}}}{{\color{roleColor}\roleFmt{r^{\roleMaybeCrashedSym}}_{#1}}}%
}%
\newcommand{\roleS}[1][]{%
  \ifempty{#1}{{\color{roleColor}\roleFmt{s}}}{{\color{roleColor}\roleFmt{s}_{\!#1}}}%
}%
\newcommand{\roleSMaybeCrashed}[1][]{%
  \ifempty{#1}{{\color{roleColor}\roleFmt{s^{\roleMaybeCrashedSym}}}}{{\color{roleColor}\roleFmt{s^{\roleMaybeCrashedSym}}_{#1}}}%
}%
\newcommand{\roleT}[1][]{%
  \ifempty{#1}{{\color{roleColor}\roleFmt{t}}}{{\color{roleColor}\roleFmt{t}_{\!#1}}}%
}%
\newcommand{\roleTCrashed}[1][]{%
  \ifempty{#1}{{\color{roleColor}\roleFmt{t^{\roleCrashedSym}}}}{{\color{roleColor}\roleFmt{t^{\roleCrashedSym}}_{#1}}}%
}%
\newcommand{\roleTMaybeCrashed}[1][]{%
  \ifempty{#1}{{\color{roleColor}\roleFmt{t^{\roleMaybeCrashedSym}}}}{{\color{roleColor}\roleFmt{t^{\roleMaybeCrashedSym}}_{#1}}}%
}%
\newcommand{\roleU}[1][]{%
  \ifempty{#1}{{\color{roleColor}\roleFmt{u}}}{{\color{roleColor}\roleFmt{u}_{\!#1}}}%
}%
\newcommand{\roleUMaybeCrashed}[1][]{%
  \ifempty{#1}{{\color{roleColor}\roleFmt{u^{\roleMaybeCrashedSym}}}}{{\color{roleColor}\roleFmt{u^{\roleMaybeCrashedSym}}_{#1}}}%
}%

\makeVars{roleA}{roleFmt}{a}
\makeVars{roleC}{roleFmt}{c}

\makeVars{roles}{roleFmt}{\mathscr P}
\makeVars{rolesR}{roleFmt}{\mathscr R}
\makeVars{rolesC}{roleFmt}{\mathscr C}
\makeVars{rolesS}{roleFmt}{\mathfrak R}
\newcommand{\rolesEmpty}[0]{\gtFmt{\emptyset}}

\definecolor{gtColor}{rgb}{0.43, 0.21, 0.1}%
\newcommand{\gtFmt}[1]{\ensuremath{{\color{gtColor}#1}}\xspace}%
\newcommand{\gtMsgFmt}[1]{\gtFmt{\labFmt{#1}}}%
\newcommand{\gtLabFmt}[1]{\gtFmt{\labFmt{#1}}}%
\newcommand{\gtLab}[1][]{%
  \ifempty{#1}{\gtMsgFmt{m}}{{\color{gtColor}\gtMsgFmt{m}_{#1}}}%
}%
\newcommand{\gtLabi}[1][]{%
  \ifempty{#1}{\gtMsgFmt{m}'}{{\color{gtColor}\gtMsgFmt{m}'_{#1}}}%
}%

\makeVarsEx{gtG}{gtFmt}{G}

\newcommand{\gtSeq}{\mathbin{\gtFmt{.}}}%

\newcommand{\gtCommRaw}[3]{%
  \gtFmt{%
    {#1} {\to} {#2}{:}%
    \left\{%
      {#3}%
    \right\}%
  }%
}%

\newcommand{\gtCommRawAnn}[4]{%
  \gtFmt{%
    {#1} {\xrightarrow{#3}} {#2}{:}%
    \left\{%
      {#4}%
    \right\}%
  }%
}%

\newcommand{\gtComm}[6]{%
  \gtFmt{%
    \gtCommRaw{#1}{#2}{%
      \gtCommChoice{#4}{#5}{#6}%
    }_{#3}%
  }%
}%
\newcommand{\gtCommAnn}[7]{%
  \gtFmt{%
    \gtCommRawAnn{#1}{#2}{#3}{%
      \gtCommChoice{#5}{#6}{#7}%
    }_{#4}%
  }%
}%
\newcommand{\gtCommSmall}[6]{%
  \gtFmt{%
    \gtCommRaw{#1}{#2}{%
      \gtCommChoiceSmall{#4}{#5}{#6}%
    }_{#3}%
  }%
}%
\newcommand{\gtCommSingle}[5]{%
  \gtFmt{%
    {#1} {\to} {#2}{:}%
    \gtCommChoice{#3}{#4}{#5}%
  }%
}%
\newcommand{\gtCommSingleAnn}[6]{%
  \gtFmt{%
    {#1} {\xrightarrow{#3}} {#2}{:}%
    \gtCommChoice{#4}{#5}{#6}%
  }%
}%
\newcommand{\gtCommTransitSingle}[5]{%
  \gtFmt{%
    {#1} {\rightsquigarrow} {#2}{:}%
    \gtCommChoice{#3}{#4}{#5}%
  }%
}%
\newcommand{\gtCommTransitRaw}[4]{%
  \gtFmt{%
    {#1} {\rightsquigarrow} {#2}{:}{#4}%
    \left\{%
      {#3}%
    \right\}%
  }%
}%

\newcommand{\gtCommTransit}[7]{%
  \gtFmt{%
    \gtCommTransitRaw{#1}{#2}{%
      \gtCommChoice{#4}{#5}{#6}%
    }{#7}_{#3}%
  }%
}%

\newcommand{\gtCommSingleErr}[6]{%
  \gtFmt{%
    {#1} {\to} {#2}{:}%
    \{\gtCommChoice{#3}{#4}{#5},\quad \gtErrKFmt{#6}\}%
  }%
}%
\newcommand{\gtCommChoice}[3]{%
  \gtFmt{%
    \gtMsgFmt{#1}\ifempty{#2}{}{({#2})}%
    \ifempty{#3}{}{\vphantom{x} \!\gtSeq\! {#3}}%
  }%
}%
\newcommand{\gtCommChoiceSmall}[3]{%
  \gtFmt{%
    \gtMsgFmt{#1}\ifempty{#2}{}{({#2})}%
    \ifempty{#3}{}{\vphantom{x} \!\gtSeq\! {#3}}%
  }%
}%

\newcommand{\gtEnd}{\gtFmt{\mathsf{end}}}%

\newcommand{\gtRec}[2]{\gtFmt{\mu{#1}.{#2}}}%
\newcommand{\gtRecVarBase}{\gtFmt{\mathbf{t}}}%
\makeVars{gtRecVar}{gtFmt}{\gtRecVarBase}

\newcommand{\gtRoles}[1]{{\color{roleColor} \operatorname{roles}(\gtFmt{#1})}}%
\newcommand{\gtRolesCrashed}[1]{{\color{roleColor}\operatorname{roles}^{\roleCrashedSym}(\gtFmt{#1})}}%

\newcommand{\gtProj}[3][]{%
  {\color{stColor}\gtFmt{#2} \ifempty{#1}{\upharpoonright}{\upharpoonright_{#1}} \roleFmt{#3}}%
}%

\newcommand{\gtStopLab}[0]{\gtMsgFmt{\mathsf{crash}}}
\newcommand{\gtCrashLab}[0]{\gtMsgFmt{\mathsf{crash}}}

\newcommand{\gtWithCrashedRoles}[2]{\langle{#1};{#2}\rangle}

\newcommand{\gtCrashRole}[2]{\gtFmt{#1}\gtFmt{\lightning}\roleFmt{#2}}

\newcommand{\gtMove}[2][\phantom{\stEnvAnnotGenericSym}]{\xrightarrow{#1}_{#2}} %
\newcommand{\gtMoveStar}[1][]{\ifempty{#1}{\gtMove[]{}^{\!\!\!*}}{\gtMove[]{#1}^{\!*}}} %

\newcommand{\iruleGtMove}[1]{GR-{#1}}
\newcommand{\iruleGtMoveCrash}[0]{\iruleGtMove{$\lightning$}}

\newcommand{\iruleGtMoveOut}[0]{\iruleGtMove{$\oplus$}}
\newcommand{\iruleGtMoveIn}[0]{\iruleGtMove{$\&$}}
\newcommand{\iruleGtMoveRec}[0]{\iruleGtMove{$\mu$}}
\newcommand{\iruleGtMoveCrDe}[0]{\iruleGtMove{$\odot$}}
\newcommand{\iruleGtMoveOrph}[0]{\iruleGtMove{$\lightning\gtLab$}}
\newcommand{\iruleGtMoveCtx}[0]{\iruleGtMove{Ctx-i}}
\newcommand{\iruleGtMoveCtxi}[0]{\iruleGtMove{Ctx-ii}}

\newcommand{\labFmt}[2][]{\ensuremath{\ifempty{#1}{\mathtt{#2}}{\mathtt{#2}\textsubscript{#1}}}\xspace}%
\newcommand{\labSet}{\mathcal{M}}

\newcommand{\gtErrKFmt}[1]{\gtStopLab.{#1}}

\definecolor{stColor}{rgb}{0, 0, 0.9}%
\newcommand{\stFmt}[1]{\ensuremath{{\color{stColor}#1}}\xspace}%
\newcommand{\stFmtC}[1]{\stFmt{\operatorname{#1}}}
\newcommand{\stMsgFmt}[1]{\stFmt{\labFmt{#1}}}%

\newcommand{\stInNB}[4]{\ifempty{#1}{}{\roleFmt{#1}}\stFmt{\&{#2}\ifempty{#3}{}{({#3})} \ifempty{#4}{}{\stSeq #4}}}%
\newcommand{\stOut}[3]{\ifempty{#1}{}{\roleFmt{#1}}\stFmt{\oplus{#2}\ifempty{#3}{}{({#3})}}}%
\newcommand{\stOutAnn}[4]{\ifempty{#1}{}{\roleFmt{#1}}\stFmt{\oplus^{#2}{#3}\ifempty{#4}{}{({#4})}}}%

\newcommand{\stChoice}[2]{\stLabFmt{#1}\ifempty{#2}{}{\stFmt{({#2})}}}%

\newcommand{\stSeq}{\mathbin{\!\stFmt{.}\!}}%
\newcommand{\stIntC}{\mathbin{\stFmt{\oplus}}}%
\newcommand{\stIntSum}[3]{\roleFmt{#1}\stFmt{\oplus\!\left\{#3\right\}_{#2}}}%
\newcommand{\stIntSumAnn}[4]{\roleFmt{#1}\stFmt{\oplus^{#2}\!\left\{#4\right\}_{#3}}}%
\newcommand{\stExtC}{\mathbin{\stFmt{\&}}}%
\newcommand{\stExtSum}[3]{\roleFmt{#1}\stFmt{\&\!\left\{#3\right\}_{#2}}}%
\newcommand{\stExtSumAnn}[4]{\roleFmt{#1}\stFmt{\&^{#2}\!\left\{#4\right\}_{#3}}}%
\newcommand{\stRec}[2]{\stFmt{\mu{#1}.{#2}}}%
\newcommand{\stEnd}{\stFmt{\mathsf{end}}}%
\newcommand{\stStop}{\stFmt{\mathsf{stop}}}%

\newcommand{\stLabFmt}[1]{\stFmt{\labFmt{#1}}}%
\newcommand{\stLab}[1][]{%
  \ifempty{#1}{\stLabFmt{m}}{\stLabFmt{m}_{{\color{stColor}#1}}}
}%
\newcommand{\stLabi}[1][]{%
  \ifempty{#1}{\stLabFmt{m'}}{\stLabFmt{m'}_{{\color{stColor}#1}}}
}%
\newcommand{\stLabii}[1][]{%
  \ifempty{#1}{\stLabFmt{m''}}{\stLabFmt{m''}_{{\color{stColor}#1}}}
}%

\newcommand{\stCrashLab}[0]{\stLabFmt{\mathsf{crash}}}

\makeVars{stS}{stFmt}{S}
\makeVars{stT}{stFmt}{T}
\makeVars{stU}{stFmt}{U}

\newcommand{\stRecVarBase}{\stFmt{\mathbf{t}}}%
\makeVars{stRecVar}{stFmt}{\stRecVarBase}

\newcommand{\stMerge}[2]{\stFmt{\bigsqcap_{#1}{#2}}}%
\newcommand{\stBinMerge}{\mathbin{\stFmt{\sqcap}}}%

\newcommand{\stSub}{\mathrel{\stFmt{\leqslant}}}%

\newcommand{\ltsCrDe}[3]{{#2}\stFmt{\mathord{\odot}}\roleFmt{#3}}
\newcommand{\ltsCrash}[2]{{#2}\stFmt{\lightning}}
\newcommand{\ltsCrashSmall}[2]{{#2}\stFmt{\lightning}}

\newcommand{\ltsSubject}[1]{{\color{roleColor} \operatorname{subj}({#1})}}%

\definecolor{mpColor}{rgb}{0, 0, 0}%
\newcommand{\mpFmt}[1]{{\color{mpColor}#1}}%

\newcommand{\mpLab}[1][]{%
  \mpFmt{\ifempty{#1}{\labFmt{m}}{{\labFmt{m}}_{\mathnormal #1}}}%
}%
\newcommand{\mpLabi}[1][]{%
  \mpFmt{\ifempty{#1}{\labFmt{m}'}{\labFmt{m}'_{\mathnormal #1}}}%
}%

\newcommand{\mpCrashLab}[0]{\mpFmt{\labFmt{\mathsf{crash}}}}

\newcommand{\mpTrue}{\mpFmt{\text{\textit{\texttt{true}}}}}%
\newcommand{\mpFalse}{\mpFmt{\text{\textit{\texttt{false}}}}}%

\newcommand{\mpChanRole}[2]{\mpFmt{{#1}[{#2}]}}%

\newcommand{\mpNil}{\mpFmt{\mathbf{0}}}%
\newcommand{\mpIf}[3]{%
  \mpFmt{\mathsf{if}\,{#1}\,\mathsf{then}}\,{#2}\,\mathsf{else}\,{#3}%
}%
\newcommand{\mpPar}{\mathbin{\mpFmt{\mid}}}%

\newcommand{\mpPart}[2]{#1 \triangleleft  #2}
\newcommand{\mpCrash}[0]{\mpFmt{\ensuremath{\lightning}}}

\makeVars{mpECtx}{mpFmt}{\mathbb{E}}

\newcommand{\mpV}[1][]{\mpFmt{\ifempty{#1}{v}{v_{#1}}}}

\newcommand{\mpW}[1][]{\mpFmt{\ifempty{#1}{w}{w_{#1}}}}

\makeVars{mpE}{mpFmt}{e}
\makeVars{mpM}{mpFmt}{\mathcal{M}}
\makeVars{mpH}{mpFmt}{h}
\makeVars{mpx}{mpFmt}{x}
\newcommand{\mpQEmpty}{\mpFmt{\epsilon}}%
\newcommand{\mpQUnavail}{\mpFmt{\oslash}}

\newcommand{\mpS}[1][]{\mpFmt{\ifempty{#1}{s}{s_{#1}}}}%

\newcommand{\mpX}[1][]{\mpFmt{\ifempty{#1}{X}{X_{#1}}}}%

\newcommand{\mpP}[1][]{\mpFmt{\ifempty{#1}{P}{P_{#1}}}}%
\newcommand{\mpPi}[1][]{\mpFmt{\ifempty{#1}{P'}{P'_{#1}}}}%
\newcommand{\mpQ}[1][]{\mpFmt{\ifempty{#1}{Q}{Q_{#1}}}}%

\newcommand{\mpMove}[1][]{\ifempty{#1}{\to}{\to_{#1}}}%

\newcommand{\mpMoveStar}[1][]{\mathrel{\mpMove[#1]{}^{\!\!\!*}}}%
\newcommand{\mpNotMove}[1][]{\mathrel{\not{\!\!\mpMove[#1]}}}

\newcommand{\iruleSafeComm}{S-${\stIntC}{\stExtC}$}%
\newcommand{\iruleSafeCrash}{S-${\stFmt{\lightning}}{\stExtC}$}%
\newcommand{\iruleSafeRec}{S-$\stFmt{\mu}$}%
\newcommand{\iruleSafeMove}{S-$\stEnvMoveMaybeCrash$}%

\newcommand{\iruleStSubEnd}{Sub-$\stEnd$}
\newcommand{\iruleStSubStop}{Sub-$\stStop$}
\newcommand{\iruleStSubRecL}{Sub-$\stFmt{\mu}$L}
\newcommand{\iruleStSubRecR}{Sub-$\stFmt{\mu}$R}
\newcommand{\iruleStSubOut}{Sub-$\stFmt{\oplus}$}
\newcommand{\iruleStSubIn}{Sub-$\stFmt{\&}$}

\newcommand{\stStopSym}{\stFmt{\ensuremath{\lightning}}}

\newcommand{\iruleTCtxOut}{$\stEnv$-$\stFmt{\oplus}$}%
\newcommand{\iruleTCtxIn}{$\stEnv$-$\stFmt{\&}$}%
\newcommand{\iruleTCtxRec}{$\stEnv$-$\mu$}%
\newcommand{\iruleTCtxCrash}{$\stEnv$-$\lightning$}
\newcommand{\iruleTCtxCrashDetect}{$\stEnv$-$\odot$}

\makeVars{stEnv}{stFmt}{\Gamma}
\newcommand{\stEnvEmpty}{\stFmt{\emptyset}}%
\newcommand{\stEnvMap}[2]{\stFmt{\mpFmt{#1}\mathbin{\!\triangleright\!}{#2}}}%
\newcommand{\stEnvComp}{\mathpunct{\stFmt{,}}}%
\newcommand{\stEnvApp}[2]{\stFmt{#1\!\left(\mpFmt{#2}\right)}}%
\newcommand{\stEnvUpd}[3]{\stFmt{#1 {} [#2 \mapsto #3]}}

\newcommand{\stEnvAssoc}[3]{\stFmt{{#2} \mathrel{\stFmt{\sqsubseteq}_{#3}} {#1}}}

\newcommand{\stEnvMove}{\mathrel{\stFmt{\to}}}%
\newcommand{\stEnvMoveMaybeCrash}[1][]{%
  \ifempty{#1}{%
    \mathrel{\stFmt{\to_{\!\lightning}}}%
  }{%
    \mathrel{\stFmt{\to_{\!{#1}}}}%
  }
}%
\newcommand{\stEnvAnnotOutSym}{\stFmt{\oplus}}%
\newcommand{\stEnvAnnotInSym}{\stFmt{\&}}%
\newcommand{\stEnvAnnotGenericSym}[1][]{\stFmt{\ifempty{#1}{\alpha}{\alpha_{#1}}}}%
\newcommand{\stEnvMoveAnnot}[1]{\mathrel{\stFmt{\xrightarrow{#1}}}}
\newcommand{\stEnvMoveGenAnnot}{\stEnvMoveAnnot{\stEnvAnnotGenericSym}}%
\newcommand{\stEnvMoveInAnnot}[3]{%
  \stEnvMoveAnnot{\stEnvInAnnot{#1}{#2}{#3}}%
}%
\newcommand{\stEnvMoveOutAnnot}[3]{%
  \stEnvMoveAnnot{\stEnvOutAnnot{#1}{#2}{#3}}%
}%

\newcommand{\stEnvInAnnot}[3]{{#1}{\stEnvAnnotInSym}{#2}:{#3}}%
\newcommand{\stEnvOutAnnot}[3]{{#1}{\stEnvAnnotOutSym}{#2}:{#3}}%
\newcommand{\stEnvInAnnotSmall}[3]{{#1}{\stEnvAnnotInSym}{#2}\!:\!{#3}}%
\newcommand{\stEnvOutAnnotSmall}[3]{{#1}{\stEnvAnnotOutSym}{#2}\!:\!{#3}}%
\newcommand{\stEnvMoveP}[1]{{#1}\!\stEnvMove}%
\newcommand{\stEnvMoveMaybeCrashP}[2][]{{#2}\!\!\stEnvMoveMaybeCrash[#1]}%
\newcommand{\stEnvNotMoveP}[1]{{#1}\!\not\stEnvMove}%
\newcommand{\stEnvNotMoveMaybeCrashP}[2][]{{#2}\!\not\stEnvMoveMaybeCrash[#1]}%
\newcommand{\stEnvMoveStar}{\mathrel{\stFmt{\stEnvMove{}^{\!\!\!*}}}}%
\newcommand{\stEnvMoveMaybeCrashStar}[1][]{%
  \ifempty{#1}{%
    \mathrel{\stFmt{\to^*_{\!\lightning}}}%
  }{%
    \mathrel{\stFmt{\to^*_{\!{#1}}}}%
  }
}%
\newcommand{\stEnvMoveAnnotP}[2]{{#1}\!\stEnvMoveAnnot{#2}}%
\newcommand{\stEnvMoveGenAnnotP}[1]{{#1}\!\stEnvMoveGenAnnot}%

\newcommand{\stIdxRemoveCrash}[2]{{#1}^{{#2} \setminus \stCrashLab}}

\makeVars{qEnv}{stFmt}{\Delta}
\newcommand{\qApp}[3]{{#1}({#2}, {#3})}
\makeVars{stQ}{stFmt}{\tau}
\makeVars{stM}{stFmt}{M}

\makeVars{qEnvPartial}{stFmt}{\delta}

\newcommand{\stQEmpty}{\stFmt{\epsilon}}%
\newcommand{\stQMsg}[2]{\stFmt{{#1}\ifempty{#2}{}{({#2})}}}%
\newcommand{\stQCons}[2]{\stFmt{{#1}\mathbin{\!\cdot\!}{#2}}}%
\newcommand{\stQUnavail}{\stFmt{\oslash}}

\definecolor{tyColorCustom}{rgb}{0.0, 0.0, 0.85}

\definecolor{purpleish}{rgb}{0.41, 0.16, 0.38}

\makeVars{ldVarX}{ldVarFmt}{x}
\makeVars{ldVarY}{ldVarFmt}{y}
\makeVars{ldVarV}{ldVarFmt}{v}
\makeVars{ldTermM}{ldTermFmt}{M}
\makeVars{ldTermN}{ldTermFmt}{N}
\makeVars{ldValX}{ldTermFmt}{X}
\makeVars{ldValY}{ldTermFmt}{Y}
\makeVars{ldValV}{ldTermFmt}{V}
\makeVars{ldValW}{ldTermFmt}{W}

\makeVars{val}{}{v}
\makeVars{valr}{}{i}
\makeVars{valn}{}{n}

\newcommand{\procin}[3]{#1 ? #2.#3}%

\newcommand{\procout}[4]{#1 ! #2\langle #3 \rangle.#4}%

\newcommand{\procoutNoVal}[3]{#1 ! #2.#3}%
\newcommand{\eval}[2]{#1 \downarrow #2}

\newcommand{\redLabel}[1]{\mpMove}
\newcommand{\redSend}[3]{\mpMove}
\newcommand{\redRecv}[3]{\mpMove}
\newcommand{\redCrash}[2]{\mpMove[#2]}
\newcommand{\redIf}[1]{\mpMove}
\newcommand{\msg}[3]{(#1,#2(#3))}

\newcommand{\annSet}[0]{\Omega}
\newcommand{\annSetVar}[0]{\annSet_{\text{v}}}
\newcommand{\annSetSet}[0]{\annSet_{\text{s}}}
\newcommand{\annSubst}[0]{\sigma}
\newcommand{\annArg}[0]{\rho}

\newtoggle{full}
\settoggle{full}{true}
\iftoggle{full}
{
	
}{
	
}

\makeatletter
\newcommand{\iftoggleverb}[1]{%
  \ifcsdef{etb@tgl@#1}
    {\csname etb@tgl@#1\endcsname\iftrue\iffalse}
    {\etb@noglobal\etb@err@notoggle{#1}\iffalse}%
}
\makeatother

\begin{document}
\maketitle

\begin{abstract}
  Session types provide a typing discipline for message-passing systems.
  However,  most session type approaches assume an ideal world: one in which
  everything is reliable and without failures. Yet this is in stark contrast with
  distributed systems in the real world.
 To address this limitation, we introduce \theTool, a
  code generation toolchain that utilises asynchronous \emph{multiparty session
  types} (MPST) with \emph{crash-stop} semantics to support failure handling protocols.

 We augment asynchronous MPST and processes with \emph{crash handling} branches.
 Our approach requires no user-level syntax extensions for global types
 and features a formalisation of global semantics, which captures
 complex behaviours induced by crashed/crash handling processes.
 The sound and complete correspondence between global and local type semantics
 guarantees deadlock-freedom, protocol conformance, and liveness of typed processes
 in the presence of crashes.

Our theory is implemented in the toolchain \theTool,
which provides \emph{correctness by construction}.
\theTool extends the \Scribble multiparty protocol language
to generate protocol-conforming \Scala code, using the \Effpi concurrent programming
library.
We extend both \Scribble and \Effpi
to support \emph{crash-stop} behaviour.
We demonstrate the feasibility of our methodology and evaluate
\theTool with examples extended from both
session type and distributed systems literature.

\end{abstract}

\section{Introduction}
\label{sec:introduction}

\subparagraph*{Background}
As distributed programming grows increasingly prevalent,
significant research effort has been devoted to improve the reliability of
distributed systems.
A key aspect of this research focuses on studying
\emph{un}reliability (or, more specifically, failures).
Modelling unreliability and failures
enables a distributed system to be designed to be more tolerant of failures,
and thus more resilient.

In pursuit of methods to achieve safety in distributed communication systems,
\emph{session types}~\cite{ESOP98Session} provide a lightweight,
type system--based approach to message-passing concurrency.
In particular,
\emph{Multiparty Session Types} (MPST)~\cite{HYC08} facilitate the
specification and verification of communication between message-passing
processes in concurrent and distributed systems.
The typing discipline
prevents common communication-based errors,
\eg deadlocks and communication mismatches~\cite{POPL19LessIsMore,HYC16}.
On the practical side,
MPST have been implemented in various mainstream  programming
languages~\cite{
DBLP:journals/pacmpl/CastroHJNY19,
DBLP:conf/ppdp/KouzapasDPG16,FASE16EndpointAPI,
DBLP:conf/coordination/LagaillardieNY20,CYV2022,
MFYZ2021,
NY2017Actor,DBLP:journals/fmsd/DemangeonHHNY15
},
which facilitates their applications
in real-world programs.

Nevertheless, the challenge to account for unreliability and failures persists
for session types:
most session type systems assume that both participants and message transmissions
are \emph{reliable} without failures.
In a real-world setting, however, participants may crash, communications
channels may fail, and messages may be lost.
The lack of failure modelling in session type theories creates a barrier to
their applications to large-scale distributed systems.

Recent works~\cite{OOPSLA21FaultTolerantMPST, FORTE22FaultTolerant,
ECOOP22AffineMPST, ESOP23MAGPi, CONCUR22MPSTCrash}
close the gap of failure modelling in session
types with various techniques.
\cite{OOPSLA21FaultTolerantMPST} introduces \emph{failure
suspicion},
where a participant may suspect their communication partner has failed,
and act accordingly.
\cite{FORTE22FaultTolerant} introduces \emph{reliability
annotations} at type level, and fall back to a given \emph{default} value in
case of failures.
\cite{ECOOP22AffineMPST} proposes a framework of \emph{affine}
multiparty session types, where a session can terminate prematurely, \eg
in case of failures.
\cite{CONCUR22MPSTCrash} integrates \emph{crash-stop failures},
where a generalised type system validates safety and liveness
properties with model checking;
\cite{ESOP23MAGPi} takes a similar approach, modelling
more kinds of failures in a session type system, \eg message losses, reordering, and
delays.

While steady advancements are made on the theoretical side, the
implementations of those enhanced session type theories seem to lag behind.
Barring the approaches in~\cite{OOPSLA21FaultTolerantMPST,ECOOP22AffineMPST}, %
the aforementioned approaches~\cite{FORTE22FaultTolerant, ESOP23MAGPi,
CONCUR22MPSTCrash} do not
provide session type API support for programming languages.\footnotemark\
To bring the benefits of the theoretical developments into real-world distributed programming,
a gap remains to be filled on the implementation side.

\footnotetext{\cite{CONCUR22MPSTCrash} provides a prototype implementation,
  utilising the mCRL2 model checker~\cite{TACAS19mCRL2}, for verifying
  type-level properties,
instead of a library for general use.}

\subparagraph*{This Paper}
We introduce a \emph{top-down}  %
methodology for designing asynchronous multiparty protocols
with crash-stop failures:
\begin{enumerate*}[label=\emph{(\arabic*)}]
  \item
  We use an extended asynchronous MPST theory, which models \emph{crash-stop}
  failures, and show that the usual session type
  guarantees remain valid, \ie communication safety, deadlock-freedom, and liveness;
  \item
  We present a toolchain for implementing asynchronous multiparty protocols,
  under our new asynchronous MPST theory, in \Scala, using the
  \Effpi concurrency library~\cite{PLDI19Effpi}.
\end{enumerate*}

The top-down design methodology comes from the original MPST
theory~\cite{HYC08},
where the design of multiparty protocols begins with a given \emph{global} type
(top),
and implementations rely on \emph{local} types (bottom) obtained from the
global type.
The global and local types reflect the global and local communication
behaviours respectively.
Well-typed implementations that conform to a global type are
guaranteed to be \emph{correct by construction},
enjoying full guarantees (safety, deadlock-freedom, liveness) from the theory.
This remains the predominant approach for implementing MPST theories,
and is also followed by some aforementioned
systems~\cite{OOPSLA21FaultTolerantMPST,ECOOP22AffineMPST}.

We model \emph{crash-stop} failures~\cite[\S 2.2]{DBLP:books/daglib/0025983},
\ie a process may fail arbitrarily and cease to interact with others.
This model is simple and expressive,
and has been adopted by other approaches~\cite{CONCUR22MPSTCrash, ESOP23MAGPi}.
Using global types in our design for handling failures in multiparty protocols
presents two distinct advantages:
\begin{enumerate*}[label=\emph{(\arabic*)}]
  \item global types provide a simple, high-level means to both specify a protocol
    abstractly and automatically derive local types; and,
  \item desirable behavioural properties such as communication safety,
    deadlock-freedom, and liveness are guaranteed by construction.
  \end{enumerate*}
In contrast to the \emph{synchronous} semantics in~\cite{CONCUR22MPSTCrash},
we model an \emph{asynchronous} semantics,
where messages are buffered whilst in transit.
We focus on asynchronous systems
since most communication in the real distributed world is asynchronous.
In \cite{ESOP23MAGPi}, %
although the authors develop a generic typing system incorporating asynchronous semantics,
their approach results in the type-level properties
becoming undecidable~\cite[\S 4.4]{ESOP23MAGPi}.
With global types,
we restore the decidability %
at a minor cost to
expressivity.

To address the gap on the practical side, %
we present a code generator toolchain, \theTool,
to implement our MPST theory.
Our toolchain takes an asynchronous multiparty protocol as input,
using the protocol description language \Scribble~\cite{YHNN2013},
and generates \Scala code using the \Effpi~\cite{PLDI19Effpi} concurrency
library as output.

The \Scribble Language~\cite{YHNN2013} is %
designed for describing multiparty communication protocols,
and is closely connected to MPST theory (\cf \cite{FeatherweightScribble}).
This language enables a programmatic approach for expressing global types and
designing multiparty protocols.
The \Effpi concurrency library \cite{PLDI19Effpi} offers an embedded Domain
Specific Language (DSL) that provides a simple actor-based API\@.
The library
offers both type-level and value-level
constructs for processes and channels.
Notably,
the type-level constructs reflect the behaviour of programs (\ie processes)
and can be used as specifications.
Our code generation technique, as well as the \Effpi library itself,
utilises the type system features introduced in \Scala~3,
including match types and dependent function types,
to encode local types in \Effpi.
This approach enables us to specify and verify program behaviour at the type level,
resulting in a more powerful and flexible method for handling concurrency.

By extending
\Scribble and \Effpi to support \emph{crash detection and handling},
our toolchain \theTool provides a lightweight way for developers to take
advantage of our theory, bridging the gap on the practical side.
We evaluate the expressivity and feasibility of \theTool with examples incorporating
crash handling behaviour, extended from session type literature.

\subparagraph*{Outline}
We begin with an overview of our methodology in
\cref{sec:overview}.
We introduce an asynchronous multiparty session
calculus %
in \cref{sec:process} with semantics of crashing and crash handling.
    We introduce an extended theory of asynchronous multiparty session
    types with semantic modelling of crash-stop failures in \cref{sec:gtype}.
    We present a typing system for the multiparty session calculus in
    \cref{sec:typing_system}.
    We introduce \theTool, a code generation toolchain that implements our
    theory in \cref{sec:impl},
    demonstrating how our approach is applied in the  \Scala
    programming language.
   We evaluate our toolchain with examples from both session type and
   distributed systems literature in \cref{sec:eval}.
We discuss related work in \cref{sec:related} and conclude in
\cref{sec:conclusion}.
Full proofs, further definitions, examples, and more technical details of 
\theTool can be found in Appendix.  
Additionally, our toolchain and examples used in our evaluation 
are available in an artifact, 
which can be 
\href{https://doi.org/10.5281/zenodo.7714133}{\color{blue}{accessed}} on Zenodo. 
For those interested in the source files, they 
are available on \href{https://github.com/adbarwell/ECOOP23-Artefact}{\color{blue}{GitHub}}.

\section{Overview}
\label{sec:overview}
In this section, we give an overview of our methodology for
designing asynchronous multiparty protocols with crash-stop failures, and demonstrate our code generation toolchain, \theTool.

\subparagraph*{Asynchronous Multiparty Protocols with Crash-Stop Failures}
We follow a standard top-down design approach enabling \emph{correctness by construction}, 
but enrich asynchronous MPST with crash-stop semantics.  As depicted in~\Cref{fig:overview-of-topdown}, we formalise (asynchronous) multiparty protocols with crash-stop failures
as  global types with \emph{crash handling branches} ($\gtCrashLab$). 
These are projected into local types, which may similarly contain crash handling branches ($\stCrashLab$).
The projected local types are then used to type-check
processes (also with crash handling branches ($\mpCrashLab$)) that are written in a session calculus.
As an example, we consider a simple
\emph{distributed logging} scenario, which is inspired by
the logging-management protocol~\cite{ECOOP22AffineMPST}, 
but extended with a third 
participant. \iftoggle{full}{The full distributed logging protocol can be found in~\Cref{sec:appendix:eval}.}{}

\begin{figure}[t]
\centering
{\footnotesize 
  \begin{tikzpicture}
  \node (Gtext) {A Global Type $\gtG$ $\highlight{\text{with }\gtCrashLab}$};
  \node[below= 0.5mm of Gtext, xshift=-37mm, align=center] (proj) {
    \small {\bf{projection}}\,($\upharpoonright$)};
  \node[below=6mm of Gtext, xshift=-40mm] (LA) {\footnotesize Local Type for $\roleFmt{L}$
  \small \boxed{\stT_{\roleFmt{L}}}};
  \node[below=6mm of Gtext] (LB) {\footnotesize Local Type for $\roleFmt{I}$ $\highlight{\text{with }\stCrashLab}$ \small \boxed{\stT_{\roleFmt{I}}}};
  \node[below=6mm of Gtext, xshift=40mm] (LC) {\footnotesize Local Type for $\roleFmt{C}$
   \small  \boxed{\stT_{\roleFmt{C}}}};
  \draw[->] (Gtext) -- (LA);
  \draw[->] (Gtext) -- (LB);
  \draw[->] (Gtext) -- (LC);
   \node[below= 1mm of LA, xshift=-8.4mm, align=center] (typ) {
    \small {\bf{typing}}\,($\vdash$)};
   \node[below=6mm of LA, xshift=0mm] (PA) {\, \, \, \footnotesize Process for $\roleFmt{L}$ \small
    \boxed{P_{\roleFmt{L}}}};
   \node[below=6mm of LB, xshift=0mm] (PB) {\, \, \, \footnotesize Process for $\roleFmt{I}$ $\highlight{\text{with }\mpCrashLab}$ \small
    \boxed{P_{\roleFmt{I}}}};
     \node[below=6mm of LC, xshift=0mm] (PC) {\, \, \, \footnotesize Process for $\roleFmt{C}$ \small
    \boxed{P_{\roleFmt{C}}}};
   \draw[->] (LA) -- (PA);
   \draw[->] (LB) -- (PB);
   \draw[->] (LC) -- (PC);
\end{tikzpicture}
}
\caption{Top-down View of MPST with Crash}
\label{fig:overview-of-topdown}
\end{figure}

The Simpler Logging protocol consists of a \emph{logger} (\roleFmt{L}) that controls the
logging services, an \emph{interface} ($\roleFmt{I}$) that provides communications between
logger and client, and a \emph{client} ($\roleFmt{C}$) that requires logging services via interface.
Initially, $\roleFmt{L}$ sends a heartbeat message $\gtMsgFmt{trigger}$ to $\roleFmt{I}$. 
Then $\roleFmt{C}$ sends a command to
$\roleFmt{L}$ 
to read the logs (\gtMsgFmt{read}).
When a \gtMsgFmt{read} request is
sent, it is forwarded to $\roleFmt{L}$, and $\roleFmt{L}$ responds
with a \gtMsgFmt{report}, which is then forwarded onto $\roleFmt{C}$.  
Assuming all participants (logger, interface, and client) are reliable, \ie without any failures or crashes, 
this logging behaviour can be represented by the \emph{global type}
$\gtG[0]$:
\begin{equation}
\label{ex:overview-global-without-crash}
{\small{
\gtG[0] =
    \gtCommSingle{\roleFmt{L}}{\roleFmt{I}}{
  \gtMsgFmt{trigger}}{}{
    \gtCommSingle{\roleFmt{C}}{\roleFmt{I}}{
   \gtMsgFmt{read}}{}{
    \gtCommSingle{\roleFmt{I}}{\roleFmt{L}}
    {\gtMsgFmt{read}}{}{
    \gtCommSingle{\roleFmt{L}}{\roleFmt{I}}
    {\gtMsgFmt{report}}{\stFmtC{log}}{
    \gtCommSingle{\roleFmt{I}}{\roleFmt{C}}{\gtMsgFmt{report}}{\stFmtC{log}}{
    \gtEnd
    }}}}
       }
}}
\end{equation}
\noindent
Here, $\gtG[0]$ is a specification of the Simpler Logging protocol between multiple roles from
a global perspective.

In the real distributed world, all participants in the Simpler Logging system
may fail.
Ergo, we need to model protocols with failures or
crashes and handling behaviour,
\eg should the client fail after the logging has started, the interface
will inform the logger to stop and exit.
We follow~\cite[\S 2.2]{DBLP:books/daglib/0025983} to model
a \emph{crash-stop} semantics, where we assume that roles can crash \emph{at any time} unless assumed
\emph{reliable} (never crash).
For simplicity, we assume $\roleFmt{I}$ and $\roleFmt{L}$ to be reliable.
The above logging behaviour, incorporating crash-stop failures, can be represented by extending
$\gtG[0]$ with a branch handling a crash of $\roleFmt{C}$:
\begin{equation}
\label{ex:overview-global-with-crash}
{\small{
\gtG =
    \gtCommSingle{\roleFmt{L}}{\roleFmt{I}}{
  \gtMsgFmt{trigger}}{}{
    \gtCommRaw{\roleFmt{C}}{\roleFmt{I}}{
    \begin{array}{@{}l@{}}
    \gtCommChoice{\gtMsgFmt{read}}{}{
    \gtCommSingle{\roleFmt{I}}{\roleFmt{L}}
    {\gtMsgFmt{read}}{}{
    \gtCommSingle{\roleFmt{L}}{\roleFmt{I}}
    {\gtMsgFmt{report}}{\stFmtC{log}}{
    \gtCommSingle{\roleFmt{I}}{\roleFmt{C}}{\gtMsgFmt{report}}{\stFmtC{log}}{
    \gtEnd
    }}}}\\
     \gtCommChoice{\gtMsgFmt{crash}}{}{
    \gtCommSingle{\roleFmt{I}}{\roleFmt{L}}{\gtMsgFmt{fatal}}{}{
    \gtEnd
    }}
    \end{array}
   }}
   }
   }
\end{equation}
We model crash detection on receiving roles: when $\roleFmt{I}$ is
waiting to receive a message from $\roleFmt{C}$, the receiving role $\roleFmt{I}$ is
able to detect whether $\roleFmt{C}$ has $\gtCrashLab$ed.
Since crashes are detected only by the receiving role, we do not
require a crash handling branch on the communication step between $\roleFmt{I}$ and
$\roleFmt{C}$ -- nor do we require them on any interaction between $\roleFmt{L}$ and $\roleFmt{I}$ (since we
are assuming that $\roleFmt{L}$ and $\roleFmt{I}$ are reliable).

Following the MPST top-down methodology,
a global type is then \emph{projected} onto \emph{local types},
which describe communications from the perspective of a single role.
In our unreliable Simpler Logging example,
$\gtG$ is projected onto three local types (one for each role
$\roleFmt{C}$, $\roleFmt{L}$, $\roleFmt{I}$):

\smallskip
\centerline{\(
{\small{
 \begin{array}{c}
  \stT[\roleFmt{C}] =
 \roleFmt{I} \stFmt{\oplus}
   \stLabFmt{read}
 \stSeq
  \roleFmt{I}
   \stFmt{\&}
    \stLabFmt{report(\stFmtC{log})}
 \stSeq
\stEnd
  \quad
   \stT[\roleFmt{L}] =
   \roleFmt{I} \stFmt{\oplus}
   \stLabFmt{trigger}
 \stSeq
  \stExtSum{\roleFmt{I}}{}{
  \begin{array}{@{}l@{}}
\stLabFmt{read}
  \stSeq
   \roleFmt{I} \stFmt{\oplus}
    \stLabFmt{report(\stFmtC{log})}
    \stSeq
    \stEnd
    \\
    \stLabFmt{fatal}
  \stSeq
   \stEnd
  \end{array}
  }
\\[3mm]
  \stT[\roleFmt{I}]  =
    \roleFmt{L}
  \stFmt{\&}
  \stLabFmt{trigger}
  \stSeq
  \stExtSum{\roleFmt{C}}{}{
  \begin{array}{@{}l@{}}
 \stLabFmt{read}
 \stSeq
 \roleFmt{L}
  \stFmt{\oplus}
  \stLabFmt{read}
  \stSeq
 \roleFmt{L}
\stFmt{\&}
\stLabFmt{report(\stFmtC{log})}
\stSeq
\roleFmt{C}
\stFmt{\oplus}
\stLabFmt{report(\stFmtC{log})}
\stSeq
\stEnd
 \\
  \stCrashLab
  \stSeq
  \roleFmt{L}
  \stFmt{\oplus}
  \stLabFmt{fatal}
 \stSeq
\stEnd
 \end{array}
 }
\end{array}
 }
 }
\)}
\smallskip

\noindent
Here,
$\stT[\roleFmt{I}]$ states that 
 $\roleFmt{I}$ first receives a trigger message from  
 $\roleFmt{L}$; then $\roleFmt{I}$ either expects a  $\stLabFmt{read}$ request from 
 $\roleFmt{C}$, or detects the crash of $\roleFmt{C}$ and handles it (in $\stCrashLab$)
 by sending the $\stLabFmt{fatal}$ message to notify $\roleFmt{L}$.
We add
additional crash modelling and introduce a $\stStop$ type for crashed
endpoints. %
We show the operational correspondence between global and local type semantics,
and demonstrate that a
projectable global type always produces a safe, deadlock-free, and live typing
context. %

The next step in this top-down methodology is to use
local types to type-check processes $\mpP[i]$ executed by role $\roleP[i]$
in our session calculus.
For example, $\stT[\roleFmt{I}]$ can be used to type check  %
$\roleFmt{I}$ that executes the process:

\smallskip
\centerline{\(
{\small{
\procin{\roleFmt{L}}{\labFmt{trigger}}{\sum \setenum{
  \begin{array}{@{}l@{}}
 \procin{\roleFmt{C}}{\labFmt{read}}{
 \procoutNoVal{\roleFmt{L}}{\labFmt{read}}
 {\procin{\roleFmt{L}}{\labFmt{report}(\mpx)}
 {\procout{\roleFmt{C}}{\labFmt{report}}{\mpx}{\mpNil}}}
 }
\\
  \procin{\roleFmt{C}}{\mpCrashLab}{ \procoutNoVal{\roleFmt{L}}{\labFmt{fatal}}{\mpNil}}
  \end{array}
  }}
}}
\)} %
\smallskip

\noindent
 In our operational semantics (\cref{sec:process}),
we allow active processes executed by unreliable roles
to crash \emph{arbitrarily}.
Therefore, the role executing the crashed process
also crashes, and is assigned the local type $\stStop$.
To ensure that a communicating process is type-safe even in presence of crashes,
we require that its typing context %
satisfies a \emph{safety property} accounting for possible crashes (\Cref{def:mpst-env-safe}),
which is preserved by
projection. %
Additional semantics surrounding crashes adds subtleties even in standard
results. We prove subject reduction and session fidelity 
results accounting for crashes and
sets of reliable roles. %

\subparagraph*{Code Generation Toolchain: \theTool}

\begin{figure}
  \begin{center}
    \includegraphics[width=0.8\textwidth]{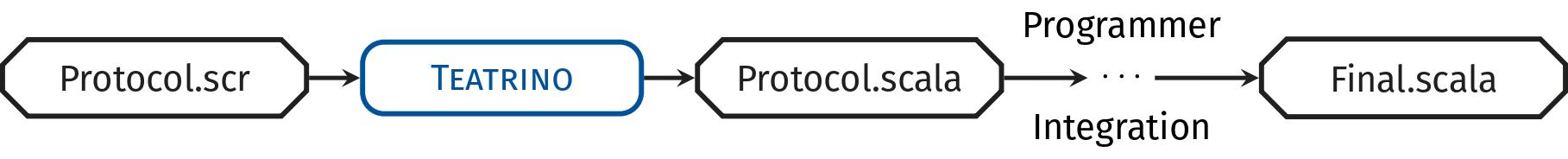}
  \end{center}
  \caption{Workflow of \theTool}
   \label{fig:overview:diagram}
\end{figure}

To complement the theory,
we present a code generation toolchain, \theTool, that generates
protocol-conforming \Scala code from a multiparty protocol.
We show the workflow diagram of our toolchain in~\Cref{fig:overview:diagram}.
\theTool takes a \Scribble protocol (\textsf{Protocol.scr})
and generates executable code (\textsf{Protocol.scala})
conforming to that protocol,
which the programmer can integrate with existing code (\textsf{Final.scala}).

\theTool implements our session type theory to handle global types expressed
using the \Scribble protocol description language~\cite{YHNN2013},
a programmer-friendly way for describing multiparty protocols.
We extend the syntax of \Scribble slightly to include constructs for \gtCrashLab
recovery branches and reliable roles.

The generated \Scala code utilises the
\Effpi concurrency
library~\cite{PLDI19Effpi}. %
\Effpi is an embedded domain specific language in \Scala~3 that offers a simple
Actor-based API for concurrency.
Our code generation technique, as well the \Effpi library itself,
leverages the type system features introduced in \Scala~3,
\eg match types and dependent function types,
to encode local types in \Effpi.
We extend \Effpi to support crash detection and handling.

As a brief introduction to \Effpi,
the concurrency library provides types for processes and channels.
For processes, an output process type \lstinline[language=effpi]+Out[A, B]+
describes a process that uses a channel of type \texttt{A} to send a value of
type \texttt{B},
and an input process type \lstinline[language=effpi]+In[A, B, C]+ describes a
process that uses a channel of type \texttt{A} to receive a value of type
$\texttt{B}$, and pass it to a continuation type \texttt{C}.
Process types can be sequentially composed by the
\lstinline[language=effpi]+>>:+ operator.
For channels,
\lstinline[language=effpi]+Chan[X]+ describes a channel that can be used
communicate values of type \texttt{X}.
More specifically, the usage of a channel can be reflected at the type level,
using the types
\lstinline[language=effpi]+InChan[X]+\slash\lstinline[language=effpi]+OutChan[X]+
for input\slash output channels.

\begin{figure}[ht]
\footnotesize
  \begin{lstlisting}[language=effpi,gobble=4]
    type I[C0 <: InChan[Trigger], C1 <: OutChan[Fatal],
           C2 <: InChan[Read], C3 <: InChan[Report], C4 <: OutChan[Report]]
    = InErr[C0, Trigger, [*\label{line:inerr}*]
            (X <: Read) =>
              Out[C3,Read] >>: In[C4, Report, (Y <: Log) => Out[C5, Report]],
            (Err <: Throwable) => Out[C2,Fatal]
      ]
  \end{lstlisting}
  \caption{\Effpi Type for $\stT[\roleFmt{I}]$}
  \label{fig:over:typeEg}
\end{figure}

As a sneak peek of the code we generate, in \cref{fig:over:typeEg},
we show the generated \Effpi
representation for the projected local type \stT[\roleFmt{I}]
from the Simpler Logging example. %
Readers may be surprised by the difference between %
$\stT[\roleFmt{I}]$ and the generated \Effpi type $\texttt{I}$.
This is because the process types need their respective channel types,
namely the type variables \texttt{C0}, \texttt{C1}, \etc bounded by
\lstinline[language=Effpi]+InChan[...]+ and
\lstinline[language=Effpi]+OutChan[...]+.
We explain the details of code generation in \cref{sec:impl:overview}, and
describe an interesting challenge posed by the
channel generation procedure in~\cref{sec:impl:boringdetail}.

For crash handling behaviour, we introduce a new type
\lstinline[language=effpi]+InErr+,
whose last argument specifies a continuation type to follow in case of a crash.
\Cref{line:inerr} in \Cref{fig:over:typeEg} shows the crash
handling behaviour: sending a message of type \texttt{Fatal},
which reflects the \stCrashLab branch in the local type $\stT[\roleFmt{I}]$.
We give more details of the generated code in \cref{sec:impl:overview}.

Code generated by \theTool is executable, protocol-conforming, and can be
specialised by the programmer to integrate with existing code.
We evaluate our toolchain on examples taken from both MPST and
distributed programming literature in \cref{sec:eval}.
Moreover, we extend each example with crash handling
behaviour to define %
unreliable variants.
We demonstrate that, with \theTool, code generation takes negligible time,
and all potential crashes are accompanied with crash handlers.

\section{Crash-Stop Asynchronous Multiparty Session Calculus} %
\label{sec:process}
In this section, we formalise the syntax %
and operational semantics %
of our asynchronous multiparty session calculus with process failures and crash detection.

\subparagraph*{Syntax}
Our asynchronous multiparty session calculus models processes that may crash arbitrarily.
Our formalisation is based on~\cite{POPL21AsyncMPSTSubtyping} -- but
in addition, follows the \emph{fail-stop} model in \cite[\S 2.7]{DBLP:books/daglib/0025983},
where processes may crash and never recover, and process failures can be detected by failure detectors~\cite[\S 2.6.2]{DBLP:books/daglib/0025983} \cite{JACM96FailureDetector} when attempting to receive messages.

We give the syntax of processes in \cref{fig:processes:syntax}.  %
In our calculus, we assume that there are basic expressions $\mpE$ (\eg
$\mpTrue, \mpFalse, 7 + 11$) that are
assigned basic types $\tyGround$ (\eg $\tyInt, \tyBool$).
We write $\eval{\mpE}{\mpV}$ to denote an expression $\mpE$ evaluates to
a value $\mpV$ (\eg $\eval{(7 < 11)}{\mpTrue}, \eval{(1 + 1)}{2}$).

\begin{figure}
{\small
\[
  \begin{array}{@{}r@{\;}c@{\;}ll@{}}
    \mpP,\mpQ & \bnfdef & & \text{\em\bf Processes} \\
            & & \sum_{i\in I}\procin{\roleP}{\mpLab_i(\mpx_i)}{\mpP_i} & \text{\em external choice}\\
            & \bnfsep & \procout{\roleP}{\mpLab}{\mpE}{\mpP} \quad
            \highlight{\text{(\em where } \mpLab \neq \mpCrashLab\text{)}}
            & \text{\em output} \\
            & \bnfsep & \mpIf{\mpE} \mpP \mpQ & \text{\em conditional} \\
            & \bnfsep & \mpX & \text{\em variable} \\
            & \bnfsep & \mu \mpX.\mpP & \text{\em recursion} \\
            & \bnfsep &  \mpNil  & \text{\em inaction} \\
            & \bnfsep & \highlight{\mpCrash} & \highlight{\text{\em crashed}}
  \end{array}%
  \quad%
  \begin{array}{@{}r@{\;}c@{\;}ll@{}}
    \mpM & \bnfdef &  & \text{\em\bf Sessions} \\
       &         & \mpPart\roleP\mpP \mpPar \mpPart\roleP \mpH & \text{\em
       participant} \\
       & \bnfsep & \mpM\mpPar\mpM  & \text{\em parallel} \\
    \mpH & \bnfdef & & \text{\em\bf %
    Queues} \\
       &         & \mpQEmpty & \text{\em empty %
       } \\
       & \bnfsep & \highlight{\mpQUnavail} & \highlight{\text{\em unavailable %
       }} \\
       & \bnfsep & \left(\roleP,\mpLab(\mpV)\right) & \text{\em message} \\
       & \bnfsep & \mpH\cdot \mpH  &\text{\em concatenation}
  \end{array}
  \]
  }
\caption{Syntax of sessions, processes, and queues.
Noticeable changes w.r.t. standard session calculus~\cite{POPL21AsyncMPSTSubtyping} are $\highlight{\text{highlighted}}$.}
\label{fig:processes:syntax}
\end{figure}

A process, ranged over by $\mpP, \mpQ$, is a communication agent within a
session.
An \emph{output} process
$\procout{\roleP}{\mpLab}{\mpE}{\mpP}$
sends a message to another role $\roleP$ in the session,
where the message is labelled $\mpLab$, and carries a payload expresion $\mpE$,
then the process continues as $\mpP$.
An \emph{external choice} (\emph{input}) process
$\sum_{i\in I}\procin{\roleP}{\mpLab_i(\mpx_i)}{\mpP_i}$
receives a message from another role $\roleP$ in the session,
among a finite set of indexes $I$,
if the message is labelled $\mpLab[i]$, then the payload would be received as
$\mpx[i]$, and process continues as $\mpP[i]$.
Note that our calculus uses $\mpCrashLab$ as a special message label
denoting that a participant (role) has crashed.
Such a label cannot be sent by any process, but a process can implement
crash detection and handling by receiving it.
Consequently, an output process \emph{cannot} send
 a $\mpCrashLab$ message (side condition $\mpLab \neq \mpCrashLab$),
 whereas an input process
may include a \emph{crash handling branch} of the form $\mpCrashLab.\mpPi$
meaning that $\mpPi$ is executed when the sending role has crashed.
A \emph{conditional} process
$\mpIf{\mpE} \mpP \mpQ$
continues as either $\mpP$ or $\mpQ$ depending on the evaluation of $\mpE$.
We allow \emph{recursion} at the process level using
$\mu \mpX.\mpP$ and $\mpX$,
and we require process recursion variables to be guarded by an input or an
output action;
we consider a recursion process structurally congruent to its unfolding
$\mu X.\mpP \equiv \mpP\subst{X}{\mu X.\mpP}$.
Finally, we write $\mpNil$ for an \emph{inactive} process, representing
a successful termination; and $\mpCrash$ for a \emph{crashed} process, representing a
termination due to failure.

An \emph{incoming queue}\footnotemark, ranged over by $\mpH, \mpHi$, is a sequence of messages
tagged with their origin.
We write $\mpQEmpty$ for an \emph{empty} queue;
$\mpQUnavail$ for an \emph{unavailable} queue;
and $(\roleP,\mpLab(\val))$
for a message sent from $\roleP$, labelled $\mpLab$, and containing a payload
value $\mpV$.
We write $\mpH[1] \cdot \mpH[2]$ to denote the concatenation of two queues
$\mpH[1]$ and $\mpH[2]$.
When describing incoming queues, we consider two messages from different
origins as swappable:
\(
  \mpH_1
  \cdot
  \msg{\roleQ_1}{\mpLab_1}{\val_1}
  \cdot
  \msg{\roleQ_2}{\mpLab_2}{\val_2}
  \cdot
  \mpH_2
  \equiv
  \mpH_1
  \cdot
  \msg{\roleQ_2}{\mpLab_2}{\val_2}
  \cdot
  \msg{\roleQ_1}{\mpLab_1}{\val_1}
  \cdot
  \mpH_2
\)
whenever $\roleQ[1] \neq \roleQ[2]$.
Moreover, we consider concatenation $(\cdot)$ as associative, and the empty
queue $\mpQEmpty$ as the identity element for concatenation.

\footnotetext{In~\cite{POPL21AsyncMPSTSubtyping}, the queues are outgoing
instead of incoming.
We use incoming queues to model our crashing semantics more easily.}

A session, ranged over by $\mpM, \mpMi$, consists of processes and their
respective incoming queue, indexed by their roles.
A single entry for a role $\roleP$ is denoted $\mpPart{\roleP}{\mpP} \mpPar
\mpPart{\roleP}{\mpH}$, where $\mpP$ is the process for $\roleP$ and $\mpH$ is
the incoming queue.
Entries are composed together in parallel as $\mpM \mpPar \mpMi$, where the
roles in $\mpM$ and $\mpMi$ are disjoint.
We consider parallel composition as commutative and associative, with
$\mpPart{\roleP}{\mpNil} \mpPar \mpPart{\roleP}{\mpQEmpty}$ as a neutral
element of the operator.
We write
$\prod_{i\in I} (\mpPart{\roleP[i]}{\mpP[i]} \mpPar \mpPart{\roleP[i]}{\mpH[i]})$
for the parallel composition of multiple entries in a set.

\subparagraph*{Operational Semantics}
of our session calculus is given in \Cref{def:session:red}, using
a standard \emph{structural congruence} $\equiv$ defined in~\cite{POPL21AsyncMPSTSubtyping}. %
\iftoggle{full}{Standard congruence rules can be found in~\cref{sec:app_congruence}. }{}%
Our semantics parameterises on a (possibly empty) set of \emph{reliable} roles $\rolesR$,
which are assumed to \emph{never crash}. %

\begin{definition}[Session Reductions]
\label{def:session:red}
Session reduction $\mpMove[\rolesR]$ is inductively
defined by the rules in \cref{fig:processes:reduction},
parameterised by a fixed set $\rolesR$ of reliable roles.
We write $\mpMove$ when $\rolesR$ is insignificant.
We write $\mpMoveStar[\rolesR]$ (resp.  $\mpMoveStar$)
for the reflexive and transitive closure of $\mpMove[\rolesR]$ (resp. $\mpMove$).
\end{definition}

Our operational semantics retains the basic rules
in~\cite{POPL21AsyncMPSTSubtyping}, but also includes
($\highlight{\text{highlighted}}$) rules for crash-stop failures and crash
handling, adapted from~\cite{CONCUR22MPSTCrash}.
Rules \inferrule{r-send} and \inferrule{r-rcv} model ordinary message delivery
and reception:
an output process located at $\roleP$ sending to $\roleQ$ appends a
message to the incoming queue of $\roleQ$; and an input process located at
$\roleP$ receiving from $\roleQ$ consumes the first message from the
incoming queue.
Rules \inferrule{r-cond-T} and \inferrule{r-cond-F} model conditionals;
and rule \inferrule{r-struct} permits reductions up to structural
congruence.

With regard to crashes and related behaviour, rule \inferrule{r-$\lightning$}
models process crashes: an active ($\mpP \neq \mpNil$) process located at an
unreliable role ($\roleP \notin \rolesR$) may reduce to a crashed process
$\mpPart{\roleP}{\mpCrash}$, with its incoming queue becoming unavailable
$\mpPart{\roleP}{\mpQUnavail}$.
Rule \inferrule{r-send-$\lightning$} models a
message delivery to a crashed role (and thus an unavailable queue),
and the message becomes lost and would not be added to the queue.
Rule \inferrule{r-rcv-$\odot$} models crash detection, which activates as a
``last resort'':
an input process at $\roleP$ receiving from $\roleQ$ would first attempt find a
message from $\roleQ$ in the incoming queue, which engages the usual rule
\inferrule{r-recv};
if none exists and $\roleQ$ has crashed
($\mpPart{\roleQ}{\mpCrash}$), then the crash handling branch in the input
process at $\roleP$ can activate.
We draw attention to the interesting fact that \inferrule{r-recv} may engage
even if $\roleQ$ has crashed, in cases where a message from $\roleQ$ in the
incoming queue may be consumed.
We now illustrate our operational semantics of sessions with an example.  %
\iftoggle{full}{More examples can be found in~\cref{sec:app_examples}.}

\begin{figure}
{\small
\[
\begin{array}{@{}llr@{}}
\highlight{\inferrule{r-$\lightning$}} &
\highlight{{
\mpPart\roleP{\mpP}
\mpPar
\mpPart\roleP{\mpH[\roleP]}
\mpPar
\mpM
\;\redCrash{\roleP}{\rolesR}\;
\mpPart\roleP{\mpCrash}
\mpPar
\mpPart\roleP{\mpQUnavail}
\mpPar
\mpM
}}
&
\hspace{-11em}
\highlight{(\mpP \neq \mpNil, \roleP \notin \rolesR)}
\\
\inferrule{r-send} &
{
\mpPart\roleP{\procout{\roleQ}{\mpLab}{\mpE}{\mpP}}
\mpPar
\mpPart\roleP{\mpH[\roleP]}
\mpPar
\mpPart\roleQ{\mpQ}
\mpPar
\mpPart\roleQ{\mpH[\roleQ]}
\mpPar
\mpM
}
\\
&
{%
\redSend{\roleP}{\roleQ}{\mpLab}\;
\mpPart\roleP{\mpP}
\mpPar
\mpPart\roleP{\mpH[\roleP]}
\mpPar
\mpPart\roleQ{\mpQ}
\mpPar
\mpPart{\roleQ}{\mpH[\roleQ]}\cdot(\roleP,\mpLab(\val))
\mpPar
\mpM
}
&
\hspace{-11em}
(\eval{\mpE} \val, \mpH[\roleQ] \neq \mpQUnavail)
\\[1mm]
\highlight{\inferrule{r-send-\mpCrash}} &
{\highlight{
\mpPart\roleP{\procout{\roleQ}{\mpLab}{\mpE}{\mpP}}
\mpPar
\mpPart\roleP{\mpH[\roleP]}
\mpPar
\mpPart\roleQ{\mpCrash}
\mpPar
\mpPart\roleQ{\mpQUnavail}
\mpPar
\mpM
\;\redSend{\roleP}{\roleQ}{\mpLab}\;
\mpPart\roleP{\mpP}
\mpPar
\mpPart\roleP{\mpH[\roleP]}
\mpPar
\mpPart\roleQ{\mpCrash}
\mpPar
\mpPart{\roleQ}{\mpQUnavail}
\mpPar
\mpM}
\hspace{-16em}
}
&
\\[1mm]
\inferrule{r-rcv}
&
\mpPart\roleP{\sum_{i\in I} \procin\roleQ{\mpLab_i(\mpx_i)}\mpP_i}
\mpPar
\mpPart\roleP{(\roleQ,\mpLab_k(\val)) \cdot \mpH[\roleP]}
\mpPar
\mpM
\redRecv{\roleP}{\roleQ}{\mpLab_k}\;
\mpPart\roleP \mpP_k\subst{\mpx_k}{\val}
\mpPar
\mpPart\roleP{\mpH_{\roleP}}
\mpPar
\mpM
\hspace{-21em}
&
(k \in I)
\\[1mm]
\highlight{\inferrule{r-rcv-$\odot$}}
&
\highlight{\mpPart\roleP{\sum_{i\in I} \procin\roleQ{\mpLab_i(\mpx_i)}\mpP_i}
\mpPar
\mpPart\roleP{\mpH[\roleP]}
\mpPar
\mpPart\roleQ\mpCrash
\mpPar
\mpPart\roleQ\mpQUnavail
\mpPar
\mpM}
\\
&
\highlight{\redRecv{\roleP}{\roleQ}{\mpLab_k}\;
\mpPart\roleP \mpP_k
\mpPar
\mpPart\roleP{\mpH_{\roleP}}
\mpPar
\mpPart\roleQ\mpCrash
\mpPar
\mpPart\roleQ\mpQUnavail
\mpPar
\mpM}
&
\hspace{-28em}
\highlight{(k \in I, \mpLab[k] = \mpCrashLab, \nexists \mpLab, \mpV: (\roleQ,
\mpLab(\mpV)) \in \mpH[\roleP])}
\\[1mm]
\inferrule{r-cond-T} &
\mpPart\roleP{\mpIf{\mpE}{\mpP}{\mpQ}}
\mpPar
\mpPart\roleP\mpH
\mpPar
\mpM
\;\redIf{\roleP}\;
\mpPart\roleP\mpP
\mpPar
\mpPart\roleP\mpH
\mpPar
\mpM
&
\hspace{-28em}
(\eval{\mpE}{\mpTrue})
\\[1mm]
\inferrule{r-cond-F} &
\mpPart\roleP{\mpIf{\mpE}{\mpP}{\mpQ}}
\mpPar
\mpPart\roleP\mpH
\mpPar
\mpM
\;\redIf{\roleP}\;
\mpPart\roleP\mpQ
\mpPar
\mpPart\roleP\mpH
\mpPar
\mpM
& 
\hspace{-28em}
(\eval{\mpE}{\mpFalse})
\\[1mm]
\inferrule{r-struct} &
\mpM_1\equiv \mpM_1'
\;\;\;\text{and}\;\;\;
\mpM_1'\mpMove \mpM_2'
\;\;\;\text{and}\;\;\;
\mpM_2'\equiv\mpM_2
\quad\implies\quad
\mpM_1\mpMove\mpM_2
\end{array}
\]
}
\caption{Reduction relation on sessions with crash-stop failures.}%
\label{fig:processes:reduction}
\end{figure}

\begin{example}
\label{ex:process_syntax_semantics}
Consider the session
{\small{$\mpM = \mpPart\roleP \mpP \mpPar  \mpPart\roleP \mpQEmpty
     \mpPar \mpPart \roleQ \mpQ \mpPar  \mpPart \roleQ \mpQEmpty$}},
where
{\small{
$ \mpP = \procout{\roleQ}{\mpLab}{\text{``\texttt{abc}''}}{\sum
\setenum{
  \begin{array}{@{}l@{}}
\procin{\roleQ}{\mpLabi(\mpx)}{\mpNil}
\\
\procin{\roleQ}{\mpCrashLab}{\mpNil}
\end{array}
}
}
$
}}
 and
{\small{
$
 \mpQ =
  \sum \setenum{
    \begin{array}{@{}l@{}}
 \procin{\roleP}{\mpLab(\mpx)}{
 \procout{\roleP}{\mpLabi}{42}{\mpNil}}
 \\
  \procin{\roleP}{\mpCrashLab}{\mpNil}
  \end{array}
  }.
  $
}}
In this session,  the process $\mpQ$ for $\roleQ$ receives a message sent from $\roleP$ to $\roleQ$;
the process $\mpP$ for $\roleP$ sends a message from $\roleP$ to $\roleQ$, and then receives a message
sent from $\roleQ$ to  $\roleP$.  %
Let each role be unreliable,
\ie $\rolesR = \emptyset$,  %
and
$\mpP$ crash before sending. %
We have
 {\small{
  $\mpM
    \redCrash{\roleP}{\emptyset}
     \mpPart\roleP{\mpCrash}
\mpPar
\mpPart\roleP{\mpQUnavail}
\mpPar
\mpPart \roleQ \mpQ \mpPar  \mpPart \roleQ \mpQEmpty
 \redSend{\roleP}{\roleQ}{\mpLab}
\mpPart\roleP{\mpCrash}
\mpPar
\mpPart\roleP{\mpQUnavail}
\mpPar
\mpPart \roleQ \mpNil
\mpPar
\mpPart\roleQ \mpQEmpty$
}}.
We observe that when the output process $\mpP$ located at an unreliable role
$\roleP$ crashes (by \inferrule{r-$\lightning$}),
the resulting entry for $\roleP$ is a crashed process ($\mpPart\roleP{\mpCrash}$) with an unavailable
queue ($\mpPart\roleP{\mpQUnavail}$).
Subsequently, the input process $\mpQ$ located at $\roleQ$ can detect and
handle the crash by \inferrule{r-rcv-$\odot$} via its $\mpCrashLab$ handling branch.
\end{example}

\section{Asynchronous Multiparty Session Types with Crash-Stop Semantics}\label{sec:gtype}\label{SEC:GTYPE}

In this section, we present our asynchronous multiparty session types
with crash-stop semantics.
We give an overview of global and local types with crashes in \Cref{sec:gtype:syntax},
including syntax, projection, subtyping, \etc;
our key additions to the classic theory are \emph{crash handling branches} in
both global and local types,
and a special local type $\stStop$ to denote crashed processes.
We give a Labelled Transition System (LTS) semantics to both global types
(\Cref{sec:gtype:lts-gt}) and configurations (\ie a collection of local types
and point-to-point communication queues, %
\Cref{sec:gtype:lts-context}%
).
We discuss alternative design options of modelling crash-stop
failures in \Cref{sec:gtype:alternative}.
We relate the two semantics in \Cref{sec:gtype:relating}, and show that
a configuration obtained via projection is safe, deadlock-free, and live in
\Cref{sec:gtype:pbp}.

\subsection{Global and Local Types with Crash-Stop Failures}\label{sec:gtype:syntax}

The top-down methodology begins with \emph{global types} to provide an
overview of
the communication between a number of \emph{roles}
($\roleP, \roleQ, \roleS, \roleT, \ldots$),
belonging to a (fixed) set $\roleSet$.
At the other end, we use \emph{local types} to describe how a \emph{single}
role communicates with other roles from a local perspective, and they are
obtained via \emph{projection} from a global type.
We give the syntax of both global and local types in \cref{fig:syntax-mpst},
which are similar to syntax used in~\cite{POPL19LessIsMore,CONCUR22MPSTCrash}.

\begin{figure}[t]%
{\small
  \[
    \begin{array}{r@{\quad}c@{\quad}l@{\quad}l}
      \tyGround & \bnfdef & \tyInt \bnfsep \tyBool \bnfsep \tyReal \bnfsep \tyUnit \bnfsep \ldots
        & \text{\footnotesize Basic types} \\
      \gtG & \bnfdef &
        \gtComm{\roleP}{\roleFmt{q^{\runtime{\roleMaybeCrashedSym}}}}{i \in
        I}{\gtLab[i]}{\tyGround[i]}{\gtG[i]}
        &
        {\footnotesize\text{Transmission}} \\
        & \bnfsep &
        \runtime{
          \gtCommTransit{\rolePMaybeCrashed}{\roleQ}{i \in
          I}{\gtLab[i]}{\tyGround[i]}{\gtG[i]}{j}
        }~(j \in I)
        &
        \runtime{\footnotesize\text{Transmission en route}} \\
        & \bnfsep & \gtRec{\gtRecVar}{\gtG} \quad \bnfsep \quad \gtRecVar \quad
        \bnfsep \quad \gtEnd &
        \text{\footnotesize Recursion, Type variable, Termination} \\
      \runtime{\roleMaybeCrashedSym} & \bnfdef & \cdot \quad
      \bnfsep \quad \runtime{\roleCrashedSym} & \runtime{\text{\footnotesize Crash
      annotation}}
        \\[1ex]
      \stS, \stT
        & \bnfdef & \stExtSum{\roleP}{i \in I}{\stChoice{\stLab[i]}{\tyGround[i]} \stSeq \stT[i]} 
        \ \bnfsep \stIntSum{\roleP}{i \in I}{\stChoice{\stLab[i]}{\tyGround[i]} \stSeq \stT[i]} 
          & \text{\footnotesize External choice, Internal choice} \\
          & \bnfsep & \stRec{\stRecVar}{\stT} \ \bnfsep \ \stRecVar \ 
           \bnfsep  \stEnd \ \bnfsep \ \runtime{\stStop} 
          &\text{\footnotesize Recursion, Type variable, Termination, $\runtime{\text{Crash}}$}
    \end{array}
  \]
  }
  \caption{Syntax of global types and local types. Runtime types are
  $\runtime{\text{shaded}}$.}
  \label{fig:syntax-global-type}%
  \label{fig:syntax-local-type}%
  \label{fig:syntax-mpst}
\end{figure}

\subparagraph*{Global Types} %
are ranged over $\gtG, \gtGi, \gtG[i], \ldots$,
and describe the behaviour for all roles from a bird's eye view.
The syntax shown in  $\runtime{\text{shade}}$ are \emph{runtime} syntax, which are not used for
describing a system at design-time, but for describing the state of a system
during execution.
The labels $\gtLab$ are taken from a fixed set of all labels $\labSet$, and
basic types $\tyGround$ (types for payloads) from a fixed set of all basic
types $\tyGroundSet$.

We explain each construct in the syntax of global types:
a transmission, denoted 
$\gtComm{\roleP}{\roleQMaybeCrashed}{i \in I}{\gtLab[i]}{\tyGround[i]}{\gtG[i]}$, 
represents a message from role $\roleP$ to role $\roleQ$ (with possible crash
annotations), with labels $\gtLab[i]$, payload types $\tyGround[i]$,
and continuations $\gtG[i]$, where $i$ is taken from an index set $I$.
We require that the index set be non-empty ($I \neq \emptyset$), labels
$\gtLab[i]$ be pair-wise distinct, and self receptions be excluded (\ie
$\roleP \neq \roleQ$), as standard in session type works.
Additionally, we require that the special $\gtCrashLab$ label (explained later)
not be the only label in a transmission, \ie $\setcomp{\gtLab[i]}{i \in I} \neq
\setenum{\gtCrashLab}$.
A transmission en route
$\gtCommTransit{\rolePMaybeCrashed}{\roleQ}{i \in
I}{\gtLab[i]}{\tyGround[i]}{\gtG[i]}{j}$
is a runtime construct representing a message $\gtLab[j]$ (index $j$) sent by
$\roleP$, and yet to be received by $\roleQ$.
Recursive types are represented via $\gtRec{\gtRecVar}{G}$ and $\gtRecVar$,
where contractive requirements apply~\cite[\S 21.8]{PierceTAPL}.
The type $\gtEnd$ describes a terminated type (omitted where
unambiguous).

To model crashes and crash handling, we use crash annotations $\roleCrashedSym$
and crash handling branches:
a \emph{crash annotation} $\roleCrashedSym$, a new addition in this work, marks
a \emph{crashed} role (only used in the \emph{runtime syntax}), and we omit
annotations for live roles, \ie $\roleP$ is a live role, $\rolePCrashed$ is a
crashed role, and $\rolePMaybeCrashed$ represents a possibly crashed role,
namely either $\roleP$ or $\rolePCrashed$.
We use a special label $\gtCrashLab$ for handling crashes: this
continuation denotes the protocol to follow when the sender of a
message is detected to have crashed by the receiver.
The special label acts as a `pseudo-message': when a sender role crashes, the
receiver can select the pseudo-message to enter crash handling.
We write $\gtRoles{\gtG}$ (resp.~$\gtRolesCrashed{\gtG}$) for the set of
\emph{active} (resp.~\emph{crashed}) roles in a global
type $\gtG$, \emph{excluding} (resp.~consisting \emph{only} of) those with a crash
annotation $\roleCrashedSym$.

\subparagraph*{Local Types}%
are ranged over $\stS, \stT, \stU,
\ldots$, and describe the behaviour of a single role.
An internal choice (selection) (resp.~an external choice (branching)), denoted 
$\stIntSum{\roleP}{i \in I}{\stChoice{\stLab[i]}{\tyGround[i]} \stSeq \stT[i]}$ 
(resp.~$\stExtSum{\roleP}{i \in I}{\stChoice{\stLab[i]}{\tyGround[i]} \stSeq \stT[i]}$%
)
indicates that the \emph{current} role is to \emph{send} to (resp.~%
\emph{receive} from) the role $\roleP$.
Similarly to global types, we require pairwise-distinct,
non-empty labels. %
Moreover, we require that the $\stCrashLab$ label not appear in \emph{internal}
choices, reflecting that a $\stCrashLab$ pseudo-message can never be
sent; and that singleton $\stCrashLab$ labels not permitted in external choices.
The type $\stEnd$ indicates a \emph{successful} termination (omitted where
unambiguous), and recursive types follow a similar fashion to global types.
We use a new \emph{runtime} type $\stStop$ to denote crashes.
\subparagraph*{Subtyping}
relation $\stSub$ on local types will be used in~\cref{sec:gtype:relating} 
to relate global and local type semantics.  
Our subtyping relation is mostly standard~\cite[Def.\@ 2.5]{POPL19LessIsMore},
except for an extra rule for $\stStop$ and additional requirements to support 
crash handling branch in external choices. 
\iftoggle{full}{The definition of $\stSub$ can be found in~\cref{sec:app:definitions}.}{} 
\subparagraph*{Projection}%
gives the local type of a participating role in a global
type,
defined as a \emph{partial} function that takes a global type
$\gtG$ and a
role $\roleP$, %
and returns a local type, given by
\cref{def:global-proj}.

\begin{definition}[Global Type Projection]%
  \label{def:global-proj}%
  \label{def:local-type-merge}%
  \label{def:removing-crash-label}%
  The \emph{projection of a global type $\gtG$ onto a role $\roleP$}, %
  with respect to a set of \emph{reliable} roles $\rolesR$,
  written \;$\gtProj[\rolesR]{\gtG}{\roleP}$,\; %
  is:

  \smallskip%
  \centerline{\(%
  \small%
  \begin{array}{c}
    \gtProj[\rolesR]{\left(%
      \gtCommSmall{\roleQ}{\roleRMaybeCrashed}
      {i \in I}{\gtLab[i]}{\tyGround[i]}{\gtG[i]}%
      \right)}{\roleP}%
    =\!%
    \left\{%
    \begin{array}{@{}l@{\hskip 2mm}r@{}}
      \stIntSum{\roleR}{\highlight{i \in \setcomp{j \in I}{\stFmt{\stLab[j]} \neq \stCrashLab}}}{ %
        \stChoice{\stLab[i]}{\tyGround[i]} \stSeq (\gtProj[\rolesR]{\gtG[i]}{\roleP})%
      }%
      \hspace{-1.5cm}
      &\text{\footnotesize%
        if\, $\roleP = \roleQ$%
      }%
      \\[1mm]%
      \stExtSum{\roleQ}{i \in I}{%
        \stChoice{\stLab[i]}{\tyGround[i]} \stSeq (\gtProj[\rolesR]{\gtG[i]}{\roleP})%
      }%
      &
      \begin{array}{@{}r@{}}
        \text{\footnotesize if \,} \roleP = \roleR,\text{\footnotesize
        \,and \,}
       \highlight{\roleQ \notin \rolesR \text{\footnotesize \, implies}}
        \\
        \highlight{\exists k \in I : \gtLab[k] = \gtCrashLab}
      \end{array}
      \\[2mm]%
      \stMerge{i \in I}{\gtProj[\rolesR]{\gtG[i]}{\roleP}}%
      &
      \text{\footnotesize%
        if \,} \roleP \neq \roleQ,\text{\footnotesize \,and \,} \roleP \neq \roleR%
    \end{array}
    \right.
    \\[8mm]%
     \gtProj[\rolesR]{\left(%
      \gtCommTransit{\roleQMaybeCrashed}{\roleR}
      {i \in I}{\gtLab[i]}{\tyGround[i]}{\gtG[i]}{j}%
      \right)}{\roleP}%
    =\!%
    \left\{%
    \begin{array}{@{}l@{\hskip 2mm}r@{}}
      \gtProj[\rolesR]{\gtG[j]}{\roleP}%
      &\text{\footnotesize%
        if\, $\roleP = \roleQ$%
      }%
      \\[1mm]%
      \stExtSum{\roleQ}{i \in I}{%
        \stChoice{\stLab[i]}{\tyGround[i]} \stSeq (\gtProj[\rolesR]{\gtG[i]}{\roleP})%
      }%
      &
      \begin{array}{@{}r@{}}
        \text{\footnotesize if \,} \roleP = \roleR,\text{\footnotesize
        \,and \,}
       \highlight{\roleQ \notin \rolesR \text{\footnotesize \, implies}}
        \\
       \highlight{\exists k \in I : \gtLab[k] = \gtCrashLab}
      \end{array}
      \\[2mm]%
     \stMerge{i \in I}{\gtProj[\rolesR]{\gtG[i]}{\roleP}}%
      &
      \text{\footnotesize%
        if \,} \roleP \neq \roleQ, \text{\footnotesize \,and \,} \roleP \neq \roleR%
    \end{array}
    \right.
    \\[8mm]%
    \gtProj[\rolesR]{(\gtRec{\gtRecVar}{\gtG})}{\roleP}%
    \;=\;%
    \left\{%
    \begin{array}{@{\hskip 0.5mm}l@{\hskip 5mm}l@{}}
      \stRec{\stRecVar}{(\gtProj[\rolesR]{\gtG}{\roleP})}%
      &%
      \text{\footnotesize%
        if\,
        $\roleP \in \gtG$ \,or\,
        $\fv{\gtRec{\gtRecVar}{\gtG}} \neq \emptyset$%
      }%
      \\%
      \stEnd%
      &%
      \text{\footnotesize%
        otherwise}
    \end{array}
    \right.%
    \quad\qquad%
    \begin{array}{@{}r@{\hskip 1mm}c@{\hskip 1mm}l@{}}
      \gtProj[\rolesR]{\gtRecVar}{\roleP}%
      &=&%
      \stRecVar%
      \\%
      \gtProj[\rolesR]{\gtEnd}{\roleP}%
      &=&%
      \stEnd%
    \end{array}
  \end{array}
  \)}%
  \smallskip%

  \noindent%
  where
  $\stMerge{}{}$ is %
  the \emph{merge operator for local types} %
  (\emph{full merging}):

  \smallskip%
  \centerline{\(%
  \small%
  \begin{array}{c}%
    \textstyle%
    \stExtSum{\roleP}{i \in I}{\stChoice{\stLab[i]}{\tyGround[i]} \stSeq \stSi[i]}%
    \!\stBinMerge\!%
    \stExtSum{\roleP}{\!j \in J}{\stChoice{\stLab[j]}{\tyGround[j]} \stSeq \stTi[j]}%
   \\
    \;=\;%
    \stExtSum{\roleP}{k \in I \cap J}{\stChoice{\stLab[k]}{\tyGround[k]} \stSeq%
      (\stSi[k] \!\stBinMerge\! \stTi[k])%
    }%
    \stExtC%
    \stExtSum{\roleP}{i \in I \setminus J}{\stChoice{\stLab[i]}{\tyGround[i]} \stSeq \stSi[i]}%
    \stExtC%
    \stExtSum{\roleP}{\!j \in J \setminus I}{\stChoice{\stLab[j]}{\tyGround[j]} \stSeq \stTi[j]}%
    \\[3mm]%
    \stIntSum{\roleP}{i \in I}{\stChoice{\stLab[i]}{\tyGround[i]} \stSeq \stSi[i]}%
    \,\stBinMerge\,%
    \stIntSum{\roleP}{i \in I}{\stChoice{\stLab[i]}{\tyGround[i]} \stSeq \stTi[i]}%
    \;=\;%
    \stIntSum{\roleP}{i \in I}{\stChoice{\stLab[i]}{\tyGround[i]} \stSeq (\stSi[i]
    \stBinMerge \stTi[i])}%
    \\[1mm]%
    \stRec{\stRecVar}{\stS} \,\stBinMerge\, \stRec{\stRecVar}{\stT}%
    \,=\,%
    \stRec{\stRecVar}{(\stS \stBinMerge \stT)}%
    \qquad%
    \stRecVar \,\stBinMerge\, \stRecVar%
    \,=\,%
    \stRecVar%
    \qquad%
    \stEnd \,\stBinMerge\, \stEnd%
    \,=\,%
    \stEnd%
  \end{array}
  \)}%
\end{definition}

We parameterise our theory on a (fixed) set of \emph{reliable} roles
$\rolesR$, \ie roles assumed to \emph{never crash}:
if $\rolesR = \emptyset$, every role is unreliable and
susceptible to crash;
if $\gtRoles{\gtG} \subseteq \rolesR$, every role in $\gtG$ is
reliable, and we simulate the results from the original MPST theory without
crashes.
We base our definition of projection on~\cite{POPL19LessIsMore},
but include more ($\highlight{\text{highlighted}}$) cases to account for reliable
roles, $\stCrashLab$ branches, and runtime global types.

When projecting a transmission from $\roleQ$ to $\roleR$, we remove the
$\stCrashLab$ label from the internal choice at $\roleQ$, reflecting our
model that a $\stCrashLab$ pseudo-message cannot be sent.
Dually, we require a $\stCrashLab$ label to be present in the external choice
at $\roleR$ -- unless the sender role $\roleQ$ is assumed to be reliable.
Our definition of projection enforces that transmissions, whenever
an unreliable role is the sender ($\roleQ \notin \rolesR$),
\emph{must include} a crash handling branch
($\exists k \in I: \gtLab[k] = \gtCrashLab$).
This requirement ensures that the receiving role $\roleR$ can \emph{always} handle
crashes whenever it happens, so that processes are not stuck when crashes occur.
We explain how these requirements help us achieve various properties by
projection
in \cref{sec:gtype:pbp}.
The rest of the rules %
are taken from the literature~\cite{POPL19LessIsMore, VanGlabbeekLICS2021},
without much modification.

\subsection{Crash-Stop Semantics of Global Types}
\label{sec:gtype:lts-gt}
We now give a Labelled Transition System (LTS) semantics to global types,
with crash-stop semantics.
To this end, we first introduce some
auxiliary definitions.
We define the transition labels in \Cref{def:mpst-env-reduction-label},
which are also used in the LTS semantics of configurations
(later in \Cref{sec:gtype:lts-context}).

\begin{definition}[Transition Labels]
  \label{def:mpst-env-reduction-label}%
  \label{def:mpst-label-subject}%
  Let $\stEnvAnnotGenericSym$ %
  be a transition label of the form:

\smallskip
\centerline{\(
  {\small{
  \begin{array}{rclllll}
  \stEnvAnnotGenericSym &\bnfdef&
    \stEnvInAnnotSmall{\roleP}{\roleQ}{\stChoice{\stLab}{\tyGround}}&
    (\text{$\roleP$ receives $\stChoice{\stLab}{\tyGround}$ from $\roleQ$})
    &\bnfsep&\stEnvOutAnnotSmall{\roleP}{\roleQ}{\stChoice{\stLab}{\tyGround}}&
    (\text{$\roleP$ sends $\stChoice{\stLab}{\tyGround}$ to $\roleQ$})
    \\
    &\bnfsep&\ltsCrash{\mpS}{\roleP}&%
    (\text{$\roleP$ crashes}) %
    &\bnfsep&\ltsCrDe{\mpS}{\roleP}{\roleQ}&%
    (\text{$\roleP$ detects the crash of $\roleQ$})
  \end{array}
  }}
\)}

\smallskip
\noindent
The subject of a transition label, written
\;$\ltsSubject{\stEnvAnnotGenericSym}$,\; is defined as:

\smallskip
\centerline{\(
{\small{
    \ltsSubject{\stEnvInAnnotSmall{\roleP}{\roleQ}{\stChoice{\stLab}{\tyGround}}}
    =
    \ltsSubject{\stEnvOutAnnotSmall{\roleP}{\roleQ}{\stChoice{\stLab}{\tyGround}}}
    =
    \ltsSubject{\ltsCrash{\mpS}{\roleP}}
    =
    \ltsSubject{\ltsCrDe{\mpS}{\roleP}{\roleQ}}
    =
    \roleP}}. 
\)}
\end{definition}

The labels
$\stEnvOutAnnotSmall{\roleP}{\roleQ}{\stChoice{\stLab}{\tyGround}}$
and
$\stEnvInAnnotSmall{\roleP}{\roleQ}{\stChoice{\stLab}{\tyGround}}$
describe sending and receiving actions respectively.
The crash of $\roleP$ is denoted by the %
label
$\ltsCrashSmall{\mpS}{\roleP}$, and
the detection of a crash by %
label $\ltsCrDe{\mpS}{\roleP}{\roleQ}$:
we model crash detection at \emph{reception}, the label
contains a \emph{detecting} role $\roleP$ and a \emph{crashed} role $\roleQ$.

We define an operator to \emph{remove} a role from a global type in
\cref{def:gtype:remove-role}:
the intuition is to remove any interaction of a
crashed role from the given global type.
When a role has crashed, we attach a \emph{crashed annotation},
and remove infeasible actions, \eg when the sender and receiver of a
transmission have both crashed.
The removal operator is a partial function that takes a global type $\gtG$ and a
live role $\roleR$ ($\roleR \in \gtRoles{\gtG}$) and gives a global type
$\gtCrashRole{\gtG}{\roleR}$.

\begin{definition}[Role Removal]%
  \label{def:gtype:remove-role}
  The removal of a live role $\roleP$ in a global type
  $\gtG$, written \;$\gtCrashRole{\gtG}{\roleP}$,\; is defined as follows:
\newlength{\gtRemovalLHSWidth}
\settowidth{\gtRemovalLHSWidth}{$
\gtCommTransit{\rolePCrashed}{\roleQ}{i \in I}{\gtLab[i]}{\tyGround[i]}{
(\gtCrashRole{\gtG[i]}{\roleR})}{j}
$}
\[\small
  \begin{array}{@{}r@{~}c@{~}l@{}}
    \gtCrashRole{
      (\gtCommSmall{\roleP}{\roleQ}{i \in I}{\gtLab[i]}{\tyGround[i]}{\gtG[i]})
    }{
      \roleR
    }
    &=&
    \left\{
      \begin{array}{@{}ll@{}}
        \gtCommTransit{\rolePCrashed}{\roleQ}{i \in I}{\gtLab[i]}{\tyGround[i]}{
          (\gtCrashRole{\gtG[i]}{\roleR})}{j}
          &
        \text{if~} \roleP = \roleR \text{~and~} \exists j \in I: \gtLab[j] =
        \gtCrashLab \\
        \gtCommSmall{\roleP}{\roleQCrashed}{i \in I}{\gtLab[i]}{\tyGround[i]}{
          (\gtCrashRole{\gtG[i]}{\roleR})
        } &
        \text{if~} \roleQ = \roleR \\
        \makebox[\gtRemovalLHSWidth][l]{%
          \gtCommSmall{\roleP}{\roleQ}{i \in I}{\gtLab[i]}{\tyGround[i]}{
            (\gtCrashRole{\gtG[i]}{\roleR})
          }
        }
        &
        \text{if~} \roleP \neq \roleR \text{~and~} \roleQ \neq \roleR
      \end{array}
    \right.
    \\[5mm]
    \gtCrashRole{
      (\gtCommTransit{\roleP}{\roleQ}{i \in
      I}{\gtLab[i]}{\tyGround[i]}{\gtG[i]}{j})
    }{
      \roleR
    }
    &=&
    \left\{
      \begin{array}{@{}ll@{}}
        \gtCommTransit{\rolePCrashed}{\roleQ}{i \in I}{\gtLab[i]}{\tyGround[i]}{
          (\gtCrashRole{\gtG[i]}{\roleR})}{j}
          &
        \text{if~} \roleP = \roleR
        \\
        \makebox[\gtRemovalLHSWidth][l]{\gtCrashRole{\gtG[j]}{\roleR}}
          &
        \text{if~} \roleQ = \roleR
        \\
        \gtCommTransit{\roleP}{\roleQ}{i \in I}{\gtLab[i]}{\tyGround[i]}{
          (\gtCrashRole{\gtG[i]}{\roleR})}{j}
          &
        \text{if~} \roleP \neq \roleR \text{~and~} \roleQ \neq \roleR
      \end{array}
    \right.
    \\[5mm]
    \gtCrashRole{
      (\gtCommSmall{\roleP}{\roleQCrashed}{i \in I}{\gtLab[i]}{\tyGround[i]}{\gtG[i]})
    }{
      \roleR
    }
    &=&
    \left\{
      \begin{array}{@{}ll@{}}
        \makebox[\gtRemovalLHSWidth][l]{\gtCrashRole{\gtG[j]}{\roleR}}
        &
        \text{if~} \roleP = \roleR \text{~and~} \exists j \in I: \gtLab[j] = \gtCrashLab \\
        \gtCommSmall{\roleP}{\roleQCrashed}{i \in I}{\gtLab[i]}{\tyGround[i]}{
          (\gtCrashRole{\gtG[i]}{\roleR})
        } &
        \text{if~} \roleP \neq \roleR \text{~and~} \roleQ \neq \roleR
      \end{array}
    \right.
    \\[3mm]
    \gtCrashRole{
      (\gtCommTransit{\rolePCrashed}{\roleQ}{i \in
      I}{\gtLab[i]}{\tyGround[i]}{\gtG[i]}{j})
    }{
      \roleR
    }
    &=&
    \left\{
      \begin{array}{@{}ll@{}}
        \makebox[\gtRemovalLHSWidth][l]{\gtCrashRole{\gtG[j]}{\roleR}} &
        \text{if~} \roleQ = \roleR
        \\
        \gtCommTransit{\rolePCrashed}{\roleQ}{i \in I}{\gtLab[i]}{\tyGround[i]}{
          (\gtCrashRole{\gtG[i]}{\roleR})}{j}
          &
        \text{if~} \roleP \neq \roleR \text{~and~} \roleQ \neq \roleR
      \end{array}
    \right.
    \\[3mm]
    \gtCrashRole{
      (\gtRec{\gtRecVar}{\gtG})
    }{
      \roleR
    }
    &=&
    \left\{
      \begin{array}{@{}ll@{}}
        \makebox[\gtRemovalLHSWidth][l]{\gtRec{\gtRecVar}{(\gtCrashRole{\gtG}{\roleR})}} & \text{if~}
        \fv{\gtRec{\gtRecVar}{\gtG}} \neq \emptyset \text{~or~}
        \gtRoles{\gtCrashRole{\gtG}{\roleR}}
        \neq \emptyset \\
        \gtEnd &
        \text{otherwise}
      \end{array}
    \right.
    \\[3mm]
    \gtCrashRole{
      \gtRecVar
    }{
      \roleR
    }
    &=&
    \gtRecVar
    \hspace{2cm}
    \gtCrashRole{\gtEnd}{\roleR} = \gtEnd
  \end{array}
\]
\end{definition}

For simple cases, the removal of a role $\gtCrashRole{\gtG}{\roleR}$ attaches
crash annotations $\roleCrashedSym$ on all occurrences of the removed role
$\roleR$ throughout global type $\gtG$ inductively.

We draw attention to some interesting cases:
when we remove the sender role $\roleP$ from a transmission prefix
$\gtFmt{\roleP \!\to\! \roleQ}$,
the result is a `pseudo-transmission' en route prefix
$\gtFmt{\rolePCrashed \!\rightsquigarrow\! \roleQ : j}$
where $\gtLab[j] = \gtCrashLab$.
This enables the receiver $\roleQ$ to `receive'  the special $\gtCrashLab$
after the crash of $\roleP$, hence triggering the crash handling branch.
Recall that our definition of projection requires that a crash handling branch
be present whenever a crash may occur ($\roleQ \notin \rolesR$).

When we remove the sender role $\roleP$ from a transmission en route prefix
$\gtFmt{\roleP \!\rightsquigarrow\! \roleQ : j}$,
the result \emph{retains} the index $j$ that was selected by $\roleP$,
instead of the index associated with $\gtCrashLab$ handling.
This is crucial to our crash modelling: when a role crashes, the messages that
the role \emph{has sent} to other roles are still available.
We discuss alternative models later in \cref{sec:gtype:alternative}.

In other cases, where removing the role $\roleR$ would render a transmission
(regardless of being en route or not) meaningless, \eg both sender and receiver
have crashed, we simply remove the prefix entirely. %
\iftoggle{full}{An example of role removal in a global type 
can be found in~\cref{sec:app_examples}.}{}

We now give an LTS semantics to a
global type $\gtG$,
by defining the semantics with a tuple $\gtWithCrashedRoles{\rolesC}{\gtG}$,
where
$\rolesC$ is a set of \emph{crashed} roles. %
The transition system is parameterised by reliability assumptions, in the form
of a fixed set of reliable roles $\rolesR$.
When unambiguous, we write
$\gtG$ as an abbreviation of $\gtWithCrashedRoles{\rolesEmpty}{\gtG}$. 
We define the reduction rules of global types in \cref{def:gtype:lts-gt}.

\begin{figure}[t]
\[
{\small 
\begin{array}{c}
 \highlight{\inference[\iruleGtMoveCrash]{
    \roleP \notin \rolesR
      &
    \roleP \in \gtRoles{\gtG}
    &
    \gtG \neq \gtRec{\gtRecVar}{\gtGi}
  }{
    \gtWithCrashedRoles{\rolesC}{\gtG}
    \gtMove[\ltsCrashSmall{\mpS}{\roleP}]{\rolesR}
    \gtWithCrashedRoles{\rolesC \cup \setenum{\roleP}}{\gtCrashRole{\gtG}{\roleP}}
  }}
  \qquad

  \inference[\iruleGtMoveRec]{
    \gtWithCrashedRoles{\rolesC}{\gtG{}\subst{\gtRecVar}{\gtRec{\gtRecVar}{\gtG}}}
    \gtMove[\stEnvAnnotGenericSym]{\rolesR}
    \gtWithCrashedRoles{\rolesCi}{\gtGi}}
  {
    \gtWithCrashedRoles{\rolesC}{\gtRec{\gtRecVar}{\gtG}}
    \gtMove[\stEnvAnnotGenericSym]{\rolesR}
    \gtWithCrashedRoles{\rolesCi}{\gtGi}
  }
  \\[1mm]
\inference[\iruleGtMoveOut]{
    j \in I
    &
    \highlight{\gtLab[j] \neq \gtCrashLab}
  }{
    \gtWithCrashedRoles{\rolesC}{
      \gtCommSmall{\roleP}{\roleQ}{i \in I}{\gtLab[i]}{\tyGround[i]}{\gtGi[i]}
    }
    \gtMove[
      \stEnvOutAnnotSmall{\roleP}{\roleQ}{\stChoice{\gtLab[j]}{\tyGround[j]}}
    ]{
      \rolesR
    }
    \gtWithCrashedRoles{\rolesC}{
      \gtCommTransit{\roleP}{\roleQ}{i \in I}{\gtLab[i]}{\tyGround[i]}{\gtGi[i]}{j}
    }
  }
  \\[1mm]
 \inference[\iruleGtMoveIn]{
    j \in I
    &
    \highlight{\gtLab[j] \neq \gtCrashLab}
  }{
    \gtWithCrashedRoles{\rolesC}{
      \gtCommTransit{\rolePMaybeCrashed}{\roleQ}{i \in I}{\gtLab[i]}{\tyGround[i]}{\gtGi[i]}{j}
    }
    \gtMove[
      \stEnvInAnnotSmall{\roleQ}{\roleP}{\stChoice{\gtLab[j]}{\tyGround[j]}}
    ]{
      \rolesR
    }
    \gtWithCrashedRoles{\rolesC}{\gtGi[j]}
  }
  \\[1mm]
\highlight{\inference[\iruleGtMoveCrDe]{
    j \in I
    &
    \gtLab[j] = \gtCrashLab
  }{
    \gtWithCrashedRoles{\rolesC}{
      \gtCommTransit{\rolePCrashed}{\roleQ}{i \in I}{\gtLab[i]}{\tyGround[i]}{\gtGi[i]}{j}
    }
    \gtMove[\ltsCrDe{\mpS}{\roleQ}{\roleP}]{\rolesR}
    \gtWithCrashedRoles{\rolesC}{\gtGi[j]}
  }}
  \\[1mm]
\highlight{\inference[\iruleGtMoveOrph]{
    j \in I
    &
    \gtLab[j] \neq \gtCrashLab
  }{
    \gtWithCrashedRoles{\rolesC}{\gtCommSmall{\roleP}{\roleQCrashed}{i \in
    I}{\gtLab[i]}{\tyGround[i]}{\gtGi[i]}}
    \gtMove[\stEnvOutAnnotSmall{\roleP}{\roleQ}{\stChoice{\gtLab[j]}{\tyGround[j]}}]{
      \rolesR
    }
    \gtWithCrashedRoles{\rolesC}{\gtGi[j]}
  }}
  \\[1mm]
 \inference[\iruleGtMoveCtx]{
    \forall i \in I :
    \gtWithCrashedRoles{\rolesC}{\gtGi[i]}
    \gtMove[\stEnvAnnotGenericSym]{\rolesR}
    \gtWithCrashedRoles{\rolesCi}{\gtGii[i]}
    &
    \ltsSubject{\stEnvAnnotGenericSym} \notin \setenum{\roleP, \roleQ}
  }{
    \gtWithCrashedRoles{\rolesC}{
      \gtCommSmall{\roleP}{\roleQMaybeCrashed}{i \in
      I}{\gtLab[i]}{\tyGround[i]}{\gtGi[i]}
    }
    \gtMove[\stEnvAnnotGenericSym]{\rolesR}
    \gtWithCrashedRoles{\rolesCi}{
      \gtCommSmall{\roleP}{\roleQMaybeCrashed}{i \in
      I}{\gtLab[i]}{\tyGround[i]}{\gtGii[i]}
    }
  }\\[1mm]
\inference[\iruleGtMoveCtxi]{
    \forall i \in I :
    \gtWithCrashedRoles{\rolesC}{\gtGi[i]}
    \gtMove[\stEnvAnnotGenericSym]{\rolesR}
    \gtWithCrashedRoles{\rolesCi}{\gtGii[i]}
    &
    \ltsSubject{\stEnvAnnotGenericSym} \neq \roleQ
  }{
    \gtWithCrashedRoles{\rolesC}{
      \gtCommTransit{\rolePMaybeCrashed}{\roleQ}{i \in
      I}{\gtLab[i]}{\tyGround[i]}{\gtGi[i]}{j}
    }
    \gtMove[\stEnvAnnotGenericSym]{\rolesR}
    \gtWithCrashedRoles{\rolesCi}{
      \gtCommTransit{\rolePMaybeCrashed}{\roleQ}{i \in
      I}{\gtLab[i]}{\tyGround[i]}{\gtGii[i]}{j}
    }
  }
\end{array}
}
\]
\caption{Global Type Reduction Rules}
\label{fig:gtype:red-rules}
\end{figure}

\begin{definition}[Global Type Reductions]
  \label{def:gtype:lts-gt}
  The global type (annotated with a set of crashed roles $\rolesC$)
  transition relation
  $\gtMove[\stEnvAnnotGenericSym]{\rolesR}$
  is inductively
  defined by the rules in \cref{fig:gtype:red-rules},
  parameterised by a fixed set $\rolesR$ of reliable roles.
  We write
  $\gtWithCrashedRoles{\rolesC}{\gtG} \gtMove{\rolesR}
  \gtWithCrashedRoles{\rolesCi}{\gtGi}$
  if there
  exists $\stEnvAnnotGenericSym$ such that
  $\gtWithCrashedRoles{\rolesC}{\gtG}
  \gtMove[\stEnvAnnotGenericSym]{\rolesR}
  \gtWithCrashedRoles{\rolesCi}{\gtGi}$;
  we write
  $\gtWithCrashedRoles{\rolesC}{\gtG} \gtMove{\rolesR}$
  if there
  exists $\rolesCi$, $\gtGi$, and $\stEnvAnnotGenericSym$ such that
  $\gtWithCrashedRoles{\rolesC}{\gtG}
  \gtMove[\stEnvAnnotGenericSym]{\rolesR}
  \gtWithCrashedRoles{\rolesCi}{\gtGi}$,
  and $\gtMoveStar[\rolesR]$ for the transitive and reflexive closure of
  $\gtMove{\rolesR}$.
\end{definition}

Rules \inferrule{\iruleGtMoveOut} and \inferrule{\iruleGtMoveIn} model sending
and receiving messages respectively, as are standard in existing
works~\cite{ICALP13CFSM}.
We add an ($\highlight{\text{highlighted}}$) extra condition that the message exchanged not be a
pseudo-message carrying the $\gtCrashLab$ label.
$\inferrule{\iruleGtMoveRec}$ is a standard rule handling recursion. %

We introduce ($\highlight{\text{highlighted}}$) rules to account for crash and consequential behaviour:
\begin{itemize}[leftmargin=*, nosep]
\item
Rule $\inferrule{\iruleGtMoveCrash}$ models crashes, where a live ($\roleP \in
\gtRoles{\gtG}$), but unreliable ($\roleP \notin \rolesR$) role $\roleP$ may crash.
The crashed role $\roleP$ is added into the set of crashed roles ($\rolesC \cup
\setenum{\roleP}$), and removed
from the global type, resulting in a global type $\gtCrashRole{\gtG}{\roleP}$.
\item
Rule $\inferrule{\iruleGtMoveCrDe}$ is for \emph{crash detection}, where a live
role $\roleQ$ may detect that $\roleP$ has crashed at reception,
and then continues with the crash handling continuation labelled $\gtCrashLab$.
This rule only applies when the message en route is a pseudo-message, since
otherwise a message rests in the queue %
of the receiver and can be received
despite the crash of the sender (\cf~\inferrule{\iruleGtMoveIn}).
\item
Rule $\inferrule{\iruleGtMoveOrph}$ models the orphaning of a message sent from a
live role $\roleP$ to a crashed role $\roleQ$.
Similar to the requirement in \inferrule{\iruleGtMoveOut}, we add the side
condition that the message sent is not a pseudo-message.
\end{itemize}

Finally, rules $\inferrule{\iruleGtMoveCtx}$ and $\inferrule{\iruleGtMoveCtxi}$
allow non-interfering reductions of (intermediate) global types
under prefix, provided that all of the continuations can be reduced by
that label.

\begin{remark}[Necessity of $\rolesC$ in Semantics]
  While we can obtain the set of crashed roles in any global type $\gtG$ via
  $\gtRolesCrashed{\gtG}$, we need a separate $\rolesC$ for bookkeeping
  purposes. To illustrate, let 
  $\gtG = \gtCommSingleErr{\roleP}{\roleQ}{\gtLab}{}{\gtEnd}{\gtEnd}$,
  we can have the following reductions:

\vspace{-0.2em}
\smallskip
 \centerline{\(
 {\small{
    \gtWithCrashedRoles{\rolesEmpty}{\gtG}
    \gtMove[\ltsCrash{\mpS}{\roleQ}]{\rolesEmpty}
    \gtWithCrashedRoles{\setenum{\roleQ}}{
      \gtCommSingleErr{\roleP}{\roleQCrashed}{\gtLab}{}{\gtEnd}{\gtEnd}
    }
    \gtMove[\stEnvOutAnnot{\roleP}{\roleQ}{\gtLab}{}]{\rolesEmpty}
    \gtWithCrashedRoles{\setenum{\roleQ}}{\gtEnd}
    }}
\)}

\vspace{-0.2em}
\smallskip
\noindent
  While we can deduce $\roleQ$ is a crashed role in the interim global type,
  the same information cannot be recovered from the final global type $\gtEnd$.
\end{remark}

\subsection{Crash-Stop Semantics of Configurations}\label{sec:gtype:lts-context}
After giving semantics to global types, %
we now give an
LTS semantics to \emph{configurations},
\ie a collection of local types and communication queues %
across roles.
We first give a definition of configurations
in \cref{def:mpst-env}, followed by their reduction rules in
\cref{def:mpst-env-reduction}.

\begin{definition}[Configurations]%
  \label{def:mpst-env}%
  \label{def:mpst-env-closed}%
  \label{def:mpst-env-comp}%
  \label{def:mpst-env-subtype}%
  A configuration is a tuple \;$\stEnv; \qEnv$,\; where
  $\stEnv$ is a \emph{typing context}, denoting a partial mapping
  from roles to local types,
  defined as:
    \(%
  \stEnv
  \,\coloncolonequals\,
  \stEnvEmpty
  \bnfsep
  \stEnv \stEnvComp \stEnvMap{\roleP}{\stT}
  \)%
  . %
  We write $\stEnvUpd{\stEnv}{\roleP}{\stT}$ for updates:
  $\stEnvApp{\stEnvUpd{\stEnv}{\roleP}{\stT}}{\roleP} = \stT$ and %
  $\stEnvApp{\stEnvUpd{\stEnv}{\roleP}{\stT}}{\roleQ} =
  \stEnvApp{\stEnv}{\roleQ}$ (where $\roleP \neq \roleQ$).

  A \emph{queue}, denoted $\stQ$, is either a (possibly empty) sequence of
  messages
  $\stM[1] \stFmt{\cdot} \stM[2] \stFmt{\cdot} \cdots \stFmt{\cdot}
  \stM[n]$,
  or unavailable $\stQUnavail$.
  We write $\stQEmpty$ for an empty queue, and
  $\stQCons{\stM}{\stQi}$
  for a non-empty queue with message $\stM$ at the beginning.
  A \emph{queue message} $\stM$ is of form
  $\stQMsg{\stLab}{\tyGround}$,
  denoting a message with label $\stLab$ and payload $\tyGround$.
  We sometimes omit $\tyGround$ when the payload is not of specific
  interest.

  We write $\qEnv$ to denote a \emph{queue environment}, a collection of
  peer-to-peer queues.
  A queue from $\roleP$ to $\roleQ$ at $\qEnv$ is denoted
  $\qApp{\qEnv}{\roleP}{\roleQ}$.
  We define updates $\stEnvUpd{\qEnv}{\roleP,
  \roleQ}{\stQ}$
  similarly.
  We write $\qEnv[\emptyset]$ for an \emph{empty} queue environment,
  where $\qApp{\qEnv[\emptyset]}{\roleP}{\roleQ} = \stQEmpty$ for any
  $\roleP$ and $\roleQ$ in the domain.

  We write $\stQCons{\stQi}{\stM}$ to append a message $\stM$ at the end
  of a queue $\stQi$:
  the message is appended to the sequence when $\stQi$ is available, or
  discarded when $\stQi$ is unavailable
  (\ie $\stQCons{\stQUnavail}{\stM} = \stQUnavail$).
  Additionally, we write
  $\stEnvUpd{\qEnv}{\cdot, \roleQ}{\stQUnavail}$
  for
  making all the queues to $\roleQ$ unavailable: \ie
   $\stEnvUpd{
    \stEnvUpd{
      \stEnvUpd{\qEnv}{\roleP[1], \roleQ}{\stQUnavail}
    }{\roleP[2], \roleQ}{\stQUnavail} \cdots
  }{\roleP[n], \roleQ}{\stQUnavail}$.
\end{definition}

We give an LTS semantics of configurations in \cref{def:mpst-env-reduction}.
Similar to that of global types, we model the semantics of configurations
in an asynchronous (\aka message passing) fashion, using a queue environment to
represent the communication queues %
among all roles.

\begin{figure}[t]
  \noindent
  \centerline{\(
  \small
  \begin{array}{@{}c@{}}
    \inference[\iruleTCtxOut]{%
      \stEnvApp{\stEnv}{%
        \roleP%
      } =
      \stIntSum{\roleQ}{i \in I}{\stChoice{\stLab[i]}{\tyGround[i]} \stSeq \stT[i]}%
      &
      k \in I%
    }{%
      \stEnv; \qEnv
      \,\stEnvMoveOutAnnot{\roleP}{\roleQ}{\stChoice{\stLab[k]}{\tyGround[k]}}\,%
      \stEnvUpd{\stEnv}{\roleP}{\stT[k]};
      \stEnvUpd{\qEnv}{\roleP, \roleQ}{
        \stQCons{
          \stEnvApp{\qEnv}{\roleP, \roleQ}
        }{
          \stQMsg{\stLab[k]}{\tyGround[k]}
        }
      }%
    }%
    \\[2mm]%
    \inference[\iruleTCtxIn]{%
      \stEnvApp{\stEnv}{%
        \roleP%
      } =
      \stExtSum{\roleQ}{i \in I}{\stChoice{\stLab[i]}{\tyGround[i]} \stSeq \stT[i]}%
      &
      k \in I%
      &
      \stLabi = \stLab[k]
      &
      \tyGroundi = \tyGround[k]
      &
      \stEnvApp{\qEnv}{\roleQ, \roleP}
      =
      \stQCons{\stQMsg{\stLabi}{\tyGroundi}}{\stQi}
      \neq
      \stQUnavail
    }{%
      \stEnv; \qEnv
      \,\stEnvMoveInAnnot{\roleP}{\roleQ}{\stChoice{\stLab[k]}{\tyGround[k]}}\,%
      \stEnvUpd{\stEnv}{\roleP}{\stT[k]};
      \stEnvUpd{\qEnv}{\roleQ, \roleP}{\stQi}
    }%
    \\[2mm]%
    \inference[\iruleTCtxRec]{%
      \stEnvApp{\stEnv}{\roleP} = \stRec{\stRecVar}{\stT}%
      &
      \stEnvUpd{\stEnv}{\roleP}{
        \stT\subst{\stRecVar}{\stRec{\stRecVar}{\stT}}%
      }; \qEnv%
      \stEnvMoveGenAnnot
      \stEnvi; \qEnvi%
    }{%
      \stEnv; \qEnv
      \stEnvMoveGenAnnot
      \stEnvi; \qEnvi%
    }%
    \qquad
    \highlight{%
    \inference[\iruleTCtxCrash]{%
      \stEnvApp{\stEnv}{\roleP} \neq \stEnd
      &
      \stEnvApp{\stEnv}{\roleP} \neq \stStop
    }{
      \stEnv; \qEnv
      \stEnvMoveAnnot{\ltsCrash{\mpS}{\roleP}}
      \stEnvUpd{\stEnv}{\roleP}{\stStop};
      \stEnvUpd{\qEnv}{\cdot, \roleP}{\stQUnavail}
    }
    }%
    \\
    \highlight{%
    \inference[\iruleTCtxCrashDetect]{%
      \stEnvApp{\stEnv}{\roleQ} =
      \stExtSum{\roleP}{i \in I}{\stChoice{\stLab[i]}{\tyGround[i]} \stSeq \stT[i]}%
      &
      \stEnvApp{\stEnv}{\roleP} = \stStop
      &
      k \in I
      &
      \stLab[k] = \stCrashLab
      &
      \stEnvApp{\qEnv}{\roleP, \roleQ} = \stQEmpty
    }{%
      \stEnv; \qEnv
      \,\stEnvMoveAnnot{\ltsCrDe{\mpS}{\roleQ}{\roleP}}\,%
      \stEnvUpd{\stEnv}{\roleQ}{\stT[k]}; \qEnv
    }%
    }%
  \end{array}
  \)}%
  \caption{Configuration Semantics}
  \label{fig:gtype:tc-red-rules}
\end{figure}

\begin{definition}[Configuration Semantics]%
  \label{def:mpst-env-reduction}%
  The \emph{configuration transition relation $\stEnvMoveGenAnnot$} %
  is defined in \cref{fig:gtype:tc-red-rules}.
  We write $\stEnvMoveGenAnnotP{\stEnv; \qEnv}$ %
  iff  $\stEnv; \qEnv \!\stEnvMoveGenAnnot\! \stEnvi; \qEnvi$ for some $\stEnvi$
  and $\qEnvi$. %
  We define two \emph{reductions} $\stEnvMove$ and
  $\stEnvMoveMaybeCrash[\rolesR]$ (where $\rolesR$ is a fixed set of
  reliable roles)
  as follows:
   \begin{itemize}[left=0pt, topsep=0pt]
   \item We write $\stEnv; \qEnv \!\stEnvMove\! \stEnvi; \qEnvi$
     for
     $\stEnv; \qEnv \stEnvMoveGenAnnot \stEnvi; \qEnvi$
     with
     \(\stEnvAnnotGenericSym \!\in\! \setenum{\stEnvInAnnotSmall{\roleP}{\roleQ}{\stChoice{\stLab}{\tyGround}},
     \stEnvOutAnnotSmall{\roleP}{\roleQ}{\stChoice{\stLab}{\tyGround}},
     \ltsCrDe{\mpS}{\roleP}{\roleQ}}\).
     We write
     $\stEnvMoveP{\stEnv; \qEnv}$
     iff
     $\stEnv; \qEnv \!\stEnvMove\! \stEnvi; \qEnvi$
     for some
     $\stEnvi; \qEnvi$,  %
     and
     $\stEnvNotMoveP{\stEnv; \qEnv}$ for its negation,  %
     and $\stEnvMoveStar$ %
     for the reflexive and transitive closure of $\stEnvMove$;

   \item We write
    $\stEnv; \qEnv \!\stEnvMoveMaybeCrash[\rolesR]\! \stEnvi; \qEnvi$
    for
    $\stEnv; \qEnv \stEnvMoveGenAnnot \stEnvi; \qEnvi$
    with
    \(
    \stEnvAnnotGenericSym \!\notin\! \setcomp{ \ltsCrash{\mpS}{\roleR} }{
      \roleR \!\in\!  \rolesR}
    \).
    We write \;$\stEnvMoveMaybeCrashP[\rolesR]{\stEnv; \qEnv}$ %
    iff
    $\stEnv; \qEnv \!\stEnvMoveMaybeCrash[\rolesR]\! \stEnvi; \qEnvi$
    for some
    $\stEnvi; \qEnvi$,
    and
    $\stEnvNotMoveMaybeCrashP[\rolesR]{\stEnv; \qEnvi}$
    for its negation.   %
    We define
    $\stEnvMoveMaybeCrashStar[\rolesR]$
    as the reflexive and transitive closure of $\stEnvMoveMaybeCrash[\rolesR]$.
 \end{itemize}
\end{definition}

We first explain the standard rules:
rule $\inferrule{\iruleTCtxOut}$ (resp.~$\inferrule{\iruleTCtxIn}$)
says that a role can perform an output (resp.~input) transition by appending
(resp.~consuming) a message at the corresponding queue.
Recall that whenever a queue is unavailable, the resulting queue remains
unavailable after appending ($\stQCons{\stQUnavail}{\stM} = \stQUnavail$).
Therefore, the rule $\inferrule{\iruleTCtxOut}$ covers delivery to both crashed
and live roles, whereas two separate rules are used in modelling global type
semantics (\inferrule{\iruleGtMoveOut} and \inferrule{\iruleGtMoveOrph}).
We also include a standard rule $\inferrule{\iruleTCtxRec}$ for recursive types.

The key innovations are the ($\highlight{\text{highlighted}}$) rules modelling crashes and crash
detection:
by rule~\inferrule{\iruleTCtxCrash}, a role $\roleP$ may crash
and become $\stStop$ at any time
(unless it is already $\stEnd$ed or $\stStop$ped).
All of $\roleP$'s receiving queues become unavailable $\stQUnavail$, so that
future messages to $\roleP$ would be discarded. %
Rule \inferrule{\iruleTCtxCrashDetect} models crash detection and handling:
if ${\roleP}$ is crashed and stopped,
another role ${\roleQ}$ attempting to receive from $\roleP$
can then take its $\stCrashLab$ handling branch. %
However, this rule only applies when the corresponding queue is empty: it is
still possible to receive messages sent before crashing via
\inferrule{\iruleTCtxIn}.

\subsection{Alternative Modellings for Crash-Stop Failures}\label{sec:gtype:alternative}
Before we dive into the relation between two semantics, let us have a short
digression to discuss our modelling choices and alternatives.
In this work, we mostly follow the assumptions laid out in
\cite{CONCUR22MPSTCrash},
where a crash is detected at reception. However, they opt to use a synchronous
(rendez-vous) semantics, whereas we give an asynchronous (message passing)
semantics, which entails interesting scenarios that would not arise in
a synchronous semantics.

Specifically,
consider the case where a role $\roleP$ sends a message to $\roleQ$, and
then $\roleP$ crashes after sending, but before $\roleQ$ receives the message.
The situation does not arise under a synchronous semantics, since sending and
receiving actions are combined into a single transmission action.

Intuitively, there are two possibilities to handle this scenario.
The questions are
whether the message sent immediately before crashing is deliverable to $\roleQ$,
and consequentially,
at what time does $\roleQ$ detect the crash of $\roleP$.

In our semantics (\cref{fig:gtype:red-rules,fig:gtype:tc-red-rules}), we
opt to answer the first question in positive:
we argue that this model is more consistent with our `passive' crash detection
design.
For example, if a role $\roleP$ never receives from another role $\roleQ$, then
$\roleP$ does not need to react in the event of $\roleQ$'s crash.
Following a similar line of reasoning,
if the message sent by $\roleP$ arrives in the receiving queue of
$\roleQ$,
then $\roleQ$ should be able to receive the message, without triggering
a crash detection (although it may be triggered later).
As a consequence, we require in \inferrule{\iruleTCtxCrashDetect} that the
queue $\stEnvApp{\qEnv}{\roleP, \roleQ}$ be empty, to reflect the idea that
crash detection should be a `last resort'.

For an alternative model,
we can opt to detect the crash after crash has occurred.
This is possibly better modelled with using outgoing queues
(\cf\cite{FSTTCS15MPSTExpressiveness}),
instead of incoming queues in the semantics presented.
Practically, this may be the scenario that a TCP connection is closed (or
reset) when a peer has crashed,
and the content in the queue %
is lost.
 It is worth noting that this kind of alternative model
will not affect our main theoretical results: the operational correspondence
between global and local type semantics, and furthermore, global type properties guaranteed by projection.

\subsection{Relating Global Type and Configuration Semantics}%
\label{sec:gtype:relating}
We have given LTS semantics for both global types
(\cref{def:gtype:lts-gt}) and configurations (\cref{def:mpst-env-reduction}),
we now relate these two semantics with the help of the projection operator
$\gtProj[]{}{}$
(\cref{def:global-proj}).

We associate configurations $\stEnv; \qEnv$ with global types $\gtG$ (as
annotated with a set of crashed roles $\rolesC$) by projection, written
$\stEnvAssoc{\gtWithCrashedRoles{\rolesC}{\gtG}}{\stEnv; \qEnv}{\rolesR}$.
Naturally, there are two components of the association: \emph{(1)} the local
types in $\stEnv$ need to correspond to the projections of the global type
$\gtG$ and the set of crashed roles $\rolesC$; and \emph{(2)} the queues in
$\qEnv$ corresponds to the transmissions en route in the global type $\gtG$ and
also the set of crashed roles $\rolesC$.

\begin{definition}[Association of Global Types and Configurations]%
  \label{def:assoc}
  \label{def:assoc-queue}
  A configuration $\stEnv; \qEnv$ is associated with a
  (well-annotated \wrt $\rolesR$) global type
  $\gtWithCrashedRoles{\rolesC}{\gtG}$,
  written
  $\stEnvAssoc{\gtWithCrashedRoles{\rolesC}{\gtG}}{\stEnv; \qEnv}{\rolesR}$,
  iff
  \begin{enumerate}[leftmargin=*, nosep]
    \item
      \label{item:assoc:stenv}
      $\stEnv$ can be split into disjoint (possibly empty) sub-contexts
      $\stEnv = \stEnv[\gtG] \stEnvComp \stEnv[\stStopSym] \stEnvComp
      \stEnv[\stEnd]$ where:
      \begin{enumerate}[label={(A\arabic*)}, leftmargin=*, ref={(A\arabic*)},
        nosep]
        \item\label{item:assoc:alive-sub}
          $\stEnv[\gtG]$\; contains projections of $\gtG$:
          $ \dom{\stEnv[\gtG]}
            =
            \gtRoles{\gtG}
          $, and
          $ \forall \roleP \in \dom{\stEnv[\gtG]}:
            \stEnvApp{\stEnv}{{\roleP}}
            \stSub
            \gtProj[\rolesR]{\gtG}{\roleP}
          $;
        \item\label{item:assoc:crash-stop}
          $\stEnv[\stStopSym]$ contains crashed roles:
          $ \dom{\stEnv[\stStopSym]} =
            \rolesC$,
            and
            $ \forall \roleP \in \dom{\stEnv[\stStopSym]}:
            \stEnvApp{\stEnv}{{\roleP}} = \stStop$;
        \item\label{item:assoc:end}
          $\stEnv[\stEnd]$ contains only $\stEnd$ endpoints:
          $\forall \roleP \in \stEnv[\stEnd]: \stEnvApp{\stEnv}{\roleP}
          = \stEnd$.
        \end{enumerate}
    \item
        \label{item:assoc:qenv}
        \begin{enumerate}[label={(A\arabic*)}, leftmargin=*, ref={(A\arabic*)},
          resume, nosep]
        \item\label{item:assoc:queue}
          $\qEnv$ is associated with global type
          $\gtWithCrashedRoles{\rolesC}{\gtG}$, given as follows:
          \begin{enumerate}[leftmargin=0pt, nosep]
            \item
              Receiving queues for a role is unavailable if and only if
              it has crashed:
              $\forall \roleQ : \roleQ \in \rolesC \iff
              \stEnvApp{\qEnv}{\cdot, \roleQ} = \stQUnavail$;
            \item
              If \;$\gtG = \gtEnd$\; or \;$\gtG = \gtRec{\gtRecVar}{\gtGi}$,\;
              then queues between all roles are empty (except receiving queue 
              for crashed roles):
              \;$\forall \roleP, \roleQ: \roleQ \notin \rolesC \implies
              \stEnvApp{\qEnv}{\roleP, \roleQ} = \stQEmpty$;
            \item
              If
              $\gtG = \gtComm{\roleP}{\roleQMaybeCrashed}{i \in I}{\gtLab[i]}{\tyGround[i]}{\gtGi[i]}$,
              or
              $\gtG =
              \gtCommTransit{\rolePMaybeCrashed}{\roleQ}{i \in I}{\gtLab[i]}{\tyGround[i]}{\gtGi[i]}{j}
              $  with $\gtLab[j] = \gtCrashLab$
              (\ie a pseudo-message is en route),
              then
              \begin{enumerate*}[label=\emph{(\roman*)}]
                \item
                  if $\roleQ$ is live, then the queue from $\roleP$ to $\roleQ$ is
                  empty:
                  $\roleQMaybeCrashed \neq \roleQCrashed \implies
                  \stEnvApp{\qEnv}{\roleP, \roleQ} = \stQEmpty$, and
                \item $\forall i \in I: \qEnv$  is associated with
                  $\gtWithCrashedRoles{\rolesC}{\gtGi[i]}$;
              \end{enumerate*}
            \item
              If
               $\gtG =
              \gtCommTransit{\rolePMaybeCrashed}{\roleQ}{i \in I}{\gtLab[i]}{\tyGround[i]}{\gtGi[i]}{j}
              $  with  $\gtLab[j] \neq \gtCrashLab$,
              then
              \begin{enumerate*}[label=\emph{(\roman*)}]
                \item
                  the queue from $\roleP$ to $\roleQ$ begins with the message
                  $\stQMsg{\gtLab[j]}{\tyGround[j]}$:
                  $\stEnvApp{\qEnv}{\roleP, \roleQ} =
                  \stQCons{\stQMsg{\gtLab[j]}{\tyGround[j]}}{\stQ}$;
                \item
                  $\forall i \in I:$ removing the message from the head of the
                  queue,
                  $\stEnvUpd{\qEnv}{\roleP, \roleQ}{\stQ}$ is associated with
                  $\gtWithCrashedRoles{\rolesC}{\gtGi[i]}$.
              \end{enumerate*}
        \end{enumerate}
      \end{enumerate}
  \end{enumerate}
  We write $\stEnvAssoc{\gtG}{\stEnv}{\rolesR}$ as an abbreviation of
  $\stEnvAssoc{\gtWithCrashedRoles{\emptyset}{\gtG}}{\stEnv; \qEnv[\emptyset]}{\rolesR}$.
  We sometimes say
  $\stEnv$ (resp.~$\qEnv$) is associated with
  $\gtWithCrashedRoles{\rolesC}{\gtG}$ for stating \cref{item:assoc:stenv}
  (resp.~\cref{item:assoc:qenv}) is
  satisfied.
\end{definition}

We demonstrate %
the relation between the two semantics via association, by showing two
main theorems:
all possible reductions of a configuration have a corresponding action in
reductions of the associated global type (\cref{thm:gtype:proj-comp});
and the reducibility of a global type is the same as its associated
configuration (\cref{thm:gtype:proj-sound}). %

\begin{restatable}[Completeness of Association]{theorem}{thmProjCompleteness}%
  \label{thm:gtype:proj-comp}
  Given associated global type $\gtG$ and configuration $\stEnv; \qEnv$:
  $\stEnvAssoc{\gtWithCrashedRoles{\rolesC}{\gtG}}{\stEnv; \qEnv}{\rolesR}$.
  If $\stEnv; \qEnv \stEnvMoveGenAnnot \stEnvi; \qEnvi$,
  where $\stEnvAnnotGenericSym \neq \ltsCrash{\mpS}{\roleP}$
  for all $\roleP \in \rolesR$,
  then there exists $\gtWithCrashedRoles{\rolesCi}{\gtGi}$ such that
  $\stEnvAssoc{\gtWithCrashedRoles{\rolesCi}{\gtGi}}{\stEnvi;
  \qEnvi}{\rolesR}$
  and
  $\gtWithCrashedRoles{\rolesC}{\gtG} \gtMove[\stEnvAnnotGenericSym]{\rolesR}
    \gtWithCrashedRoles{\rolesCi}{\gtGi}$.
\end{restatable}
\iftoggle{full}{
\begin{proof}
  By induction on configuration reductions (\cref{def:mpst-env-reduction}).
See \cref{sec:proof:relating} for detailed proof.     \qedhere       
\end{proof}
}{}
\begin{restatable}[Soundness of Association]{theorem}{thmProjSoundness}%
  \label{thm:gtype:proj-sound}
  Given associated global type $\gtG$ and configuration $\stEnv; \qEnv$:
  $\stEnvAssoc{\gtWithCrashedRoles{\rolesC}{\gtG}}{\stEnv; \qEnv}{\rolesR}$.
  If
   $\gtWithCrashedRoles{\rolesC}{\gtG} \gtMove[]{\rolesR}$,
  then there exists  $\stEnvi; \qEnvi$,  $\stEnvAnnotGenericSym$  and
  $\gtWithCrashedRoles{\rolesCi}{\gtGi}$,  such that
  $\gtWithCrashedRoles{\rolesC}{\gtG} \gtMove[\stEnvAnnotGenericSym]{\rolesR}
   \gtWithCrashedRoles{\rolesCi}{\gtGi}$,
  $\stEnvAssoc{\gtWithCrashedRoles{\rolesCi}{\gtGi}}{\stEnvi;
  \qEnvi}{\rolesR}$,
  and
  $\stEnv; \qEnv \stEnvMoveGenAnnot \stEnvi; \qEnvi$.
\end{restatable}
\iftoggle{full}{
\begin{proof}
 By induction on global type reductions (\cref{def:gtype:lts-gt}). 
 See \cref{sec:proof:relating} for detailed proof.    \qedhere         
 \end{proof}
}{}

By~\Cref{thm:gtype:proj-comp,thm:gtype:proj-sound}, we obtain, as a corollary, that
a global type $\gtG$ is in operational correspondence with the typing context
$\stEnv = \setenum{\stEnvMap{\roleP}{\gtProj[\rolesR]{\gtG}{\roleP}}}_{\roleP \in \gtRoles{\gtG}}$,
which contains the projections of all roles in $\gtG$.

\iftoggle{full}{
\begin{remark}[`Weakness' and `Sufficiency' of Soundness Theorem]
  Curious readers may wonder why we proved a `\emph{weak}' soundness theorem
  instead of one that is the dual of the completeness theorem,
  \eg as seen in the literature~\cite{ICALP13CFSM}.
  The reason is that we use the `full' subtyping (\cref{def:subtyping}, notably
  \inferrule{\iruleStSubOut}).
  A local type in the typing context may have fewer branches to choose from
  than the projected local type,
  resulting in uninhabited sending actions in the global type.

  For example, let $\gtG = \gtCommRaw{\roleP}{\roleQ}{
    \gtCommChoice{\gtLab[1]}{}{\gtEnd}; \; \;
    \gtCommChoice{\gtLab[2]}{}{\gtEnd}
  }$.
  An associated typing context $\stEnv$ (assuming $\roleP$ reliable) may have
  $\stEnvApp{\stEnv}{\roleP} =
  \stIntSum{\roleQ}{}{\stChoice{\stLab[1]}{}
  \stSeq {\stEnd}}
  \stSub
  \stIntSum{\roleQ}{}{
    \stChoice{\stLab[1]}{} \stSeq {\stEnd}; \; \;
    \stChoice{\stLab[2]}{} \stSeq {\stEnd}
  }$ (via \inferrule{\iruleStSubOut}).
  The global type $\gtG$ may make a transition
  $\stEnvOutAnnot{\roleP}{\roleQ}{\gtLab[2]}$, where an associated
  configuration $\stEnv; \qEnv[\emptyset]$ cannot.

  However, our soundness theorem is \emph{sufficient} for concluding
  desired properties guaranteed via association, \eg safety, deadlock-freedom, and liveness,
  as illustrated in~\cref{sec:gtype:pbp}.
\end{remark}
}{}

\subsection{Properties Guaranteed by Projection}\label{sec:gtype:pbp}

A key benefit of our top-down approach of multiparty protocol design is that
desirable properties are guaranteed by the methodology.
As a consequence, processes
following the local types obtained from projections are correct \emph{by
construction}.
In this subsection, we focus on three properties: \emph{communication safety},
\emph{deadlock-freedom},  and \emph{liveness}, and show that the
three properties are guaranteed from
configurations associated with global types.
\subparagraph*{Communication Safety}
\label{sec:type-system-safety}

We begin by defining communication safety for configurations
(\cref{def:mpst-env-safe}).
We focus on two safety requirements:
\begin{enumerate*}[label=(\roman*)]
  \item each role must be able to handle any message that may end up in their
    receiving queue (so that there are no label mismatches); and
 \item each receiver must be able to handle the potential crash of the
    sender, unless the sender is reliable.
\end{enumerate*}

\begin{definition}[Configuration Safety]\label{def:mpst-env-safe}%
  Given a fixed set of reliable roles $\rolesR$, we say that
  $\predP$ is an \emph{$\rolesR$-safety property} of configurations %
  iff, whenever $\predPApp{\stEnv; \qEnv}$, we have:

  \noindent%
  \begin{tabular}{@{}r@{\hskip 2mm}l}
    \inferrule{\iruleSafeComm}%
    &%
    $ \stEnvApp{\stEnv}{\roleQ} =
        \stExtSum{\roleP}{i \in I}{\stChoice{\stLab[i]}{\tyGround[i]} \stSeq \stSi[i]}%
    $
     and
    $\stEnvApp{\qEnv}{\roleP, \roleQ} \neq \stQUnavail
    $
     and
    $
      \stEnvApp{\qEnv}{\roleP, \roleQ} \neq \stQEmpty
    $
     implies  %
    $\stEnvMoveAnnotP{\stEnv;
    \qEnv}{\stEnvInAnnot{\roleQ}{\roleP}{\stChoice{\stLabi}{\tyGroundi}}}$;
    \\%
    \inferrule{\iruleSafeCrash}%
    &%
    $\stEnvApp{\stEnv}{\roleP} = \stStop$
     and
    $ \stEnvApp{\stEnv}{\roleQ} =
        \stExtSum{\roleP}{i \in I}{\stChoice{\stLab[i]}{\stS[i]} \stSeq \stSi[i]}%
    $
     and
    $
      \stEnvApp{\qEnv}{\roleP, \roleQ} = \stQEmpty
    $
     implies  %
    $\stEnvMoveAnnotP{\stEnv; \qEnv}{\ltsCrDe{\mpS}{\roleQ}{\roleP}}$;
    \\[1mm]
    \inferrule{\iruleSafeRec}%
    &%
    $
      \stEnvApp{\stEnv}{%
        \roleP%
      } =
      \stRec{\stRecVar}{\stS}%
    $ %
     implies  %
    $\predPApp{%
      \stEnvUpd{\stEnv}{%
        \roleP%
      }{%
        \stS\subst{\stRecVar}{\stRec{\stRecVar}{\stS}}%
      }; \qEnv%
    }$;
    \\[1mm]
    \inferrule{\iruleSafeMove}%
    &%
    $\stEnv; \qEnv \stEnvMoveMaybeCrash[\rolesR] \stEnvi; \qEnvi$
     implies  %
    $\predPApp{\stEnvi; \qEnvi}$.
  \end{tabular}

\smallskip
\noindent%
  We say \emph{\;$\stEnv; \qEnv$\; is $\rolesR$-safe}, %
  if \;$\predPApp{\stEnv; \qEnv}$\; holds %
  for some $\rolesR$-safety property $\predP$. %
\end{definition}

We use a coinductive view of the safety property~\cite{SangiorgiBiSimCoInd},
where the predicate of $\rolesR$-safe configurations is the %
largest $\rolesR$-safety property, by taking the union of all safety properties
$\predP$.
For a configuration $\stEnv; \qEnv$ to be $\rolesR$-safe, it has to satisfy
all clauses defined in \cref{def:mpst-env-safe}.

By clause~\inferrule{\iruleSafeComm}, whenever a role $\roleQ$
receives from another role $\roleP$,
and a message is present in the queue,
the receiving action must be possible for some label $\gtLabi$.
Clause~\inferrule{\iruleSafeCrash} states that if a role $\roleQ$ receives
from a crashed role $\roleP$, and there is nothing in the queue,
then $\roleQ$ must have a $\stCrashLab$ branch, and a crash detection action
can be fired.
(Note that $\inferrule{\iruleSafeComm}$ applies when the queue is non-empty,
despite the crash of sender $\roleP$.)
Finally,
clause~\inferrule{\iruleSafeRec} extends the previous clauses
by unfolding any recursive entries; and
clause \inferrule{\iruleSafeMove}
states that any configuration $\stEnvi; \qEnvi$ which $\stEnv; \qEnv$ transitions to
must also be $\rolesR$-safe.
By using transition $\stEnvMoveMaybeCrash[\rolesR]$, we
ignore crash transitions $\ltsCrash{\mpS}{\roleP}$ for
any reliable role $\roleP \in \rolesR$.%

\begin{example}
\label{ex:configuration_safety}
Recall the local types $\stT[\roleFmt{C}]$, $\stT[\roleFmt{L}]$, and $\stT[\roleFmt{I}]$
of the Simpler Logging example in~\cref{sec:overview}.
The configuration $\stEnv; \qEnv$, where
{\small $\stEnv =  \stEnvMap{\roleFmt{C}}{\stT[\roleFmt{C}]}  \stEnvComp
\stEnvMap{\roleFmt{L}}{\stT[\roleFmt{L}]}   \stEnvComp
\stEnvMap{\roleFmt{I}}{\stT[\roleFmt{I}]}$} and
{\small $\qEnv =  \qEnv[\emptyset]$}, is
{\small $\setenum{\roleFmt{L}, \roleFmt{I}}$}-safe.
This can be verified by checking its
reductions.  For example,
in the case where $\roleFmt{C}$ crashes immediately, we have:
{\small $\stEnv; \qEnv
\stEnvMoveAnnot{\ltsCrash{\mpS}{\roleFmt{C}}}
 \stEnvUpd{\stEnv}{\roleFmt{C}}{
\stStop};
       \stEnvUpd{\qEnv}{\cdot, \roleFmt{C}}{\stQUnavail}
       \stEnvMoveStar
    \stEnvUpd{\stEnvUpd{\stEnvUpd{\stEnv}{\roleFmt{C}}{\stStop}}{
 \roleFmt{L}}{
  \stEnd
  }}{\roleFmt{I}}{\stEnd};
  \stEnvUpd{\qEnv}{\cdot, \roleFmt{C}}{\stQUnavail}
  $}
and each reductum satisfies all clauses of~\Cref{def:mpst-env-safe}.
\iftoggle{full}{Full reductions are available in~\cref{sec:app_examples},~\cref{ex:configuration_safety_full}.}{}
\end{example}

\iffalse
We show in \cref{lem:ext-proj-safe} that any configuration associated (\wrt
reliable roles $\rolesR$) with a global type is $\rolesR$-safe.
\iftoggle{full}{The proof of~\cref{lem:safety-by-proj}
is available in~\cref{sec:proof:safety}.}{}

\begin{restatable}{lemma}{lemProjSafe}
\label{lem:safety-by-proj}
  If $\stEnvAssoc{\gtWithCrashedRoles{\rolesC}{\gtG}}{\stEnv; \qEnv}{\rolesR}$,
  then $\stEnv; \qEnv$ is  $\rolesR$-safe.%
  \label{lem:ext-proj-safe}
\end{restatable}

\begin{restatable}[Safety by Projection]{corollary}{colProjSafe}%
  \label{cor:typing:safe-proj}
  Let $\gtG$ be a global type without runtime constructs,
  and $\rolesR$ be a set of reliable roles.
  If  $\stEnv$ is a typing context
  associated to the global type $\gtG$:
  $\stEnvAssoc{\gtG}{\stEnv}{\rolesR}$,
  then $\stEnv; \qEnv[\emptyset]$ is $\rolesR$-safe.
\end{restatable}
\fi

\subparagraph*{Deadlock-Freedom}
\label{sec:type-system-deadlock-free}

The property of deadlock-freedom, sometimes also known as progress, describes
whether a configuration can keep reducing unless it is a terminal configuration.
We give its formal definition in \cref{def:mpst-env-deadlock-free}.

\begin{definition}[Configuration Deadlock-Freedom]\label{def:mpst-env-deadlock-free}%
Given a set of reliable roles $\rolesR$, we say that a configuration $\stEnv;
\qEnv$ is
\emph{$\rolesR$-deadlock-free} iff:
\begin{enumerate*}
  \item $\stEnv; \qEnv$ is $\rolesR$-safe; and,
  \item
    \label{item:df:reduces}
    If $\stEnv; \qEnv$ can reduce to a configuration $\stEnvi; \qEnvi$
    without further reductions:
    $\stEnv; \qEnv \!\stEnvMoveMaybeCrashStar[\rolesR]\! \stEnvi; \qEnvi
    \!\not\stEnvMoveMaybeCrash[\rolesR]$,  then:
    \begin{enumerate*}
    \item
      \label{item:df:context}
      $\stEnvi$ can be split into two disjoint contexts,
      one with only $\stEnd$
      entries, and one with only $\stStop$ entries:
      $\stEnvi =  \stEnvi[\stEnd] \stEnvComp \stEnvi[\stStopSym]$,  where
      $\dom{\stEnvi[\stEnd]} =
      \setcomp{\roleP}{\stEnvApp{\stEnvi}{\roleP} = \stEnd}$  and
      $\dom{\stEnvi[\stStopSym]} =
      \setcomp{\roleP}{\stEnvApp{\stEnvi}{\roleP} = \stStop}$; and,
    \item
      \label{item:df:queues}
      $\qEnvi$ is empty for all pairs of roles, except for the receiving queues
      of crashed roles, which are unavailable:
      $\forall \roleP, \roleQ:  \stEnvApp{\qEnvi}{\cdot, \roleQ} =
      \stQUnavail$
      if
       $\stEnvApp{\stEnvi}{\roleQ} = \stStop$, and
       $\stEnvApp{\qEnvi}{\roleP, \roleQ} = \stQEmpty$, otherwise.
    \end{enumerate*}
\end{enumerate*}
\end{definition}

It is worth noting that a (safe) configuration that reduces infinitely
satisfies deadlock-freedom, as \cref{item:df:reduces} in the premise does not hold.
Otherwise, whenever a terminal configuration is reached, it must satisfy
\cref{item:df:context} that all local types in the typing context be
terminated (either successfully $\stEnd$, or crashed $\stStop$), and
\cref{item:df:queues} that all queues be empty (unless unavailable due to
crash).
As a consequence,
a deadlock-free configuration $\stEnv; \qEnv$ either does not stop reducing, or
terminates in a stable configuration.

\iffalse
Similarly, we show in \cref{lem:ext-proj-deadlock-free} that any configuration
associated (\wrt reliable roles $\rolesR$) with a global type is
$\rolesR$-deadlock-free.
\iftoggle{full}{The proof of~\cref{lem:proj:df}
is available in~\cref{sec:proof:deadlockfree}.}{}

\begin{restatable}{lemma}{lemProjDeadlockFree}
\label{lem:proj:df}
  If $\stEnvAssoc{\gtWithCrashedRoles{\rolesC}{\gtG}}{\stEnv; \qEnv}{\rolesR}$,
  then $\stEnv; \qEnv$ is $\rolesR$-deadlock-free.%
  \label{lem:ext-proj-deadlock-free}
\end{restatable}

\begin{restatable}[Deadlock-Freedom by Projection]{corollary}{colProjDf}%
  \label{cor:df}
  Let $\gtG$ be a global type without runtime constructs,
  and $\rolesR$ be a set of reliable roles.
  If $\stEnv$ is a typing context
  associated to the global type $\gtG$:
  $\stEnvAssoc{\gtG}{\stEnv}{\rolesR}$,
  then $\stEnv; \qEnv[\emptyset]$ is $\rolesR$-deadlock-free.
\end{restatable}
\fi
\subparagraph*{Liveness}
\label{sec:type-system-live}
The property of liveness describes that every pending output/external choice
is eventually triggered by means of a message transmission or crash detection.
Our liveness property is based on \emph{fairness},
which guarantees that every enabled message transmission, including crash detection,
is performed successfully.
We give the definitions of non-crashing, fair,  and live paths of configurations respectively in
\Cref{def:non-crash-fair-live-path}, and use these paths to
formalise the liveness for configurations in~\Cref{def:mpst-env-live}.

\begin{definition}[Non-crashing, Fair, Live Paths]
\label{def:non-crash-fair-live-path}
 A \emph{non-crashing path} is a possibly infinite
sequence of configurations $(\stEnv[n]; \qEnv[n])_{n \in N}$, where
$N = \setenum{0, 1, 2, \ldots}$
is a set of consecutive natural numbers, and
$\forall n \in N$, $\stEnv[n]; \qEnv[n] \!\stEnvMove\!  \stEnv[n+1]; \qEnv[n+1]$.
We say that a non-crashing path $(\stEnv[n]; \qEnv[n])_{n \in N}$ is
\emph{fair} iff, $\forall n \in N$:
 \begin{enumerate}[label={(F\arabic*)}, leftmargin=*, ref={(F\arabic*)},nosep]
 \item
 \label{item:fairness_send}
 $\stEnv[n]; \qEnv[n] \stEnvMoveOutAnnot{\roleP}{\roleQ}{\stChoice{\stLab}{\tyGround}}$
 implies  $\exists k, \stLabi, \tyGroundi$ such that $n \leq k \in N$
 and
  $\stEnv[k]; \qEnv[k] \stEnvMoveOutAnnot{\roleP}{\roleQ}{\stChoice{\stLabi}{\tyGroundi}} \stEnv[k+1]; \qEnv[k+1]$;
 \item
 \label{item:fairness_receive}
 $\stEnv[n]; \qEnv[n] \stEnvMoveInAnnot{\roleP}{\roleQ}{\stChoice{\stLab}{\tyGround}}$
 implies $\exists k$ such that $n \leq k \in N$
 and
 $\stEnv[k]; \qEnv[k] \stEnvMoveInAnnot{\roleP}{\roleQ}{\stChoice{\stLab}{\tyGround}} \stEnv[k+1]; \qEnv[k+1]$;
 \item
 \label{item:fairness_crash_detection}
 $\stEnv[n]; \qEnv[n] \stEnvMoveAnnot{\ltsCrDe{\mpS}{\roleP}{\roleQ}}$
 implies $\exists k$ such that $n \leq k \in N$
 and
  $\stEnv[k]; \qEnv[k] \stEnvMoveAnnot{\ltsCrDe{\mpS}{\roleP}{\roleQ}}\stEnv[k+1]; \qEnv[k+1]$.
 \end{enumerate}

We say that a non-crashing path $(\stEnv[n]; \qEnv[n])_{n \in N}$
is \emph{live} iff, $\forall n \in N$:
 \begin{enumerate}[label={(L\arabic*)}, leftmargin=*, ref={(L\arabic*)},nosep]
 \item
 \label{item:liveness_consume}
 $\stEnvApp{\qEnv[n]}{\roleP, \roleQ}
      =
      \stQCons{\stQMsg{\stLab}{\tyGround}}{\stQ}
      \neq
      \stQUnavail$\; and \;$\stLab \neq \stCrashLab$
 implies $\exists k$ such that $n \leq k \in N$
 and
 $\stEnv[k]; \qEnv[k] \stEnvMoveInAnnot{\roleQ}{\roleP}{\stChoice{\stLab}{\tyGround}} \stEnv[k+1]; \qEnv[k+1]$;
\item
\label{item:liveness_receiving_cd}
$\stEnvApp{\stEnv[n]}{%
        \roleP%
      } =
      \stExtSum{\roleQ}{i \in I}{\stChoice{\stLab[i]}{\tyGround[i]} \stSeq \stT[i]}$     %
 implies  $\exists k, \stLabi, \tyGroundi$ such that $n \leq k \in N$  and\\
 $\stEnv[k]; \qEnv[k]
 \stEnvMoveInAnnot{\roleP}{\roleQ}{\stChoice{\stLabi}{\tyGroundi}} \stEnv[k+1];
 \qEnv[k+1]$
 or
 $\stEnv[k]; \qEnv[k]
      \stEnvMoveAnnot{\ltsCrDe{\mpS}{\roleP}{\roleQ}}%
      \stEnv[k+1]; \qEnv[k+1]$.
 \end{enumerate}
 \end{definition}
A non-crash path is a (possibly infinite) sequence of reductions of
a configuration without crashes. A non-crash path is fair if along the path, every internal choice eventually
sends a message~\ref{item:fairness_send}, every external choice eventually receives a message~\ref{item:fairness_receive},
and every crash detection is eventually performed~\ref{item:fairness_crash_detection}.  A non-crashing path is live if
along the path, every non-crash message in the queue is eventually consumed~\ref{item:liveness_consume}, and every
hanging external choice eventually consumes a message or performs a crash detection~\ref{item:liveness_receiving_cd}.

 \begin{definition}[Configuration Liveness]
 \label{def:mpst-env-live}%
Given a set of reliable roles $\rolesR$, we say that a configuration $\stEnv;
\qEnv$ is
\emph{$\rolesR$-live} iff:
\begin{enumerate*}[leftmargin=*, nosep]
  \item $\stEnv; \qEnv$ is $\rolesR$-safe; and,
  \item
    \label{item:live:reduces}
  $\stEnv; \qEnv \stEnvMoveMaybeCrashStar[\rolesR]
  \stEnvi; \qEnvi$  implies all non-crashing paths starting with
  $\stEnvi; \qEnvi$ that are fair are also live.
   \end{enumerate*}
\end{definition}

A configuration $\stEnv; \qEnv$ is $\rolesR$-live when it is $\rolesR$-safe and
any reductum of $\stEnv; \qEnv$ (via transition $\stEnvMoveMaybeCrashStar[\rolesR]$)
consistently leads to a live path if it is fair. 

\iffalse
We show in~\Cref{lem:ext-proj-live} that any
configuration associated (\wrt reliable roles $\rolesR$) with
a global type is $\rolesR$-live. %
\iftoggle{full}{The proofs of~\cref{lem:ext-proj-live}
is available in~\cref{sec:proof:safe, sec:proof:df, sec:proof:liveness}.}{}

\begin{restatable}{lemma}{lemProjLive}
 \label{lem:ext-proj-live}
  If $\stEnvAssoc{\gtWithCrashedRoles{\rolesC}{\gtG}}{\stEnv; \qEnv}{\rolesR}$,
  then $\stEnv; \qEnv$ is $\rolesR$-live.%
\end{restatable}
\fi

\subparagraph*{Properties by Projection}
We conclude by showing the guarantee of safety, 
deadlock-freedom, and liveness in configurations 
associated with global types in~\cref{lem:ext-proj}. 
Furthermore, as a corollary, \Cref{cor:allproperties} demonstrates that
a typing context projected from a global type (without runtime
constructs) is inherently safe, deadlock-free, and live by construction.  %
\iftoggle{full}{The detailed proofs for~\Cref{lem:ext-proj,cor:allproperties}
are available in~\cref{sec:proof:propbyproj}.}{}

\begin{restatable}{lemma}{lemProj}
 \label{lem:ext-proj}
  If $\stEnvAssoc{\gtWithCrashedRoles{\rolesC}{\gtG}}{\stEnv; \qEnv}{\rolesR}$,
  then $\stEnv; \qEnv$ is $\rolesR$-safe, $\rolesR$-deadlock-free, and $\rolesR$-live.%
\end{restatable}

\begin{restatable}[Safety, Deadlock-Freedom, and Liveness by Projection]{theorem}{colProjAll}%
  \label{cor:allproperties}
  Let $\gtG$ be a global type without runtime constructs,
  and $\rolesR$ be a set of reliable roles.
  If $\stEnv$ is a typing context
  associated with the global type $\gtG$:
  $\stEnvAssoc{\gtG}{\stEnv}{\rolesR}$,
  then $\stEnv; \qEnv[\emptyset]$ is $\rolesR$-safe, $\rolesR$-deadlock-free, and $\rolesR$-live.
\end{restatable}

\section{Typing System with Crash-Stop Semantics}
\label{sec:typing_system}
In this section, we present a type system for our asynchronous multiparty session
calculus.
Our typing system is extended from the one in~\cite{POPL21AsyncMPSTSubtyping}
with
crash-stop failures. %
We introduce the typing rules in \cref{sec:typingrules}, and show various
properties of typed sessions: subject reduction, session
fidelity, deadlock-freedom,
and liveness in \Cref{sec:type:system:results}.

\subsection{Typing Rules}
\label{sec:typingrules}

\begin{figure}
\[
{\small
\begin{array}{@{}c@{}}
  \inference[{t-$\mpQEmpty$}]{}{\vdash \mpQEmpty : \stQEmpty}
  \qquad
  \highlight{\inference[{t-$\mpQUnavail$}]{ }{\vdash \mpQUnavail : \stQUnavail}}
\qquad
  \inference[{t-$\cdot$}]{
    \vdash \mpH_1:\qEnvPartial[1] \quad \vdash \mpH_2:\qEnvPartial[2]
  }{
    \vdash \mpH_1 \cdot \mpH_2: \qEnvPartial[1] \cdot \qEnvPartial[2]
  }
\\[1mm]
    \inference[{t-msg}]{
    \vdash \mpV:\tyGround
    &
    \stEnvApp{\qEnvPartial}{\roleQ} = \stQMsg\stLab\tyGround
    &
    \forall \roleR \neq \roleQ:
    \stEnvApp{\qEnvPartial}{\roleR} = \stQEmpty
  }{
    \vdash (\roleQ , \mpLab(\mpV)) : \qEnvPartial
  }
\\[1mm]
\highlight{\inference[{t-$\mpCrash$}]{$\;$}{\Theta \vdash \mpCrash : \stStop}}
  \qquad
\inference[{t-$\mpNil$}]{$\;$}{\Theta \vdash \mpNil : \stEnd}
  \qquad
\inference[{t-out}]{
  \Theta \vdash \mpE:\tyGround
  \quad
  \Theta \vdash \mpP:\stT
}{
  \Theta  \vdash \procout\roleQ{\mpLab}{\mpE}{\mpP}:
  \stOut{\roleQ}{\stLab}{\tyGround}\stSeq{\stT}
}
  \\[1mm]%
\inference[{t-ext}]{
  \forall i\in I\;\;\; \Theta, x_i:\tyGround_i \vdash \mpP_i:\stT_i
}{
  \Theta \vdash  \sum_{i\in I}\procin{\roleQ}{\mpLab_i(\mpx_i)}{\mpP_i}:
  \stExtSum{\roleQ}{i\in I}{\stChoice{\stLab[i]}{\tyGround_i}\stSeq{\stT_i}}
}
  \qquad
\inference[{t-cond}]{\Theta \vdash \mpE:\tyBool
\quad \Theta  \vdash \mpP_i:\stT \ \text{\tiny $(i=1,2)$}
}{\Theta \vdash \mpIf\mpE{\mpP_1}{\mpP_2}:\stT}
\\[1mm]
\inference[{t-rec}]{\Theta, X:\stT  \vdash \mpP:\stT}
{\Theta  \vdash  \mu X.\mpP: \stT}
\qquad
\inference[{t-var}]{$\;$}{\Theta, X:\stT  \vdash X:\stT}
\qquad
\inference[{t-sub}]{\Theta \vdash \mpP:\stT \quad \stT\stSub \stT' }{\Theta \vdash \mpP:\stT'}
\\[1mm]
\inference[{t-sess}]{
   \highlight{\stEnvAssoc{\gtWithCrashedRoles{\rolesC}{\gtG}}{\stEnv; \qEnv}{\rolesR}}
  \qquad
  \forall i\in I %
  \quad \highlight{\vdash \mpP_i:\stEnvApp{\stEnv}{\roleP[i]}} \qquad  %
  \highlight{\vdash \mpH[i]: \stEnvApp{\qEnv}{-, \roleP[i]}} %
  &
  \highlight{\dom{\stEnv} \subseteq
  \setcomp{\roleP[i]}{i \in I}}
}{
  \gtWithCrashedRoles{\rolesC}{\gtG}
  \vdash \prod_{i\in I} (\mpPart{\roleP[i]}{\mpP[i]} \mpPar
  \mpPart{\roleP[i]}{\mpH[i]})
}
\end{array}
}
\]
\caption{
  Typing rules for queues,
  processes, and sessions.%
}%
\label{fig:processes:typesystem}
\end{figure}

Our type system uses three kinds of typing judgements:  \emph{(1)} for processes;
\emph{(2)} for queues;
and \emph{(3)} for sessions,
and is defined inductively by the typing rules in~\cref{fig:processes:typesystem}.
Typing judgments for processes are of form $\Theta \vdash \mpP : \stT$, where
$\Theta$ is a typing context for variables, defined as $\Theta \bnfdef
\emptyset \bnfsep \Theta , \mpx : \tyGround \bnfsep \Theta , \mpX : \stT$.

With regard to queues, we use judgments of the form $\vdash \mpH : \qEnvPartial$,
where we use $\qEnvPartial$ to denote a partially applied queue lookup function.
We write $\qEnvPartial = \stEnvApp{\qEnv}{-, \roleP}$ to describe the incoming
queue for a role $\roleP$, as a partially
applied function $\qEnvPartial = \stEnvApp{\qEnv}{-, \roleP}$ such that
$\stEnvApp{\qEnvPartial}{\roleQ} = \stEnvApp{\qEnv}{\roleQ,\roleP}$.
We write $\qEnvPartial[1] \cdot \qEnvPartial[2]$ to denote the point-wise
application of concatenation.
For empty queues ($\mpQEmpty$), unavailable queues ($\mpQUnavail$), and queue
concatenations ($\cdot$), we simply lift the process-level queue constructs
to type-level counterparts.
For a singleton message $(\roleQ, \mpLab(\mpV))$, the appropriate partial queue
$\qEnvPartial$ would be a singleton of $\stQMsg{\stLab}{\tyGround}$ (where
$\tyGround$ is the type of $\mpV$) for $\roleQ$, and an empty queue
($\stQEmpty$) for any other role.

Finally, we use judgments of the form $\gtWithCrashedRoles{\rolesC}{\gtG} \vdash
\mpM$ for sessions.
We use a global type-guided judgment, effectively asserting that all
participants in the session respect the prescribed global type, as is the case
in~\cite{JLAMP19SyncSubtyping}.
As $\highlight{$highlighted$}$, the global type with crashed roles 
$\gtWithCrashedRoles{\rolesC}{\gtG}$ must
have some associated configuration $\stEnv; \qEnv$, used to type the
processes and the queues respectively.
Moreover, all the entries in the configuration must be present in the session.

Most rules in~\cref{fig:processes:typesystem}
assign the corresponding session type according to the behaviour of the
process.
For example, ($\highlight{\text{highlighted}}$)
rule \inferrule{t-$\mpQUnavail$}
assigns the unavailable queue type $\stQUnavail$
 to a unavailable queue
$\mpQUnavail$;
rules $\inferrule{t-out}$ and $\inferrule{t-ext}$ assign internal
and external choice types to input and output processes;
($\highlight{\text{highlighted}}$) rule
\inferrule{t-$\mpCrash$} (resp.~\inferrule{t-$\mpNil$})
assigns the crash termination
$\stStop$ (resp.~successful termination $\stEnd$)
to a crashed process $\mpCrash$ (resp.~inactive process $\mpNil$).

\begin{example}
\label{ex:typing_system}
Consider the process that acts as the role $\roleFmt{C}$ in
our Simpler Logging example (\Cref{sec:overview} and \Cref{ex:configuration_safety}):
{\small $\mpP[\roleFmt{C}]  =
\procoutNoVal{\roleFmt{I}}{\labFmt{read}}{\procin{\roleFmt{I}}{\labFmt{report}(\mpx)}{\mpNil}}$},
and a message queue {\small $\mpH[\roleFmt{C}] = \mpQEmpty$}.
Process $\mpP[\roleFmt{C}]$
has the type $\stT[\roleFmt{C}]$,
and
queue $\mpH[\roleFmt{C}]$
has the type $\stQEmpty$,
which can be verified in the standard way.
If we follow a crash reduction, \eg by the rule \inferrule{r-$\lightning$},
  the session evolves as {\small $\mpPart{\roleFmt{C}}{\mpP[\roleFmt{C}]} \mpPar
  \mpPart{\roleFmt{C}}{\mpH[\roleFmt{C}]}
  \;\redCrash{\roleP}{\rolesR}\;
\mpPart{\roleFmt{C}}{\mpCrash} \mpPar
  \mpPart{\roleFmt{C}}{\mpQUnavail}$}, where, by
 \inferrule{t-$\mpCrash$}, $\mpP[\roleFmt{C}]$ is typed by $\stStop$,  and
 $\mpH[\roleFmt{C}]$ is typed by $\stQUnavail$.  %
\iftoggle{full}{Full example is available in~\cref{sec:app_examples},~\cref{ex:typing_system_full}.}{}
  \end{example}

\subsection{Properties of Typed Sessions}
\label{sec:type:system:results}
We present the main properties of typed sessions: \emph{subject reduction} (\cref{lem:sr}),
\emph{session fidelity} (\cref{lem:sf}), \emph{deadlock-freedom} (\cref{lem:session_deadlock_free}), and
\emph{liveness} (\cref{lem:session_live}).

\emph{Subject reduction} states that well-typedness of sessions are preserved
by reduction.
In other words, a session governed by a global type continues to be
governed by a global type.

\begin{restatable}[Subject Reduction]{theorem}{lemSubjectReduction}%
  \label{lem:sr}
  If
  \;$\gtWithCrashedRoles{\rolesC}{\gtG} \vdash \mpM$\;
  and
  \;$\mpM \mpMove[\rolesR] \mpMi$,\;
  then either
  \;$\gtWithCrashedRoles{\rolesC}{\gtG} \vdash \mpMi$,\;
  or there exists \;$\gtWithCrashedRoles{\rolesCi}{\gtGi}$\; such that
  \;$\gtWithCrashedRoles{\rolesC}{\gtG} \gtMove{\rolesR}
  \gtWithCrashedRoles{\rolesCi}{\gtGi}$\;
  and
  \;$\gtWithCrashedRoles{\rolesCi}{\gtGi} \vdash \mpMi$.
\end{restatable}
\iftoggle{full}{
\begin{proof}
By induction on the derivation of $\mpM \mpMove[\rolesR] \mpMi$. 
 See \cref{sec:proof:typesystem} for detailed proof.    \qedhere 
\end{proof}}{}

\emph{Session fidelity} states the opposite implication with regard to subject reduction:
sessions respect the progress of the governing global type .

\begin{restatable}[Session Fidelity]{theorem}{lemSessionFidelity}
  \label{lem:sf}
  If
  \;$\gtWithCrashedRoles{\rolesC}{\gtG} \vdash \mpM$\; and
  \;$\gtWithCrashedRoles{\rolesC}{\gtG} \gtMove{\rolesR}$,\;
  then there exists $\mpMi$ and \;$\gtWithCrashedRoles{\rolesCi}{\gtGi}$\; such that
  \;$\gtWithCrashedRoles{\rolesC}{\gtG} \gtMove{\rolesR}
  \gtWithCrashedRoles{\rolesCi}{\gtGi}$,
  \;$\mpM \mpMoveStar[\rolesR] \mpMi$\; and
  \;$\gtWithCrashedRoles{\rolesCi}{\gtGi} \vdash \mpMi$.
\end{restatable}
\iftoggle{full}{
\begin{proof}
By induction on the derivation of $\gtWithCrashedRoles{\rolesC}{\gtG} \gtMove{\rolesR}$.
See \cref{sec:proof:typesystem} for detailed proof.      \qedhere  
\end{proof}}{}

Session \emph{deadlock-freedom} means that
the `successful' termination of a session may include crashed processes and their respective unavailable
incoming queues -- but reliable roles (which cannot crash) can only successfully terminate by reaching
inactive processes with empty incoming queues.
We formalise the definition of deadlock-free sessions in~\Cref{def:session_df} and show that a well-typed session is deadlock-free
in~\cref{lem:session_deadlock_free}. %
\iftoggle{full}{The proof of~\cref{lem:session_deadlock_free} is available in~\cref{sec:proof:typesystem}.}{}

\begin{definition}[Deadlock-Free Sessions]
\label{def:session_df}
A session $\mpM$ is \emph{deadlock-free} iff
\;$\mpM \mpMoveStar[\rolesR] \mpMi \mpNotMove[\rolesR]$\;
 implies
either \;$\mpMi \equiv \mpPart\roleP\mpNil \mpPar \mpPart\roleP \mpQEmpty$,\;
or
\;$\mpMi \equiv \mpPart\roleP\mpCrash \mpPar \mpPart\roleP \mpQUnavail$.
\end{definition}

\begin{restatable}[Session Deadlock-Freedom]{theorem}{lemSessionDF}%
\label{lem:session_deadlock_free}
If
 \;$\gtWithCrashedRoles{\rolesC}{\gtG} \vdash \mpM$,\;
 then $\mpM$ is deadlock-free.
\end{restatable}

Finally, we show that well-typed sessions guarantee the property of \emph{liveness}:
a session is \emph{live} when all its input processes will be performed eventually,
and all its queued messages will be consumed eventually.
We formalise the definition of live sessions in~\Cref{def:session_live} and conclude by showing
that a well-typed session is live
in~\cref{lem:session_live}. %
\iftoggle{full}{The proof of~\cref{lem:session_live} is available in~\cref{sec:proof:typesystem}.}{}

\begin{definition}[Live Sessions]
\label{def:session_live}
A session $\mpM$ is \emph{live} iff
\;$\mpM \mpMoveStar[\rolesR] \mpMi
\equiv \mpPart\roleP\mpP \mpPar \mpPart\roleP \mpH[\roleP] \mpPar \mpMii$\;
 implies:
 \begin{enumerate*}
 \item if \;$\mpH[\roleP] = \left(\roleQ,\mpLab(\mpV)\right)\cdot\mpHi[\roleP]$, \;then $\exists \mpPi, \mpMiii:
 \mpMi \mpMoveStar[\rolesR] \mpPart\roleP\mpPi \mpPar \mpPart\roleP \mpHi[\roleP] \mpPar \mpMiii$; and
 \item if \;$\mpP =  \sum_{i\in I}\procin{\roleQ}{\mpLab_i(\mpx_i)}{\mpP_i}$,  %
 \;then $\exists k \in I, \mpW, \mpHi[\roleP], \mpMiii:
 \mpMi \mpMoveStar[\rolesR] \mpPart\roleP \mpP_k\subst{\mpx_k}{\mpW} \mpPar \mpPart\roleP \mpHi[\roleP]
 \mpPar \mpMiii$.
 \end{enumerate*}
\end{definition}

\begin{restatable}[Session Liveness]{theorem}{lemSessionLive}%
\label{lem:session_live}
If
 \;$\gtWithCrashedRoles{\rolesC}{\gtG} \vdash \mpM$,\;
 then $\mpM$ is live.
\end{restatable}

\section{\theTool: Generating Scala Programs from Protocols}
\label{sec:impl}

In this section, we present our toolchain \theTool that implements our
extended MPST theory with crash-stop failures.
\theTool processes protocols represented in the \Scribble protocol
description language,
and generates protocol-conforming Scala code
that uses the \Effpi concurrency library.
A user specifies a multiparty protocol in \Scribble
as input, introduced in \cref{sec:impl:prelims}.
We show the style of our generated code in \cref{sec:impl:overview}, and how a
developer can use the generated code to implement multiparty protocols.
As mentioned in \cref{sec:overview},
generating channels for each process and type poses an interesting challenge,
explained in \cref{sec:impl:boringdetail}.

\subsection{Specifying a Multiparty Protocol in \Scribble}
\label{sec:impl:prelims}

The \Scribble Language~\cite{YHNN2013} is a multiparty protocol
description language that relates closely to MPST theory
(\cf \cite{FeatherweightScribble}),
and provides a programmatic way to express global types.
As an example, \cref{fig:impl:egProtocol} describes the following global type
of a simple distributed logging protocol:

\smallskip
  {\centerline{
  \( {\small{
    \gtG =
    \gtRec{\gtFmt{\mathbf{t_{0}}}}{
      \gtCommRaw{\roleFmt{u}}{\roleFmt{l}}{
        \begin{array}{@{}l@{}}
          \gtCommChoice{\gtMsgFmt{write}}{\tyString}{ \gtFmt{\mathbf{t_{0}}} }
          \gtFmt{,}\quad
          \gtCommChoice{\gtMsgFmt{read}}{}{
            \gtCommSingle{\roleFmt{l}}{\roleFmt{u}}{\gtMsgFmt{report}}
                         {\stFmtC{Log}}{ \gtFmt{\mathbf{t_{0}}} }
          }
          \gtFmt{,}\quad
          \gtCommChoice{\gtCrashLab}{}{ \gtEnd }
        \end{array}
      }
    }
    }}
    .
  \)
  }}

{\small
\begin{figure}[t]
  \begin{lstlisting}[language=Scribble,gobble=4]
    global protocol SimpleLogger(role U, reliable role L) [*\label{line:scribble:reliable}*]
    { rec t0 { choice at U { write(String) from U to L; [*\label{line:scribble:choice}*]
                             continue t0;               }
                        or { read from U to L;
                             report(Log) from L to U;
                             continue t0;               }
                        or { crash from U to L;         } } } [*\label{line:scribble:crash}*]
  \end{lstlisting}
  \caption{A Simple Logger protocol in \Scribble.}
  \label{fig:impl:egProtocol}
\end{figure}
}

\smallskip
The global type is described by a \Scribble
\lstinline[language=Scribble]+global protocol+, with roles declared on
\cref{line:scribble:reliable}.
A transmission in the global type
(\eg $\gtCommRaw{\roleFmt{u}}{\roleFmt{l}}{\cdots}$)
is in the form of an interaction statement
(\eg \lstinline[language=Scribble]+... from U to L;+),
except that choice
(\ie with an index set $|I| > 1$) must be marked explicitly by a
\lstinline[language=Scribble]+choice+ construct (\cref{line:scribble:choice}).
Recursions and type variables in the global types are in the forms of
\lstinline[language=Scribble]+rec+ and \lstinline[language=Scribble]+continue+
statements, respectively.

In order to express our new theory,
we need two extensions to the language:
\begin{enumerate*}[label=\emph{(\arabic*)}]
  \item a reserved label \lstinline[language=Scribble]+crash+ to mark
    crash handling branches (\cf the special label $\gtCrashLab$ in the
    theory), \eg on \cref{line:scribble:crash}; and
  \item a \lstinline[language=Scribble]+reliable+ keyword to mark the reliable
    roles in the protocol (\cf the reliable role set $\rolesR$ in the theory).
    Roles are assumed unreliable unless declared using the
    \lstinline[language=Scribble]+reliable+ keyword, \eg \texttt{L} on
    \cref{line:scribble:reliable}.
\end{enumerate*}

\subsection{Generating \Scala Code from \Scribble Protocols}
\label{sec:impl:overview}

\subparagraph*{The \Effpi Concurrency Library}\!\!\!\!\cite{PLDI19Effpi}
provides an embedded Domain Specific Language (DSL) offering a simple
actor-based API\@.
The library utilises advanced type system features in Scala~3,
and provides both type-level and value-level constructs for processes and
channels.
In particular,
the type-level constructs reflect the behaviour of programs
(\ie processes), and thus can be used as specifications.
Following this intuition,
we generate process types that reflect local types from our theory,
as well as a tentative process implementing that type (by providing some
default values where necessary).

\begin{figure}
  {\small
  \begin{lstlisting}[language=effpi, gobble=4]
    // (i) label and payload declarations
    case class Log() // payload type
    case class Read() // label types
    case class Report(x : Log)
    case class Write(x : String) [*\hrule*]
    // (ii) recursion variable declarations
    sealed abstract class RecT0[A]() extends RecVar[A]("RecT0")
    case object RecT0 extends RecT0[Unit] [*\hrule*]
    // (iii) local type declarations
    type U[C0 <: OutChan[Read | Write], C1 <: InChan[Report]] = [*\label{line:effpi:U}*]
      Rec[RecT0,
          ( (Out[C0, Read] >>: In[C1, Report, (x0 : Report) => Loop[RecT0]]) [*\label{line:effpi:out_c0}*]
          | (Out[C0, Write] >>: Loop[RecT0]) )]

    type L[C0 <: InChan[Read | Write], C1 <: OutChan[Report]] = [*\label{line:effpi:L}*]
      Rec[RecT0,
          InErr[C0, Read | Write, (x0 : Read | Write) => L0[x0.type, C1], [*\label{line:effpi:recv-union}*]
                                  (err : Throwable) => PNil]]

    type L0[X0 <: Read | Write, C1 <: OutChan[Report]] <: Process =
      X0 match { case Read => Out[C1, Report] >>: Loop[RecT0] [*\label{line:effpi:recv-match}*]
                 case Write => Loop[RecT0] } [*\hrule*]
    // (iv) role-implementing functions
    def u(c0 : OutChan[Read | Write],
          c1 : InChan[Report]) : U[c0.type, c1.type] = {
      rec(RecT0) {
        val x0 = 0
        if (x0 == 0) {
          send(c0, new Read()) >> receive(c1) {(x1 : Report) => loop(RecT0) }
        } else {
          send(c0, new Write("")) >> loop(RecT0)
      } } }

    def l(c0 : InChan[Read | Write],
          c1 : OutChan[Report]) : L[c0.type, c1.type] =
      rec(RecT0) {
        receiveErr(c0)((x0 : Read | Write) => l0(x0, c1),
                       (err : Throwable) => nil) }

    def l0(x : Read | Write, c1 : OutChan[Report]) : L0[x.type, c1.type] =
      x match { case y : Read => send(c1, new Report(new Log())) >> loop(RecT0)
                case y : Write => loop(RecT0) } [*\hrule*]
    // (v) an entry point (main object)
    object Main {
      def main() : Unit = {
        var c0 = Channel[Read | Write]() [*\label{line:effpi:channel-create}*]
        var c1 = Channel[Report]()
        eval(par(u(c0, c1), l(c0, c1))) [*\label{line:effpi:channel-use}*]
    } }
  \end{lstlisting}
  }
  \caption{Generated \Scala code for the Simple Logger protocol in \Cref{fig:impl:egProtocol}}
  \label{fig:impl:egEffpi}
\end{figure}

\subparagraph*{Generated Code}
To illustrate our approach,
we continue with the Simple Logger example from
\cref{sec:impl:prelims},
and show the generated code in \cref{fig:impl:egEffpi}.
The generated code can be divided into five sections:
\begin{enumerate*}[label=\emph{(\roman*)}]
\item label and payload declarations,
\item recursion variable declarations,
\item local type declarations,
\item role-implementing functions, and
\item an entry point.
\end{enumerate*}

Sections \emph{(i)} and \emph{(ii)} contain boilerplate code,
where we generate type declarations for various constructs needed for
expressing local types and processes.
We draw attention to the \emph{key} sections \emph{(iii)} and \emph{(iv)},
where we generate a representation of local types for each role,
as well as a tentative process inhabiting that type.

\subparagraph*{Local Types and \Effpi Types}
We postpone the discussion about channels in \Effpi to
\cref{sec:impl:boringdetail}.
For now, we
compare the generated \Effpi type and the projected local type,
and also give a quick primer\footnotemark\ on \Effpi constructs.
The projected local types of the roles $\roleFmt{u}$ and $\roleFmt{l}$ are
shown as follows:

\footnotetext{A more detailed description of constructs can be found
in~\cite{SCALA19Effpi}.}

\smallskip
\centerline{
\(
{\small{
\begin{array}{rcl}
\gtProj[\setenum{\roleFmt{l}}]{\gtG}{\roleFmt{u}}
& = &
\stRec{\stFmt{\mathbf{t_{0}}}}{
  \stIntSum{\roleFmt{l}}{}{
    \begin{array}{@{}l@{}}
      \stChoice{\stMsgFmt{write}}{\tyString} \stSeq \stFmt{\mathbf{t_{0}}}
      \stFmt{,}\quad
      \stChoice{\stMsgFmt{read}}{}
      \stSeq
      \stInNB{\roleFmt{l}}{\stMsgFmt{report}}{\stFmtC{Log}}{
        \stFmt{\mathbf{t_{0}}}
      }

    \end{array}
  }
} \\
\gtProj[\setenum{\roleFmt{l}}]{\gtG}{\roleFmt{l}}
& = &
\stRec{\stFmt{\mathbf{t_{0}}}}{
  \stExtSum{\roleFmt{u}}{}{
    \begin{array}{@{}l@{}}
      \stChoice{\stMsgFmt{write}}{\tyString} \stSeq \stFmt{\mathbf{t_{0}}}
      \stFmt{,}\quad
      \stChoice{\stMsgFmt{read}}{}
      \stSeq
      \stOut{\roleFmt{u}}{\stMsgFmt{report}}{\stFmtC{Log}}
      \stSeq
      \stFmt{\mathbf{t_{0}}}
      \stFmt{,}\quad
      \stChoice{\stCrashLab}{} \stSeq {\stEnd}
    \end{array}
  }
}
\end{array}
}
}
\)
} %

\smallskip
\noindent
The local types are recursive,
and the \Effpi type implements recursion with
\lstinline[language=Effpi]+Rec[RecT0, ...]+
and
\lstinline[language=Effpi]+Loop[RecT0]+,
using the recursion variable \texttt{RecT0} declared in section \emph{(ii)}.

For role \roleFmt{u},
The inner local type is a sending type towards role $\roleFmt{l}$,
and we use an \Effpi process output type
\lstinline[language=Effpi]+Out[A, B]+,
which describes a process that uses a channel of type \texttt{A} to send a value
of type \texttt{B}.
For each branch, we use a separate output type, and connect it to
the type of the continuation using a
sequential composition operator (\lstinline[language=Effpi]+>>:+).
The different branches are then composed together
using a union type (\texttt{|}) from the \Scala~3 type system.

Recall that the role $\roleFmt{l}$ is declared
\lstinline[language=Scribble]+reliable+,
and thus the reception labelled \texttt{report} from \roleFmt{l} at \roleFmt{u}
does not need to contain a crash handler.
We use an \Effpi process input type
\lstinline[language=Effpi]+In[A, B, C]+,
which describes a process that uses a channel of type \texttt{A} to receive a
value of type \texttt{B}, and uses the received value in a continuation of type
\texttt{C}.

For role \roleFmt{l}, the reception type is more complex for two reasons:
\begin{enumerate*}[label=\emph{(\arabic*)}]
  \item role $\roleFmt{u}$ is unreliable, necessitating crash handling; and
  \item the reception contains branching behaviour (\cf the reception
    \roleFmt{u} being a singleton), with labels \texttt{write} and
    \texttt{read}.
\end{enumerate*}
For \emph{(1)}, we extend \Effpi with a variant of the input process type
\lstinline[language=Effpi]+InErr[A, B, C, D]+,
where \texttt{D} is the type of continuation in case of a crash.
For \emph{(2)}, the payload type is first received as an union
(\cref{line:effpi:recv-union}), and then \lstinline[language=Effpi]+match+ed
to select the
correct continuation according to the type (\cref{line:effpi:recv-match}).

\subparagraph*{From Types To Implementations}
Since \Effpi type-level and value-level constructs are closely related,
we can easily generate the processes from the processes types.
Namely, by matching
the type \lstinline[language=effpi]+Out[..., ...]+
with the process \lstinline[language=effpi]+send(..., ...)+;
the type \lstinline[language=effpi]+In[..., ...]+
with the process \lstinline[language=effpi]+receive(...) {... => ...}+;
and similarly for other constructs.
Whilst executable, the generated code represents a skeleton implementation,
and the programmer is expected to alter the code according to their
requirements.

We also introduce a new crash handling receive process
\lstinline[language=effpi]+receiveErr+,
to match the new \lstinline[language=Effpi]+InErr+ type.
Process crashes are modelled by (caught) exceptions and errors in
role-implementing functions, %
and crash detection is achieved via timeouts.
Timeouts are set by the programmer in an
(implicit) argument to each \lstinline[language=effpi]+receiveErr+ call.

Finally,
the entry point (main object) in section \emph{(v)}
composes the role-implementing functions together with
\lstinline[language=effpi]+par+ construct in \Effpi,
and connects the processes with channels.

\subsection{Generating \Effpi Channels from Scribble Protocols}
\label{sec:impl:boringdetail}

As previously mentioned,
\Effpi processes use \emph{channels} to communicate,
and the type of the channel is reflected in the type of the process.
However, our local types do not have any channels;
instead, they contain a partner role with which to communicate.
This poses an interesting challenge,
and we explain the channel generation procedure in this section.

We draw attention to the generated code in \cref{fig:impl:egEffpi} again,
where we now focus on the parameters \texttt{C0} in the generated types
\texttt{U} and \texttt{L}.
In the type \texttt{U}, the channel type \texttt{C0} needs to be a subtype of
\lstinline[language=effpi]+OutChan[Read | Write]+ (\cref{line:effpi:U}),
and we see the channel is used in the output processes types, \eg
\lstinline[language=effpi]+Out[C0, Read]+
(\cref{line:effpi:out_c0},
note that output channels subtyping is covariant on the payload type).
Dually, in the type \texttt{L}, the channel type \texttt{C0} needs to be a
subtype of
\lstinline[language=effpi]+InChan[Read | Write]+ (\cref{line:effpi:L}),
and we see the channel is used in the input process type, \ie
\lstinline[language=effpi]+InErr[C0, Read | Write, ..., ...]+
(\cref{line:effpi:recv-union}).

Similarly,
a channel \texttt{c0} is needed in the role-implementing functions
\texttt{u} and \texttt{l} as arguments,
and the channel is used in processes
\lstinline[language=effpi]+send(c0, ...)+ and
\lstinline[language=effpi]+receiveErr(c0) ...+.
Finally, in the entry point, we create a bidirectional channel
\lstinline[language=effpi]+c0 = Channel[Read | Write]()+
(\cref{line:effpi:channel-create}),
and pass it as an argument to the role-implementing functions \texttt{u} and \texttt{l}
(\cref{line:effpi:channel-use}),
so that the channel can be used to link two role-implementing processes
together for communication.

Generating the channels correctly is crucial to the correctness of our
approach,
but non-trivial since channels are implicit in the protocols.
In order to do so, a simple approach is to
traverse each interaction in the global protocol,
and assign a channel to each accordingly.

This simple approach would work for the example we show in
\cref{fig:impl:egProtocol};
however, it would not yield the correct result when \emph{merging} occurs
during projection, which we explain using an example.
For clarity and convenience, we use \emph{annotated} global and local types,
where we assign an identifier for each interaction to signify the channel to
use,
and consider the following global type:
\(
{\small
  \gtG = \gtCommRawAnn{\roleP}{\roleQ}{0}{
    \gtCommChoice{\gtLabFmt{left}}{}{
      \gtCommSingleAnn{\roleP}{\roleR}{1}{\gtLabFmt{left}}{}{\gtEnd}
    }
    \gtFmt{,} \quad
    \gtCommChoice{\gtLabFmt{right}}{}{
      \gtCommSingleAnn{\roleP}{\roleR}{2}{\gtLabFmt{right}}{}{\gtEnd}
    }
  }
  }
\).

The global type describes a simple protocol, where role $\roleP$ selects a
label \gtLabFmt{left} or \gtLabFmt{right} to $\roleQ$, and $\roleQ$ passes on
the same label to $\roleR$.
As a result, the projection on $\roleR$ (assuming all roles reliable) should
be a reception from $\roleQ$ with branches labelled $\gtLabFmt{left}$ or
$\gtLabFmt{right}$, \ie
\(
  \stExtSumAnn{\roleP}{1,2}{}{
      \stChoice{\stLabFmt{left}}{} \stSeq \stEnd
      \stFmt{,}\quad
      \stChoice{\stLabFmt{right}}{} \stSeq \stEnd
  }
\).
Here, we notice that the interaction between $\roleQ$ and $\roleR$ should take
place on a single channel, instead of two separate channels annotated $1$ and
$2$.

When \emph{merging} behaviour occurs during projection,
we need to use the same channel in those interactions
to achieve the correct behaviour.
After traversing the global type to annotate each interaction,
we %
merge annotations involved in merges
during projection. \iftoggle{full}{Additional technical details are included in \cref{sec:appendix:channel-gen}.}{}

\section{Evaluation}\label{sec:eval}

We evaluate our toolchain \theTool from two perspectives: \emph{expressivity}
and \emph{feasibility}.
For expressivity, we use examples from %
 session type literature,
and extend them to include crash handling behaviour using two
patterns: failover and graceful failure.
For feasibility, we show that our tool
generates \Scala code within negligible time.

We note that we do \emph{not} evaluate the performance of the generated code.
The generated code uses the \Effpi concurrency library to
implement protocols,
and any performance indication would depend and reflect on the performance of
\Effpi,
instead of \theTool.

\subparagraph*{Expressivity}\label{sec:eval:expressivity}

\begin{table}[t]
  \caption{%
    Overview of All Variants for Each Example.
  }
  \vspace{-1em}
  \begin{center}
  \footnotesize
  \begin{tabular}{lCCccccc}
    \toprule
    Name & \text{Var.} & \rolesR & Comms.\@ & Crash Branches & Max Cont.\@
    Len.
    \\
    \midrule
    \multirow{2}{*}{%
      \begin{tabular}{@{}l@{}}
      \exampleName{PingPong} \\
      $\roleSet = \roleSetOf{p,q}$
      \end{tabular}
    } & (a) & \roleSet &	2 &	0 &	4 \\
     & (b) & \rolesEmpty &	2 &	2 &	4 \\\midrule
    \multirow{2}{*}{%
      \begin{tabular}{@{}l@{}}
      \exampleName{Adder} \\
      $\roleSet = \roleSetOf{p,q}$
      \end{tabular}
    } & (c) & \roleSet &	5 &	0 &	6 \\
     & (d) & \rolesEmpty &	5 & 5 &	6 \\\midrule
    \multirow{2}{*}{%
      \begin{tabular}{@{}l@{}}
      \exampleName{TwoBuyer} \\
      $\roleSet = \roleSetOf{p,q,r}$
      \end{tabular}
    } & (e) & \roleSet &	7 &	0 &	8 \\
     & (f) & \roleFmt{\{\roleR\}} &	18 & 6 &	12 \\\midrule
    \multirow{4}{*}{%
      \begin{tabular}{@{}l@{}}
      \exampleName{OAuth}\\
      $\roleSet = \roleSetOf{c,a,s}$
      \end{tabular}
    } & (g) & \roleSet &	12 &  0 &	11 \\
     & (h) & \roleFmt{\{s,a\}} &	21 & 8 &	11 \\
     & (i) & \roleFmt{\{s\}} &	26 & 13 &	11 \\
     & (j) & \rolesEmpty &	30 &	28 &	11 \\\midrule
    \multirow{3}{*}{%
      \begin{tabular}{@{}l@{}}
      \exampleName{TravelAgency}\\
      $\roleSet = \roleSetOf{c,a,s}$
      \end{tabular}
    } & (k) & \roleSet &	8 &	0 &	6 \\
     & (l) & \roleFmt{\{a,s\}} &	9 &	3 &	6 \\
     & (m) & \roleFmt{\{a\}} &	9 &	4 &	6 \\\midrule
    \multirow{3}{*}{%
      \begin{tabular}{@{}l@{}}
      \exampleName{DistLogger} \\
      $\roleSet = \roleSetOf{l,c,i}$
      \end{tabular}
    } & (n) & \roleSet &	10 &	0 &	7 \\
     & (o) & \roleFmt{\{i,c\}} &	15 &	2 &	7 \\
     & (p) & \roleFmt{\{i\}} &	16 &	4 &	7 \\\midrule
    \multirow{3}{*}{%
      \begin{tabular}{@{}l@{}}
      \exampleName{CircBreaker}\\
      $\roleSet = \roleSetOf{s,a,r}$
      \end{tabular}
    } & (q) & \roleSet &	18 &	0 &	10 \\
     & (r) & \roleFmt{\{a,s\}} &	24 &	3 &	10 \\
     & (s) & \roleFmt{\{a,s\}} &	23 &	3 &	11 \\
    \bottomrule
  \end{tabular}
  \end{center}
  \label{tab:eval:table}
\end{table}

We evaluate our approach with examples in session type literature:
\exampleName{PingPong}, \exampleName{Adder},
\exampleName{TwoBuyer}~\cite{HYC16},
\exampleName{OAuth}~\cite{RV13SPy},
\exampleName{TravelAgency}~\cite{ECOOP08SessionJava},
\exampleName{DistLogger}~\cite{ECOOP22AffineMPST}, and
\exampleName{CircBreaker}~\cite{ECOOP22AffineMPST}.
Notably,
the last two are
inspired by real-world patterns in distributed computing.
\iftoggle{full}{We give details of the examples in \cref{sec:eval:examples}.}{}

We begin with the \emph{fully reliable} version of the examples,
and extend them to include crash handling behaviour.
Recall that our extended theory subsumes the original theory,
when all roles are assumed \emph{reliable}.
Therefore,
the fully reliable versions can act both as a sanity check,
to ensure the code generation does not exclude good protocols in the original
theory,
and as a baseline to compare against.

To add crash handling behaviour, we employ two patterns: \emph{failover} and
\emph{graceful failure}.
In the former scenario,
a crashed role has its functions taken over by another role,
acting as a substitute to the crashed role~\cite{CONCUR22MPSTCrash}.
In the latter scenario,
the protocol is terminated peacefully,
possibly involving additional messages for notification purposes.
Using the example from \cref{sec:overview},
the fully reliable protocol in \cref{ex:overview-global-without-crash}
is extended to one with graceful failure in
\cref{ex:overview-global-with-crash}.

We show a summary of the examples in \cref{tab:eval:table}.
For each example, we give the set of all roles $\roleSet$
and vary the set of reliable roles ($\rolesR$).
Each variant is given an identifier (Var.),
and each example always has a fully
reliable variant where $\rolesR = \roleSet$.
We give the number of communication interactions (Comms.),
the number of \gtCrashLab branches added (Crash Branches),
and the length of the longest continuation (Max Cont.\@ Len.)
in the given global type.

The largest of our examples in terms of concrete interactions is
\exampleName{OAuth}, with Variant $(i)$ having 26 interactions and $(j)$ having
30 interactions. This represents a $2.17\times$ and $2.5\times$ increase over
the size of the original protocol, and is a consequence of the confluence of two
factors: the graceful failure pattern, and low degree of branching in the
protocol itself. The \exampleName{TwoBuyer} Variant $(f)$ represents the
greatest increase ($2.57\times$) in interactions, a result of implementing the
failover pattern. The \exampleName{CircBreaker} variants are also notable in
that they are large in terms of both interactions and branching degree -- both
affect generation times.

\subparagraph*{Feasibility}\label{sec:eval:feasibility}

In order to demonstrate the feasibility of our tool \theTool, we give
generation times using our prototype for all protocol variants and examples,
plotted in \cref{fig:eval:stackedRuntimes}.
We show that \theTool is able to complete the code generation within
milliseconds,
which does not pose any major overhead for a developer.

\begin{figure}[t]
{\small
  \begin{center}
  \begin{tikzpicture}
    \begin{axis}[
        width=0.8\textwidth,
        height=2.5cm,
        enlargelimits=0.05,
        ybar stacked,
        bar width=0.7,
        ymajorgrids,
        scale only axis,
        xticklabels={$(a)$,$(b)$,$(c)$,$(d)$,$(e)$,$(k)$,$(l)$,$(m)$,$(n)$,$(g)$,$(o)$,$(p)$,$(f)$,$(q)$,$(h)$,$(s)$,$(r)$,$(i)$,$(j)$},
        xtick=data,
        ytick={0,0.5,...,2},
        xticklabel style={text height=2ex},
        minor y tick num=1,
        ylabel={Time (ms)},
        legend style={
          at={(0.3,0.95)},
          anchor=north,
          legend columns=-1,
          /tikz/every even column/.append style={column sep=1mm},
        },
      ]
      \pgfplotstableread{results/milliseconds.txt}\loadedtable
      \addplot table[x=Order, y=ReadFile] {\loadedtable};
      \addplot table[x=Order, y=EffpiGen] {\loadedtable};
      \addplot table[x=Order, y=ScalaGen] {\loadedtable};
      \legend{Parsing,EffpiIR,CodeGen};
    \end{axis}
  \end{tikzpicture}
  \end{center}
  }
  \vspace{-2em}
  \caption{Average Generation Times for All Variants in
  \Cref{tab:eval:table}.}
 \vspace{-1em}
  \label{fig:eval:stackedRuntimes}
\end{figure}

In addition to total generation times,
we report measurements for three main
constituent phases of \theTool: parsing, EffpiIR
generation, and code generation.
EffpiIR generation projects and transforms a parsed global type into an
intermediate representation,
which is then used to generate concrete \Scala code.

For all variants, the code generation phase is the most expensive
phase.
This is likely a consequence of traversing the given EffpiIR representation of
a protocol twice -- once for local type declarations and once for role-implementing
functions.
\iftoggle{full}{More details on generation time, including how they are measured, can be found in
\cref{sec:eval:times}.}{}

\section{Related Work}\label{sec:related}

We summarise related work on both theory and implementations of session types
with failure handling,
as well as other MPST implementations targeting \Scala without failures.

We first discuss closest related
work~\cite{OOPSLA21FaultTolerantMPST,FORTE22FaultTolerant,CONCUR22MPSTCrash,ESOP23MAGPi},
where multiparty session types are extended to model crashes or failures.
Both~\cite{FORTE22FaultTolerant} and~\cite{ESOP23MAGPi} are exclusively theoretical.

\cite{FORTE22FaultTolerant}~proposes an MPST framework to
model fine-grained unreliability:
each transmission in a global type is parameterised by a
reliability annotation,
which can be one of unreliable (sender/receiver can crash, and messages can be
lost),
weakly reliable (sender/receiver can crash, messages are not lost), or
reliable (no crashes or message losses).
\cite{OOPSLA21FaultTolerantMPST}~utilises MPST as a guidance for
fault-tolerant distributed system with recovery mechanisms.
Their framework includes various features, such as sub-sessions,
event-driven programming, dynamic role assignments, and, most importantly, 
failure handling.
\cite{CONCUR22MPSTCrash}~develops a theory of multiparty session types with crash-stop
failures:
they model crash-stop failures in the semantics of processes
and session
types, where the type system uses a model checker to validate type safety.
\cite{ESOP23MAGPi}~follow a similar framework
to~\cite{CONCUR22MPSTCrash}:
they model an asynchronous semantics, and support more patterns of failure,
including message losses, delays, reordering, as well as link failures and
network partitioning.
However, their typing system suffers from its genericity, when type-level
properties become undecidable~\cite[\S 4.4]{ESOP23MAGPi}.

Other session type works on modelling failures can be briefly
categorised into
two: using affine types or exceptions~\cite{ECOOP22AffineMPST, LMCS18Affine,
DBLP:journals/pacmpl/FowlerLMD19}, %
and using coordinators or supervision%
~\cite{DBLP:conf/forte/AdameitPN17,ESOP18CrashHandling}.
The former adapts session types to an \emph{affine}
representation, in which endpoints may cease prematurely;
the latter, instead, are usually reliant on one or more \emph{reliable} processes that \emph{coordinate}
in the event of failure.
The works~\cite{LMCS18Affine,DBLP:conf/forte/AdameitPN17,ESOP18CrashHandling}
are limited to theory.

\cite{LMCS18Affine}~first proposes the affine approach
to failure handling.
Their extension is primarily comprised of a \emph{cancel operator}, which is
semantically similar to our crash construct: it represents a process
that has terminated early.
\cite{DBLP:journals/pacmpl/FowlerLMD19}~presents a concurrent
$\lambda$-calculus based on~\cite{LMCS18Affine}, with asynchronous session-typed communication
and exception handling, and implements their approach as parts of the
\textsc{Links} language.
\cite{ECOOP22AffineMPST}~proposes a framework of
\emph{affine} multiparty session types, and provides an implementation of affine
MPST in the \Rust programming language.
They utilise the affine type
system and \texttt{Result} types of \Rust,
so that the type system enforces that failures are handled.

Coordinator model approaches~\cite{DBLP:conf/forte/AdameitPN17,ESOP18CrashHandling}
often incorporate \emph{interrupt blocks} (or similar constructs)
to model crashes and failure handling.
\cite{DBLP:conf/forte/AdameitPN17}~extends the standard MPST syntax with
\emph{optional blocks},
representing regions of a protocol that are susceptible
to communication failures.
In their approach, if a process $\mpP$ expects a value from an optional block
which fails, then a
\emph{default value} is provided to $\mpP$, so $\mpP$ can continue running.
This ensures termination and deadlock-freedom.
Although this approach does not feature an explicit reliable coordinator
process, we describe it here due to the inherent coordination required for
multiple processes to start and end an optional block. %
\cite{ESOP18CrashHandling}~similarly extends the standard
global type syntax with a \emph{try-handle} construct, which is facilitated by the
presence of a reliable coordinator process, and via a construct to specify
reliable processes.
When the coordinator detects a failure, it broadcasts notifications to all
remaining live processes;
then, the protocol proceeds according to the failure handling continuation
specified as part of the try-handle construct.

Other related MPST implementations
include~\cite{ECOOP21MPSTActor,ECOOP22MPSTScala,DBLP:journals/corr/abs-2005-09520}.
\cite{ECOOP21MPSTActor}~designs a framework for
MPST-guided, safe actor programming.
Whilst the MPST protocol does not include any failure handling, the actors may
fail or raise exceptions, which are handled in a similar way to what we
summarise as the affine technique.
\cite{ECOOP22MPSTScala} revisits API generation techniques in
\Scala for MPST\@.
In addition to the traditional local type/automata-based code
generation~\cite{ECOOP17MPST, FASE16EndpointAPI}, they propose a
new technique based on sets of pomsets, utilising \Scala 3 match
types~\cite{POPL22TypeLevelMatch}.
\cite{DBLP:journals/corr/abs-2005-09520}~presents \Choral,
a programming language for choreographies (multiparty protocols).
\Choral supports the handling of \emph{local exceptions} in choreographies,
which can be used to program reliable channels over unreliable networks,
supervision mechanisms, \etc for fallible communication.
They utilise automatic retries to implement channel APIs.

\section{Conclusion and Future Work}
\label{sec:conclusion}
To overcome the challenge of accounting for failure handling in
distributed systems using session types,
we propose \theTool, a code generation toolchain.
It is built on asynchronous MPST with crash-stop semantics,
enabling the implementation of multiparty protocols that are resilient to failures.
Desirable global type properties such as deadlock-freedom, protocol conformance,
and liveness are preserved by construction in typed processes,
even in the presence of crashes.
Our toolchain \theTool, extends \Scribble and \Effpi to support crash detection
and handling,
providing developers with a lightweight way to leverage our theory.
The evaluation of \theTool demonstrates that it can generate \Scala code with
minimal overhead, which is made possible by the guarantees provided by our
theory.

This work is a new step towards modelling and handling real-world failures using
session types,
bridging the gap between their theory and applications.
As future work, we plan to studys different crash models (\eg crash-recover) 
and failures of other components (\eg link failures). 
These further steps will contribute to our long-term objective 
of modelling and type-checking well-known consensus algorithms used in large-scale distributed systems.

\bibliography{references}

\iftoggle{full}
{
\newpage
\appendix
\section{Structural Congruence}
\label{sec:app_congruence}
The congruence relation of our session calculus is formalised in~\cref{tab:congruence}. 
\begin{figure}
{\small
$\begin{array}{c}
\mpH_1
\cdot
\msg{\roleQ_1}{\mpLab_1}{\val_1}
\cdot
\msg{\roleQ_2}{\mpLab_2}{\val_2}
\cdot
\mpH_2
\equiv
\mpH_1
\cdot
\msg{\roleQ_2}{\mpLab_2}{\val_2}
\cdot
\msg{\roleQ_1}{\mpLab_1}{\val_1}
\cdot
\mpH_2
\quad\text{(if $\roleQ_1\neq \roleQ_2$)}\\[1mm]
\mpQEmpty \cdot \mpH \equiv  \mpH \qquad\qquad
\mpH \cdot \mpQEmpty \equiv  \mpH \qquad\qquad
\mpH[1] \cdot (\mpH[2] \cdot \mpH[3])
\equiv
(\mpH[1] \cdot \mpH[2]) \cdot \mpH[3]
\qquad\qquad
\mu X.\mpP \equiv \mpP\subst{X}{\mu X.\mpP}
\\[1mm]
\mpPart\roleP\mpNil
\mpPar
\mpPart\roleP\mpQEmpty
\mpPar
\mpM
\equiv
\mpM
\qquad
\mpM_1 \mpPar \mpM_2
\equiv
\mpM_2 \mpPar \mpM_1
\qquad
(\mpM_1 \mpPar \mpM_2) \mpPar \mpM_3
\,\equiv\,
\mpM_1 \mpPar (\mpM_2 \mpPar \mpM_3) \qquad
\\[1mm]
\mpP \equiv \mpQ
\;\text{ and }\;
\mpH_1 \equiv \mpH_2
\;\;\implies\;\;
\mpPart\roleP\mpP
\mpPar
\mpPart\roleP\mpH_1
\mpPar
\mpM
\,\equiv\,
\mpPart\roleP\mpQ
\mpPar
\mpPart\roleP\mpH_2
\mpPar
\mpM
\end{array}
$
}
\caption{Structural congruence rules for queues, processes, and sessions.}
\label{tab:congruence}
\end{figure}

\section{Auxiliary Definitions}
\label{sec:app:definitions}

\subsection{Well-Annotated Global Types}
We introduce an auxiliary concept of \emph{well-annotated} global
types in \cref{def:globaltypes:well-anno}, as a consistency requirement for crash
annotations $\roleCrashedSym$ in a global type $\gtG$, and the set of crashed
roles $\rolesC$, and a fixed set of reliable roles $\rolesR$.
We show that well-annotatedness \wrt \rolesR  is preserved by global type
reductions in \cref{lem:well-annotated-preserve}.
It follows that, a global type $\gtG$ without runtime construct is trivially
well-annotated, and all reduct global types $\gtG \gtMoveStar[\rolesR]{}
\gtGi$ are also well-annotated.
\iftoggle{full}{The proof of \cref{lem:well-annotated-preserve} is available in~\Cref{sec:proof:semantics:gty}.
}{}

\begin{definition}[Well-Annotated Global Types]\label{def:globaltypes:well-anno}
  A global type $\gtG$ with crashed roles $\rolesC$ is \emph{well-annotated}
  \wrt %
  a (fixed) set of reliable roles $\rolesR$, iff:
  \begin{enumerate}[label=(WA\arabic*), leftmargin=12mm, nosep]
    \item No reliable roles are crashed, \;$\gtRolesCrashed{\gtG} \cap \rolesR =
      \emptyset$;\; and,
      \label{item:wa:reliable-no-crash}
    \item All roles with crash annotations are in the crashed set,
      \;$\gtRolesCrashed{\gtG} \subseteq \rolesC$;\; and,
      \label{item:wa:crash-annot-crash}
    \item A role cannot be live and crashed simultaneously, \;$\gtRoles{\gtG}
      \cap \gtRolesCrashed{\gtG} = \emptyset$.
      \label{item:wa:live-no-crash}
  \end{enumerate}
\end{definition}
\begin{restatable}[Preservation of Well-Annotated Global Types]{lemma}{lemWellAnnoPreserve}%
  \label{lem:well-annotated-preserve}
  If
  \;$\gtWithCrashedRoles{\rolesC}{\gtG}
  \gtMove[\stEnvAnnotGenericSym]{\rolesR}
  \gtWithCrashedRoles{\rolesCi}{\gtGi}
  $,\; and
  \;$\gtWithCrashedRoles{\rolesC}{\gtG}$\; is well-annotated \wrt $\rolesR$, then
  \;$\gtWithCrashedRoles{\rolesCi}{\gtGi}$\; is also well-annotated \wrt
  $\rolesR$.
\end{restatable}

\subsection{Local Type Subtyping}
We define a \emph{subtyping} relation $\stSub$ on local types in
\cref{def:subtyping}. %
Our subtyping relation is mostly standard~\cite[Def.\@ 2.5]{POPL19LessIsMore},
except for the ($\highlight{\text{highlighted}}$) addition of the rule $\inferrule{\iruleStSubStop}$ and
extra requirements in \inferrule{\iruleStSubIn}.
In \inferrule{\iruleStSubIn}, we add two additional requirements:
\emph{(1)} the supertype cannot be a ``pure'' crash handling branch;
and \emph{(2)} if the subtype has a crash handling branch, then the supertype
must also have one.
For simplicity, we do not consider subtyping on basic types $\tyGround$.

\begin{definition}[Subtyping]\label{def:subtyping}
The subtyping relation $\stSub$ is coinductively defined:
{\small{
\[
\begin{array}{c}
\cinference[\iruleStSubEnd]{}{
  \stEnd \stSub \stEnd
}
\quad

\cinference[\iruleStSubIn]{
  \forall i \in I
  &
  \stT[i] \stSub \stTi[i]
  &
  \highlight{\setcomp{\stLab[k]}{k \in I} \neq \setenum{\stCrashLab}}
  &
 \highlight{\nexists j \in J: \stLab[j] = \stCrashLab}
}{
  \stExtSum{\roleP}{i \in I \cup J}{\stChoice{\stLab[i]}{\tyGround[i]} \stSeq \stT[i]}%
  \stSub
  \stExtSum{\roleP}{i \in I}{\stChoice{\stLab[i]}{\tyGround[i]} \stSeq \stTi[i]}%
}

\\[2ex]
\highlight{
\cinference[\iruleStSubStop]{}{
  \stStop \stSub \stStop
}}
\quad

\cinference[\iruleStSubOut]{
  \forall i \in I
  &
  \stT[i] \stSub \stTi[i]
}{
  \stIntSum{\roleP}{i \in I}{\stChoice{\stLab[i]}{\tyGround[i]} \stSeq \stT[i]}
  \stSub
  \stIntSum{\roleP}{i \in I \cup J}{\stChoice{\stLab[i]}{\tyGround[i]} \stSeq \stTi[i]}
}

\\[2ex]
\cinference[\iruleStSubRecL]{
  \stT{}\subst{\stRecVar}{\stRec{\stRecVar}{\stT}} \stSub \stTi
}
{
  \stRec{\stRecVar}{\stT} \stSub \stTi
}
\quad

\cinference[\iruleStSubRecR]{
  \stT \stSub \stTi{}\subst{\stRecVar}{\stRec{\stRecVar}{\stTi}}
}{
  \stT \stSub \stRec{\stRecVar}{\stTi}
}
\end{array}
\]
}}
\end{definition}

\section{Additional Examples}
\label{sec:app_examples}
\begin{example}
\label{ex:process_syntax_semantics_full}
We now illustrate our operational semantics of sessions with an example. 
Consider the session:

\smallskip
\centerline{\(  
{\small{\mpPart\roleP\mpP \mpPar  \mpPart\roleP \mpQEmpty
     \mpPar \mpPart \roleQ \mpQ \mpPar  \mpPart \roleQ \mpQEmpty}}
\)}

\smallskip 
\noindent
where 
{\small{
$ \mpP = \procout{\roleQ}{\mpLab}{\text{``\texttt{abc}''}}{\sum
\setenum{
\begin{array}{l}
\procin{\roleQ}{\mpLabi(\mpx)}{\mpNil}
\\
\procin{\roleQ}{\mpCrashLab}{\mpNil}
\end{array}
}
}
$
}}
 and 
{\small{
$
 \mpQ = 
  \sum \setenum{
 \begin{array}{l}
 \procin{\roleP}{\mpLab(\mpx)}{
 \procout{\roleP}{\mpLabi}{42}{\mpNil}}
 \\
  \procin{\roleP}{\mpCrashLab}{\mpNil}
  \end{array}
  }
  $
}}. 

In this session,  the process $\mpQ$ for $\roleQ$ receives a message sent from $\roleP$ to $\roleQ$; 
the process $\mpP$ for $\roleP$ sends a message from $\roleP$ to $\roleQ$, and then receives a message 
sent from $\roleQ$ to  $\roleP$. 

 On a successful reduction (without crashes),  %
     we have:

 \smallskip
 \centerline{\(
 {\small{
 \begin{array}{rll}
  \mpPart\roleP\mpP \mpPar  \mpPart\roleP \mpQEmpty
     \mpPar \mpPart \roleQ \mpQ \mpPar  \mpPart \roleQ \mpQEmpty
     & \redSend{\roleP}{\roleQ}{\mpLab} &
    \mpPart \roleP \sum
\setenum{
\begin{array}{l}
\procin{\roleQ}{\mpLabi(\mpx)}{\mpNil}
\\
\procin{\roleQ}{\mpCrashLab}{\mpNil}
\end{array}
}
\mpPar
\mpPart\roleP \mpQEmpty  \mpPar
\mpPart \roleQ \mpQ \mpPar  \mpPart \roleQ (\roleP,\mpLab(\text{``\texttt{abc}''}))
\\
& \redSend{\roleP}{\roleQ}{\mpLab} &
  \mpPart \roleP \sum
\setenum{
\begin{array}{l}
\procin{\roleQ}{\mpLabi(\mpx)}{\mpNil}
\\
\procin{\roleQ}{\mpCrashLab}{\mpNil}
\end{array}
}
\mpPar
\mpPart\roleP \mpQEmpty
\mpPar
\mpPart \roleQ \procout{\roleP}{\mpLabi}{42}{\mpNil} \mpPar
\mpPart\roleQ \mpQEmpty
\\
& \redSend{\roleP}{\roleQ}{\mpLab} &
  \mpPart \roleP \sum
\setenum{
\begin{array}{l}
\procin{\roleQ}{\mpLabi(\mpx)}{\mpNil}
\\
\procin{\roleQ}{\mpCrashLab}{\mpNil}
\end{array}
}
\mpPar
\mpPart\roleP (\roleQ, \mpLabi(42))
\mpPar
\mpPart \roleQ \mpNil
\mpPar
\mpPart\roleQ \mpQEmpty
\\
& \redSend{\roleP}{\roleQ}{\mpLab} &
\mpPart\roleP \mpNil
\mpPar
\mpPart\roleP \mpQEmpty
\mpPar
\mpPart \roleQ \mpNil
\mpPar
\mpPart\roleQ \mpQEmpty
\end{array}
}}
\)}

\smallskip
\noindent
Let $\rolesR = \emptyset$ (\ie each role is unreliable).  
Suppose that
$\mpP$ crashes before sending, which leads to the reduction:

\smallskip
 \centerline{\(
 {\small{
  \mpPart\roleP\mpP \mpPar  \mpPart\roleP \mpQEmpty
     \mpPar \mpPart \roleQ \mpQ \mpPar  \mpPart \roleQ \mpQEmpty
    \redCrash{\roleP}{\rolesR}
     \mpPart\roleP{\mpCrash}
\mpPar
\mpPart\roleP{\mpQUnavail}
\mpPar
\mpPart \roleQ \mpQ \mpPar  \mpPart \roleQ \mpQEmpty
 \redSend{\roleP}{\roleQ}{\mpLab}
\mpPart\roleP{\mpCrash}
\mpPar
\mpPart\roleP{\mpQUnavail}
\mpPar
\mpPart \roleQ \mpNil
\mpPar
\mpPart\roleQ \mpQEmpty
}}
\)}

\smallskip
\noindent
We can observe that when the output (sending) process $\mpP$ located at an unreliable role $\roleP$ crashes (by
\inferrule{r-$\lightning$}), $\roleP$ also crashes ($\mpPart\roleP{\mpCrash}$),  with an unavailable incoming message
queue ($\mpPart\roleP{\mpQUnavail}$). Subsequently, the input (receiving) process $\mpQ$ located at $\roleQ$ can detect and
handle the crash by \inferrule{r-rcv-$\odot$} via its handling branch.
\end{example}

\begin{example}
\label{ex:role_removal}
We remove role $\roleFmt{C}$ in the global type $\gtG$ in~\eqref{ex:overview-global-with-crash} (defined in \Cref{sec:overview}).

\smallskip
\centerline{\(
{\small{
\gtCrashRole{\gtG}{\roleFmt{C}} =
    \gtCommSingle{\roleFmt{L}}{\roleFmt{I}}
    {\gtMsgFmt{trigger}}{}{\gtCommTransitSingle{\roleFmt{C^\lightning}}{\roleFmt{I}}
     {\gtCrashLab}{}
     {\gtCommSingle{\roleFmt{I}}{\roleFmt{L}}{\gtMsgFmt{fatal}}{}{
    \gtEnd
    }
}}}}
\)}

\smallskip
\noindent
Role $\roleFmt{C}$ now carries a crash annotation $\roleCrashedSym$ in the
 resulting global type, denoting it has crashed.
 Crash annotations change the reductions available for global types.
  \end{example}

\begin{example}
\label{ex:configuration_safety_full}
Recall the local types of the Simpler Logging example in~\cref{sec:overview}:

\smallskip
\centerline{\(
{\small{
 \begin{array}{c}
  \stT[\roleFmt{C}]  =
 \roleFmt{I} \stFmt{\oplus}
   \stLabFmt{read}
 \stSeq
  \roleFmt{I}
   \stFmt{\&}
    \stLabFmt{report(\stFmtC{log})}
 \stSeq
\stEnd
  \quad
   \stT[\roleFmt{L}] =
   \roleFmt{I} \stFmt{\oplus}
   \stLabFmt{trigger}
 \stSeq
  \stExtSum{\roleFmt{I}}{}{
  \begin{array}{@{}l@{}}
\stLabFmt{fatal}
  \stSeq
   \stEnd
\\
\stLabFmt{read}
  \stSeq
   \roleFmt{I} \stFmt{\oplus}
    \stLabFmt{report(\stFmtC{log})}
    \stSeq
    \stEnd
  \end{array}
  }
\\[3mm]
  \stT[\roleFmt{I}]  =
    \roleFmt{L}
  \stFmt{\&}
  \stLabFmt{trigger}
  \stSeq
  \stExtSum{\roleFmt{C}}{}{
  \begin{array}{@{}l@{}}
 \stLabFmt{read}
 \stSeq
 \roleFmt{L}
  \stFmt{\oplus}
  \stLabFmt{read}
  \stSeq
 \roleFmt{L}
\stFmt{\&}
\stLabFmt{report(\stFmtC{log})}
\stSeq
\roleFmt{C}
\stFmt{\oplus}
\stLabFmt{report(\stFmtC{log})}
\stSeq
\stEnd
 \\
  \stCrashLab
  \stSeq
  \roleFmt{L}
  \stFmt{\oplus}
  \stLabFmt{fatal}
 \stSeq
\stEnd
 \end{array}
 }
\end{array}
 }
 }
\)}

\smallskip
\noindent
The configuration $\stEnv; \qEnv$, where
$\stEnv =  \stEnvMap{\roleFmt{C}}{\stT[\roleFmt{C}]}  \stEnvComp
\stEnvMap{\roleFmt{L}}{\stT[\roleFmt{L}]}   \stEnvComp
\stEnvMap{\roleFmt{I}}{\stT[\roleFmt{I}]}$ and
$\qEnv =  \qEnv[\emptyset]$, is
$\setenum{\roleFmt{L}, \roleFmt{I}}$-safe.
This can be verified by checking its possible
reductions.  For example,
in the case where $\roleFmt{C}$ crashes immediately, we have:

\smallskip
{\small{
\centerline{\(
\begin{array}{@{}r@{}cl}
\stEnv; \qEnv
&
\stEnvMoveAnnot{\ltsCrash{\mpS}{\roleFmt{C}}}
&
 \stEnvUpd{\stEnv}{\roleFmt{C}}{
\stStop};
       \stEnvUpd{\qEnv}{\cdot, \roleFmt{C}}{\stQUnavail}
       \\
&
\stEnvMoveOutAnnot{\roleFmt{L}}{\roleFmt{I}}{\stChoice{\stLabFmt{trigger}}{}}
&
\stEnvUpd{\stEnvUpd{\stEnv}{\roleFmt{C}}{\stStop}}{
 \roleFmt{L}}{
   \stExtSum{\roleFmt{I}}{}{
  \begin{array}{@{}l@{}}
\stLabFmt{fatal}
  \stSeq
   \stEnd
\\
\stLabFmt{read}
  \stSeq
   \roleFmt{I} \stFmt{\oplus}
    \stLabFmt{report(\stFmtC{log})}
    \stSeq
    \stEnd
  \end{array}
  }
};  \stEnvUpd{\stEnvUpd{\qEnv}{\cdot, \roleFmt{C}}{\stQUnavail}}{\roleFmt{L}, \roleFmt{I}}{\stLabFmt{trigger}}
\\
    &
     \stEnvMoveInAnnot{\roleFmt{I}}{\roleFmt{L}}{\stChoice{\stLabFmt{trigger}}{}}
     &
 \stEnvUpd{\stEnvUpd{\stEnv}{\roleFmt{C}}{\stStop}}{
 \roleFmt{L}}{
   \stExtSum{\roleFmt{I}}{}{
  \begin{array}{@{}l@{}}
\stLabFmt{fatal}
  \stSeq
   \stEnd
\\
\stLabFmt{read}
  \stSeq
   \roleFmt{I} \stFmt{\oplus}
    \stLabFmt{report(\stFmtC{log})}
    \stSeq
    \stEnd
  \end{array}
  }}
  \\
  &
  &
 \, \, \, \,\stFmt{[} \roleFmt{I}
 \stFmt{\mapsto} \stExtSum{\roleFmt{C}}{}{
  \begin{array}{@{}l@{}}
 \stLabFmt{read}
 \stSeq
 \roleFmt{L}
  \stFmt{\oplus}
  \stLabFmt{read}
  \stSeq
 \roleFmt{L}
\stFmt{\&}
\stLabFmt{report(\stFmtC{log})}
\stSeq
\roleFmt{C}
\stFmt{\oplus}
\stLabFmt{report(\stFmtC{log})}
\stSeq
\stEnd
 \\
  \stCrashLab
  \stSeq
  \roleFmt{L}
  \stFmt{\oplus}
  \stLabFmt{fatal}
 \stSeq
\stEnd
 \end{array}
 }\stFmt{]}
; \stEnvUpd{\qEnv}{\cdot, \roleFmt{C}}{\stQUnavail}
 \\
 &
 \stEnvMoveAnnot{\ltsCrDe{\mpS}{\roleFmt{I}}{\roleFmt{C}}}
 &
  \stEnvUpd{\stEnvUpd{\stEnvUpd{\stEnv}{\roleFmt{C}}{\stStop}}{
 \roleFmt{L}}{
   \stExtSum{\roleFmt{I}}{}{
  \begin{array}{@{}l@{}}
\stLabFmt{fatal}
  \stSeq
   \stEnd
\\
\stLabFmt{read}
  \stSeq
   \roleFmt{I} \stFmt{\oplus}
    \stLabFmt{report(\stFmtC{log})}
    \stSeq
    \stEnd
  \end{array}
  }}}{\roleFmt{I}}{\roleFmt{L} \stFmt{\oplus} \stLabFmt{fatal}
  \stSeq \stEnd}; \stEnvUpd{\qEnv}{\cdot, \roleFmt{C}}{\stQUnavail}
 \\
  &
     \stEnvMoveOutAnnot{\roleFmt{I}}{\roleFmt{L}}{\stChoice{\stLabFmt{fatal}}{}}
     &
    \stEnvUpd{\stEnvUpd{\stEnvUpd{\stEnv}{\roleFmt{C}}{\stStop}}{
 \roleFmt{L}}{
   \stExtSum{\roleFmt{I}}{}{
  \begin{array}{@{}l@{}}
\stLabFmt{fatal}
  \stSeq
   \stEnd
\\
\stLabFmt{read}
  \stSeq
   \roleFmt{I} \stFmt{\oplus}
    \stLabFmt{report(\stFmtC{log})}
    \stSeq
    \stEnd
  \end{array}
  }}}{\roleFmt{I}}{\stEnd};\\
  & &
   \stEnvUpd{\stEnvUpd{\qEnv}{\cdot, \roleFmt{C}}{\stQUnavail}}{\roleFmt{I}, \roleFmt{L}}{\stLabFmt{fatal}}
   \\
   &
    \stEnvMoveInAnnot{\roleFmt{L}}{\roleFmt{I}}{\stChoice{\stLabFmt{fatal}}{}}
     &
    \stEnvUpd{\stEnvUpd{\stEnvUpd{\stEnv}{\roleFmt{C}}{\stStop}}{
 \roleFmt{L}}{
  \stEnd
  }}{\roleFmt{I}}{\stEnd};
  \stEnvUpd{\qEnv}{\cdot, \roleFmt{C}}{\stQUnavail}
 \end{array}
\)}
}}

\smallskip
\noindent
and each reductum satisfies all clauses of~\Cref{def:mpst-env-safe}.
The cases where $\roleFmt{C}$ crashes after sending the $\stLabFmt{read}\text{ing}$ message
to $\roleFmt{I}$ are similar.  There are no other crash reductions to consider, since both
$\roleFmt{L}$ and $\roleFmt{I}$ are assumed to be reliable.
The cases where no crashes occur are similar as well,
except that \inferrule{\iruleTCtxCrashDetect} and \inferrule{\iruleTCtxCrash} are not applied in the non-crash reductions.
\end{example}

\begin{example}
\label{ex:typing_system_full}
Consider our Simpler Logging example (\Cref{sec:overview} and \Cref{ex:configuration_safety}).  
Specifically, consider the processes that act as the roles $\roleFmt{C}$, $\roleFmt{L}$, and 
$\roleFmt{I}$ respectively: 

\smallskip
 \centerline{\(
 {\small{
 \begin{array}{c}
\mpP[\roleFmt{C}]  =
\procoutNoVal{\roleFmt{I}}{\labFmt{read}}{\procin{\roleFmt{I}}{\labFmt{report}(\mpx)}{\mpNil}}
\quad
\mpP[\roleFmt{L}] =
\procoutNoVal{\roleFmt{I}}{\labFmt{trigger}}{\sum \setenum{
 \begin{array}{l}
 \procin{\roleFmt{I}}{\labFmt{fatal}}{\mpNil}
 \\
 \procin{\roleFmt{I}}{\labFmt{read}}{
 \procout{\roleFmt{I}}{\labFmt{report}}{\mpFmt{log}}{\mpNil}}
  \end{array}
  }
  }
  \\[3mm]
 \mpP[\roleFmt{I}] =
\procin{\roleFmt{L}}{\labFmt{trigger}}{\sum \setenum{
 \begin{array}{l}
 \procin{\roleFmt{C}}{\labFmt{read}}{
 \procoutNoVal{\roleFmt{L}}{\labFmt{read}}
 {\procin{\roleFmt{L}}{\labFmt{report}(\mpx)}
 {\procout{\roleFmt{C}}{\labFmt{report}}{\mpFmt{log}}{\mpNil}}}
 }
 \\
  \procin{\roleFmt{C}}{\mpCrashLab}{ \procoutNoVal{\roleFmt{L}}{\labFmt{fatal}}{\mpNil}}
  \end{array}
  }
  } 
\end{array}
}}
\)}

\smallskip
\noindent
and message queues $\mpH[\roleFmt{C}] = \mpH[\roleFmt{L}] = \mpH[\roleFmt{I}] = \mpQEmpty$,
and %
 the configuration $\stEnv; \qEnv$ as in \Cref{ex:configuration_safety}. 

Process $\mpP[\roleFmt{C}]$ (resp. $\mpP[\roleFmt{L}]$, $\mpP[\roleFmt{I}]$)
has the type $\stEnvApp{\stEnv}{\roleFmt{C}}$ (resp. $\stEnvApp{\stEnv}{\roleFmt{L}}$, 
$\stEnvApp{\stEnv}{\roleFmt{I}}$),
and
queue $\mpH[\roleFmt{C}]$ (resp. $\mpH[\roleFmt{L}]$, $\mpH[\roleFmt{I}]$) %
has the type
$\stEnvApp{\qEnv}{-, \roleFmt{C}}$ (resp. $\stEnvApp{\qEnv}{-, \roleFmt{L}}$, $\stEnvApp{\qEnv}{-, \roleFmt{I}}$),
which can be verified in the standard way. Then, together with the association 
 $\stEnvAssoc{\gtWithCrashedRoles{\emptyset}{\gtG[s]}}{\stEnv; \qEnv}{\setenum{\roleFmt{L}, \roleFmt{I}}}$, 
we can use \inferrule{{t-sess}} to assert that the session
$\mpPart{\roleFmt{C}}{\mpP[\roleFmt{C}]} \mpPar
  \mpPart{\roleFmt{C}}{\mpH[\roleFmt{C}]} \mpPar
  \mpPart{\roleFmt{L}}{\mpP[\roleFmt{L}]} \mpPar
  \mpPart{\roleFmt{L}}{\mpH[\roleFmt{L}]} \mpPar
  \mpPart{\roleFmt{I}}{\mpP[\roleFmt{I}]} \mpPar
 \mpPart{\roleFmt{I}}{\mpH[\roleFmt{I}]}$
 is governed by the global type
  $\gtWithCrashedRoles{\emptyset}{\gtG[s]}$.
  If we follow a crash reduction, \eg by the rule \inferrule{r-$\lightning$}, 
  the session evolves as $\mpPart{\roleFmt{C}}{\mpP[\roleFmt{C}]} \mpPar
  \mpPart{\roleFmt{C}}{\mpH[\roleFmt{C}]} \mpPar
  \mpPart{\roleFmt{L}}{\mpP[\roleFmt{L}]} \mpPar
  \mpPart{\roleFmt{L}}{\mpH[\roleFmt{L}]} \mpPar
  \mpPart{\roleFmt{I}}{\mpP[\roleFmt{I}]} \mpPar
 \mpPart{\roleFmt{I}}{\mpH[\roleFmt{I}]}
\;\redCrash{\roleP}{\rolesR}\;
\mpPart{\roleFmt{C}}{\mpCrash} \mpPar
  \mpPart{\roleFmt{C}}{\mpQUnavail} \mpPar
  \mpPart{\roleFmt{L}}{\mpP[\roleFmt{L}]} \mpPar
  \mpPart{\roleFmt{L}}{\mpH[\roleFmt{L}]} \mpPar
  \mpPart{\roleFmt{I}}{\mpP[\roleFmt{I}]} \mpPar
 \mpPart{\roleFmt{I}}{\mpH[\roleFmt{I}]}$, where, by 
 \inferrule{t-$\mpCrash$}, $\mpP[\roleFmt{C}]$ is typed by $\stStop$,  and 
 $\mpH[\roleFmt{C}]$ is typed by $\stQUnavail$. 
  \end{example}

\section{Additional Related Work}

\textbf{Peters \etal~\cite{FORTE22FaultTolerant}} propose an MPST framework to
model fine-grained unreliability.
In their work, each transmission in a global type is parameterised by a
reliability annotation,
which can be one of unreliable (sender/receiver can crash, and messages can be
lost),
weakly reliable (sender/receiver can crash, messages are not lost), or
reliable (no crashes or message losses).
The design choice taken in our work roughly falls under weakly
reliable in their work, where a $\gtCrashLab$ handling branch (in their work, a
default branch) needs to be present to handle failures.
In our work, the reliability assumptions operate on a coarse level, but
nonetheless are \emph{consistent} within a given global type -- if a role
$\roleP$ is assumed reliable, $\roleP \in \rolesR$, then it does not crash for
the duration of the protocol, and vice versa.
Therefore, in a transmission $\roleP \to \roleQ$, our model allows \emph{one}
of the two roles to be unreliable, whereas their work does not permit the
`mixing' of reliability of sending and receiving roles.

\textbf{Viering \etal~\cite{OOPSLA21FaultTolerantMPST}} utilise MPST as a guidance for
fault-tolerant distributed system with recovery mechanisms.
Their framework includes various features, such as sub-sessions,
event-driven programming, dynamic role assignments, and most importantly
failure handling.
Our work handles unreliability in distributed programming, with the
following differences:
\begin{enumerate}[nosep, leftmargin=*]
  \item Failure detection \emph{assumptions} and \emph{models} are different:
    in our work, we assume a perfect failure detector, where all detected
    crashes are genuine.
    Their work uses a less strict assumption to allow \emph{false suspicions}.
    This difference subsequently gives rise to how failures are handled in both
    approaches:
    in our work, we use a special $\gtCrashLab$ handling branch in %
    global types to specify how the global protocol should progress after a
    crash has been detected.
    In contrast, they use a \emph{try-catch} construct in global types.
    In such a try-catch construct, crash detection within a sub-session
    $\gtG[1]$ is specified
    $\gtG[1]~\mathsf{with}~\roleP @ \roleQ~.~\gtG[2]$, where $\gtG[2]$ is a
    global type for failure handling (where $\roleP$ cannot occur), and
    $\roleQ$ is a \emph{monitor} that monitors $\roleP$ for possible failures.
    Moreover, the well-formedness condition (2) in~\cite[\S
    4.1]{OOPSLA21FaultTolerantMPST} requires the first message in $\gtG[2]$ to
    be a message broadcast of failure notification from the monitor $\roleQ$
    to all roles participating in the sub-session (except the crashed role
    $\roleP$).

    On this matter, we consider our framework more \emph{flexible} when
    detecting and handling crashes: every communication construct can have a
    crash handling branch (when the receiver is not assumed reliable), and the
    failure broadcast is not necessary (failure detection only occurs when
    receiving from a crashed role).

  \item The \emph{merge} operators, used when projecting global types to obtain
    local types, are different: we use a more expressive \emph{full} merge
    operator (\cref{def:local-type-merge}), whereas they use a \emph{plain}
    merge operator, \ie requiring all continuations to project to the
    \emph{same} local type.

  \item  \emph{Reliability assumptions} are different: in our work, we
    support a range of assumptions from every role being unreliable, to totally
    reliable (as in the literature).
    In their work, they require \emph{at least one} reliable role, because
    they use a \emph{monitoring tree} for detecting crashes.
    Our work allows a role to detect the crash of its communication partner
    during reception, thus requiring neither such trees nor reliable roles.

 \item They provide a toolchain for
implement fault-tolerant protocols targeting \Scala.
Specifically, they use an \emph{event-driven} style of API, in a similar style
to~\cite{OOPSLA20VerifiedRefinements}, in contrast to previous work in
\Scala~\cite{ECOOP17MPST}.
We instead use the \effpi library, which fits our needs to encode
local types, and it can be easily adapted to implement crash detection.
\end{enumerate}

\textbf{Barwell \etal~\cite{CONCUR22MPSTCrash}}
develop a theory of multiparty session types with crash-stop
failures. Their theory models crash-stop failures in the semantics of processes
and session
types, where the type system uses a model checker to validate type safety.
Our theory follows a similar model of crash-stop failures, but differs in the
following:
\begin{enumerate}[nosep, leftmargin=*]
  \item We model an asynchronous (message-passing) semantics, whereas they
    model a synchronous (rendezvous) semantics.
    We focus on the asynchronous systems, where a message can be buffered while in
    transit, as most of the interactions in the real distributed world are asynchronous;
  \item We follow a top-down methodology, beginning with protocol specification
    using global types, whereas they follow~\cite{POPL19LessIsMore} to analyse
    only local types.
    Our method dispenses with the need to use a model checker.
    More specifically, it is not feasible to model check asynchronous systems
    with buffers, since the model may be infinite~\cite[Appendix \S
    G]{POPL19LessIsMore}.
  \item We present an toolkit in our work to implement multiparty session type APIs
that can handle crash-stop failures in \Scala, whereas they do not
produce any APIs for end-users to program with, as their toolchain mainly
concerns with validating types for safety and liveness properties via model-checking.
Their framework cannot be scaled to real distributed systems
since model-checking does not allow verification of infinite asynchronous queues.
\end{enumerate}

\textbf{Le Brun and Dardha~\cite{ESOP23MAGPi}} follow a similar framework
to~\cite{CONCUR22MPSTCrash}.
They model an asynchronous semantics, and support more patterns of failure,
including message losses, delays, reordering, as well as link failures and
network partitioning.
However, their typing system suffers from its genericity, when type-level
properties become undecidable~\cite[\S 4.4]{ESOP23MAGPi}.
Our work uses global types for guidance, and recovers decidability of
properties at a small expense of expressivity.

\textbf{Mostrous and Vasconcelos~\cite{LMCS18Affine}}  first proposed the affine approach
to failure handling.
They present a $\pi$-calculus with affine binary sessions (\ie limited to two participants)
and concomitant reduction rules that represent a
minimal extension to standard (binary) session types.
Their extension is primarily comprised of a \emph{cancel operator}, which is
semantically similar to our crash construct: it represents a process
that has terminated early.
Besides these similarities, our work differs from \cite{LMCS18Affine} in several ways:
\begin{enumerate}[nosep, leftmargin=*]
  \item We address multiparty protocols and sessions rather than binary session types;
  \item In \cite{LMCS18Affine}, a cancel operator can have an arbitrary session type;
  consequently, crashes are not visible at the type level.
  Instead, we type crashed session endpoints with the special type $\stStop$,
  which lets us model crashes in the type semantics, and
  helps us in ensuring that a process implements its failure handling as expected in its (global or local) type;
  \item The reduction rules of \cite{LMCS18Affine}  do not permit a process to
    terminate early arbitrarily: cancellations must be raised explicitly by the
    programmer (or automatically by attempting to receive messages from crashed
    endpoints);
  \item Finally, cancellations in \cite{LMCS18Affine} may be caught and handled via a \emph{do-catch}
construct. This construct catches only the first cancellation and cannot be nested, thus providing little help in handling failure across
multiple roles. Our global and local types seamlessly support protocols where the failure of a role is detected (and handled) while handling the failure of another role.
\end{enumerate}

\vspace{0.5em}
\noindent
\textbf{Fowler \etal~\cite{DBLP:journals/pacmpl/FowlerLMD19}} present a concurrent
$\lambda$-calculus, \emph{EGV}, with asynchronous session-typed communication
and exception handling, and implement their approach as parts of the
\textsc{Links} language.
Their approach is based on~\cite{LMCS18Affine}, and
therefore shares many of the same differences to our approach:
the use of the cancel operator and binary
session types, and the lack of a reduction rule enabling a process to crash
arbitrarily;
the cancel operator used in EGV takes an arbitrary session type
-- whereas we reflect the crashed status with the dedicated $\stStop$ type.
Similar to  our work, \cite{DBLP:journals/pacmpl/FowlerLMD19}
has asynchronous communication channels:
messages are queued when sent, and delivered at a later stage.

\vspace{0.5em}
\noindent
\textbf{Lagaillardie \etal~\cite{ECOOP22AffineMPST}} propose a framework of
\emph{affine} multiparty session types, and provide an implementation of affine
MPST in the \Rust programming language.
They utilise the affine type
system and \texttt{Result} types of \Rust, so that the type system enforces
that failures are handled.
In their system, a multiparty session can terminate prematurely.
While their theory can be used to model crash-stop failures, such failures
are not built in the semantics, so manual encoding of failures is necessary.
Moreover, there is no way to recover from a cancellation (\ie failure) besides
propagating the cancellation.
In our work, we provide the ability to follow a \emph{different} protocol
when a crash is detected, which gives rise to more flexibility and
expressivity.

\vspace{0.5em}
\noindent
\textbf{Coordinator model} approaches~\cite{DBLP:conf/forte/AdameitPN17,ESOP18CrashHandling}
often incorporate \emph{interrupt blocks} (or similar constructs)
to model crashes and failure handling.

\vspace{0.5em}
\noindent
\textbf{Adameit \etal~\cite{DBLP:conf/forte/AdameitPN17}} extend the standard MPST syntax with  \emph{optional blocks}, representing regions of a protocol that are susceptible to communication failures.
In their approach, if a process $\mpP$ expects a value from an optional block which fails, then a
\emph{default value} is provided to $\mpP$, so $\mpP$ can continue running.
This ensures termination and deadlock-freedom.
Although this approach does not feature an explicit reliable coordinator
process, we describe it here due to the inherent coordination required for multiple processes to start and end an optional block. %
Our approach differs in three key ways:
\begin{enumerate}[leftmargin=*, nosep]
  \item we model crash-stop process failures instead of impermanent link failures;
  \item we extend the semantics of communications in lieu of introducing a
    new syntactic construct to enclose the potentially crashing regions of a protocol.
    Our global type projections and typing context safety ensure that crash detection is performed at every pertinent communication point;
  \item we allow crashes to significantly affect the evolution of protocols:
    our global and local types can have crash detection branches specifying
    significantly different behaviours \wrt non-crashing executions.
    Conversely, the approach in~\cite{DBLP:conf/forte/AdameitPN17} does not
    discriminate between the presence and absence of failures: both have the
    same protocol in the optional block's continuation.
\end{enumerate}

\vspace{0.5em}
\noindent
\textbf{Viering \etal~\cite{ESOP18CrashHandling}} similarly extend the standard
global type syntax with a \emph{try-handle} construct, which is facilitated by the
presence of a reliable coordinator process, and via a construct to specify
reliable processes.
When the coordinator detects a failure, it broadcasts notifications to all remaining live processes;
then, the protocol proceeds according to the failure handling continuation
specified as part of the try-handle construct.
Our approach and \cite{ESOP18CrashHandling} share several modelling choices: crash-stop
semantics, perfect links, and the possibility of specifying reliable processes.
However, unlike \cite{ESOP18CrashHandling}, our approach does \emph{not} depend on a reliable coordinator that broadcasts failure notifications: all roles in a protocol can be unreliable, all processes may crash.

Besides the differences discussed above, we decided not to adopt coordinator processes nor failure broadcasts in order to avoid their inherent drawbacks.
The use of a coordinator requires additional run-time resources and increases the overall complexity of a distributed system.  Furthermore, the broadcasting of failure notifications introduces an effective synchronisation point for all roles, with additional overheads. Such synchronisation points may also make it harder to extend the theory to support scenarios
with unreliable communication.

\vspace{0.5em}
\noindent
\textbf{Other related MPST implementations} include~\cite{ECOOP21MPSTActor,ECOOP22MPSTScala}.
\textbf{Harvey \etal~\cite{ECOOP21MPSTActor}} design a framework for
MPST-guided, safe actor programming.
Whilst the MPST protocol does not include any failure handling, the actors may
fail or raise exceptions, which are handled in a similar way to what we
summarise as the affine technique.
\textbf{Cledou \etal~\cite{ECOOP22MPSTScala}} revisit API generation techniques in
\Scala for MPST\@.
In addition to the traditional local type/automata-based code
generation~\cite{ECOOP17MPST, FASE16EndpointAPI}, they propose a
technique based on sets of pomsets, utilising match
types~\cite{POPL22TypeLevelMatch} in \Scala 3.
Our approach sticks to the traditional approach guided by local types, but also
utilises match types.
Whereas the global and local types in their work allow for sequential
composition (\ie semi-colon operator $;$), we opt to stick to traditional MPST
for simplicity, since the main goal of our work is to model crash-stop
failures.

\vspace{0.5em}
\noindent
\textbf{Giallorenzo \etal~\cite{DBLP:journals/corr/abs-2005-09520}} present \Choral,
a programming language for choreographies (multiparty protocols).
\Choral supports the handling of \emph{local exceptions} in choreographies, which can be used to program reliable channels over unreliable networks, supervision mechanisms, and other methods for fallible communication.
They utilise automatic retries in the implementation of channel APIs . Additionally,
these APIs have equivalent versions that wrap results in \texttt{Result} objects.
These objects are essentially sum types that combine the transmitted value type with an error type, similar to the approach used in \Go and \Rust. As future work, we plan to integrate our theory in \Choral, enabling us to develop more reliable distributed applications.

\section{Proofs for \cref{sec:gtype}}

With regard to recursive global types, we define their \emph{unfolding} as
$\unfoldOne{\gtRec{\gtRecVar}{\gtG}} =
\unfoldOne{\gtG\subst{\gtRecVar}{\gtRec{\gtRecVar}{\gtG}}}$, and
$\unfoldOne{\gtG} = \gtG$ otherwise.
A recursive type $\gtRec{\gtRecVar}{\gtG}$ must be guarded (or
contractive), \ie the unfolding leads to a progressive prefix, \eg a
transmission.
Unguarded types, such as $\gtRec{\gtRecVar}{\gtRecVar}$ and
$\gtRec{\gtRecVar}{\gtRec{\gtRecVari}{\gtRecVar}}$, are excluded.
Similar definitions and requirements apply for local types.

\subsection{Active and Crashed Roles}

\begin{definition}[Active and Crashed Roles]
\label{def:active_crashed_roles}
We define $\gtRoles{\gtG}$ for the set of \emph{active} roles in a global
type $\gtG$, and $\gtRolesCrashed{\gtG}$ for the set of \emph{crashed} roles,
as the smallest set such that:
\[
{\small{
  \begin{array}{r@{\;}c@{\;}lr@{\;}c@{\;}l}
    \gtRoles{\gtComm{\roleP}{\roleQ}{i \in I}{\gtLab[i]}{\tyGround[i]}{\gtG[i]}
    }
    & =
    & \setenum{\roleP, \roleQ} \cup \bigcup\limits_{i \in I}{\gtRoles{\gtG[i]}}
    & \gtRolesCrashed{\gtComm{\roleP}{\roleQ}{i \in I}{\gtLab[i]}{\tyGround[i]}{\gtG[i]}
    }
    & =
    & \bigcup\limits_{i \in I}{\gtRolesCrashed{\gtG[i]}}
    \\
    \gtRoles{
      \gtComm{\roleP}{\roleQCrashed}{i \in I}{\gtLab[i]}{\tyGround[i]}{\gtG[i]}
    }
    & =
    & \setenum{\roleP} \cup \bigcup\limits_{i \in I}{\gtRoles{\gtG[i]}}
    & \gtRolesCrashed{
      \gtComm{\roleP}{\roleQCrashed}{i \in I}{\gtLab[i]}{\tyGround[i]}{\gtG[i]}
    }
    & =
    & \setenum{\roleQ} \cup \bigcup\limits_{i \in I}{\gtRolesCrashed{\gtG[i]}}
    \\
    \gtRoles{
      \gtCommTransit{\rolePMaybeCrashed}{\roleQ}{i \in I}{\gtLab[i]}{\tyGround[i]}{\gtG[i]}{j}
    }
    & =
    & \setenum{\roleQ} \cup \bigcup\limits_{i \in I}{\gtRoles{\gtG[i]}}
    & \gtRolesCrashed{
      \gtCommTransit{\rolePMaybeCrashed}{\roleQ}{i \in I}{\gtLab[i]}{\tyGround[i]}{\gtG[i]}{j}
    }
    & =
    & \bigcup\limits_{i \in I}{\gtRolesCrashed{\gtG[i]}}
    \\
    \gtRoles{\gtEnd} = \gtRoles{\gtRecVar}
    & =
    & \emptyset
    & \gtRolesCrashed{\gtEnd} = \gtRolesCrashed{\gtRecVar}
    & =
    & \emptyset
    \\
    \gtRoles{\gtRec{\gtRecVar}{\gtG}}
    & =
    & \gtRoles{\gtG\subst{\gtRecVar}{\gtRec{\gtRecVar}{\gtG}}}
    & \gtRolesCrashed{\gtRec{\gtRecVar}{\gtG}}
    & =
    & \gtRolesCrashed{\gtG\subst{\gtRecVar}{\gtRec{\gtRecVar}{\gtG}}}
  \end{array}
  }}
\]
\end{definition}
\begin{remark}
  Even if $\roleP$ occurs crashed in a communication in transit, we do not
  consider $\roleP$ as crashed unless it appears crashed in a continuation:
  \(
  \gtRolesCrashed{
    \gtCommTransit{\rolePCrashed}{\roleQ}{i \in I}{\gtLab[i]}{\tyGround[i]}{\gtG[i]}{j}
  }
  = \bigcup\limits_{i \in I}{\gtRolesCrashed{\gtG[i]}}.
  \)
\end{remark}

\begin{lemma}\label{lem:gtype:crash-remove-role}
  If \;$\roleP \in \gtRoles{\gtG}$, \;then \;$\roleP \notin
  \gtRoles{\gtCrashRole{\gtG}{\roleP}}$ and
  $\gtRoles{\gtCrashRole{\gtG}{\roleP}} \subseteq \gtRoles{\gtG}$.
\end{lemma}
\begin{proof}
  By induction on \cref{def:gtype:remove-role}. We detail interesting cases
  here:
\begin{enumerate}[leftmargin=*]
  \item
\[
  \textstyle
  \gtRoles{
    \gtCrashRole{
      (\gtCommSmall{\roleP}{\roleQ}{i \in I}{\gtLab[i]}{\tyGround[i]}{\gtG[i]})
    }{\roleP}
  }
  =
  \gtRoles{
    \gtCommTransit{\rolePCrashed}{\roleQ}{i \in I}{\gtLab[i]}{\tyGround[i]}{
      (\gtCrashRole{\gtG[i]}{\roleP})
    }{j}
  }
  =
  \setenum{\roleQ}
  \cup
  \bigcup\limits_{i \in I}{\gtRoles{\gtCrashRole{\gtG[i]}{\roleP}}}.
\]
\noindent
The required result follows by inductive hypothesis that $\roleP \notin
\gtRoles{\gtCrashRole{\gtG[i]}{\roleP}}$, and
$\gtRoles{\gtCrashRole{\gtG[i]}{\roleP}} \subseteq \gtRoles{\gtG[i]}$.

  \item
\[
  \gtRoles{
    \gtCrashRole{
      (\gtCommTransit{\rolePCrashed}{\roleQ}{i \in
      I}{\gtLab[i]}{\tyGround[i]}{\gtG[i]}{j})
    }{\roleQ}
  }
  =
  \gtRoles{
    \gtCrashRole{\gtG[j]}{\roleQ}
  }
\]
The required result follows by inductive hypothesis that $\roleQ \notin
\gtRoles{\gtCrashRole{\gtG[j]}{\roleQ}}$, and
$\gtRoles{\gtCrashRole{\gtG[j]}{\roleQ}} \subseteq \gtRoles{\gtG[j]}$.

\end{enumerate}
\noindent
The rest of the cases are similar or straightforward. \qedhere
\end{proof}

\begin{lemma}\label{lem:gtype:crashed-crash-remove-role}
  If \;$\roleP \in \gtRoles{\gtG}$, \;then\; %
  $\gtRolesCrashed{\gtCrashRole{\gtG}{\roleP}} \setminus \setenum{\roleP} \subseteq \gtRolesCrashed{\gtG}$.
\end{lemma}
\begin{proof}
By induction on \cref{def:gtype:remove-role}. We detail interesting cases here:
\begin{enumerate}[leftmargin=*]
 \item
\[
  \textstyle
  \gtRolesCrashed{
  \gtCrashRole{
      (\gtCommSmall{\roleP}{\roleQ}{i \in I}{\gtLab[i]}{\tyGround[i]}{\gtG[i]})
  }{\roleP}
  }
 =
 \gtRolesCrashed{
  \gtCommTransit{\rolePCrashed}{\roleQ}{i \in I}{\gtLab[i]}{\tyGround[i]}{
    (\gtCrashRole{\gtG[i]}{\roleP})
   }{j}
 }
  =
  \bigcup\limits_{i \in I}{\gtRolesCrashed{\gtCrashRole{\gtG[i]}{\roleP}}}.
\]
\noindent
The required result follows by inductive hypothesis that
$\gtRolesCrashed{\gtCrashRole{\gtG[i]}{\roleP}} \setminus \setenum{\roleP} \subseteq \gtRolesCrashed{\gtG[i]}$.
 \item
\[
  \gtRolesCrashed{
   \gtCrashRole{
         (\gtCommTransit{\rolePCrashed}{\roleQ}{i \in
      I}{\gtLab[i]}{\tyGround[i]}{\gtG[i]}{j})
    }{\roleQ}
 }
  =
  \gtRolesCrashed{
    \gtCrashRole{\gtG[j]}{\roleQ}
  }
\]
The required result follows by inductive hypothesis that
$\gtRolesCrashed{\gtCrashRole{\gtG[j]}{\roleQ}} \setminus \setenum{\roleQ} \subseteq \gtRolesCrashed{\gtG[j]}$.
\end{enumerate}
\noindent
The rest of the cases are similar or straightforward. \qedhere
\end{proof}

\begin{lemma}\label{lem:in-roles-not-end}
If \;$\gtG \neq \gtRec{\gtRecVar}{\gtGi}$ and\;
$\roleP \in  \gtRoles{\gtG}$, then
\;$\gtProj[\rolesR]{\gtG}{\roleP} \neq \stEnd$.
\end{lemma}
\begin{proof}
We know that $\gtG \neq \gtEnd$; otherwise, we may have $\gtRoles{\gtG} = \emptyset$, a contradiction
to $\roleP \in  \gtRoles{\gtG}$. By induction on the structure of $\gtG$:
\begin{itemize}[leftmargin=*]
\item Case $\gtG = \gtComm{\roleQ}{\roleR}{i \in I}{\gtLab[i]}{\tyGround[i]}{\gtG[i]}$: we perform case analysis on $\roleP$:
\begin{itemize}[leftmargin=*]
\item $\roleP = \roleQ$: we have  $\gtProj[\rolesR]{\gtG}{\roleP} =
 \stIntSum{\roleR}{i \in \setcomp{j \in I}{\stFmt{\stLab[j]} \neq \stCrashLab}}{ %
        \stChoice{\stLab[i]}{\tyGround[i]} \stSeq (\gtProj[\rolesR]{\gtG[i]}{\roleP})%
      } \neq \stEnd$.
 \item $\roleP = \roleR$: we have $\gtProj[\rolesR]{\gtG}{\roleP} =
 \stExtSum{\roleQ}{i \in I}{%
        \stChoice{\stLab[i]}{\tyGround[i]} \stSeq (\gtProj[\rolesR]{\gtG[i]}{\roleP})%
      }
      \neq \stEnd$.
  \item $\roleP \neq \roleQ$ and $\roleP \neq \roleR$: we have
  $\gtProj[\rolesR]{\gtG}{\roleP} =
   \stMerge{i \in I}{\gtProj[\rolesR]{\gtG[i]}{\roleP}}$.
   Since $\roleP \in \gtRoles{\gtG}$, $\roleP \neq \roleQ$, and $\roleP \neq \roleR$,
   there exists $j \in I$ such that $\roleP \in \gtRoles{\gtG[j]}$. Then, by applying inductive hypothesis,
   $\gtProj[\rolesR]{\gtG[j]}{\roleP} \neq \gtEnd$, and therefore, we have
   $\gtProj[\rolesR]{\gtG}{\roleP} =
    \stMerge{i \in I}{\gtProj[\rolesR]{\gtG[i]}{\roleP}} =
    \gtProj[\rolesR]{\gtG[j]}{\roleP} \,\stBinMerge\,
    \stMerge{i \in I \setminus \{j\} }{\gtProj[\rolesR]{\gtG[i]}{\roleP}} \neq \stEnd$.
  \end{itemize}
      
\item Case $\gtG =  \gtCommTransit{\roleQ}{\roleR}{i \in
          I}{\gtLab[i]}{\tyGround[i]}{\gtG[i]}{j}$: we perform case analysis on $\roleP$:
\begin{itemize}[leftmargin=*]
  \item $\roleP \neq \roleQ$ and $\roleP \neq \roleR$: we have
  $\gtProj[\rolesR]{\gtG}{\roleP} =
   \stMerge{i \in I}{\gtProj[\rolesR]{\gtG[i]}{\roleP}}$.
   Since $\roleP \in \gtRoles{\gtG}$, $\roleP \neq \roleQ$, and $\roleP \neq \roleR$,
   there exists $k \in I$ such that $\roleP \in \gtRoles{\gtG[k]}$. Then, by applying inductive hypothesis,
   $\gtProj[\rolesR]{\gtG[k]}{\roleP} \neq \gtEnd$, and therefore, we have
   $\gtProj[\rolesR]{\gtG}{\roleP} =
    \stMerge{i \in I}{\gtProj[\rolesR]{\gtG[i]}{\roleP}} =
    \gtProj[\rolesR]{\gtG[k]}{\roleP} \,\stBinMerge\,
    \stMerge{i \in I \setminus \{k\} }{\gtProj[\rolesR]{\gtG[i]}{\roleP}} \neq \stEnd$, as desired. Meanwhile, we obtain that 
    $\forall l \in I: \gtProj[\rolesR]{\gtG[l]}{\roleP} \neq \stEnd$. 
    
\item $\roleP = \roleQ$: we have  $\gtProj[\rolesR]{\gtG}{\roleP} = \gtProj[\rolesR]{\gtG[j]}{\roleP}$, which follows that 
$\gtProj[\rolesR]{\gtG}{\roleP} \neq \stEnd$ by the fact that $\forall l \in I: \gtProj[\rolesR]{\gtG[l]}{\roleP} \neq \stEnd$. 

\item $\roleP = \roleR$: we have $\gtProj[\rolesR]{\gtG}{\roleP} =
 \stExtSum{\roleQ}{i \in I}{%
        \stChoice{\stLab[i]}{\tyGround[i]} \stSeq (\gtProj[\rolesR]{\gtG[i]}{\roleP})%
      }
      \neq \stEnd$.
  \end{itemize}
   
\item Other cases are similar. 
  \qedhere     
\end{itemize}
\end{proof}

\begin{lemma}\label{lem:not-in-roles-end}
If \;$\roleP \notin  \gtRoles{\gtG}$ and\; 
$\roleP \notin \gtRolesCrashed{\gtG}$, then
\;$\gtProj[\rolesR]{\gtG}{\roleP} = \stEnd$.
\end{lemma}
\begin{proof}
By induction on the structure of $\gtG$:
\begin{itemize}[leftmargin=*]
     \item Case $\gtG = \gtComm{\roleQ}{\roleR}{i \in I}{\gtLab[i]}{\tyGround[i]}{\gtG[i]}$:
     since $\roleP \notin \gtRoles{\gtG}$ and $\roleP \notin \gtRolesCrashed{\gtG}$, we have $\roleP \neq \roleQ$,
     $\roleP \neq \roleR$, and for all $i \in I$, $\roleP \notin \gtRoles{\gtG[i]}$ and
     $\roleP \notin \gtRolesCrashed{\gtG[i]}$ by \Cref{def:active_crashed_roles}.
     Thus, $\gtProj[\rolesR]{\gtG}{\roleP} = \stMerge{i \in I}{\gtProj[\rolesR]{\gtG[i]}{\roleP}} = \stEnd$ by applying inductive hypothesis
     and $ \stEnd \,\stBinMerge\, \stEnd%
    \,=\,%
    \stEnd$.
         \item Case $\gtG = \gtRec{\gtRecVar}{\gtGi}$: since $\roleP \notin \gtRoles{\gtG}$ and $\roleP \notin \gtRolesCrashed{\gtG}$, we have
      $\roleP \notin \gtRoles{\gtGi}$ and $\roleP \notin \gtRolesCrashed{\gtGi}$ by \Cref{def:active_crashed_roles}.
      We have two further subcases to consider:
       \begin{itemize}[leftmargin=*]
       \item If $\fv{\gtRec{\gtRecVar}{\gtGi}} \neq \emptyset$, we have $\gtProj[\rolesR]{\gtG}{\roleP} =
       \stRec{\stRecVar}{(\gtProj[\rolesR]{\gtGi}{\roleP})} =
       \stRec{\stRecVar}{\stEnd} = \stEnd$ by applying inductive hypothesis.
       \item Otherwise, we have $\gtProj[\rolesR]{\gtG}{\roleP} = \stEnd$ immediately.
       \end{itemize}
\item Other cases are similar or trivial. 
\qedhere
\end{itemize}
\end{proof}

\subsection{Subtyping}
\begin{lemma}[Subtyping is Reflexive]\label{lem:reflexive-subtyping}
  For any closed, well-guarded local type $\stT$, $\stT \stSub \stT$ holds.
\end{lemma}
\begin{proof}
By induction on the structure of local type $\stT$. \qedhere
\end{proof}

\begin{lemma}[Subtyping is Transitive]\label{transitive-subtyping}
  For any closed, well-guarded local type $\stS$, $\stT$, $\stU$,
  if $\stS \stSub \stT$ and $\stT \stSub \stU$ hold, then $\stS \stSub \stU$
  holds.
\end{lemma}
\begin{proof}
By induction on the structure of local type $\stS$. \qedhere
\end{proof}

\begin{lemma}\label{lem:unfold-subtyping}
  For any closed, well-guarded local type $\stT$,
  \begin{enumerate*}
    \item $\unfoldOne{\stT} \stSub \stT$; and
    \item $\stT \stSub \unfoldOne{\stT}$.
  \end{enumerate*}
\end{lemma}
 \begin{proof}
 \begin{enumerate*}
     \item If $\stT = \stRec{\stRecVar}{\stTi}$,  $\unfoldOne{\stT} \stSub \stT$ holds by $\inferrule{\iruleStSubRecR}$. Otherwise,
      by \Cref{lem:reflexive-subtyping}.
    \item If $\stT = \stRec{\stRecVar}{\stTi}$, $\stT \stSub \unfoldOne{\stT}$ holds by $\inferrule{\iruleStSubRecL}$.  Otherwise,
       by \Cref{lem:reflexive-subtyping}.
       \qedhere
  \end{enumerate*}
 \end{proof}

\begin{lemma}\label{lem:merge-subtyping}
  Given a collection of mergable local types $\stT[i]$ ($i \in I$).
  For all $j \in I$, $\stMerge{i \in I}{\stT[i]} \stSub \stT[j]$ holds.
\end{lemma}
\begin{proof}
By \cref{transitive-subtyping}, we have that subtyping is transitive.
Hence, it is sufficient to show that a relation
$R = \setcomp{(\stT[1] \!\stBinMerge\! \stT[2], \stT[1]),
(\stT[1] \!\stBinMerge\! \stT[2], \stT[2])}
{\stT[1], \stT[2] \text{ are local types s.t. } \stT[1] \stBinMerge
\stT[2] \text{ is defined}}$
is a subset of the subtyping relation $\stSub$.  Then we apply induction on 
on the structure of $\stT[1] \,\stBinMerge\, \stT[2]$.  \qedhere
\end{proof}

\begin{lemma}\label{lem:merge-lower-bound}
  Given a collection of mergable local types $\stT[i]$ ($i \in I$).
  If for all $i \in I$, $\stT[i] \stSub \stS$ for some local type $\stS$,
  then $\stMerge{i \in I}{\stT[i]} \stSub \stS$.
\end{lemma}
\begin{proof}
The proof is similar to that of \cref{lem:merge-subtyping}. 
By induction on the structure of $\stT[1] \,\stBinMerge\, \stT[2]$.  \qedhere
 \end{proof}

\begin{lemma}\label{lem:merge-upper-bound}
  Given a collection of mergable local types $\stT[i]$ ($i \in I$).
  If for all $i \in I$, $\stS \stSub \stT[i]$ for some local type $\stS$,
  then $\stS \stSub \stMerge{i \in I}{\stT[i]}$.
\end{lemma}
\begin{proof}
The proof is similar to that of \cref{lem:merge-subtyping}. 
By induction on the structure of $\stS$.  \qedhere%
\end{proof}

\begin{lemma}\label{lem:subtype:merge-subty}
  Given two collections of mergable local types $\stS[i],  \stT[i]$ ($i \in I$).
  If for all $i \in I$, $\stS[i] \stSub \stT[i]$, then
  $\stMerge{i \in I} {\stS[i]} \stSub \stMerge{i \in I}{\stT[i]}$.
\end{lemma}
\begin{proof}
The proof is similar to that of \cref{lem:merge-subtyping}. %
By induction on the structure of $\stS[1] \,\stBinMerge\, \stS[2]$.   \qedhere
\end{proof}

\begin{lemma}\label{lem:proj-non-crashing-role-preserve}
  If \; $\roleP, \roleQ \in \gtRoles{\gtG}$ \; with \; $\roleP \neq \roleQ$ \;
  and \; $\roleQ \notin \rolesR$, \;
  then \;
  $
    \gtProj[\rolesR]{\gtG}{\roleP}
    \stSub
    \gtProj[\rolesR]{(\gtCrashRole{\gtG}{\roleQ})}{\roleP}
  $.
\end{lemma}
\begin{proof}
  We construct a relation $
  R = \setcomp{
    (\gtProj[\rolesR]{\gtG}{\roleP}
    ,
    \gtProj[\rolesR]{(\gtCrashRole{\gtG}{\roleQ})}{\roleP})
  }{
    \roleP, \roleQ \in \roleSet, \roleP \neq \roleQ, \roleQ \in \rolesR
  }$, and show that $R \subseteq \;\stSub$. By induction on the structure of $\gtG$:
  \begin{itemize}[leftmargin=*]
      \item Case $\gtG = \gtComm{\roleP}{\roleQ}{i \in I}{\gtLab[i]}{\tyGround[i]}{\gtG[i]}$: we perform case analysis on the role being projected upon:
      \begin{itemize}[leftmargin=*]
        \item
          On LHS,
          we have $\gtProj[\rolesR]{\gtG}{\roleP} =
          \stIntSum{\roleQ}{i \in \setcomp{j \in I}{\stFmt{\stLab[j]} \neq \stCrashLab}}{ %
            \stChoice{\stLab[i]}{\tyGround[i]} \stSeq (\gtProj[\rolesR]{\gtG[i]}{\roleP})%
          }%
          $.

          On RHS, we perform case analysis on the role being removed:

          \begin{enumerate}[leftmargin=*]
            \item we have
          $\gtCrashRole{\gtG}{\roleQ} =
          \gtComm{\roleP}{\roleQCrashed}{i \in I}{\gtLab[i]}{\tyGround[i]}{
            (\gtCrashRole{\gtG[i]}{\roleQ})
          }
          $, and thus\\ 
          $
          \gtProj[\rolesR]{(\gtCrashRole{\gtG}{\roleQ})}{\roleP} =
          \stIntSum{\roleQ}{i \in \setcomp{j \in I}{\stFmt{\stLab[j]} \neq \stCrashLab}}{ %
            \stChoice{\stLab[i]}{\tyGround[i]} \stSeq
            (\gtProj[\rolesR]{(\gtCrashRole{\gtG[i]}{\roleQ})}{\roleP})%
          }%
          $,
          apply \inferrule{\iruleStSubOut} and coinductive
          hypothesis.

            \item ($\roleR \neq \roleQ$) we have
          $\gtCrashRole{\gtG}{\roleR} =
          \gtComm{\roleP}{\roleQ}{i \in I}{\gtLab[i]}{\tyGround[i]}{
            (\gtCrashRole{\gtG[i]}{\roleR})
          }
          $, and thus\\
          $
          \gtProj[\rolesR]{(\gtCrashRole{\gtG}{\roleR})}{\roleP} =
          \stIntSum{\roleQ}{i \in \setcomp{j \in I}{\stFmt{\stLab[j]} \neq \stCrashLab}}{ %
            \stChoice{\stLab[i]}{\tyGround[i]} \stSeq
            (\gtProj[\rolesR]{(\gtCrashRole{\gtG[i]}{\roleR})}{\roleP})%
          }
          $, 
          apply $\inferrule{\iruleStSubOut}$ and coinductive
          hypothesis.
          \end{enumerate}
        \item
          On LHS,
          we have $\gtProj[\rolesR]{\gtG}{\roleQ} =
          \stExtSum{\roleP}{i \in I}{%
            \stChoice{\stLab[i]}{\tyGround[i]} \stSeq (\gtProj[\rolesR]{\gtG[i]}{\roleQ})%
          }%
          $.

          On RHS, we perform case analysis on the role being removed:
          \begin{enumerate}[leftmargin=*]
          \item we have $
          \gtCrashRole{\gtG}{\roleP} =
          \gtCommTransit{\rolePCrashed}{\roleQ}{i \in I}{\gtLab[i]}{\tyGround[i]}{
            (\gtCrashRole{\gtG[i]}{\roleP})
          }{j}
          $, and thus\\
          $
          \gtProj[\rolesR]{(\gtCrashRole{\gtG}{\roleP})}{\roleQ} =
          \stExtSum{\roleP}{i \in I}{
            \stChoice{\stLab[i]}{\tyGround[i]} \stSeq
            (\gtProj[\rolesR]{(\gtCrashRole{\gtG[i]}{\roleP})}{\roleQ})%
          }
          $,
          apply \inferrule{\iruleStSubIn} and coinductive
          hypothesis.

          \item ($\roleR \neq \roleP$) we have $
          \gtCrashRole{\gtG}{\roleR} =
          \gtComm{\roleP}{\roleQ}{i \in I}{\gtLab[i]}{\tyGround[i]}{
            (\gtCrashRole{\gtG[i]}{\roleR})
          }
          $, and thus\\
          $
          \gtProj[\rolesR]{(\gtCrashRole{\gtG}{\roleR})}{\roleQ} =
          \stExtSum{\roleP}{i \in I}{
            \stChoice{\stLab[i]}{\tyGround[i]} \stSeq
            (\gtProj[\rolesR]{(\gtCrashRole{\gtG[i]}{\roleR})}{\roleQ})%
          }
          $,
          apply $\inferrule{\iruleStSubIn}$ and coinductive
          hypothesis.
          \end{enumerate}
        \item ($\roleR \notin \setenum{\roleP, \roleQ}$)
          On LHS, we have $
          \gtProj[\rolesR]{\gtG}{\roleR} =
          \stMerge{i \in I}{\gtProj[\rolesR]{\gtG[i]}{\roleR}}%
          $.

          On RHS, we perform case analysis on the role being removed:

          \begin{enumerate}[leftmargin=*]
          \item we have $
          \gtCrashRole{\gtG}{\roleP} =
          \gtCommTransit{\rolePCrashed}{\roleQ}{i \in I}{\gtLab[i]}{\tyGround[i]}{
            (\gtCrashRole{\gtG[i]}{\roleP})
          }{j}
          $, and thus $
          \gtProj[\rolesR]{(\gtCrashRole{\gtG}{\roleP})}{\roleR} =
          \stMerge{i \in I}{
            (\gtProj[\rolesR]{(\gtCrashRole{\gtG[i]}{\roleP})}{\roleR})%
          }
          $,
          apply \cref{lem:subtype:merge-subty} and coinductive
          hypothesis.

          \item we have
          $\gtCrashRole{\gtG}{\roleQ} =
          \gtComm{\roleP}{\roleQCrashed}{i \in I}{\gtLab[i]}{\tyGround[i]}{
            (\gtCrashRole{\gtG[i]}{\roleQ})
          }
          $, and thus $
          \gtProj[\rolesR]{(\gtCrashRole{\gtG}{\roleQ})}{\roleR} =
          \stMerge{i \in I}{
            (\gtProj[\rolesR]{(\gtCrashRole{\gtG[i]}{\roleQ})}{\roleR})%
          }
          $,
          apply \cref{lem:subtype:merge-subty} and coinductive
          hypothesis.

          \item ($\roleS \notin \setenum{\roleP, \roleQ, \roleR}$) we have $
          \gtCrashRole{\gtG}{\roleS} =
          \gtComm{\roleP}{\roleQ}{i \in I}{\gtLab[i]}{\tyGround[i]}{
            (\gtCrashRole{\gtG[i]}{\roleS})
          }
          $, and thus $
          \gtProj[\rolesR]{(\gtCrashRole{\gtG}{\roleS})}{\roleR} =
          \stMerge{i \in I}{
            (\gtProj[\rolesR]{(\gtCrashRole{\gtG[i]}{\roleS})}{\roleR})%
          }
          $,
          apply \cref{lem:subtype:merge-subty} and coinductive
          hypothesis.

          \end{enumerate}
      \end{itemize}
          \item Case $\gtG = \gtRec{\gtRecVar}{\gtGi}$: 
      by coinductive hypothesis.
   \item  Other cases are similar or trivial. 
    \qedhere
  \end{itemize}
\end{proof}

\begin{lemma}[Inversion of Subtyping]\label{lem:subtyping-invert}
  ~
  \begin{enumerate}
    \item If
      $\stS \stSub
       \stIntSum{\roleP}{i \in I}{\stChoice{\stLab[i]}{\tyGround[i]} \stSeq \stT[i]}
      $, then
      $\unfoldOne{\stS} =
       \stIntSum{\roleP}{j \in J}{\stChoice{\stLabi[j]}{\tyGroundi[j]} \stSeq
       \stTi[j]}
      $, and $J \subseteq I$,
      and $\forall i \in J: \stLab[i] = \stLabi[i], \tyGround[i] =
      \tyGroundi[i]$ and $\stTi[i] \stSub \stT[i]$.
    \item If
      $\stS \stSub
       \stExtSum{\roleP}{i \in I}{\stChoice{\stLab[i]}{\tyGround[i]} \stSeq \stT[i]}
      $, then
      $\unfoldOne{\stS} =
       \stExtSum{\roleP}{j \in J}{\stChoice{\stLabi[j]}{\tyGroundi[j]} \stSeq
       \stTi[j]}
      $, and $I \subseteq J$,
      and $\forall i \in I: \stLab[i] = \stLabi[i], \tyGround[i] =
      \tyGroundi[i]$ and $\stTi[i] \stSub \stT[i]$.
  \end{enumerate}
\end{lemma}
\begin{proof}
By \cref{lem:unfold-subtyping}, the transitivity of subtyping, and \cref{def:subtyping} (\inferrule{\iruleStSubIn}, \inferrule{\iruleStSubOut}).   \qedhere
\end{proof}

\subsection{Semantics of Global Types}
\label{sec:proof:semantics:gty}

\begin{lemma}[No Revival Or Sudden Death]\label{lem:no-revival-roles}
  If
  \;$\gtWithCrashedRoles{\rolesC}{\gtG}
  \gtMove[\stEnvAnnotGenericSym]{\rolesR}
  \gtWithCrashedRoles{\rolesCi}{\gtGi}
  $,\;
  \begin{enumerate}
    \item If $\roleP \in \gtRolesCrashed{\gtGi}$ and $\stEnvAnnotGenericSym \neq
      \ltsCrash{}{\roleP}$, then $\roleP \in \gtRolesCrashed{\gtG}$;
      \label{item:crash-roles-remain-crash}
    \item If $\roleP \in \gtRoles{\gtGi}$ and $\stEnvAnnotGenericSym \neq
      \ltsCrash{}{\roleP}$, then $\roleP \in \gtRoles{\gtG}$;
      \label{item:live-roles-remain-live}
    \item If $\roleP \in \gtRolesCrashed{\gtGi}$ and $\stEnvAnnotGenericSym =
      \ltsCrash{}{\roleP}$, then $\roleP \in \gtRoles{\gtG}$.
      \label{item:crash-role-crash}
  \end{enumerate}
\end{lemma}
\begin{proof}
\begin{enumerate}[leftmargin=*]
  \item By induction on global type reductions: since $\stEnvAnnotGenericSym \neq
      \ltsCrash{}{\roleP}$, we start from \inferrule{\iruleGtMoveRec}.
      \begin{itemize}[leftmargin=*]
      \item Case \inferrule{\iruleGtMoveRec}: we have
      $\gtG = \gtRec{\gtRecVar}{\gtGii}$ and
      $\gtWithCrashedRoles{\rolesC}{\gtGii{}\subst{\gtRecVar}{\gtRec{\gtRecVar}{\gtGii}}}
      \gtMove[\stEnvAnnotGenericSym]{\rolesR} \gtWithCrashedRoles{\rolesCi}{\gtGi}$
    by \inferrule{\iruleGtMoveRec} and its inversion.
    Hence, by $\gtWithCrashedRoles{\rolesC}{\gtGii{}\subst{\gtRecVar}{\gtRec{\gtRecVar}{\gtGii}}}
    \gtMove[\stEnvAnnotGenericSym]{\rolesR} \gtWithCrashedRoles{\rolesCi}{\gtGi}$, $\roleP \in  \gtRolesCrashed{\gtGi}$, and inductive hypothesis,
    we have $\roleP \in \gtRolesCrashed{\gtGii{}\subst{\gtRecVar}{\gtRec{\gtRecVar}{\gtGii}}}$.
    Therefore,
    by $\gtRolesCrashed{\gtRec{\gtRecVar}{\gtGii}} =  \gtRolesCrashed{\gtGii{}\subst{\gtRecVar}{\gtRec{\gtRecVar}{\gtGii}}}$,
    we conclude with $\roleP \in \gtRolesCrashed{\gtG}$, as desired.
     \item Case \inferrule{\iruleGtMoveIn}: we have $\gtG =  \gtCommTransit{\rolePMaybeCrashed}{\roleQ}{i \in I}{\gtLab[i]}{\tyGround[i]}{\gtGi[i]}{j}$ and $\gtGi = \gtGi[j]$ by \inferrule{\iruleGtMoveIn}. It follows that $\gtRolesCrashed{\gtG} = \bigcup\limits_{i \in I}{\gtRolesCrashed{\gtGi[i]}}$
       and $\gtRolesCrashed{\gtGi} = \gtRolesCrashed{\gtGi[j]}$ with $j \in I$, and hence, $\gtRolesCrashed{\gtGi}  \subseteq \gtRolesCrashed{\gtG}$.
       Therefore, by $\roleP \in \gtRolesCrashed{\gtGi}$,
       we conclude with $\roleP \in \gtRolesCrashed{\gtG}$, as desired.
       
    \item Case \inferrule{\iruleGtMoveCtx}: we have $\gtG =  \gtCommSmall{\roleP}{\roleQMaybeCrashed}{i \in
      I}{\gtLab[i]}{\tyGround[i]}{\gtGi[i]}$, $\gtGi =  \gtCommSmall{\roleP}{\roleQMaybeCrashed}{i \in
      I}{\gtLab[i]}{\tyGround[i]}{\gtGii[i]}$,  $\forall i \in I :
    \gtWithCrashedRoles{\rolesC}{\gtGi[i]}
    \gtMove[\stEnvAnnotGenericSym]{\rolesR}
    \gtWithCrashedRoles{\rolesCi}{\gtGii[i]}$, and $\ltsSubject{\stEnvAnnotGenericSym} \notin \setenum{\roleP, \roleQ}$
    by \inferrule{\iruleGtMoveCtx} and its inversion. It follows that
    $\gtRolesCrashed{\gtG} = \bigcup\limits_{i \in I}{\gtRolesCrashed{\gtGi[i]}}$, $\gtRolesCrashed{\gtGi} = \bigcup\limits_{i \in I}{\gtRolesCrashed{\gtGii[i]}}$,
    and $\stEnvAnnotGenericSym \neq
      \ltsCrash{}{\roleP}$. Then by $\forall i \in I :
    \gtWithCrashedRoles{\rolesC}{\gtGi[i]}
    \gtMove[\stEnvAnnotGenericSym]{\rolesR}
    \gtWithCrashedRoles{\rolesCi}{\gtGii[i]}$, $\stEnvAnnotGenericSym \neq
      \ltsCrash{}{\roleP}$, and inductive hypothesis, we have $\forall i \in I : \text{if } \roleP \in \gtRolesCrashed{\gtGii[i]}, \text{then } \roleP \in \gtRolesCrashed{\gtGi[i]}$.
      Therefore, by $\roleP \in \bigcup\limits_{i \in I}{\gtRolesCrashed{\gtGii[i]}}$, we conclude with $\roleP \in \bigcup\limits_{i \in I}{\gtRolesCrashed{\gtGi[i]}} = \gtRolesCrashed{\gtG}$, as desired.
\item Other cases are similar.
      \end{itemize}
 \item Similar to the proof of (1).
 \item The proof is trivial by \inferrule{\iruleGtMoveCrash} and its inversion.
 \qedhere
\end{enumerate}
\end{proof}

\lemWellAnnoPreserve*
\begin{proof}
  By induction on global type reductions: 

  \begin{itemize}[leftmargin=*]
    \item Case \inferrule{\iruleGtMoveCrash}:  %
we have $\rolesCi = \rolesC \cup \setenum{\roleP}$, $\gtGi = \gtCrashRole{\gtG}{\roleP}$,  $\roleP \notin \rolesR$,
    $\roleP \in \gtRoles{\gtG}$, and $\gtG \neq \gtRec{\gtRecVar}{\gtGi}$ by \inferrule{\iruleGtMoveCrash} and its inversion.

    \cref{item:wa:reliable-no-crash}: %
from the premise, we have
    $\gtRolesCrashed{\gtG} \cap \rolesR = \emptyset$. Since $\roleP \in \gtRoles{\gtG}$, by \cref{lem:gtype:crashed-crash-remove-role},
    we have $\gtRolesCrashed{\gtCrashRole{\gtG}{\roleP}} \setminus \setenum{\roleP} \subseteq \gtRolesCrashed{\gtG}$. Then we consider two cases:
    \begin{itemize}[leftmargin=*]
    \item if $\roleP \in \gtRolesCrashed{\gtCrashRole{\gtG}{\roleP}}$, then
    $\gtRolesCrashed{\gtCrashRole{\gtG}{\roleP}} = \setenum{\roleP}
    \cup (\gtRolesCrashed{\gtCrashRole{\gtG}{\roleP}} \setminus \setenum{\roleP})$. Hence, by $\roleP \notin \rolesR$,
    $\gtRolesCrashed{\gtCrashRole{\gtG}{\roleP}} \setminus \setenum{\roleP} \subseteq \gtRolesCrashed{\gtG}$, and $\gtRolesCrashed{\gtG} \cap \rolesR = \emptyset$,
    we have $\gtRolesCrashed{\gtCrashRole{\gtG}{\roleP}} \cap \rolesR  = \emptyset$.
    \item  if $\roleP \notin \gtRolesCrashed{\gtCrashRole{\gtG}{\roleP}}$, then
    $\gtRolesCrashed{\gtCrashRole{\gtG}{\roleP}} = \gtRolesCrashed{\gtCrashRole{\gtG}{\roleP}} \setminus \setenum{\roleP}$. Hence, by
    $\gtRolesCrashed{\gtCrashRole{\gtG}{\roleP}} \setminus \setenum{\roleP} \subseteq \gtRolesCrashed{\gtG}$ and $\gtRolesCrashed{\gtG} \cap \rolesR = \emptyset$,
    we have $\gtRolesCrashed{\gtCrashRole{\gtG}{\roleP}} \cap \rolesR  = \emptyset$.
     \end{itemize}
    Therefore, by  $\gtGi = \gtCrashRole{\gtG}{\roleP}$, we conclude with $\gtRolesCrashed{\gtGi} \cap \rolesR = \emptyset$, as desired.

    \cref{item:wa:crash-annot-crash}:
   from the premise, we have
    $\gtRolesCrashed{\gtG} \subseteq \rolesC$. Since $\roleP \in \gtRoles{\gtG}$, by \cref{lem:gtype:crashed-crash-remove-role},
     we have $\gtRolesCrashed{\gtCrashRole{\gtG}{\roleP}} \setminus \setenum{\roleP} \subseteq \gtRolesCrashed{\gtG}$. Then we consider two cases:
    \begin{itemize}[leftmargin=*]
    \item if $\roleP \in \gtRolesCrashed{\gtCrashRole{\gtG}{\roleP}}$, then
    $\gtRolesCrashed{\gtCrashRole{\gtG}{\roleP}} = \setenum{\roleP}
    \cup (\gtRolesCrashed{\gtCrashRole{\gtG}{\roleP}} \setminus \setenum{\roleP})$. Hence, by
    $\gtRolesCrashed{\gtCrashRole{\gtG}{\roleP}} \setminus \setenum{\roleP} \subseteq \gtRolesCrashed{\gtG}$ and $\gtRolesCrashed{\gtG} \subseteq \rolesC$,
    we have $\gtRolesCrashed{\gtCrashRole{\gtG}{\roleP}} \subseteq \rolesC \cup \setenum{\roleP}$.
    \item  if $\roleP \notin \gtRolesCrashed{\gtCrashRole{\gtG}{\roleP}}$, then
    $\gtRolesCrashed{\gtCrashRole{\gtG}{\roleP}} = \gtRolesCrashed{\gtCrashRole{\gtG}{\roleP}} \setminus \setenum{\roleP}$. Hence, by
    $\gtRolesCrashed{\gtCrashRole{\gtG}{\roleP}} \setminus \setenum{\roleP} \subseteq \gtRolesCrashed{\gtG}$ and $\gtRolesCrashed{\gtG} \subseteq \rolesC$,
    we have $\gtRolesCrashed{\gtCrashRole{\gtG}{\roleP}} \subseteq \rolesC \cup \setenum{\roleP}$.
     \end{itemize}
   Therefore, by $\gtGi = \gtCrashRole{\gtG}{\roleP}$ and
    $\rolesCi = \rolesC \cup \{\roleP\}$, we conclude with $\gtRolesCrashed{\gtGi} \subseteq \rolesCi$, as desired.

    \cref{item:wa:live-no-crash}:
  from the premise, we have $\gtRoles{\gtG}
      \cap \gtRolesCrashed{\gtG} = \emptyset$. Since $\roleP \in \gtRoles{\gtG}$, by \cref{lem:gtype:crash-remove-role},  \cref{lem:gtype:crashed-crash-remove-role},
      we have $\gtRoles{\gtCrashRole{\gtG}{\roleP}} \subseteq \gtRoles{\gtG}$,
      $\roleP \notin \gtRoles{\gtCrashRole{\gtG}{\roleP}}$,
      and $\gtRolesCrashed{\gtCrashRole{\gtG}{\roleP}} \setminus \setenum{\roleP} \subseteq \gtRolesCrashed{\gtG}$. Then we consider two cases:
       \begin{itemize}[leftmargin=*]
    \item if $\roleP \in \gtRolesCrashed{\gtCrashRole{\gtG}{\roleP}}$, then
    $\gtRolesCrashed{\gtCrashRole{\gtG}{\roleP}} = \setenum{\roleP}
    \cup (\gtRolesCrashed{\gtCrashRole{\gtG}{\roleP}} \setminus \setenum{\roleP})$. Hence, by $\gtRoles{\gtG}
      \cap \gtRolesCrashed{\gtG} = \emptyset$, $\roleP \notin \gtRoles{\gtCrashRole{\gtG}{\roleP}}$, $\gtRolesCrashed{\gtCrashRole{\gtG}{\roleP}} \setminus \setenum{\roleP} \subseteq \gtRolesCrashed{\gtG}$,
      and $\gtRoles{\gtCrashRole{\gtG}{\roleP}} \subseteq \gtRoles{\gtG}$, we have $\gtRolesCrashed{\gtCrashRole{\gtG}{\roleP}}  \cap \gtRoles{\gtCrashRole{\gtG}{\roleP}} = \emptyset$.
     \item if $\roleP \notin \gtRolesCrashed{\gtCrashRole{\gtG}{\roleP}}$, then
    $\gtRolesCrashed{\gtCrashRole{\gtG}{\roleP}} = \gtRolesCrashed{\gtCrashRole{\gtG}{\roleP}} \setminus \setenum{\roleP}$. Hence, by $\gtRoles{\gtG}
      \cap \gtRolesCrashed{\gtG} = \emptyset$,
    $\gtRolesCrashed{\gtCrashRole{\gtG}{\roleP}} \setminus \setenum{\roleP} \subseteq \gtRolesCrashed{\gtG}$, and
     $\gtRoles{\gtCrashRole{\gtG}{\roleP}} \subseteq \gtRoles{\gtG}$,  we have $\gtRolesCrashed{\gtCrashRole{\gtG}{\roleP}}  \cap \gtRoles{\gtCrashRole{\gtG}{\roleP}} = \emptyset$.
    \end{itemize}
  Therefore, by $\gtGi = \gtCrashRole{\gtG}{\roleP}$, we conclude with
      $\gtRoles{\gtGi}
      \cap \gtRolesCrashed{\gtGi} = \emptyset$, as desired.

      \item Case \inferrule{\iruleGtMoveRec}:
     we have
      $\gtG = \gtRec{\gtRecVar}{\gtGii}$ and
      $\gtWithCrashedRoles{\rolesC}{\gtGii{}\subst{\gtRecVar}{\gtRec{\gtRecVar}{\gtGii}}}    
      \gtMove[\stEnvAnnotGenericSym]{\rolesR} \gtWithCrashedRoles{\rolesCi}{\gtGi}$
    by \inferrule{\iruleGtMoveRec} and its inversion.
     From the premise, we also have
    $\gtWithCrashedRoles{\rolesC}{\gtRec{\gtRecVar}{\gtGii}}$ is well-annotated.
    Hence, by $\gtRolesCrashed{\gtRec{\gtRecVar}{\gtGii}} =
    \gtRolesCrashed{\gtGii{}\subst{\gtRecVar}{\gtRec{\gtRecVar}{\gtGii}}}$,
    $\gtRoles{\gtRec{\gtRecVar}{\gtGii}} =
    \gtRoles{\gtGii{}\subst{\gtRecVar}{\gtRec{\gtRecVar}{\gtGii}}}$,
    and $\gtWithCrashedRoles{\rolesC}{\gtRec{\gtRecVar}{\gtGii}}$ is well-annotated,
    we have
    $\gtRolesCrashed{\gtGii{}\subst{\gtRecVar}{\gtRec{\gtRecVar}{\gtGii}}} \cap \rolesR = \emptyset$,
    $\gtRolesCrashed{\gtGii{}\subst{\gtRecVar}{\gtRec{\gtRecVar}{\gtGii}}} \subseteq \rolesC$,
    and $\gtRoles{\gtGii{}\subst{\gtRecVar}{\gtRec{\gtRecVar}{\gtGii}}} \cap
    \gtRolesCrashed{\gtGii{}\subst{\gtRecVar}{\gtRec{\gtRecVar}{\gtGii}}}= \emptyset$.
    It follows directly that
     $\gtWithCrashedRoles{\rolesC}{\gtGii{}\subst{\gtRecVar}{\gtRec{\gtRecVar}{\gtGii}}}$ is well-annotated.
     Therefore, by $\gtWithCrashedRoles{\rolesC}{\gtGii{}\subst{\gtRecVar}{\gtRec{\gtRecVar}{\gtGii}}}
    \gtMove[\stEnvAnnotGenericSym]{\rolesR} \gtWithCrashedRoles{\rolesCi}{\gtGi}$
    and inductive hypothesis, we conclude with $\gtWithCrashedRoles{\rolesCi}{\gtGi}$ is well-annotated, as desired.
     \item  Other cases are similar. 
  \qedhere
  \end{itemize}
\end{proof}

\begin{lemma}\label{lem:gt-lts-unfold}
  $\gtWithCrashedRoles{\rolesC}{\gtG}
  \gtMove[\stEnvAnnotGenericSym]{\rolesR}
  \gtWithCrashedRoles{\rolesCi}{\gtGi}$, \;iff
  \;$
  \gtWithCrashedRoles{\rolesC}{\unfoldOne{\gtG}}
  \gtMove[\stEnvAnnotGenericSym]{\rolesR}
  \gtWithCrashedRoles{\rolesCi}{\gtGi}$.
\end{lemma}
\begin{proof}
  By inverting or applying $\inferrule{\iruleGtMoveRec}$ when necessary.  \qedhere
\end{proof}

\begin{lemma}[Progress of Global Types]\label{lem:gtype:progress}
  If $\gtWithCrashedRoles{\rolesC}{\gtG}$ (where $\gtG$ is a projectable
  global type) is well-annotated, and $\gtG \neq
  \gtEnd$, then there exists $\gtGi, \rolesCi$ such that
  $\gtWithCrashedRoles{\rolesC}{\gtG}
  \gtMove{\rolesR}
  \gtWithCrashedRoles{\rolesCi}{\gtGi}$.
\end{lemma}
\begin{proof}
  By \cref{lem:gt-lts-unfold}, we only consider unfoldings.

  If $\unfoldOne{\gtG} = \gtEnd$, the premise does not hold.

  Case $
    \unfoldOne{\gtG} =
    \gtComm{\roleP}{\roleQ}{i \in I}{\gtLab[i]}{\tyGround[i]}{\gtG[i]}
  $, apply \inferrule{\iruleGtMoveOut}.
  We can
  pick any $\gtLab[i] \neq \gtCrashLab$ to reduce the global type (note that
  our syntax prohibits singleton $\gtCrashLab$ branches).

  Case $
    \unfoldOne{\gtG} =
    \gtComm{\roleP}{\roleQCrashed}{i \in I}{\gtLab[i]}{\tyGround[i]}{\gtG[i]}
  $, apply \inferrule{\iruleGtMoveOrph}.
  We can
  pick any $\gtLab[i] \neq \gtCrashLab$ to reduce the global type (note that
  our syntax prohibits singleton $\gtCrashLab$ branches).

  Case $
    \unfoldOne{\gtG} =
    \gtCommTransit{\rolePMaybeCrashed}{\roleQ}{i \in I}{\gtLab[i]}{\tyGround[i]}{\gtG[i]}{j}
  $.
  If $\gtLab[j] = \gtCrashLab$, then it $\roleP$ must have crashed, apply \inferrule{\iruleGtMoveCrDe}.
  Otherwise, apply \inferrule{\iruleGtMoveIn}.
\end{proof}

\subsection{Semantics of Configurations}

\begin{lemma}\label{lem:stenv-red:trivial-2}
  If \; $\stEnv; \qEnv \stEnvMoveMaybeCrash \stEnvi; \qEnvi$ \; with \;
  $\stEnv; \qEnv
  \stEnvMoveGenAnnot \stEnvi; \qEnvi$, \; then \;
  \begin{enumerate}
    \item $\dom{\stEnv} = \dom{\stEnvi}$; \; and \;
    \item for all \; $\roleP \in \dom{\stEnv}$ \; with \;
      $\roleP \neq \ltsSubject{\stEnvAnnotGenericSym}$, \; we have \;
      $\stEnvApp{\stEnv}{\roleP} =
      \stEnvApp{\stEnvi}{\roleP}$.
  \end{enumerate}
\end{lemma}
\begin{proof}
Trivial by induction on reductions of configuration. \qedhere
\end{proof}

\begin{lemma}\label{lem:stenv-red:trivial-3}
 If \,$\stEnv; \qEnv \stEnvMoveGenAnnot \stEnvi; \qEnvi$, \,then\,
 for any $\roleP, \roleQ \in \dom{\stEnv}$ with
 $\ltsSubject{\stEnvAnnotGenericSym} \notin \setenum{\roleP, \roleQ}$, we have
 $\stEnvApp{\qEnv}{\roleP, \roleQ} = \stEnvApp{\qEnvi}{\roleP, \roleQ}$.
\end{lemma}
\begin{proof}
Trivial by induction on reductions of configuration.
\end{proof}

\begin{lemma}[Inversion of Typing Context Reduction]\label{lem:stenv-red:inversion-basic}
  ~
  \begin{enumerate}
    \item
      If \,$
        \stEnv; \qEnv
        \stEnvMoveOutAnnot{\roleP}{\roleQ}{\stChoice{\stLab[k]}{\tyGround[k]}}
        \stEnvi; \qEnvi
      $\,, then \,
      $
        \unfoldOne{\stEnvApp{\stEnv}{\roleP}} =
        \stIntSum{\roleQ}{i \in I}{\stChoice{\stLab[i]}{\tyGround[i]} \stSeq \stT[i]}%
      $\,, \,$k \in I$\,, and \,
      $
        \stEnvApp{\stEnvi}{\roleP} = \stT[k]
      $;
    \item
      If \,$
        \stEnv; \qEnv
        \stEnvMoveInAnnot{\roleQ}{\roleP}{\stChoice{\stLab[k]}{\tyGround[k]}}
        \stEnvi; \qEnvi
      $\,, then \,
      $
        \unfoldOne{\stEnvApp{\stEnv}{\roleQ}} =
        \stExtSum{\roleP}{i \in I}{\stChoice{\stLab[i]}{\tyGround[i]} \stSeq \stT[i]}%
      $\,, \,$k \in I$\,, and \,
      $
        \stEnvApp{\stEnvi}{\roleQ} = \stT[k]
      $.
  \end{enumerate}
\end{lemma}
\begin{proof}
By applying and inverting \inferrule{\iruleTCtxOut} and \inferrule{\iruleTCtxIn}.
\end{proof}

\begin{lemma}[Determinism of Configuration Reduction]\label{lem:stenv-red:det}
  If \,$\stEnv; \qEnv \stEnvMoveGenAnnot \stEnvi; \qEnvi$ \,and\,
  $\stEnv; \qEnv \stEnvMoveGenAnnot
  \stEnvii; \qEnvii$, then $\stEnvi = \stEnvii$ and $\qEnvi = \qEnvii$.
\end{lemma}
\begin{proof}
Trivial by induction on reductions of configuration.
\end{proof}

\begin{lemma}\label{lem:stenv-queue-red:det}
 If \,$\stEnv[1]; \qEnv \stEnvMoveGenAnnot \stEnvi[1]; \qEnvi[1]$ \,and\,
  $\stEnv[2]; \qEnv \stEnvMoveGenAnnot
  \stEnvi[2]; \qEnvi[2]$, then $\qEnvi[1]= \qEnvi[2]$.
\end{lemma}
\begin{proof}
Trivial by induction on reductions of configuration.
\end{proof}

\subsection{Relating Semantics}%
\label{sec:proof:relating}
\begin{proposition}\label{prop:gt-lts-unfold-once}
  $\stEnvAssoc{
    \gtWithCrashedRoles{\rolesC}{\gtRec{\gtRecVar}{\gtG}}
  }{\stEnv; \qEnv}{\rolesR}$
  \,if and only if\,
 $\stEnvAssoc{
    \gtWithCrashedRoles{\rolesC}{\gtG{}\subst{\gtRecVar}{\gtRec{\gtRecVar}{\gtG}}}
    }{\stEnv; \qEnv}{\rolesR}$.
\end{proposition}
\begin{proof}
  By \cref{lem:unfold-subtyping}.
\end{proof}

\begin{proposition}\label{prop:lt-lts-unfold-once}
  If\;
  $\stEnvAssoc{
    \gtWithCrashedRoles{\rolesC}{\gtG}
  }{\stEnv; \qEnv}{\rolesR}$
  and
  $\stEnvApp{\stEnv}{\roleP} = \stRec{\stRecVar}{\stT}$,
  \,then\,
  $\stEnvAssoc{
    \gtWithCrashedRoles{\rolesC}{\gtG}
  }{
      \stEnvUpd{\stEnv}{\roleP}{
        \stT\subst{\stRecVar}{\stRec{\stRecVar}{\stT}}%
      }
  ; \qEnv}{\rolesR}$.
\end{proposition}
\begin{proof}
  By \cref{lem:unfold-subtyping}.
\end{proof}

\begin{proposition}\label{prop:gt-lts-unfold}
  $\stEnvAssoc{\gtWithCrashedRoles{\rolesC}{\gtG}}{\stEnv; \qEnv}{\rolesR}$
  \,if and only if\,
  $\stEnvAssoc{\gtWithCrashedRoles{\rolesC}{\unfoldOne{\gtG}}}{\stEnv; \qEnv}{\rolesR}$.
\end{proposition}
\begin{proof}
  By applying \cref{prop:gt-lts-unfold-once} as many times as necessary.
\end{proof}

\begin{lemma}[Inversion of Projection]\label{lem:inv-proj}
  Given a local type $\stS$, which is a subtype of projection from a global type
  $\gtG$ on a role $\roleP$ with respect to a set of reliable roles $\rolesR$,
  \ie $\stS \stSub (\gtProj[\rolesR]{\gtG}{\roleP})$,
  then:
  \begin{enumerate}
    \item
      \label{item:proj-inv:send}
      If\;
      $\unfoldOne{\stS}
      =
      \stIntSum{\roleQ}{i \in I}{\stChoice{\stLab[i]}{\tyGround[i]} \stSeq
        \stSi[i]}$, then either
        \begin{enumerate}
        \item
          $\unfoldOne{\gtG} =
            \gtComm{\roleP}{\roleQMaybeCrashed}{i \in I'}{\gtLab[i]}{\tyGroundi[i]}{\gtG[i]}$,
          where
          $I \subseteq I'$, and
          for all $i \in I$:
          $\stLab[i] = \gtLab[i]$,
          $\stSi[i] \stSub (\gtProj[\rolesR]{\gtG[i]}{\roleP})$,
          and
          $\tyGround[i] = \tyGroundi[i]$;
          or,
        \item
          $\unfoldOne{\gtG} =
            \gtComm{\roleS}{\roleTMaybeCrashed}{j \in J}{\gtLab[j]}{\tyGroundi[j]}{\gtG[j]}$,
          or
          $\unfoldOne{\gtG} =
            \gtCommTransit{\roleSMaybeCrashed}{\roleT}{j \in
            J}{\gtLab[j]}{\tyGroundi[j]}{\gtG[j]}{k}$,
          where for all $j \in J$:
          $\stS \stSub (\gtProj[\rolesR]{\gtG[j]}{\roleP})$,
          with $\roleP \neq \roleS$ and $\roleP \neq \roleT$.
        \end{enumerate}
    \item
      \label{item:proj-inv:recv}
      If\;
      $\unfoldOne{\stS}
      =
      \stExtSum{\roleQ}{i \in I}{\stChoice{\stLab[i]}{\tyGround[i]} \stSeq
        \stSi[i]}$, then either
      \begin{enumerate}
        \item\label{item:proj-inv:recv-equal}
          $\unfoldOne{\gtG} =
            \gtComm{\roleQ}{\rolePMaybeCrashed}{i \in I'}{\gtLab[i]}{\tyGroundi[i]}{\gtG[i]}$,
          or
          $\unfoldOne{\gtG} =
            \gtCommTransit{\roleQMaybeCrashed}{\roleP}{i \in
            I'}{\gtLab[i]}{\tyGroundi[i]}{\gtG[i]}{j}$,
          where
          $I' \subseteq I$, and
          for all $i \in I'$:
          $\stLab[i] = \gtLab[i]$,
          $\stSi[i] \stSub (\gtProj[\rolesR]{\gtG[i]}{\roleP})$,
          $\tyGroundi[i] = \tyGround[i]$,
          and $\roleQ \notin \rolesR$ implies $\exists k \in I': \gtLab[k] =
          \gtCrashLab$;
          or,
        \item
          $\unfoldOne{\gtG} =
            \gtComm{\roleS}{\roleTMaybeCrashed}{j \in J}{\gtLab[j]}{\tyGroundi[j]}{\gtG[j]}$,
          or
          $\unfoldOne{\gtG} =
            \gtCommTransit{\roleSMaybeCrashed}{\roleT}{j \in
            J}{\gtLab[j]}{\tyGroundi[j]}{\gtG[j]}{k}$,
          where for all $j \in J$:
          $\stS \stSub (\gtProj[\rolesR]{\gtG[j]}{\roleP})$,
          with $\roleP \neq \roleS$ and $\roleP \neq \roleT$.
      \end{enumerate}
    \item
      \label{item:proj-inv:end}
      If\;
      $\stS = \stEnd$,
      then $\roleP \notin \gtRoles{\gtG}$.
  \end{enumerate}
\end{lemma}
\begin{proof}
By the definition of global type projection (\cref{def:global-proj}). 
\end{proof}

\begin{lemma}[Existence of Crash Handling Branch Under Projection]\label{lem:crash-lab-exists}
  If role $\roleQ$ is unreliable, $\roleQ \notin \rolesR$, and a global type
  projects onto $\roleP$ with an external choice from $\roleQ$,
  $ \unfoldOne{\gtProj[\rolesR]{\gtG}{\roleP}} =
    \stExtSum{\roleQ}{i \in I}{\stChoice{\stLab[i]}{\stS[i]} \stSeq \stSi[i]}
  $, then there must be a crash handling branch,
  $\exists j \in I : \stLab[j] = \stCrashLab$.
\end{lemma}
\begin{proof}
  By induction on \cref{item:proj-inv:recv} of \cref{lem:inv-proj}.
\end{proof}

\begin{lemma}
\label{lem:eventual-global-form}
If a local type $\stT$ is a subtype of an external choice with a matching role, obtained
via projection from a global type $\gtG$, i.e.,
$\stT = \stExtSum{\roleP}{i \in I}{\stChoice{\stLab[i]}{\tyGround[i]} \stSeq \stT[i]}
\stSub
\gtProj[\rolesR]{\gtG}{\roleQ}$, $\unfoldOne{\gtG}$ is of the form
$ \gtComm{\roleS}{\roleTMaybeCrashed}{j \in J}{\gtLab[j]}{\tyGroundi[j]}{\gtG[j]}$ or
$  \gtCommTransit{\roleSMaybeCrashed}{\roleT}{j \in
            J}{\gtLab[j]}{\tyGroundi[j]}{\gtG[j]}{k}$, and a queue environment $\qEnv$ is associated
 with $\gtWithCrashedRoles{\rolesC}{\gtG}$, then there exists a global type $\gtWithCrashedRoles{\rolesCi}{\gtGi}$ and a queue environment $\qEnvi$ such that
 $\gtProj[\rolesR]{\gtG}{\roleQ} \stSub
 \gtProj[\rolesR]{\gtGi}{\roleQ}$, $\unfoldOne{\gtGi}$ is of the form
 $ \gtCommTransit{\rolePMaybeCrashed}{\roleQ}{i \in
            I'}{\gtLab[i]}{\tyGroundi[i]}{\gtG[i]}{j}$ or
 $ \gtComm{\roleP}{\roleQMaybeCrashed}{i \in I'}{\gtLab[i]}{\tyGroundi[i]}{\gtG[i]}$, and $\qEnvi$ is
 associated with $\gtWithCrashedRoles{\rolesCi}{\gtGi}$ with $\stEnvApp{\qEnvi}{\roleP, \roleQ} = \stEnvApp{\qEnv}{\roleP, \roleQ}$.
\end{lemma}
\begin{proof}
 By induction on \cref{item:proj-inv:recv} of \cref{lem:inv-proj}.

 Apply \cref{item:proj-inv:recv} of \cref{lem:inv-proj} on the premise, we have $\forall j \in J: \stT =
 \stExtSum{\roleP}{i \in I}{\stChoice{\stLab[i]}{\tyGround[i]} \stSeq \stT[i]}
\stSub
\gtProj[\rolesR]{\gtG[j]}{\roleQ}$ with $\roleQ \neq \roleS$ and $\roleQ \neq \roleT$, which follows that
$\gtProj[\rolesR]{\gtG}{\roleQ} = \stMerge{j \in J}{\gtProj[\rolesR]{\gtG[j]}{\roleQ}}$.
Then, by \cref{lem:merge-subtyping}, we get $\forall j \in J: \gtProj[\rolesR]{\gtG}{\roleQ} \stSub
\gtProj[\rolesR]{\gtG[j]}{\roleQ}$.
 We take an arbitrary $\gtG[j]$ with $j \in J$. By applying \cref{item:proj-inv:recv} of \cref{lem:inv-proj} again,
we have two cases.
\begin{itemize}[leftmargin=*]
\item Case (1):
$\unfoldOne{\gtG[j]} = \gtCommTransit{\rolePMaybeCrashed}{\roleQ}{i \in
            I'}{\gtLab[i]}{\tyGroundi[i]}{\gtG[i]}{l}$ or
$\unfoldOne{\gtG[j]} =  \gtComm{\roleP}{\roleQMaybeCrashed}{i \in I'}{\gtLab[i]}{\tyGroundi[i]}{\gtG[i]}$.

Since $\gtProj[\rolesR]{\gtG}{\roleQ} \stSub \gtProj[\rolesR]{\gtG[j]}{\roleQ} $,
we only need to show that there exists $\qEnvi$ such that
$\qEnvi$ is associated with $\gtWithCrashedRoles{\rolesC}{\gtG[j]}$ and
$\stEnvApp{\qEnvi}{%
      \roleP,  \roleQ%
      }  = \stEnvApp{\qEnv}{%
     \roleP,   \roleQ%
      }$. We consider two subcases:
      \begin{itemize}[leftmargin=*]
      \item $\unfoldOne{\gtG} =
      \gtComm{\roleS}{\roleTMaybeCrashed}{j \in J}{\gtLab[j]}{\tyGroundi[j]}{\gtG[j]}$: %
      by \cref{def:assoc-queue}, we have that  $\qEnv$ is associated with $\gtG[j]$. We take $\qEnvi$ as $\qEnv$.
      \item $\unfoldOne{\gtG} =  \gtCommTransit{\roleSMaybeCrashed}{\roleT}{j \in
            J}{\gtLab[j]}{\tyGroundi[j]}{\gtG[j]}{k}$, we perform case analysis on $\gtLab[k] = \gtCrashLab$ and
            $\gtLab[k] \neq \gtCrashLab$:
            \begin{itemize}[leftmargin=*]
            \item $\gtLab[k] = \gtCrashLab$: by \cref{def:assoc-queue}, we have that  $\qEnv$ is associated with $\gtWithCrashedRoles{\rolesC}{\gtG[j]}$.
            We take $\qEnvi$ as $\qEnv$.
            \item $\gtLab[k] \neq \gtCrashLab$: by \cref{def:assoc-queue}, we have that  $\stEnvApp{\qEnv}{\roleS, \roleT} =
        \stQCons{\stQMsg{\gtLab[k]}{\tyGroundi[k]}}{\stQ}$ and
        $\stEnvUpd{\qEnv}{\roleS, \roleT}{\stQ}$ is associated with 
        $\gtWithCrashedRoles{\rolesC}{\gtG[j]}$. We take $\qEnvi = \stEnvUpd{\qEnv}{\roleS, \roleT}{\stQ}$. Since $\roleQ \neq \roleS$
        and $\roleQ \neq \roleT$, it is straightforward that $\stEnvApp{\qEnvi}{\roleP, \roleQ} = \stEnvApp{\qEnv}{\roleP, \roleQ}$, as required.
            \end{itemize}
       \end{itemize}
\item Case (2): $\unfoldOne{\gtG[j]} =  \gtComm{\roleR}{\roleUMaybeCrashed}{l \in L}{\gtLab[l]}{\tyGroundi[l]}{\gtG[l]}$ or
$  \gtCommTransit{\roleRMaybeCrashed}{\roleU}{l \in
            L}{\gtLab[l]}{\tyGroundi[l]}{\gtG[l]}{k}$.

            We can construct a queue environment $\qEnvi$ as in case (1), which is associated with $ \gtWithCrashedRoles{\rolesC}{\gtG[j]}$.
            The thesis is then proved by applying inductive hypothesis on $\gtWithCrashedRoles{\rolesC}{\gtG[j]}$ and $\qEnvi$.
\qedhere
\end{itemize}
\end{proof}

\begin{lemma}\label{lem:queue-assoc-crash}
  If $\qEnv$ is associated with $\gtWithCrashedRoles{\rolesC}{\gtG}$ and
  $\roleP \in \gtRoles{\gtG}$, then
  $\stEnvUpd{\qEnv}{\cdot, \roleP}{\stQUnavail}$ is associated with
  $\gtWithCrashedRoles{\rolesC \cup
  \setenum{\roleP}}{\gtCrashRole{\gtG}{\roleP}}$.
\end{lemma}

\begin{proof}
  We denote $\qEnvi = \stEnvUpd{\qEnv}{\cdot, \roleP}{\stQUnavail}$ in the
  subsequent proof.
  To show association, there are two parts: namely a shape-dependent
  part, and a shape-independent part.

  We first show shape-independent part, which shared for all cases:
  that a crashed role $\roleR$ is in $\rolesC \cup
  \setenum{\roleP}$ iff $\stEnvApp{\qEnvi}{\cdot, \roleR} = \stQUnavail$.
  It follows that the roles $\roleR \in \rolesC$ have the requirements
  satisfied from the premise, and we set $\qEnvi = \stEnvUpd{\qEnv}{\cdot,
  \roleP}{\stQUnavail}$, and that $\roleP$ is in the new set of crashed
  roles.

  Shape-dependent part are by induction on the definition of
  $\gtCrashRole{\gtG}{\roleP}$:
  \begin{itemize}[leftmargin=*]
    \item Case
    \(
    \gtCrashRole{
      (\gtCommSmall{\roleP}{\roleQ}{i \in I}{\gtLab[i]}{\tyGround[i]}{\gtG[i]})
    }{
      \roleP
    }
    =
    \gtCommTransit{\rolePCrashed}{\roleQ}{i \in I}{\gtLab[i]}{\tyGround[i]}{
      (\gtCrashRole{\gtG[i]}{\roleP})}{j}
    \) where $j \in I$ and $\gtLab[j] = \gtCrashLab$:

    Since $\qEnv$ is associated with $\gtWithCrashedRoles{\rolesC}{\gtG}$, we that
    $\forall i \in I:$ $\qEnv$ is associated with
    $\gtWithCrashedRoles{\rolesC}{\gtGi[i]}$, and $\stEnvApp{\qEnv}{\roleP, \roleQ} =
    \stQEmpty$.

    By inductive hypothesis, we have $\qEnvi$ is associated with
    $\gtWithCrashedRoles{\rolesC}{(\gtCrashRole{\gtG[i]}{\roleP})}$.
    Moreover, since $\gtLab[j] = \gtCrashLab$, we can show that
    $\stEnvApp{\qEnvi}{\roleP, \roleQ} = \stEnvApp{\qEnv}{\roleP, \roleQ} =
    \stQEmpty$.
       \item Case
    \(
    \gtCrashRole{
      (\gtCommTransit{\roleP}{\roleQ}{i \in
      I}{\gtLab[i]}{\tyGround[i]}{\gtG[i]}{j})
    }{
      \roleQ
    }
    =
    \gtCrashRole{\gtG[j]}{\roleQ}
    \): 

    Subcase $\gtLab[j] = \gtCrashLab$: by inductive hypothesis, we know $\qEnvi$ is associated with
    $\gtWithCrashedRoles{\rolesC \cup \setenum{\roleQ}}{\gtGi[j]}$, as
    required.

    Subcase $\gtLab[j] \neq \gtCrashLab$: from association, we have
    $\stEnvApp{\qEnv}{\roleP, \roleQ} =
    \stQCons{\stQMsg{\gtLab[j]}{\tyGround[j]}}{\stQ}$
    and
    $\stEnvUpd{\qEnv}{\roleP, \roleQ}{\stQ}$
    is associated with
    $\gtWithCrashedRoles{\rolesC}{\gtGi[j]}$.

    By inductive hypothesis, we know $
    \stEnvUpd{
      \stEnvUpd{\qEnv}{\roleP, \roleQ}{\stQ}
    }{\cdot, \roleQ}{\stQUnavail}$
    is associated with
    $\gtWithCrashedRoles{\rolesC \cup
    \setenum{\roleQ}}{\gtCrashRole{\gtGi[j]}{\roleQ}}$.

    Since
    $ \stEnvUpd{
      \stEnvUpd{\qEnv}{\roleP, \roleQ}{\stQ}
    }{\cdot, \roleQ}{\stQUnavail} =
    \stEnvUpd{\qEnv}{\cdot, \roleQ}{\stQUnavail} = \qEnvi$,
    so $\qEnvi$ is associated with $\gtWithCrashedRoles{\rolesC \cup
    \setenum{\roleQ}}{\gtCrashRole{\gtGi[j]}{\roleQ}}$.
    \item Case
    \(
    \gtCrashRole{
      (\gtCommSmall{\roleP}{\roleQCrashed}{i \in I}{\gtLab[i]}{\tyGround[i]}{\gtG[i]})
    }{
      \roleP
    }
    =
    \gtCrashRole{\gtG[j]}{\roleP}
    \), where $j \in I$ and $\gtLab[j] = \gtCrashLab$: 

    Since $\qEnv$ is associated with $\gtWithCrashedRoles{\rolesC}{\gtG}$, we
    have that $\qEnv$ is associated with
    $\gtWithCrashedRoles{\rolesC}{\gtG[j]}$.

    By inductive hypothesis, we have $\qEnvi$ is associated with
    $\gtWithCrashedRoles{\rolesC \cup
    \setenum{\roleP}}{\gtCrashRole{\gtG[j]}{\roleP}}$, as required.
 \item Other cases follows directly by inductive hypothesis or trivially.
    \qedhere
  \end{itemize}
\end{proof}

\begin{lemma}\label{lem:stenv-red:non-trivial}
  If\,
   $\stEnv; \qEnv
  \stEnvMoveGenAnnot \stEnvi; \qEnvi$,\,
  $\stEnvAssoc{\gtWithCrashedRoles{\rolesC}{\gtG}}{\stEnv; \qEnv}{\rolesR}$,\,
  $\stEnvAssoc{\gtWithCrashedRoles{\rolesC}{\gtG[1]}}{\stEnv[1]; \qEnv}{\rolesR}$,\,
  $\stEnvApp{\stEnv}{\ltsSubject{\stEnvAnnotGenericSym}} = \stEnvApp{\stEnv[1]}{\ltsSubject{\stEnvAnnotGenericSym}}$,\,
  and\,
  $\stEnvi[1] =  \stEnvUpd{\stEnv[1]}{\ltsSubject{\stEnvAnnotGenericSym}}{\stEnvApp{\stEnvi}{\ltsSubject{\stEnvAnnotGenericSym}}}$,\,
   then\,
$\stEnv[1]; \qEnv
  \stEnvMoveGenAnnot \stEnvi[1]; \qEnvi$.
 \end{lemma}
  \begin{proof}
Trivial by induction on reductions of configuration. %
\end{proof}

We define
$\stIdxRemoveCrash{I}{\stLab[i]}$ to be an index set with the special label
$\stCrashLab$ removed, \ie
$\stIdxRemoveCrash{I}{\stLab[i]} = \setcomp{i \in I}{\stLab[i] \neq \stCrashLab}$.

\thmProjSoundness*
\begin{proof}
  By induction on reductions of global type
  \(
    \gtWithCrashedRoles{\rolesC}{\gtG}
    \gtMove[\stEnvAnnotGenericSym]{\rolesR}
    \gtWithCrashedRoles{\rolesCi}{\gtGi}
  \).
  \begin{itemize}[leftmargin=*]
    \item Case \inferrule{\iruleGtMoveCrash}:

      From the premise we have
      \begin{gather}
        \stEnvAssoc{\gtWithCrashedRoles{\rolesC}{\gtG}}{\stEnv; \qEnv}{\rolesR}
        \label{eq:soundness:movecrash_assoc_pre}
        \\
        \gtWithCrashedRoles{\rolesC}{\gtG}
        \gtMove{\rolesR}
        \\
        \roleP \notin \rolesR
        \\
        \roleP \in \gtRoles{\gtG}
        \\
        \gtG \neq \gtRec{\gtRecVar}{\gtGi}
        \\
        \stEnvAnnotGenericSym = \ltsCrashSmall{\mpS}{\roleP}
        \\
        \rolesCi = \rolesC \cup \setenum{\roleP}
        \label{eq:soundness:movecrash_rolesCi}
        \\
        \gtGi = \gtCrashRole{\gtG}{\roleP}
      \end{gather}

      Let $\stEnvi; \qEnvi = \stEnvUpd{\stEnv}{\roleP}{\stStop};
      \stEnvUpd{\qEnv}{\cdot, \roleP}{\stQUnavail}$.
      We show
        $\stEnvAssoc{\gtWithCrashedRoles{\rolesCii}{\gtGii}}{\stEnvi;
        \qEnvi}{\rolesR}$\;
      and
      \;$\stEnv; \qEnv \stEnvMoveGenAnnot \stEnvi; \qEnvi$.

      For the first part:
      By association~\eqref{eq:soundness:movecrash_assoc_pre}, we know that
      $\stEnvApp{\stEnv}{\roleP} \stSub \gtProj[\rolesR]{\gtG}{\roleP}$.
      By   $\gtG \neq \gtRec{\gtRecVar}{\gtGi}$, $\roleP \in \gtRoles{\gtG}$, and \cref{lem:in-roles-not-end}, we know that
      $\gtProj[\rolesR]{\gtG}{\roleP} \neq \stEnd$, which gives $\stEnvApp{\stEnv}{\roleP} \neq \stEnd$. 
      By $\roleP \in \gtRoles{\gtG}$, we also know that $\roleP \notin \rolesC$, which gives 
      $\stEnvApp{\stEnv}{\roleP} \neq \stStop$.
      Thus we apply \inferrule{\iruleTCtxCrash} to obtain the thesis.

      For the second part:
        \begin{enumerate}[label=(A\arabic*)]
          \item
            From association~\eqref{eq:soundness:movecrash_assoc_pre}, we know
            that
            $\forall \roleQ \in \gtRoles{\gtG} : \stEnvApp{\stEnv}{\roleQ}
            \stSub \gtProj[\rolesR]{\gtG}{\roleQ}$.

            By \cref{lem:gtype:crash-remove-role}, we know
            $\gtRoles{\gtCrashRole{\gtG}{\roleP}} \subseteq \gtRoles{\gtG}$.

            For any $\roleQ \in \gtRoles{\gtCrashRole{\gtG}{\roleP}}$,
            we apply \cref{lem:proj-non-crashing-role-preserve}, to obtain
            \(
            \gtProj[\rolesR]{\gtG}{\roleQ}
            \stSub
            \gtProj[\rolesR]{(\gtCrashRole{\gtG}{\roleP})}{\roleQ}
            \)

            Thus we have, by transitivity of subtyping,
            $\forall \roleQ \in \gtRoles{\gtCrashRole{\gtG}{\roleP}} :
            \stEnvApp{\stEnv}{\roleQ}
            \stSub \gtProj[\rolesR]{\gtG}{\roleQ}
            \stSub \gtProj[\rolesR]{(\gtCrashRole{\gtG}{\roleP})}{\roleQ}
            $, as required.

          \item
            From association~\eqref{eq:soundness:movecrash_assoc_pre}, we know
            that
            $\forall \roleQ \in \rolesC : \stEnvApp{\stEnv}{\roleQ} = \stStop$,
            and they are unchanged in $\stEnvi$.

            Moreover, we have updated $\stEnvApp{\stEnvi}{\roleP} = \stStop$.

            This completes the consideration of the set $\rolesCi$ \eqref{eq:soundness:movecrash_rolesCi}.
          \item
            No change here.
          \item
            By \cref{lem:queue-assoc-crash}.
        \end{enumerate}
    \item Case \inferrule{\iruleGtMoveRec}:

      From the premise we have
      \begin{gather}
        \stEnvAssoc{\gtWithCrashedRoles{\rolesC}{\gtG}}{\stEnv; \qEnv}{\rolesR}
        \\
        \gtG = \gtRec{\gtRecVar}{\gtG[0]}
        \\
        \gtWithCrashedRoles{\rolesC}{\gtG}
        \gtMove{\rolesR}
        \\
        \gtWithCrashedRoles{\rolesC}{\gtG[0]{}\subst{\gtRecVar}{\gtRec{\gtRecVar}{\gtG[0]}}}
        \gtMove[\stEnvAnnotGenericSym]{\rolesR}
        \gtWithCrashedRoles{\rolesCi}{\gtGi}
      \end{gather}

      By \cref{prop:gt-lts-unfold-once}, we have
      $\stEnvAssoc{
        \gtWithCrashedRoles{\rolesC}{\gtG[0]{}\subst{\gtRecVar}{\gtRec{\gtRecVar}{\gtG[0]}}}
      }{\stEnv; \qEnv}{\rolesR}
      $, and we can apply inductive hypothesis to obtain the desired result.

    \item Case \inferrule{\iruleGtMoveOut}:

      From the premise we have
      \begin{gather}
        \stEnvAssoc{\gtWithCrashedRoles{\rolesC}{\gtG}}{\stEnv; \qEnv}{\rolesR}
        \label{eq:soundness:moveout_assoc_pre}
        \\
        \gtG =
          \gtCommSmall{\roleP}{\roleQ}{i \in I}{\gtLab[i]}{\tyGround[i]}{\gtG[i]}
        \\
        \gtWithCrashedRoles{\rolesC}{\gtG}
        \gtMove{\rolesR}
        \\
        j \in I
        \\
        \gtLab[j] \neq \gtCrashLab
        \\
        \stEnvAnnotGenericSym =
          \stEnvOutAnnotSmall{\roleP}{\roleQ}{\stChoice{\gtLab[j]}{\tyGround[j]}}
        \\
        \rolesCi = \rolesC
        \\
        \gtGi =
          \gtCommTransit{\roleP}{\roleQ}{i \in I}{\gtLab[i]}{\tyGround[i]}{\gtG[i]}{j}
      \end{gather}
      By association \eqref{eq:soundness:moveout_assoc_pre} and $\roleP \in \gtRoles{\gtG}$,
      we know that $\stEnvApp{\stEnv}{\roleP} \stSub \gtProj[\rolesR]{\gtG}{\roleP} =
        \stIntSum{\roleQ}{i \in \stIdxRemoveCrash{I}{\stLab[i]}}{
        \stChoice{\stLab[i]}{\tyGround[i]} \stSeq (\gtProj[\rolesR]{\gtG[i]}{\roleP})%
      }$. Then by \cref{lem:subtyping-invert}, we obtain that 
      $\stEnvApp{\stEnv}{\roleP} =
      \stIntSum{\roleQ}{i \in I'}{
        \stChoice{\stLab[i]}{\tyGround[i]} \stSeq \stT[i]%
      }$, where $I' \subseteq \stIdxRemoveCrash{I}{\stLab[i]}$ and $\forall i \in I': \stT[i] \stSub \gtProj[\rolesR]{\gtG[i]}{\roleP}$.
      Note that here for any $i \in I': \stLab[i] = \gtLab[i]$. 
      
      Since the $\stCrashLab$ label cannot appear in the internal choices, it holds that for any $i \in I'$, $\stLab[i] \neq \stCrashLab$. 
      Therefore, with $\forall i \in I': \stLab[i] = \gtLab[i]$, we can set $j \in I'$ with $\stLab[j] = \gtLab[j] \neq \stCrashLab$ and 
      $  \stEnvAnnotGenericSym =
          \stEnvOutAnnotSmall{\roleP}{\roleQ}{\stChoice{\gtLab[j]}{\tyGround[j]}}
          = 
          \stEnvOutAnnotSmall{\roleP}{\roleQ}{\stChoice{\stLab[j]}{\tyGround[j]}}$. 
            
        Let $\stEnvi; \qEnvi = \stEnvUpd{\stEnv}{\roleP}{\stT[j]};  \stEnvUpd{\qEnv}{\roleP, \roleQ}{ \stQCons{
          \stEnvApp{\qEnv}{\roleP, \roleQ}
        }{
          \stQMsg{\stLab[j]}{\tyGround[j]}
        }}$.
      We show
      $\stEnv; \qEnv \stEnvMoveGenAnnot \stEnvi; \qEnvi$\;
      and
        \;$\stEnvAssoc{\gtWithCrashedRoles{\rolesCi}{\gtGi}}{\stEnvi;
        \qEnvi}{\rolesR}$.

    For the first part: we apply \inferrule{\iruleTCtxOut} on those we have: 
    $\stEnvApp{\stEnv}{\roleP} =
      \stIntSum{\roleQ}{i \in I'}{
        \stChoice{\stLab[i]}{\tyGround[i]} \stSeq \stT[i]%
      }$ and $j \in I'$, to obtain the thesis. 
      
   For the second part:
     \begin{enumerate}[label=(A\arabic*)]
     \item   We want to show $\forall \roleR \in \gtRoles{\gtGi}: \stEnvApp{\stEnvi}{\roleR}
           \stSub \gtProj[\rolesR]{\gtGi}{\roleR}$. We consider three subcases:
           \begin{itemize}[leftmargin=*]
           \item $\roleR = \roleP$: since
           $\stEnvApp{\stEnvi}{\roleP} = \stT[j]$, $j \in I'$, and $\forall i \in I': \stT[i] \stSub \gtProj[\rolesR]{\gtG[i]}{\roleP}$, 
           we have $\stEnvApp{\stEnvi}{\roleP} \stSub
           \gtProj[\rolesR]{\gtG[j]}{\roleP} =  \gtProj[\rolesR]{\gtGi}{\roleP}$, as desired.
           
           \item $\roleR = \roleQ$:
                     by association~\eqref{eq:soundness:moveout_assoc_pre}, 
                     we have 
                     $\stEnvApp{\stEnvi}{\roleQ} = \stEnvApp{\stEnv}{\roleQ} 
                     \stSub  \gtProj[\rolesR]{\gtG}{\roleQ} 
                     = \stExtSum{\roleP}{i \in I}{
        \stChoice{\stLab[i]}{\tyGround[i]} \stSeq (\gtProj[\rolesR]{\gtG[i]}{\roleQ})%
      } =  \gtProj[\rolesR]{\gtGi}{\roleQ} 
      $, as desired.
           \item $\roleR \neq \roleQ$ and $\roleR \neq \roleP$: 
           by association~\eqref{eq:soundness:moveout_assoc_pre}, 
            it holds that $\stEnvApp{\stEnvi}{\roleR} =  \stEnvApp{\stEnv}{\roleR} \stSub
           \gtProj[\rolesR]{\gtG}{\roleR} =  \stMerge{i \in I}{\gtProj[\rolesR]{\gtG[i]}{\roleR}} 
           = 
            \gtProj[\rolesR]{\gtGi}{\roleR}$, as desired. 
           \end{itemize}

     \item No change here. 
     \item No change here. 
     \item  We are left to show 
      $\qEnvi =  \stEnvUpd{\qEnv}{\roleP, \roleQ}{ \stQCons{
          \stEnvApp{\qEnv}{\roleP, \roleQ}
        }{
          \stQMsg{\stLab[j]}{\tyGround[j]}
        }}$ is associated with $\gtWithCrashedRoles{\rolesC} 
     {\gtCommTransit{\roleP}{\roleQ}{i \in I}{\gtLab[i]}{\tyGround[i]}{\gtG[i]}{j}}$. 
     
     Since $\qEnv$ is associated with $\gtWithCrashedRoles{\rolesC} 
     {\gtCommSmall{\roleP}{\roleQ}{i \in I}{\gtLab[i]}{\tyGround[i]}{\gtG[i]}}$, 
     by \cref{def:assoc-queue}, we have $\stEnvApp{\qEnv}{\roleP, \roleQ} = \stQEmpty$ and 
     $\forall i \in I: \qEnv \text{~is associated with~}
      \gtWithCrashedRoles{\rolesC}{\gtG[i]}$. 
     
     Since $\gtLab[j] \neq \gtCrashLab$, we just need to show that 
     $\stEnvApp{\qEnvi}{\roleP, \roleQ} = \stQMsg{\gtLab[j]}{\tyGround[j]}$, 
     which follows directly from 
     $\qEnvi =  \stEnvUpd{\qEnv}{\roleP, \roleQ}{ \stQCons{
          \stEnvApp{\qEnv}{\roleP, \roleQ}
        }{
          \stQMsg{\stLab[j]}{\tyGround[j]}
        }}$, $\gtLab[j] = \stLab[j]$, and $\stEnvApp{\qEnv}{\roleP, \roleQ} = \stQEmpty$,  
        and 
      $\forall i \in I: \stEnvUpd{\qEnvi}{\roleP, \roleQ}{\stQEmpty}$ is associated with 
      $\gtWithCrashedRoles{\rolesC}{\gtG[i]}$, which follows from 
      $\stEnvUpd{\qEnvi}{\roleP, \roleQ}{\stQEmpty}  = \qEnv$ and 
       $\forall i \in I: \qEnv \text{~is associated with~}
      \gtWithCrashedRoles{\rolesC}{\gtG[i]}$. 
     \end{enumerate}

 \item Case \inferrule{\iruleGtMoveIn}:

      From the premise we have
      \begin{gather}
        \stEnvAssoc{\gtWithCrashedRoles{\rolesC}{\gtG}}{\stEnv; \qEnv}{\rolesR}
         \label{eq:soundness:movein_assoc_pre}
        \\
        \gtG =
          \gtCommTransit{\rolePMaybeCrashed}{\roleQ}{i \in I}{\gtLab[i]}{\tyGround[i]}{\gtG[i]}{j}
        \\
        \gtWithCrashedRoles{\rolesC}{\gtG}
        \gtMove{\rolesR}
        \\
        j \in I
        \\
        \gtLab[j] \neq \gtCrashLab
        \\
        \stEnvAnnotGenericSym =
          \stEnvInAnnotSmall{\roleQ}{\roleP}{\stChoice{\gtLab[j]}{\tyGround[j]}}
        \\
        \rolesCi = \rolesC
        \\
        \gtGi = \gtG[j]
      \end{gather}
     By association \eqref{eq:soundness:movein_assoc_pre} and $\roleQ \in \gtRoles{\gtG}$,
      we know that $\stEnvApp{\stEnv}{\roleQ} \stSub \gtProj[\rolesR]{\gtG}{\roleQ} =
        \stExtSum{\roleP}{i \in I}{
        \stChoice{\stLab[i]}{\tyGround[i]} \stSeq (\gtProj[\rolesR]{\gtG[i]}{\roleQ})%
      }$. Note that here for any $i \in I: \stLab[i] = \gtLab[i]$. Furthermore, by \cref{lem:subtyping-invert},  
      we obtain that $\stEnvApp{\stEnv}{\roleQ} =
      \stExtSum{\roleP}{i \in I'}{
        \stChoice{\stLab[i]}{\tyGround[i]} \stSeq \stT[i]%
      }$, where $I \subseteq I'$ and $\forall i \in I: \stT[i] \stSub \gtProj[\rolesR]{\gtG[i]}{\roleQ}$.
      From association \eqref{eq:soundness:movein_assoc_pre}, we also get that 
      $\stEnvApp{\qEnv}{\roleP, \roleQ} = \stQCons{\stQMsg{\gtLab[j]}{\tyGround[j]}}{\stQ} 
      = \stQCons{\stQMsg{\stLab[j]}{\tyGround[j]}}{\stQ}$. 
      
        Let $\stEnvi; \qEnvi = \stEnvUpd{\stEnv}{\roleQ}{\stT[j]};  \stEnvUpd{\qEnv}{\roleP, \roleQ}{\stQ}$.
      We show
      $\stEnv; \qEnv \stEnvMoveGenAnnot \stEnvi; \qEnvi$\;
      and
        \;$\stEnvAssoc{\gtWithCrashedRoles{\rolesCi}{\gtGi}}{\stEnvi;
        \qEnvi}{\rolesR}$.
        
    For the first part: we apply \inferrule{\iruleTCtxIn} on those we have: 
    $\stEnvApp{\stEnv}{\roleQ} =
      \stExtSum{\roleP}{i \in I'}{
        \stChoice{\stLab[i]}{\tyGroundi[i]} \stSeq \stT[i]
      }$, $j \in I \subseteq I'$, $\stLab[j] = \gtLab[j]$, and $\stEnvApp{\qEnv}{\roleP, \roleQ} = 
      \stQCons{\stQMsg{\stLab[j]}{\tyGround[j]}}{\stQ}$,  
    to obtain the thesis.   
    
    For the second part:
     \begin{enumerate}[label=(A\arabic*)]
     \item 
       We want to show $\forall \roleR \in \gtRoles{\gtG[j]}: \stEnvApp{\stEnvi}{\roleR}
           \stSub \gtProj[\rolesR]{\gtG[j]}{\roleR}$. We consider three subcases:
           \begin{itemize}[leftmargin=*]
           \item $\roleR = \roleQ$ (meaning that $\roleQ \in \gtRoles{\gtG[j]}$): since
           $\stEnvApp{\stEnvi}{\roleQ} = \stT[j]$ and $\forall i \in I: \stT[i] \stSub \gtProj[\rolesR]{\gtG[i]}{\roleQ}$, 
           we have $\stEnvApp{\stEnvi}{\roleQ} \stSub
           \gtProj[\rolesR]{\gtG[j]}{\roleQ}$, as desired.
           \item $\roleR = \roleP$ (meaning that $\rolePMaybeCrashed = \roleP$ and $\roleP \in \gtRoles{\gtG[j]}$): 
                     by association~\eqref{eq:soundness:movein_assoc_pre}, 
                     we have 
                     $\stEnvApp{\stEnvi}{\roleP} = \stEnvApp{\stEnv}{\roleP} 
                     \stSub  \gtProj[\rolesR]{\gtG}{\roleP} =  \gtProj[\rolesR]{\gtG[j]}{\roleP}$, as desired.
           \item $\roleR \neq \roleQ$ and $\roleR \neq \roleP$: 
           by association~\eqref{eq:soundness:movein_assoc_pre}, 
            it holds that $\stEnvApp{\stEnvi}{\roleR} =  \stEnvApp{\stEnv}{\roleR} \stSub
           \gtProj[\rolesR]{\gtG}{\roleR} =  \stMerge{i \in I}{\gtProj[\rolesR]{\gtG[i]}{\roleR}}$.
           Then, by applying \Cref{lem:merge-subtyping} and transitivity of subtyping,  we conclude with $\stEnvApp{\stEnvi}{\roleR}
        \stSub \gtProj[\rolesR]{\gtG[j]}{\roleR}$, as desired.
           \end{itemize}
     \item No change here. 
     \item No change here if $\roleQ \in \gtRoles{\gtG[j]}$ and $\roleP \in \gtRoles{\gtG[j]}$. Otherwise, if $\roleQ \notin \gtRoles{\gtG[j]}$: with
            $\roleQ \notin \gtRolesCrashed{\gtG[j]}$, by \Cref{lem:not-in-roles-end}, we have $\gtProj[\rolesR]{\gtG[j]}{\roleQ} = \stEnd$.
            Furthermore, by the fact that $\forall i \in I: \stT[i] \stSub \gtProj[\rolesR]{\gtG[i]}{\roleQ} $,
            it holds that $\stT[j] = \stEnd$, and thus,
          $\stEnvApp{\stEnvi}{\roleQ} = \stEnd$, as desired. The argument for $\roleP \notin  \gtRoles{\gtG[j]}$ and 
          $\rolePMaybeCrashed = \roleP$ follows similarly. 
     \item  Since $\qEnv$ is associated with $\gtWithCrashedRoles{\rolesC} 
     {\gtCommTransit{\rolePMaybeCrashed}{\roleQ}{i \in I}{\gtLab[i]}{\tyGround[i]}{\gtG[i]}{j}}$ 
     and $\gtLab[j] \neq \gtCrashLab$, by \Cref{def:assoc-queue},  we have that
     $\stEnvApp{\qEnv}{\roleP, \roleQ} = \stQCons{\stQMsg{\gtLab[j]}{\tyGround[j]}}{\stQ}$ and 
            $\forall i \in I: \stEnvUpd{\qEnv}{\roleP, \roleQ}{\stQ} \text{~is associated with~}
      \gtWithCrashedRoles{\rolesC}{\gtG[i]}$, which follows that $\qEnvi = \stEnvUpd{\qEnv}{\roleP, \roleQ}{\stQ}$ 
      is associated with
       $\gtWithCrashedRoles{\rolesC}{\gtG[j]}$, as desired.
     \end{enumerate} 
    
    \item Case \inferrule{\iruleGtMoveCrDe}:

      From the premise we have
      \begin{gather}
        \stEnvAssoc{\gtWithCrashedRoles{\rolesC}{\gtG}}{\stEnv; \qEnv}{\rolesR}
        \label{eq:soundness:detectcrash_assoc_pre}
        \\
        \gtG =
          \gtCommTransit{\rolePCrashed}{\roleQ}{i \in I}{\gtLab[i]}{\tyGround[i]}{\gtG[i]}{j}
        \\
        \gtWithCrashedRoles{\rolesC}{\gtG}
        \gtMove{\rolesR}
        \\
        j \in I
        \label{eq:soundness:detectcrash_index}
        \\
        \gtLab[j] = \gtCrashLab
        \label{eq:soundness:detectcrash_crash_branch}
        \\
        \stEnvAnnotGenericSym =
          \ltsCrDe{\mpS}{\roleQ}{\roleP}
        \\
        \rolesCi = \rolesC
        \\
        \gtGi = \gtG[j]
      \end{gather}
      By association \eqref{eq:soundness:detectcrash_assoc_pre} and $\roleQ \in \gtRoles{\gtG}$,
      we know that $\stEnvApp{\stEnv}{\roleQ} \stSub \gtProj[\rolesR]{\gtG}{\roleQ} =
        \stExtSum{\roleP}{i \in I}{
        \stChoice{\stLab[i]}{\tyGround[i]} \stSeq (\gtProj[\rolesR]{\gtG[i]}{\roleQ})%
      }$. Note that here for any $i \in I: \stLab[i] = \gtLab[i]$. Furthermore, by \cref{lem:subtyping-invert}, 
      we obtain that $\stEnvApp{\stEnv}{\roleQ} =
      \stExtSum{\roleP}{i \in I'}{
        \stChoice{\stLab[i]}{\tyGroundi[i]} \stSeq \stT[i]%
      }$, where $I \subseteq I'$ and $\forall i \in I: \tyGround[i] =
      \tyGroundi[i]$ and $\stT[i] \stSub \gtProj[\rolesR]{\gtG[i]}{\roleQ}$.

       Let $\stEnvi; \qEnvi = \stEnvUpd{\stEnv}{\roleQ}{\stT[j]}; \qEnv$.
      We show
      $\stEnv; \qEnv \stEnvMoveGenAnnot \stEnvi; \qEnvi$\;
      and
        \;$\stEnvAssoc{\gtWithCrashedRoles{\rolesCi}{\gtGi}}{\stEnvi;
        \qEnvi}{\rolesR}$.

        For the first part: By association~\eqref{eq:soundness:detectcrash_assoc_pre},
        $\roleP \in \rolesC$, and
        $\gtLab[j] = \gtCrashLab$~\eqref{eq:soundness:detectcrash_crash_branch},
        we know that $\stEnvApp{\stEnv}{\roleP} = \stStop$ and
        $\stEnvApp{\qEnv}{\roleP, \roleQ} = \stQEmpty$.
        Since $j \in I$~\eqref{eq:soundness:detectcrash_index},
         $\gtLab[j] = \gtCrashLab$~\eqref{eq:soundness:detectcrash_crash_branch}, and
         $\forall i \in I: \gtLab[i] = \stLab[i]$, we have $\stLab[j] = \stCrashLab$.  We also get
         $j \in I'$ from $j \in I$ and $I \subseteq I'$. Thus, together with $\stEnvApp{\stEnv}{\roleQ} =
      \stExtSum{\roleP}{i \in I'}{
        \stChoice{\stLab[i]}{\tyGroundi[i]} \stSeq \stT[i]%
      }$, we apply \inferrule{\iruleTCtxCrashDetect} to obtain the thesis.

     For the second part:
        \begin{enumerate}[label=(A\arabic*)]
          \item
          We want to show $\forall \roleR \in \gtRoles{\gtG[j]}: \stEnvApp{\stEnvi}{\roleR}
           \stSub \gtProj[\rolesR]{\gtG[j]}{\roleR}$. We consider two subcases:
           \begin{itemize}[leftmargin=*]
           \item $\roleR = \roleQ$ (meaning that $\roleQ \in \gtRoles{\gtG[j]}$): since
           $\stEnvApp{\stEnvi}{\roleQ} = \stT[j]$, $\forall i \in I: \stT[i] \stSub \gtProj[\rolesR]{\gtG[i]}{\roleQ}$, and
           $j \in I$~\eqref{eq:soundness:detectcrash_index}, we have $\stEnvApp{\stEnvi}{\roleQ} \stSub
           \gtProj[\rolesR]{\gtG[j]}{\roleQ}$, as desired.
           \item $\roleR \neq \roleQ$: since $\roleP \in \rolesC$, we know $\roleP \notin \gtRoles{\gtG[j]}$, and thus, $\roleR \neq \roleP$.
           Furthermore, by association~\eqref{eq:soundness:detectcrash_assoc_pre}, it holds that 
           $\stEnvApp{\stEnvi}{\roleR} = \stEnvApp{\stEnv}{\roleR} \stSub
           \gtProj[\rolesR]{\gtG}{\roleR} =  \stMerge{i \in I}{\gtProj[\rolesR]{\gtG[i]}{\roleR}}$.
           Then, by applying \Cref{lem:merge-subtyping} and transitivity of subtyping,  we can conclude with $\stEnvApp{\stEnvi}{\roleR}
        \stSub \gtProj[\rolesR]{\gtG[j]}{\roleR}$, as desired.
           \end{itemize}

          \item
            No change here.
          \item
            No change here if $\roleQ \in \gtRoles{\gtG[j]}$. Otherwise, if $\roleQ \notin \gtRoles{\gtG[j]}$: with
            $\roleQ \notin \gtRolesCrashed{\gtG[j]}$, by \Cref{lem:not-in-roles-end}, we have $\gtProj[\rolesR]{\gtG[j]}{\roleQ} = \stEnd$.
            Furthermore, by the fact that $\forall i \in I: \stT[i] \stSub \gtProj[\rolesR]{\gtG[i]}{\roleQ} $,
            it holds that $\stT[j] = \stEnd$, and thus,
          $\stEnvApp{\stEnvi}{\roleQ} = \stEnd$, as desired.

          \item
             Since $\qEnv$ is associated with $\gtWithCrashedRoles{\rolesC}{\gtCommTransit{\rolePCrashed}{\roleQ}{i \in
            I}{\gtLab[i]}{\tyGround[i]}{\gtG[i]}{j} }$ and $\gtLab[j] = \gtCrashLab$, by \Cref{def:assoc-queue},  we have that
            $\forall i \in I: \qEnv \text{~is associated with~}
      \gtWithCrashedRoles{\rolesC}{\gtG[i]}$, which follows that $\qEnvi = \qEnv$ is associated with
       $\gtWithCrashedRoles{\rolesC}{\gtG[j]}$, as desired.
        \end{enumerate}

       \item Case \inferrule{\iruleGtMoveOrph}:

      From the premise we have
      \begin{gather}
        \stEnvAssoc{\gtWithCrashedRoles{\rolesC}{\gtG}}{\stEnv; \qEnv}{\rolesR}
        \label{eq:soundness:moveorph_assoc_pre}
        \\
        \gtG =
          \gtCommSmall{\roleP}{\roleQCrashed}{i \in I}{\gtLab[i]}{\tyGround[i]}{\gtG[i]}
        \\
        \gtWithCrashedRoles{\rolesC}{\gtG}
        \gtMove{\rolesR}
        \\
        j \in I
        \\
        \gtLab[j] \neq \gtCrashLab
        \\
        \stEnvAnnotGenericSym =
          \stEnvOutAnnotSmall{\roleP}{\roleQ}{\stChoice{\gtLab[j]}{\tyGround[j]}}
        \\
        \rolesCi = \rolesC
        \\
        \gtGi = \gtG[j]
      \end{gather}
        By association \eqref{eq:soundness:moveorph_assoc_pre} and $\roleP \in \gtRoles{\gtG}$,
      we know that $\stEnvApp{\stEnv}{\roleP} \stSub \gtProj[\rolesR]{\gtG}{\roleP} =
        \stIntSum{\roleQ}{i \in \stIdxRemoveCrash{I}{\stLab[i]}}{
        \stChoice{\stLab[i]}{\tyGround[i]} \stSeq (\gtProj[\rolesR]{\gtG[i]}{\roleP})%
      }$. Then by \cref{lem:subtyping-invert}, we obtain that 
      $\stEnvApp{\stEnv}{\roleP} =
      \stIntSum{\roleQ}{i \in I'}{
        \stChoice{\stLab[i]}{\tyGround[i]} \stSeq \stT[i]%
      }$, where $I' \subseteq \stIdxRemoveCrash{I}{\stLab[i]}$ and $\forall i \in I': \stT[i] \stSub \gtProj[\rolesR]{\gtG[i]}{\roleP}$.
      Note that here for any $i \in I': \stLab[i] = \gtLab[i]$. 
      
      Since the $\stCrashLab$ label cannot appear in the internal choices, it holds that for any $i \in I'$, $\stLab[i] \neq \stCrashLab$. 
      Therefore, with $\forall i \in I': \stLab[i] = \gtLab[i]$, we can set $j \in I'$ with $\stLab[j] = \gtLab[j] \neq \stCrashLab$ and 
      $  \stEnvAnnotGenericSym =
          \stEnvOutAnnotSmall{\roleP}{\roleQ}{\stChoice{\gtLab[j]}{\tyGround[j]}}
          = 
          \stEnvOutAnnotSmall{\roleP}{\roleQ}{\stChoice{\stLab[j]}{\tyGround[j]}}$. 
            
        Let $\stEnvi; \qEnvi = \stEnvUpd{\stEnv}{\roleP}{\stT[j]};  \stEnvUpd{\qEnv}{\roleP, \roleQ}{ \stQCons{
          \stEnvApp{\qEnv}{\roleP, \roleQ}
        }{
          \stQMsg{\stLab[j]}{\tyGround[j]}
        }}$.
      We show
      $\stEnv; \qEnv \stEnvMoveGenAnnot \stEnvi; \qEnvi$\;
      and
        \;$\stEnvAssoc{\gtWithCrashedRoles{\rolesCi}{\gtGi}}{\stEnvi;
        \qEnvi}{\rolesR}$.

    For the first part: we apply \inferrule{\iruleTCtxOut} on those we have: 
    $\stEnvApp{\stEnv}{\roleP} =
      \stIntSum{\roleQ}{i \in I'}{
        \stChoice{\stLab[i]}{\tyGround[i]} \stSeq \stT[i]%
      }$ and $j \in I'$, to obtain the thesis. 
      
   For the second part:
     \begin{enumerate}[label=(A\arabic*)]
     \item   We want to show $\forall \roleR \in \gtRoles{\gtGi}: \stEnvApp{\stEnvi}{\roleR}
           \stSub \gtProj[\rolesR]{\gtGi}{\roleR}$. We consider two subcases:
           \begin{itemize}[leftmargin=*]
           \item $\roleR = \roleP$ (meaning that $\roleP \in \gtRoles{\gtG[j]}$): since
           $\stEnvApp{\stEnvi}{\roleP} = \stT[j]$, $\forall i \in I': \stT[i] \stSub \gtProj[\rolesR]{\gtG[i]}{\roleP}$, and
           $j \in I'$, we have $\stEnvApp{\stEnvi}{\roleP} \stSub
           \gtProj[\rolesR]{\gtG[j]}{\roleP}$, as desired.
          
           \item $\roleR \neq \roleP$: since $\roleQ \in \rolesC$, we know $\roleQ \notin \gtRoles{\gtG[j]}$, and thus, $\roleR \neq \roleQ$.
           Furthermore, by association~\eqref{eq:soundness:moveorph_assoc_pre}, it holds that 
           $\stEnvApp{\stEnvi}{\roleR} = \stEnvApp{\stEnv}{\roleR} \stSub
           \gtProj[\rolesR]{\gtG}{\roleR} =  \stMerge{i \in I}{\gtProj[\rolesR]{\gtG[i]}{\roleR}}$.
           Then, by applying \Cref{lem:merge-subtyping} and transitivity of subtyping,  we can conclude with $\stEnvApp{\stEnvi}{\roleR}
        \stSub \gtProj[\rolesR]{\gtG[j]}{\roleR}$, as desired.
           \end{itemize}

     \item No change here. 
     \item 
     No change here if $\roleP \in \gtRoles{\gtG[j]}$. Otherwise, if $\roleP \notin \gtRoles{\gtG[j]}$: with
            $\roleP \notin \gtRolesCrashed{\gtG[j]}$, by \Cref{lem:not-in-roles-end}, we have $\gtProj[\rolesR]{\gtG[j]}{\roleP} = \stEnd$.
            Furthermore, by the fact that $\forall i \in I': \stT[i] \stSub \gtProj[\rolesR]{\gtG[i]}{\roleP} $,
            it holds that $\stT[j] = \stEnd$, and thus,
          $\stEnvApp{\stEnvi}{\roleP} = \stEnd$, as desired.
     \item  We are left to show 
      $\qEnvi =  \stEnvUpd{\qEnv}{\roleP, \roleQ}{ \stQCons{
          \stEnvApp{\qEnv}{\roleP, \roleQ}
        }{
          \stQMsg{\stLab[j]}{\tyGround[j]}
        }}$ is associated with $\gtWithCrashedRoles{\rolesC} 
     {\gtG[j]}$. 
     
     Since $\qEnv$ is associated with $\gtWithCrashedRoles{\rolesC} 
     {\gtCommSmall{\roleP}{\roleQCrashed}{i \in I}{\gtLab[i]}{\tyGround[i]}{\gtG[i]}}$, 
     by \cref{def:assoc-queue}, we have $\stEnvApp{\qEnv}{\roleP, \roleQ} = \stQUnavail$ and 
     $\forall i \in I: \qEnv \text{~is associated with~}
      \gtWithCrashedRoles{\rolesC}{\gtG[i]}$. It follows from $\stEnvApp{\qEnv}{\roleP, \roleQ} = \stQUnavail$
       that 
      $\qEnvi = \stEnvUpd{\qEnv}{\roleP, \roleQ}{ \stQCons{
          \stEnvApp{\qEnv}{\roleP, \roleQ}
        }{
          \stQMsg{\stLab[j]}{\tyGround[j]}
        }}
        = 
        \stEnvUpd{\qEnv}{\roleP, \roleQ}{ \stQCons{
        \stQUnavail
        }{\stQMsg{\stLab[j]}{\tyGround[j]}
        } }
        = 
         \stEnvUpd{\qEnv}{\roleP, \roleQ}{\stQUnavail}
         =
        \qEnv$. Hence, with $\forall i \in I: \qEnv \text{~is associated with~}
      \gtWithCrashedRoles{\rolesC}{\gtG[i]}$, we conclude that $\qEnvi = \qEnv$ is associated with 
      $\gtWithCrashedRoles{\rolesC} 
     {\gtG[j]}$, as desired. 
  \end{enumerate}  

 \item Case \inferrule{\iruleGtMoveCtx}:

      From the premise we have
      \begin{gather}
        \stEnvAssoc{\gtWithCrashedRoles{\rolesC}{\gtG}}{\stEnv; \qEnv}{\rolesR}
        \label{eq:soundness:context_1_assoc_pre}
        \\
        \gtG =
          \gtCommSmall{\roleP}{\roleQMaybeCrashed}{i \in I}{\gtLab[i]}{\tyGround[i]}{\gtG[i]}
        \\
        \gtWithCrashedRoles{\rolesC}{\gtG}
        \gtMove{\rolesR}
        \\
        \forall i \in I :
        \gtWithCrashedRoles{\rolesC}{\gtG[i]}
        \gtMove[\stEnvAnnotGenericSym]{\rolesR}
        \gtWithCrashedRoles{\rolesCi}{\gtGi[i]}
        \\
        \ltsSubject{\stEnvAnnotGenericSym} \notin \setenum{\roleP, \roleQ}
        \\
        \gtGi =
          \gtCommSmall{\roleP}{\roleQMaybeCrashed}{i \in I}{\gtLab[i]}{\tyGround[i]}{\gtGi[i]}
      \end{gather}
      We consider two subcases depending on whether $\roleQ$ has crashed.
      \begin{itemize}[leftmargin=*]
      \item $\roleQMaybeCrashed = \roleQ$:

      For $\roleP$, we have
     $\gtProj[\rolesR]{\gtG}{\roleP} =
       \stIntSum{\roleQ}{i \in \stIdxRemoveCrash{I}{\stLab[i]}}{ %
        \stChoice{\stLab[i]}{\tyGround[i]} \stSeq (\gtProj[\rolesR]{\gtG[i]}{\roleP})%
      }$ and $ \gtProj[\rolesR]{\gtGi}{\roleP} =
       \stIntSum{\roleQ}{i \in \stIdxRemoveCrash{I}{\stLab[i]}}{ %
        \stChoice{\stLab[i]}{\tyGround[i]} \stSeq (\gtProj[\rolesR]{\gtGi[i]}{\roleP})%
      }$. For $\roleQ$, we have $\gtProj[\rolesR]{\gtG}{\roleQ} =
       \stExtSum{\roleP}{i \in I}{%
        \stChoice{\stLab[i]}{\tyGround[i]} \stSeq (\gtProj[\rolesR]{\gtG[i]}{\roleQ})%
      }$ and $\gtProj[\rolesR]{\gtGi}{\roleQ} =
       \stExtSum{\roleP}{i \in I}{%
        \stChoice{\stLab[i]}{\tyGround[i]} \stSeq (\gtProj[\rolesR]{\gtGi[i]}{\roleQ})%
      }$. Take an arbitrary $j \in I$. %
      Let $\stEnvApp{\stEnv[j]}{\roleP} =
      \gtProj[\rolesR]{\gtG[j]}{\roleP}$, $\stEnvApp{\stEnv[j]}{\roleQ} =
      \gtProj[\rolesR]{\gtG[j]}{\roleQ}$, $\stEnvApp{\stEnv[j]}{\roleR} =
     \stEnvApp{\stEnv}{\roleR}$ for $\roleR \in \dom{\stEnv} \setminus \setenum{\roleP, \roleQ}$, $\qEnv[j] = \qEnv$.
     We show $\stEnvAssoc{\gtWithCrashedRoles{\rolesC}{\gtG[j]}}{\stEnv[j]; \qEnv[j]}{\rolesR}$.
      \begin{enumerate}[label=(A\arabic*)]
      \item We want to show $\forall \roleS \in \gtRoles{\gtG[j]}: \stEnvApp{\stEnv[j]}{\roleS} \stSub \gtProj[\rolesR]{\gtG[j]}{\roleS}$. We consider 
      three subcases:
      \begin{itemize}[leftmargin=*]
      \item $\roleS = \roleP$ (meaning that $\roleP \in \gtRoles{\gtG[j]}$): trivial by $\stEnvApp{\stEnv[j]}{\roleP} =
      \gtProj[\rolesR]{\gtG[j]}{\roleP}$ and the reflexivity of subtyping. 
      \item $\roleS = \roleQ$ (meaning that $\roleQ \in \gtRoles{\gtG[j]}$): trivial by $\stEnvApp{\stEnv[j]}{\roleQ} =
      \gtProj[\rolesR]{\gtG[j]}{\roleQ}$ and the reflexivity of subtyping. 
      \item $\roleS \neq \roleP$ and $\roleS \neq \roleQ$: by association \eqref{eq:soundness:context_1_assoc_pre} and 
      $\stEnvApp{\stEnv[j]}{\roleS} =
     \stEnvApp{\stEnv}{\roleS}$, we have $\stEnvApp{\stEnv[j]}{\roleS} \stSub 
       \stMerge{i \in I}{\gtProj[\rolesR]{\gtG[i]}{\roleS}}$. 
        Then, by \Cref{lem:merge-subtyping} and transitivity of subtyping,  we conclude  $\stEnvApp{\stEnv[j]}{\roleS}
        \stSub \gtProj[\rolesR]{\gtG[j]}{\roleS}$, as desired. 
      \end{itemize}
       \item No change here.
      \item No change here if $\roleP \in \gtRoles{\gtG[j]}$ and $\roleQ \in \gtRoles{\gtG[j]}$. Otherwise, consider the case that $\roleP \notin \gtRoles{\gtG[j]}$: 
      with $\roleP \notin \gtRolesCrashed{\gtG[j]}$, by \Cref{lem:not-in-roles-end}, we have $\gtProj[\rolesR]{\gtG[j]}{\roleP} = \stEnd$. 
      Therefore, we know from $\stEnvApp{\stEnv[j]}{\roleP} =
      \gtProj[\rolesR]{\gtG[j]}{\roleP}$ that $ \stEnvApp{\stEnv[j]}{\roleP} = \stEnd$, as required. The argument for 
      $\roleQ \notin  \gtRoles{\gtG[j]}$ follows similarly. 
      \item Trivial by association~\eqref{eq:soundness:context_1_assoc_pre}, \cref{def:assoc-queue}, and $\qEnv[j] = \qEnv$. 
      \end{enumerate}
      By inductive hypothesis, there exists $\stEnvi[j]; \qEnvi[j]$ such that
      $\stEnv[j]; \qEnv[j]  \stEnvMoveGenAnnot
      \stEnvi[j]; \qEnvi[j]$ and
      $\stEnvAssoc{\gtWithCrashedRoles{\rolesCi}{\gtGi[j]}}{\stEnvi[j]; \qEnvi[j]}{\rolesR}$. Since
      $\ltsSubject{\stEnvAnnotGenericSym} \notin \setenum{\roleP, \roleQ}$, we apply \cref{lem:stenv-red:trivial-2}, which gives
      $\stEnvApp{\stEnv[j]}{\roleP} = \stEnvApp{\stEnvi[j]}{\roleP}$ and
      $\stEnvApp{\stEnv[j]}{\roleQ} = \stEnvApp{\stEnvi[j]}{\roleQ}$.

      We now construct a configuration $\stEnvi; \qEnvi$ and show 
      $\stEnv; \qEnv \stEnvMoveGenAnnot \stEnvi; \qEnvi$ and 
      $\stEnvAssoc{\gtWithCrashedRoles{\rolesCi}{\gtGi}}{\stEnvi; \qEnvi}{\rolesR}$.

      Let  $\stEnvApp{\stEnvi}{\roleP}
      = \stEnvApp{\stEnv}{\roleP}$, 
       $\stEnvApp{\stEnvi}{\roleQ}
      =
      \stEnvApp{\stEnv}{\roleQ}$, 
       $\stEnvApp{\stEnvi}{\roleR} =
      \stMerge{i \in I}{\stEnvApp{\stEnvi[i]}{\roleR}}$ for
    $\roleR \in \dom{\stEnv} \setminus
      \setenum{\roleP, \roleQ}$, $\qEnvi = \qEnvi[j]$ with an arbitrary $j \in I$.
      
      For the first part that $\stEnv; \qEnv \stEnvMoveGenAnnot \stEnvi; \qEnvi$: 
      
      We know that for any $i \in I$, $\stEnvi[i]$ is obtained from $\stEnv[i]$ 
      by updating $\ltsSubject{\stEnvAnnotGenericSym}$ to a fixed type $\stT$. It follows that 
      for any $i, k \in I$ and for any $\roleR \notin \setenum{\roleP, \roleQ}$, $\stEnvApp{\stEnvi[i]}{\roleR} 
      = 
      \stEnvApp{\stEnvi[k]}{\roleR}$, and hence, $\stMerge{i \in I}{\stEnvApp{\stEnvi[i]}{\roleR}} = 
      \stEnvApp{\stEnvi[j]}{\roleR}$ with an arbitrary $j \in J$. Therefore, we have that $\stEnvi$ is obtained from 
      $\stEnv$ by updating $\ltsSubject{\stEnvAnnotGenericSym}$ to $\stEnvApp{\stEnvi[j]}{\ltsSubject{\stEnvAnnotGenericSym}}$,  
      i.e., $\stEnvi = \stEnvUpd{\stEnv}{\ltsSubject{\stEnvAnnotGenericSym}}{\stEnvApp{\stEnvi[j]}{\ltsSubject{\stEnvAnnotGenericSym}}}$.     
      
      We apply \cref{lem:stenv-red:non-trivial} on $\stEnv[j]; \qEnv[j]
  \stEnvMoveGenAnnot \stEnvi[j]; \qEnvi[j]$,
  $\stEnvAssoc{\gtWithCrashedRoles{\rolesC}{\gtG}}{\stEnv; \qEnv}{\rolesR}$, 
  $\stEnvAssoc{\gtWithCrashedRoles{\rolesC}{\gtG[j]}}{\stEnv[j]; \qEnv[j]}{\rolesR}$, 
  $\stEnvApp{\stEnv}{\ltsSubject{\stEnvAnnotGenericSym}} = \stEnvApp{\stEnv[j]}{\ltsSubject{\stEnvAnnotGenericSym}}$, 
  and 
  $\stEnvi =  \stEnvUpd{\stEnv}{\ltsSubject{\stEnvAnnotGenericSym}}{\stEnvApp{\stEnvi[j]}{\ltsSubject{\stEnvAnnotGenericSym}}}$ to get the thesis. 
      
      For the second part that  $\stEnvAssoc{\gtWithCrashedRoles{\rolesCi}{\gtGi}}{\stEnvi; \qEnvi}{\rolesR}$: 
      \begin{enumerate}[label=(A\arabic*)]
      \item We want to show $\forall \roleR \in \gtRoles{\gtGi}: \stEnvApp{\stEnvi}{\roleR}
           \stSub \gtProj[\rolesR]{\gtGi}{\roleR}$. We consider three subcases:
           \begin{itemize}[leftmargin=*]
           \item $\roleR = \roleP$: by association \eqref{eq:soundness:context_1_assoc_pre} and 
           \cref{lem:subtyping-invert}, we obtain $\stEnvApp{\stEnv}{\roleP} =  
           \stIntSum{\roleQ}{i \in I'}{ %
        \stChoice{\stLab[i]}{\tyGround[i]} \stSeq \stT[i]%
      }$ where $I' \subseteq \stIdxRemoveCrash{I}{\stLab[i]}$ and $\forall i \in I': \stT[i] \stSub 
      \gtProj[\rolesR]{\gtG[i]}{\roleP}$. Since $\stEnvApp{\stEnv[i]}{\roleP} = 
       \gtProj[\rolesR]{\gtG[i]}{\roleP} = 
      \stEnvApp{\stEnvi[i]}{\roleP}$ for each $i \in I$, with association 
      $\stEnvAssoc{\gtWithCrashedRoles{\rolesCi}{\gtGi[i]}}{\stEnvi[i]; \qEnvi[i]}{\rolesR}$, we get 
      $\gtProj[\rolesR]{\gtG[i]}{\roleP} = \stEnvApp{\stEnvi[i]}{\roleP} \stSub 
      \gtProj[\rolesR]{\gtGi[i]}{\roleP}$. Then by transitivity of subtyping, 
      it holds that $\forall i \in I': \stT[i] \stSub \gtProj[\rolesR]{\gtGi[i]}{\roleP}$, which follows that 
      $ \stIntSum{\roleQ}{i \in I'}{ %
        \stChoice{\stLab[i]}{\tyGround[i]} \stSeq \stT[i]%
      } \stSub 
      \stIntSum{\roleQ}{i \in \stIdxRemoveCrash{I}{\stLab[i]}}{ %
        \stChoice{\stLab[i]}{\tyGround[i]} \stSeq (\gtProj[\rolesR]{\gtGi[i]}{\roleP})%
      }$, as desired. 
     
      \item $\roleR = \roleQ$: by association \eqref{eq:soundness:context_1_assoc_pre} and 
           \cref{lem:subtyping-invert}, we obtain $\stEnvApp{\stEnv}{\roleQ} =  
          \stExtSum{\roleP}{i \in I'}{%
        \stChoice{\stLab[i]}{\tyGround[i]} \stSeq \stT[i]%
      }$ where $I \subseteq I'$ and $\forall i \in I: \stT[i] \stSub 
      \gtProj[\rolesR]{\gtG[i]}{\roleQ}$. Since $\stEnvApp{\stEnv[i]}{\roleQ} = 
       \gtProj[\rolesR]{\gtG[i]}{\roleQ} = 
      \stEnvApp{\stEnvi[i]}{\roleQ}$ for each $i \in I$, with association 
      $\stEnvAssoc{\gtWithCrashedRoles{\rolesCi}{\gtGi[i]}}{\stEnvi[i]; \qEnvi[i]}{\rolesR}$, we get 
      $\gtProj[\rolesR]{\gtG[i]}{\roleQ} = \stEnvApp{\stEnvi[i]}{\roleQ} \stSub 
      \gtProj[\rolesR]{\gtGi[i]}{\roleQ}$. Then by transitivity of subtyping, 
      it holds that $\forall i \in I: \stT[i] \stSub \gtProj[\rolesR]{\gtGi[i]}{\roleQ}$, which follows that 
      $ \stExtSum{\roleP}{i \in I'}{ %
        \stChoice{\stLab[i]}{\tyGround[i]} \stSeq \stT[i]%
      } \stSub 
      \stExtSum{\roleP}{i \in I}{ %
        \stChoice{\stLab[i]}{\tyGround[i]} \stSeq (\gtProj[\rolesR]{\gtGi[i]}{\roleQ})%
      }$, as desired. 
     
      \item $\roleR \neq \roleP$ and $\roleR \neq \roleQ$: by association 
      $\stEnvAssoc{\gtWithCrashedRoles{\rolesCi}{\gtGi[i]}}{\stEnvi[i]; \qEnvi[i]}{\rolesR}$ for each 
      $i \in I$, it holds that $\stEnvApp{\stEnvi[i]}{\roleR} \stSub \gtProj[\rolesR]{\gtGi[i]}{\roleR}$ for each 
      $i \in I$. Then by applying \cref{lem:subtype:merge-subty}, we conclude with $\stEnvApp{\stEnvi}{\roleR} = \stMerge{i \in I}{\stEnvApp{\stEnvi[i]}{\roleR}}
      \stSub  \stMerge{i \in I}{\gtProj[\rolesR]{\gtGi[i]}{\roleR}} = \gtProj[\rolesR]{\gtGi}{\roleR}$, as desired. 
     \end{itemize}
      
       \item By association $\stEnvAssoc{\gtWithCrashedRoles{\rolesCi}{\gtGi[i]}}{\stEnvi[i]; \qEnvi[i]}{\rolesR}$ for each 
      $i \in I$, it holds that $\forall \roleR \in \rolesCi: \stEnvApp{\stEnvi[i]}{\roleR}  = \stStop$, which follows that for any 
      $\roleR \in \rolesCi$, $\stEnvApp{\stEnvi}{\roleR} =  \stMerge{i \in I}{\stEnvApp{\stEnvi[i]}{\roleR}} = \stStop$, as desired.
      \item For any endpoint $\roleR$ in $\gtGi$, $\roleR$ is an endpoint in each $\gtGi[i]$ with $i \in I$. Then by association 
      $\stEnvAssoc{\gtWithCrashedRoles{\rolesCi}{\gtGi[i]}}{\stEnvi[i]; \qEnvi[i]}{\rolesR}$, we have $\stEnvApp{\stEnvi[i]}{\roleR} = 
      \stEnd$, which follows $\stEnvApp{\stEnvi}{\roleR} =  \stMerge{i \in I}{\stEnvApp{\stEnvi[i]}{\roleR}} = \stEnd$, as desired. 
      \item By \cref{lem:stenv-queue-red:det}, it holds that for any $i, j \in I$, $\qEnvi[i] = \qEnvi[j]$. Take an arbitrary 
      $\qEnvi[j]$ with $j \in I$. Note here $\qEnvi = \qEnvi[j]$. With the fact 
      that $\qEnvi[i]$ is associated with $\gtGi[i]$ for any $i \in I$, we have that $\forall i \in I: \qEnvi[j] \text{~is associated with~} \gtGi[i]$, and hence, 
       $\forall i \in I: \qEnvi \text{~is associated with~} \gtGi[i]$. 
       We are left to show that $\stEnvApp{\qEnvi}{\roleP, \roleQ} = \stQEmpty$, which is obtained by applying \cref{lem:stenv-red:trivial-3}
      on $\stEnvApp{\qEnv}{\roleP, \roleQ} = \stQEmpty$ and  $\stEnv[j]; \qEnv
  \stEnvMoveGenAnnot \stEnvi[j]; \qEnvi[j]$\,
  with\, 
  $\ltsSubject{\stEnvAnnotGenericSym} \notin \setenum{\roleP, \roleQ}$. 
      \end{enumerate}
   
 \item $\roleQMaybeCrashed = \roleQCrashed$:
     
      Note that $\roleQ \in \rolesC$.
      For $\roleP$, we have
     $\gtProj[\rolesR]{\gtG}{\roleP} =
       \stIntSum{\roleQ}{i \in \stIdxRemoveCrash{I}{\stLab[i]}}{ %
        \stChoice{\stLab[i]}{\tyGround[i]} \stSeq (\gtProj[\rolesR]{\gtG[i]}{\roleP})%
      }$ and $ \gtProj[\rolesR]{\gtGi}{\roleP} =
       \stIntSum{\roleQ}{i \in \stIdxRemoveCrash{I}{\stLab[i]}}{ %
        \stChoice{\stLab[i]}{\tyGround[i]} \stSeq (\gtProj[\rolesR]{\gtGi[i]}{\roleP})%
      }$. 
      Take an arbitrary $j \in I$. %
      Let $\stEnvApp{\stEnv[j]}{\roleP} =
      \gtProj[\rolesR]{\gtG[j]}{\roleP}$, $\stEnvApp{\stEnv[j]}{\roleQ} =
      \stStop$, $\stEnvApp{\stEnv[j]}{\roleR} =
     \stEnvApp{\stEnv}{\roleR}$ for $\roleR \in \dom{\stEnv} \setminus \setenum{\roleP, \roleQ}$, $\qEnv[j] = \qEnv$.
     We show $\stEnvAssoc{\gtWithCrashedRoles{\rolesC}{\gtG[j]}}{\stEnv[j]; \qEnv[j]}{\rolesR}$.
      \begin{enumerate}[label=(A\arabic*)]
      \item We want to show $\forall \roleS \in \gtRoles{\gtG[j]}: \stEnvApp{\stEnv[j]}{\roleS} \stSub \gtProj[\rolesR]{\gtG[j]}{\roleS}$. We consider 
      two subcases:
      \begin{itemize}[leftmargin=*]
      \item $\roleS = \roleP$ (meaning that $\roleP \in \gtRoles{\gtG[j]}$): trivial by $\stEnvApp{\stEnv[j]}{\roleP} =
      \gtProj[\rolesR]{\gtG[j]}{\roleP}$ and the reflexivity of subtyping. 
      \item $\roleS \neq \roleP$: since $\roleQ \in \rolesC$, we have $\roleS \neq \roleQ$. By association \eqref{eq:soundness:context_1_assoc_pre} and 
      $\stEnvApp{\stEnv[j]}{\roleS} =
     \stEnvApp{\stEnv}{\roleS}$, we have $\stEnvApp{\stEnv[j]}{\roleS} \stSub 
       \stMerge{i \in I}{\gtProj[\rolesR]{\gtG[i]}{\roleS}}$. 
        Then, by \Cref{lem:merge-subtyping} and transitivity of subtyping,  we conclude  $\stEnvApp{\stEnv[j]}{\roleS}
        \stSub \gtProj[\rolesR]{\gtG[j]}{\roleS}$, as desired. 
      \end{itemize}
       \item No change here.
      \item No change here if $\roleP \in \gtRoles{\gtG[j]}$. Otherwise, consider the case that $\roleP \notin \gtRoles{\gtG[j]}$: 
      with $\roleP \notin \gtRolesCrashed{\gtG[j]}$, by \Cref{lem:not-in-roles-end}, we have $\gtProj[\rolesR]{\gtG[j]}{\roleP} = \stEnd$. 
      Therefore, we know from $\stEnvApp{\stEnv[j]}{\roleP} =
      \gtProj[\rolesR]{\gtG[j]}{\roleP}$ that $ \stEnvApp{\stEnv[j]}{\roleP} = \stEnd$, as required. 
     \item Trivial by association~\eqref{eq:soundness:context_1_assoc_pre}, \cref{def:assoc-queue}, and $\qEnv[j] = \qEnv$. 
      \end{enumerate}
      By inductive hypothesis, there exists $\stEnvi[j]; \qEnvi[j]$ such that
      $\stEnv[j]; \qEnv[j]  \stEnvMoveGenAnnot
      \stEnvi[j]; \qEnvi[j]$ and
      $\stEnvAssoc{\gtWithCrashedRoles{\rolesCi}{\gtGi[j]}}{\stEnvi[j]; \qEnvi[j]}{\rolesR}$. Since
      $\ltsSubject{\stEnvAnnotGenericSym} \notin \setenum{\roleP, \roleQ}$, we apply \cref{lem:stenv-red:trivial-2}, which gives
      $\stEnvApp{\stEnv[j]}{\roleP} = \stEnvApp{\stEnvi[j]}{\roleP}$ and
      $\stEnvApp{\stEnv[j]}{\roleQ} = \stEnvApp{\stEnvi[j]}{\roleQ} = \stStop$.
      We know from $\stEnvApp{\stEnvi[j]}{\roleQ} = \stStop$ and   
      $\stEnvAssoc{\gtWithCrashedRoles{\rolesCi}{\gtGi[j]}}{\stEnvi[j]; \qEnvi[j]}{\rolesR}$ that $\roleQ \in \rolesCi$. 

      We now construct a configuration $\stEnvi; \qEnvi$ and show 
      $\stEnv; \qEnv \stEnvMoveGenAnnot \stEnvi; \qEnvi$ and 
      $\stEnvAssoc{\gtWithCrashedRoles{\rolesCi}{\gtGi}}{\stEnvi; \qEnvi}{\rolesR}$.

      Let  $\stEnvApp{\stEnvi}{\roleP}
      = \stEnvApp{\stEnv}{\roleP}$, 
       $\stEnvApp{\stEnvi}{\roleQ}
      =
      \stEnvApp{\stEnv}{\roleQ}
      = \stStop$, 
       $\stEnvApp{\stEnvi}{\roleR} =
      \stMerge{i \in I}{\stEnvApp{\stEnvi[i]}{\roleR}}$ for
    $\roleR \in \dom{\stEnv} \setminus
      \setenum{\roleP, \roleQ}$, $\qEnvi = \qEnvi[j]$ with an arbitrary $j \in I$.
      
      For the first part that $\stEnv; \qEnv \stEnvMoveGenAnnot \stEnvi; \qEnvi$: 
      
        We know that for any $i \in I$, $\stEnvi[i]$ is obtained from $\stEnv[i]$ 
      by updating $\ltsSubject{\stEnvAnnotGenericSym}$ to a fixed type $\stT$. It follows that 
      for any $i, k \in I$ and for any $\roleR \notin \setenum{\roleP, \roleQ}$, $\stEnvApp{\stEnvi[i]}{\roleR} 
      = 
      \stEnvApp{\stEnvi[k]}{\roleR}$, and hence, $\stMerge{i \in I}{\stEnvApp{\stEnvi[i]}{\roleR}} = 
      \stEnvApp{\stEnvi[j]}{\roleR}$ with an arbitrary $j \in J$. Therefore, we have that $\stEnvi$ is obtained from 
      $\stEnv$ by updating $\ltsSubject{\stEnvAnnotGenericSym}$ to $\stEnvApp{\stEnvi[j]}{\ltsSubject{\stEnvAnnotGenericSym}}$,  
      i.e., $\stEnvi = \stEnvUpd{\stEnv}{\ltsSubject{\stEnvAnnotGenericSym}}{\stEnvApp{\stEnvi[j]}{\ltsSubject{\stEnvAnnotGenericSym}}}$.     
      
      We apply  \cref{lem:stenv-red:non-trivial} on $\stEnv[j]; \qEnv[j]
  \stEnvMoveGenAnnot \stEnvi[j]; \qEnvi[j]$,
  $\stEnvAssoc{\gtWithCrashedRoles{\rolesC}{\gtG}}{\stEnv; \qEnv}{\rolesR}$, 
  $\stEnvAssoc{\gtWithCrashedRoles{\rolesC}{\gtG[j]}}{\stEnv[j]; \qEnv[j]}{\rolesR}$, 
  $\stEnvApp{\stEnv}{\ltsSubject{\stEnvAnnotGenericSym}} = \stEnvApp{\stEnv[j]}{\ltsSubject{\stEnvAnnotGenericSym}}$, 
  and 
  $\stEnvi =  \stEnvUpd{\stEnv}{\ltsSubject{\stEnvAnnotGenericSym}}{\stEnvApp{\stEnvi[j]}{\ltsSubject{\stEnvAnnotGenericSym}}}$ to get the thesis. 
      
      For the second part that  $\stEnvAssoc{\gtWithCrashedRoles{\rolesCi}{\gtGi}}{\stEnvi; \qEnvi}{\rolesR}$: 
      \begin{enumerate}[label=(A\arabic*)]
      \item We want to show $\forall \roleR \in \gtRoles{\gtGi}: \stEnvApp{\stEnvi}{\roleR}
           \stSub \gtProj[\rolesR]{\gtGi}{\roleR}$. We consider two subcases:
           \begin{itemize}[leftmargin=*]
           \item $\roleR = \roleP$: by association \eqref{eq:soundness:context_1_assoc_pre} and 
           \cref{lem:subtyping-invert}, we obtain $\stEnvApp{\stEnv}{\roleP} =  
           \stIntSum{\roleQ}{i \in I'}{ %
        \stChoice{\stLab[i]}{\tyGround[i]} \stSeq \stT[i]%
      }$ where $I' \subseteq \stIdxRemoveCrash{I}{\stLab[i]}$ and $\forall i \in I': \stT[i] \stSub 
      \gtProj[\rolesR]{\gtG[i]}{\roleP}$. Since $\stEnvApp{\stEnv[i]}{\roleP} = 
       \gtProj[\rolesR]{\gtG[i]}{\roleP} = 
      \stEnvApp{\stEnvi[i]}{\roleP}$ for each $i \in I$, with association 
      $\stEnvAssoc{\gtWithCrashedRoles{\rolesCi}{\gtGi[i]}}{\stEnvi[i]; \qEnvi[i]}{\rolesR}$, we get 
      $\gtProj[\rolesR]{\gtG[i]}{\roleP} = \stEnvApp{\stEnvi[i]}{\roleP} \stSub 
      \gtProj[\rolesR]{\gtGi[i]}{\roleP}$. Then by transitivity of subtyping, 
      it holds that $\forall i \in I': \stT[i] \stSub \gtProj[\rolesR]{\gtGi[i]}{\roleP}$, which follows that 
      $ \stIntSum{\roleQ}{i \in I'}{ %
        \stChoice{\stLab[i]}{\tyGround[i]} \stSeq \stT[i]%
      } \stSub 
      \stIntSum{\roleQ}{i \in \stIdxRemoveCrash{I}{\stLab[i]}}{ %
        \stChoice{\stLab[i]}{\tyGround[i]} \stSeq (\gtProj[\rolesR]{\gtGi[i]}{\roleP})%
      }$, as desired.

      \item $\roleR \neq \roleP$: since $\roleQ \in \rolesCi$, we have $\roleR \neq \roleQ$. By association 
      $\stEnvAssoc{\gtWithCrashedRoles{\rolesCi}{\gtGi[i]}}{\stEnvi[i]; \qEnvi[i]}{\rolesR}$ for each 
      $i \in I$, it holds that $\stEnvApp{\stEnvi[i]}{\roleR} \stSub \gtProj[\rolesR]{\gtGi[i]}{\roleR}$ for each 
      $i \in I$. Then by applying \cref{lem:subtype:merge-subty}, we conclude with $\stEnvApp{\stEnvi}{\roleR} = \stMerge{i \in I}{\stEnvApp{\stEnvi[i]}{\roleR}}
      \stSub  \stMerge{i \in I}{\gtProj[\rolesR]{\gtGi[i]}{\roleR}} = \gtProj[\rolesR]{\gtGi}{\roleR}$, as desired. 
     \end{itemize}
      
       \item We want to show $\forall \roleR \in \rolesCi: \stEnvApp{\stEnvi}{\roleR} = \stStop$. 
          We consider two subcases:
           \begin{itemize}[leftmargin=*]
       \item $\roleR = \roleQ$: trivial by $\stEnvApp{\stEnvi}{\roleQ} = \stStop$. 
       \item $\roleR \neq \roleQ$: 
       by association $\stEnvAssoc{\gtWithCrashedRoles{\rolesCi}{\gtGi[i]}}{\stEnvi[i]; \qEnvi[i]}{\rolesR}$ for each 
      $i \in I$, it holds that  $\forall \roleR \in \rolesCi: \stEnvApp{\stEnvi[i]}{\roleR}  = \stStop$, which follows that for any 
      $\roleR \in \rolesCi$ with $\roleR \neq \roleQ$, $\stEnvApp{\stEnvi}{\roleR} =  \stMerge{i \in I}{\stEnvApp{\stEnvi[i]}{\roleR}} = \stStop$, as desired.
      \end{itemize}
     
      \item For any endpoint $\roleR$ in $\gtGi$, $\roleR$ is an endpoint in each $\gtGi[i]$ with $i \in I$. Then by association 
      $\stEnvAssoc{\gtWithCrashedRoles{\rolesCi}{\gtGi[i]}}{\stEnvi[i]; \qEnvi[i]}{\rolesR}$, we have $\stEnvApp{\stEnvi[i]}{\roleR} = 
      \stEnd$, which follows $\stEnvApp{\stEnvi}{\roleR} =  \stMerge{i \in I}{\stEnvApp{\stEnvi[i]}{\roleR}} = \stEnd$, as desired. 
     
      \item By \cref{lem:stenv-queue-red:det}, it holds that for any $i, j \in I$, $\qEnvi[i] = \qEnvi[j]$. Take an arbitrary 
      $\qEnvi[j]$ with $j \in I$. Note here $\qEnvi = \qEnvi[j]$. With the fact 
      that $\qEnvi[i]$ is associated with $\gtGi[i]$ for any $i \in I$, we have that $\forall i \in I: \qEnvi[j] \text{~is associated with~} \gtGi[i]$, and hence, 
       $\forall i \in I: \qEnvi \text{~is associated with~} \gtGi[i]$. 
       We are left to show that $\stEnvApp{\qEnvi}{\roleP, \roleQ} = \stQUnavail$, which is obtained by applying \cref{lem:stenv-red:trivial-3}
      on $\stEnvApp{\qEnv}{\roleP, \roleQ} = \stQUnavail$ and  $\stEnv[j]; \qEnv
  \stEnvMoveGenAnnot \stEnvi[j]; \qEnvi[j]$\,
  with\, 
  $\ltsSubject{\stEnvAnnotGenericSym} \notin \setenum{\roleP, \roleQ}$. 
      \end{enumerate}

      \end{itemize}
  \item Case \inferrule{\iruleGtMoveCtxi}: similar to the case  \inferrule{\iruleGtMoveCtx}. 
     \qedhere
    
       \end{itemize}
\end{proof}

\thmProjCompleteness*
\begin{proof}
  By induction on reductions of configuration $\stEnv; \qEnv \stEnvMoveGenAnnot
  \stEnvi; \qEnvi$.
  \begin{itemize}[leftmargin=*]
    \item Case $\inferrule{\iruleTCtxOut}$:

    From the premise, we have:
    \begin{gather}
      \stEnvAssoc{\gtWithCrashedRoles{\rolesC}{\gtG}}{\stEnv; \qEnv}{\rolesR}
      \\
      \stEnvApp{\stEnv}{%
        \roleP%
      } =
      \stIntSum{\roleQ}{i \in I}{\stChoice{\stLab[i]}{\tyGround[i]} \stSeq \stT[i]}%
      \label{eq:stenvp-send}
      \\
      \stEnvAnnotGenericSym = \stEnvOutAnnotSmall{\roleP}{\roleQ}{\stChoice{\stLab[k]}{\tyGround[k]}}
      \\
      k \in I
      \\
      \stEnvi =
      \stEnvUpd{\stEnv}{\roleP}{\stT[k]}
      \label{eq:stenvp-send-type-cxt-cons}
      \\
      \qEnvi =
      \stEnvUpd{\qEnv}{\roleP, \roleQ}{%
        \stQCons{%
          \stEnvApp{\qEnv}{\roleP, \roleQ}
        }{%
          \stQMsg{\stLab[k]}{\tyGround[k]}
        }
      }
      \label{eq:stenvp-send-queue-cons}%
    \end{gather}

    Apply \cref{lem:inv-proj} \cref{item:proj-inv:send} on
    \eqref{eq:stenvp-send}, we have two cases.
    \begin{itemize}[leftmargin=*]
      \item Case (1):
    \begin{gather}
      \unfoldOne{\gtG} =
        \gtComm{\roleP}{\roleQMaybeCrashed}{i \in I'}{\gtLab[i]}{\tyGroundi[i]}{\gtG[i]}
      \\
      I \subseteq I'
      \\
      \forall i \in I:
      \stLab[i] = \gtLab[i],
      \stT[i] \stSub (\gtProj[\rolesR]{\gtG[i]}{\roleP}),
      \tyGround[i] = \tyGroundi[i]
      \label{eq:stenvp-send-sub-cons}
    \end{gather}

    We have two further subcases here: namely $\roleQMaybeCrashed =
    \roleQ$ and $\roleQMaybeCrashed = \roleQCrashed$.

      \begin{itemize}[leftmargin=*]
        \item
        In the case of $\roleQMaybeCrashed = \roleQ$,
        we have $k \in I \subseteq I'$ and $\stCrashLab$ does not appear in
        internal choices, we apply $\inferrule{\iruleGtMoveOut}$ (via
        \cref{lem:gt-lts-unfold}) to get:
        \begin{gather}
          \gtWithCrashedRoles{\rolesC}{\gtG}
          \gtMove[\stEnvAnnotGenericSym]{\rolesR}
          \gtWithCrashedRoles{\rolesC}{
            \gtCommTransit{\roleP}{\roleQ}{i \in I'}{\gtLab[i]}{\tyGroundi[i]}{\gtG[i]}{k}
          }
        \end{gather}

        We are now left to show $\stEnvAssoc{
          \gtWithCrashedRoles{\rolesC}{
            \gtCommTransit{\roleP}{\roleQ}{i \in I'}{\gtLab[i]}{\tyGroundi[i]}{\gtG[i]}{k}
          }
        }{\stEnvi; \qEnvi}{\rolesR}$.

        Note that $\gtGi = \gtCommTransit{\roleP}{\roleQ}{i \in I'}{\gtLab[i]}{\tyGroundi[i]}{\gtG[i]}{k}$ here.
        By \cref{def:active_crashed_roles}, we have  $\gtRoles{\gtG} = \gtRoles{\gtGi}$ and
        $\gtRolesCrashed{\gtG} =  \gtRolesCrashed{\gtGi}$. Furthermore, with $\dom{\stEnv} = \dom{\stEnvi}$ and
        $\roleP \in \gtRoles{\gtG}$, we can set $\stEnvi =  \stEnvUpd{\stEnv[\gtG]}{\roleP}{\stT[k]}\stEnvComp \stEnv[\stStopSym] \stEnvComp
  \stEnv[\stEnd]$.

        \begin{enumerate}[label=(A\arabic*)]
          \item First, we want to show that $\dom{\stEnvUpd{\stEnv[\gtG]}{\roleP}{\stT[k]}} =  \setcomp{{\roleP}}{\roleP \in \gtRoles{\gtGi})}$, which follows directly
          from the fact that $\dom{\stEnv[\gtG]} = \setcomp{{\roleP}}{\roleP \in \gtRoles{\gtG})}$, $\gtRoles{\gtG} = \gtRoles{\gtGi}$, and
          $\dom{\stEnvUpd{\stEnv[\gtG]}{\roleP}{\stT[k]}} = \dom{\stEnv[\gtG]}$.

          Then, we are left to show $ \forall \roleR \in \gtRoles{\gtGi}:
        \stEnvApp{\stEnvUpd{\stEnv[\gtG]}{\roleP}{\stT[k]}}{{\roleR}}
        \stSub
        \gtProj[\rolesR]{\gtGi}{\roleR}
      $. We consider two cases:
      \begin{itemize}[leftmargin=*]
      \item $\roleR = \roleP$: we have $\gtProj[\rolesR]{\gtGi}{\roleP} = \gtProj[\rolesR]{\gtG[k]}{\roleP}$. By  \eqref{eq:stenvp-send-sub-cons},
      we obtain that $ \stEnvApp{\stEnvUpd{\stEnv[\gtG]}{\roleP}{\stT[k]}}{{\roleP}} = \stT[k] \stSub \gtProj[\rolesR]{\gtG[k]}{\roleP} = \gtProj[\rolesR]{\gtGi}{\roleP}$, as desired.
      \item $\roleR \neq \roleP$: since $\stEnvAssoc{\gtWithCrashedRoles{\rolesC}{\gtG}}{\stEnv; \qEnv}{\rolesR}$,
      we have that $ \stEnvApp{\stEnv[\gtG]}{{\roleR}} \stSub \gtProj[\rolesR]{\gtG}{\roleR} =   \stMerge{i \in I'}{\gtProj[\rolesR]{\gtG[i]}{\roleR}}$.
      Furthermore, by $\stEnvApp{\stEnvUpd{\stEnv[\gtG]}{\roleP}{\stT[k]}}{{\roleR}} = \stEnvApp{\stEnv[\gtG]}{{\roleR}}$ and
      $\gtProj[\rolesR]{\gtGi}{\roleR} =   \stMerge{i \in I'}{\gtProj[\rolesR]{\gtG[i]}{\roleR}}$, we have that
      $\stEnvApp{\stEnvUpd{\stEnv[\gtG]}{\roleP}{\stT[k]}}{{\roleR}} \stSub \gtProj[\rolesR]{\gtGi}{\roleR}$, as desired.
      \end{itemize}

           \item Trivial by $\stEnvAssoc{\gtWithCrashedRoles{\rolesC}{\gtG}}{\stEnv; \qEnv}{\rolesR}$.
          \item Trivial by $\stEnvAssoc{\gtWithCrashedRoles{\rolesC}{\gtG}}{\stEnv; \qEnv}{\rolesR}$.
          \item By \eqref{eq:stenvp-send-queue-cons}, we have  $ \stEnvApp{\qEnvi}{\roleP, \roleQ} =
        \stQCons{%
          \stEnvApp{\qEnv}{\roleP, \roleQ}
        }{%
          \stQMsg{\stLab[k]}{\tyGround[k]}
        }
      $. Then we only need to show that $\forall i \in I':
      \stEnvUpd{\qEnvi}{\roleP, \roleQ}{ \stEnvApp{\qEnv}{\roleP, \roleQ}} \text{~(note that~} \stEnvUpd{\qEnvi}{\roleP, \roleQ}{ \stEnvApp{\qEnv}{\roleP, \roleQ}} = \qEnv) \text{~is associated with~} \gtWithCrashedRoles{\rolesC}{\gtG[i]}$, which follows directly from
       the fact that $\qEnv$ is associated with $\gtWithCrashedRoles{\rolesC}{\gtG}$ and \cref{def:assoc-queue}.

        \end{enumerate}

        \item
        In the case of $\roleQMaybeCrashed = \roleQCrashed$,
        we have $k \in I \subseteq I'$, $\roleQ \in \rolesC$, and
        $\stCrashLab$ does not appear in
        internal choices.
        We apply $\inferrule{\iruleGtMoveOrph}$ (via
        \cref{lem:gt-lts-unfold}) to get:
        \begin{gather}
          \gtWithCrashedRoles{\rolesC}{\gtG}
          \gtMove[\stEnvAnnotGenericSym]{\rolesR}
          \gtWithCrashedRoles{\rolesC}{
            \gtG[k]
          }
        \end{gather}
        We are now left to show $\stEnvAssoc{
          \gtWithCrashedRoles{\rolesC}{
            \gtG[k]
          }
        }{\stEnvi; \qEnvi}{\rolesR}$.

        Note that $\gtGi = \gtG[k]$ here.
 \begin{enumerate}[label=(A\arabic*)]
          \item We need to show that $\forall \roleR \in \gtRoles{\gtG[k]}:
          \stEnvApp{\stEnvi}{\roleR} \stSub
        \gtProj[\rolesR]{\gtG[k]}{\roleR}$.
        We consider two subcases:
        \begin{itemize}[leftmargin=*]
        \item $\roleR = \roleP$, which means that $\roleP \in \gtRoles{\gtG[k]}$: by \eqref{eq:stenvp-send-type-cxt-cons} and \eqref{eq:stenvp-send-sub-cons}, we have
          $\stEnvApp{\stEnvi}{\roleP} =  \stT[k] \stSub \gtProj[\rolesR]{\gtG[k]}{\roleP}$, as desired.
        \item $\roleR \neq \roleP$: with $\roleQ \notin \gtRoles{\gtG[k]}$, we have $\roleR \neq \roleQ$, and hence,
        $\gtProj[\rolesR]{\gtG}{\roleR} = \stMerge{i \in I'}{\gtProj[\rolesR]{\gtG[i]}{\roleR}}$.
        Furthermore, by \eqref{eq:stenvp-send-type-cxt-cons} and $\stEnvAssoc{\gtWithCrashedRoles{\rolesC}{\gtG}}{\stEnv; \qEnv}{\rolesR}$,
        it holds that $\stEnvApp{\stEnvi}{\roleR} = \stEnvApp{\stEnv}{\roleR} \stSub \gtProj[\rolesR]{\gtG}{\roleR} = \stMerge{i \in I'}{\gtProj[\rolesR]{\gtG[i]}{\roleR}}$.
        Then, by \Cref{lem:merge-subtyping} and transitivity of subtyping,  we can conclude that $\stEnvApp{\stEnvi}{\roleR}
        \stSub \gtProj[\rolesR]{\gtG[k]}{\roleR}$, as desired.
        \end{itemize}
         \item Trivial by $\stEnvAssoc{\gtWithCrashedRoles{\rolesC}{\gtG}}{\stEnv; \qEnv}{\rolesR}$.
          \item Trivial by $\stEnvAssoc{\gtWithCrashedRoles{\rolesC}{\gtG}}{\stEnv; \qEnv}{\rolesR}$ and the
                   fact that if $\roleP \notin \gtRoles{\gtG[k]}$, then $\stEnvApp{\stEnvi}{\roleP} = \stEnd$:
                   with $\roleP \notin \gtRolesCrashed{\gtG[k]}$, by \Cref{lem:not-in-roles-end}, we have $\gtProj[\rolesR]{\gtG[k]}{\roleP} = \stEnd$.  Furthermore, by \eqref{eq:stenvp-send-sub-cons}, it holds that $\stT[k] = \stEnd$, and thus,
          $\stEnvApp{\stEnvi}{\roleP} = \stEnd$, as desired.

          \item Since $\roleQ \in \rolesC$, by \Cref{def:assoc-queue}, we have
          $\stEnvApp{\qEnv}{\cdot, \roleQ} = \stQUnavail$.
          Hence,  $\qEnvi =
      \stEnvUpd{\qEnv}{\roleP, \roleQ}{%
        \stQCons{%
          \stEnvApp{\qEnv}{\roleP, \roleQ}
        }{%
          \stQMsg{\stLab[k]}{\tyGround[k]}
        }
      } =
      \stEnvUpd{\qEnv}{\roleP, \roleQ}{%
        \stQCons{%
         \stQUnavail
        }{%
          \stQMsg{\stLab[k]}{\tyGround[k]}
        }
      } =  \stEnvUpd{\qEnv}{\roleP, \roleQ}{\stQUnavail} = \qEnv$.
      Then we only need to show that $\qEnv$ is associated with
       $\gtWithCrashedRoles{\rolesC}{
            \gtG[k]}$, which follows directly from the fact that
            $\qEnv$ is associated with $\gtWithCrashedRoles{\rolesC}{\gtG}$ and \cref{def:assoc-queue}.

        \end{enumerate}
      \end{itemize}
      \item Case (2):
      \begin{gather}
       \unfoldOne{\gtG} =
            \gtComm{\roleS}{\roleTMaybeCrashed}{j \in J}{\gtLab[j]}{\tyGroundi[j]}{\gtG[j]}
          \text{ \;or\; }
          \unfoldOne{\gtG} =
            \gtCommTransit{\roleSMaybeCrashed}{\roleT}{j \in
            J}{\gtLab[j]}{\tyGroundi[j]}{\gtG[j]}{k}
            \\
       \forall j \in J:
        \stEnvApp{\stEnv}{%
        \roleP%
        } \stSub \gtProj[\rolesR]{\gtG[j]}{\roleP}
        \\
        \roleP \neq \roleS \text{ \;and\; } \roleP \neq \roleT
      \end{gather}

  We consider two subcases:  $\unfoldOne{\gtG} =
            \gtComm{\roleS}{\roleTMaybeCrashed}{j \in J}{\gtLab[j]}{\tyGroundi[j]}{\gtG[j]}$ and
          $\unfoldOne{\gtG} =
            \gtCommTransit{\roleSMaybeCrashed}{\roleT}{j \in
            J}{\gtLab[j]}{\tyGroundi[j]}{\gtG[j]}{k}$.

  \begin{itemize}[leftmargin=*]
  \item In the case of $\unfoldOne{\gtG} =
            \gtComm{\roleS}{\roleTMaybeCrashed}{j \in J}{\gtLab[j]}{\tyGroundi[j]}{\gtG[j]}$: %

            First, we take an arbitrary index $j \in J$ and construct a configuration $\stEnv[j]; \qEnv[j]$ such that
  $\stEnv[j]; \qEnv[j]
      \,\stEnvMoveOutAnnot{\roleP}{\roleQ}{\stChoice{\stLab[k]}{\tyGround[k]}}\,
      \stEnvi[j]; \qEnvi[j]$ and
       $\stEnvAssoc{\gtWithCrashedRoles{\rolesC}{\gtG[j]}}{\stEnv[j]; \qEnv[j]}{\rolesR}$.

      We know from $\roleS \in \gtRoles{\gtG}$ and
           $\stEnvAssoc{\gtWithCrashedRoles{\rolesC}{\gtG}}{\stEnv; \qEnv}{\rolesR}$, that
           $
        \stEnvApp{\stEnv}{\roleS}
        \stSub
         \gtProj[\rolesR]{\unfoldOne{\gtG}}{\roleS}$ and
         $  \gtProj[\rolesR]{\unfoldOne{\gtG}}{\roleS}
         =
        \stIntSum{\roleT}{j \in \stIdxRemoveCrash{J}{\stLab[j]}}{%
          \stChoice{\stLab[j]}{\tyGroundi[j]}
                 \stSeq (\gtProj[\rolesR]{\gtG[j]}{\roleS})}$.
      By inverting $\inferrule{\iruleStSubOut}$ (applying
      \cref{lem:unfold-subtyping} where necessary), we have
      $ \unfoldOne{\stEnvApp{\stEnv}{\roleS}}
        =
        \stIntSum{\roleT}{j \in J_{\roleS}}{%
          \stChoice{\stLab[j]}{\tyGroundii[j]} \stSeq \stTii[j]
        }
      $, where $J_{\roleS} \subseteq \stIdxRemoveCrash{J}{\stLab[j]}$,
      and $\forall j \in J_{\roleS}:
       \stTii[j] \stSub (\gtProj[\rolesR]{\gtG[j]}{\roleS})$.

      To construct $\stEnv[j]$, let
      $ \stEnvApp{\stEnv[j]}{\roleS}
        =
        \stTii[j]
      $ if $j \in J_{\roleS}$ and
      $ \stEnvApp{\stEnv[j]}{\roleS}
        =
        \gtProj[\rolesR]{\gtG[j]}{\roleS}
      $ otherwise.
      In either case, we have
      $ \stEnvApp{\stEnv[j]}{\roleS}
      \stSub
      \gtProj[\rolesR]{\gtG[j]}{\roleS}$, as required.

      We have two further subcases here: namely $\roleTMaybeCrashed = \roleT$ and $\roleTMaybeCrashed = \roleTCrashed$.
      \begin{itemize}[leftmargin=*]
       \item If $\roleTMaybeCrashed = \roleT$, we know
      from $\roleT \in \gtRoles{\gtG}$ and
      $\stEnvAssoc{\gtWithCrashedRoles{\rolesC}{\gtG}}{\stEnv; \qEnv}{\rolesR}$, that
      $ \stEnvApp{\stEnv}{\roleT} \stSub
      \gtProj[\rolesR]{\unfoldOne{\gtG}}{\roleT}
        =
        \stExtSum{\roleS}{j \in J}{\stChoice{\stLab[j]}{\tyGroundi[j]}
                 \stSeq (\gtProj[\rolesR]{\gtG[j]}{\roleT})}
      $.
      By inverting $\inferrule{\iruleStSubIn}$ (applying
      \cref{lem:unfold-subtyping} where necessary), we have
      $ \unfoldOne{\stEnvApp{\stEnv}{\roleT}}
        =
        \stExtSum{\roleS}{j \in J_{\roleT}}{%
          \stChoice{\stLab[j]}{\tyGroundii[j]} \stSeq \stUii[j]}$,
      where $J \subseteq J_{\roleT}$,
      and $\forall j \in J:
      \stUii[j] \stSub (\gtProj[\rolesR]{\gtG[j]}{\roleT})$.

      To construct $\stEnv[j]$, let
      $ \stEnvApp{\stEnv[j]}{\roleT}
        =
        \stUii[j]
      $, and we have
      $ \stEnvApp{\stEnv[j]}{\roleT}
        \stSub
        \gtProj[\rolesR]{\gtG[j]}{\roleT}
        $, as required.
        \item If $\roleTMaybeCrashed = \roleTCrashed$, let
        $ \stEnvApp{\stEnv[j]}{\roleT}
        =
        \stEnvApp{\stEnv}{\roleT} = \stStop$,
      as required.
       \end{itemize}
       For roles $\roleR \in (\gtRoles{\gtG[j]} \cup \rolesC)$,
       where $\roleR \notin
  \setenum{\roleS, \roleT}$, their typing context
  entry do not change, \ie $\stEnvApp{\stEnv[j]}{\roleR}
  = \stEnvApp{\stEnv}{\roleR}$.
  For crashed roles $\roleR \in \rolesC$, we have
  $ \stStop =
    \stEnvApp{\stEnv}{\roleR} =
    \stEnvApp{\stEnv[j]}{\roleR} =
    \stStop
  $, as required.
  For non-crashed roles $\roleR \in \gtRoles{\gtG[j]}$, we have
  $\stEnvApp{\stEnv[j]}{\roleR}
    =
     \stEnvApp{\stEnv}{\roleR}
     \stSub
     \gtProj[\rolesR]{\gtG}{\roleR}
     =
     \stMerge{j \in J}(\gtProj[\rolesR]{\gtG[j]}{\roleR})
     \stSub
     \gtProj[\rolesR]{\gtG[j]}{\roleR}
  $ (applying \cref{lem:merge-subtyping}).

We know from $\stEnvAssoc{\gtWithCrashedRoles{\rolesC}{\gtG}}{\stEnv; \qEnv}{\rolesR}$ and
$\unfoldOne{\gtG} =
            \gtComm{\roleS}{\roleTMaybeCrashed}{j \in J}{\gtLab[j]}{\tyGroundi[j]}{\gtG[j]}$,  that
$\qEnv$ is associated with $\gtG[j]$. Hence, to construct $\qEnv[j]$, just let $\qEnv[j] = \qEnv$.
Notice that $\setenum{\roleP, \roleQ} \in \gtRoles{\gtG[j]}$, so they are still able to
 perform the communication action
$\stEnv[j]; \qEnv[j]
      \,\stEnvMoveOutAnnot{\roleP}{\roleQ}{\stChoice{\stLab[k]}{\tyGround[k]}}\,
      \stEnvi[j]; \qEnvi[j]$.

  We apply inductive hypothesis on $\stEnv[j]; \qEnv[j]$, and obtain
  $\gtWithCrashedRoles{\rolesC}{\gtG[j]}
          \gtMove[\stEnvOutAnnotSmall{\roleP}{\roleQ}{\stChoice{\stLab[k]}{\tyGround[k]}}]{\rolesR}
         \gtWithCrashedRoles{\rolesC}{\gtGi[j]} $    and
    $\stEnvAssoc{\gtWithCrashedRoles{\rolesC}{\gtGi[j]}}{\stEnvi[j]; \qEnvi[j]}{\rolesR}$.

 We apply $\inferrule{\iruleGtMoveCtx}$ (via
        \cref{lem:gt-lts-unfold}) to get:
        \begin{gather}
          \gtWithCrashedRoles{\rolesC}{\gtG}
          \gtMove[\stEnvOutAnnotSmall{\roleP}{\roleQ}{\stChoice{\stLab[k]}{\tyGround[k]}}]{\rolesR}
          \gtWithCrashedRoles{\rolesC}{\gtComm{\roleS}{\roleTMaybeCrashed}{j \in J}{\gtLab[j]}{\tyGroundi[j]}{\gtGi[j]}}
        \end{gather}
We are now left to show $\stEnvAssoc{\gtWithCrashedRoles{\rolesC}{\gtComm{\roleS}{\roleTMaybeCrashed}{j \in J}{\gtLab[j]}{\tyGroundi[j]}{\gtGi[j]}}}{\stEnvi; \qEnvi}{\rolesR}$.

Note that $\gtGi = \gtComm{\roleS}{\roleTMaybeCrashed}{j \in J}{\gtLab[j]}{\tyGroundi[j]}{\gtGi[j]}$ here. %

        \begin{enumerate}[label=(A\arabic*)]
        \item  For role $\roleS$, %
        we know that
        $ \unfoldOne{\stEnvApp{\stEnv}{\roleS}}
        =
        \stIntSum{\roleT}{j \in J_{\roleS}}{%
          \stChoice{\stLab[j]}{\tyGroundii[j]} \stSeq \stTii[j]
        }
      $, where $J_{\roleS} \subseteq \stIdxRemoveCrash{J}{\stLab[j]}$,
       $\forall j \in J_{\roleS}:
       \stTii[j] \stSub (\gtProj[\rolesR]{\gtG[j]}{\roleS})$, and
       $\tyGroundii[j] = \tyGroundi[j]$.

  Since $\roleS \notin
  \ltsSubject{\stEnvOutAnnotSmall{\roleP}{\roleQ}{\stChoice{\stLab[k]}{\tyGround[k]}}}$,
  we apply
  \cref{lem:stenv-red:trivial-2} on $\stEnv$ and $\stEnv[j]$ for all
  $j \in J_{\roleS}$.
  For all $j \in J_{\roleS}$,
  We have $
    \stTii[j]
    =
    \stEnvApp{\stEnv[j]}{\roleS}
    =
    \stEnvApp{\stEnvi[j]}{\roleS}
  $ (from \cref{lem:stenv-red:trivial-2}) and $
    \stEnvApp{\stEnvi[j]}{\roleS}
    \stSub
    \gtProj[\rolesR]{\gtGi[j]}{\roleS}
    $ (from inductive hypothesis).
  Therefore, we have $\stTii[j] \stSub \gtProj[\rolesR]{\gtGi[j]}{\roleS}$.
  We now apply \cref{lem:stenv-red:trivial-2} on $\stEnv$,
  which gives
  $ \unfoldOne{\stEnvApp{\stEnv}{\mpChanRole{\mpS}{\roleS}}}
    =
    \unfoldOne{\stEnvApp{\stEnvi}{\mpChanRole{\mpS}{\roleS}}}
    =
    \stIntSum{\roleT}{j \in J_{\roleS}}{%
      \stChoice{\stLab[j]}{\tyGroundii[j]} \stSeq \stTii[j]
    }
  $.
  We can now apply $\inferrule{\iruleStSubOut}$ to conclude $
 \stEnvApp{\stEnvi}{\roleS}
 \stSub
    \gtProj[\rolesR]{\gtGi}{\roleS}
  $, as required.

   For role $\roleT$ (where $\roleTMaybeCrashed = \roleT$), we know that
  $ \unfoldOne{\stEnvApp{\stEnv}{\roleT}}
    =
    \stExtSum{\roleS}{j \in J_{\roleT}}{%
      \stChoice{\stLab[j]}{\tyGroundii[j]} \stSeq \stUii[j]
    }
  $, where $J \subseteq J_{\roleT}$,
   $\forall j \in J:
     \stUii[j] \stSub (\gtProj[\rolesR]{\gtG[j]}{\roleT})$, and $
    \tyGroundi[j] = \tyGroundii[j]
  $.
  Since $\roleT \notin
  \ltsSubject{\stEnvOutAnnotSmall{\roleP}{\roleQ}{\stChoice{\stLab[k]}{\tyGround[k]}}}$,
  we apply
  \cref{lem:stenv-red:trivial-2} on $\stEnv$ and $\stEnv[j]$ for all
  $j \in J$.
  For all $j \in J$
  We have $
    \stUii[j]
    =
    \stEnvApp{\stEnv[j]}{\roleT}
    =
    \stEnvApp{\stEnvi[j]}{\roleT}
  $ (from \cref{lem:stenv-red:trivial-2}) and $
  \stEnvApp{\stEnvi[j]}{\roleT}
    \stSub
    \gtProj[\rolesR]{\gtGi[j]}{\roleT}
  $ (from inductive hypothesis).
  Therefore, we have $\stUii[j] \stSub \gtProj[\rolesR]{\gtGi[j]}{\roleT}$.
  We now apply \cref{lem:stenv-red:trivial-2} on $\stEnv$,
  which gives
  $ \unfoldOne{\stEnvApp{\stEnv}{\roleT}}
    =
    \unfoldOne{\stEnvApp{\stEnvi}{\roleT}}
    =
    \stExtSum{\roleS}{j \in J_{\roleT}}{%
      \stChoice{\stLab[j]}{\tyGroundii[j]} \stSeq \stUii[j]
    }
  $.
  We can now apply $\inferrule{\iruleStSubIn}$ to conclude $
    \stEnvApp{\stEnvi}{\roleT}
    \stSub
    \gtProj[\rolesR]{\gtGi}{\roleT}
$, as required.

 For other role $\roleR \in \gtRoles{\gtGi}$ (where
  $\roleR \notin \setenum{\roleS, \roleT}$),
  we need to show
  $\stEnvApp{\stEnvi}{\roleR}
  \stSub
   \gtProj[\rolesR]{\gtGi}{\roleR}
  $.
  We know what $
    \gtProj[\rolesR]{\gtGi}{\roleR} =
    \stMerge{j \in J}{\gtProj[\roles]{\gtGi[j]}{\roleR}}
  $.
  If $\roleR \notin \setenum{\roleP, \roleQ}$, we apply
  \cref{lem:stenv-red:trivial-2} on $\stEnv[j]$, obtaining
  $
    \stEnvApp{\stEnv[j]}{\roleR}
    =
    \stEnvApp{\stEnvi[j]}{\roleR}
  $. The inductive hypothesis gives
  $  \stEnvApp{\stEnvi[j]}{\roleR}
  \stSub
  \gtProj[\rolesR]{\gtGi[j]}{\roleR}$,
  we apply \cref{lem:merge-lower-bound} to obtain $
    \stEnvApp{\stEnvi[j]}{\roleR}
    \stSub
    \gtProj[\rolesR]{\gtGi}{\roleR}
    =
    \stMerge{j \in J}{\gtProj[\roles]{\gtGi[j]}{\roleR}}
  $.
  Note that $
    \stEnvApp{\stEnvi[j]}{\roleR} =
    \stEnvApp{\stEnv[j]}{\roleR}
  $ by \cref{lem:stenv-red:trivial-2},
  and $
    \stEnvApp{\stEnv[j]}{\roleR} =
    \stEnvApp{\stEnv}{\roleR}
  $ by construction.
  Therefore, we have $
  \stEnvApp{\stEnvi}{\roleR}
=
    \stEnvApp{\stEnv}{\roleR}
    \stSub
    \gtProj[\rolesR]{\gtGi}{\roleR}
      $, as required.

 We are left to consider the cases of $\roleP$ and $\roleQ$.
  We know what $
    \gtProj[\rolesR]{\gtGi}{\roleP} =
    \stMerge{j \in J}{\gtProj[\roles]{\gtGi[j]}{\roleP}}
  $ and $\setenum{\roleP, \roleQ} \subseteq \gtRoles{\gtGi}$.
  The inductive hypothesis gives
  $ \stEnvApp{\stEnvi[j]}{\roleP}
  \stSub
  \gtProj[\rolesR]{\gtGi[j]}{\roleP}
      $,
  we apply \cref{lem:merge-lower-bound} to obtain $
    \stEnvApp{\stEnvi[j]}{\roleP}
    \stSub
    \gtProj[\rolesR]{\gtGi}{\roleP}
    =
    \stMerge{j \in J}{\gtProj[\roles]{\gtGi[j]}{\roleP}}
  $.
  We now apply \cref{lem:stenv-red:det} on $\stEnv$ and all $\stEnv[j]$, which
  gives $\stEnvi[j] = \stEnvi$ for all $j$.
  Therefore, we have $
  \stEnvApp{\stEnvi}{\roleP}
  \stSub
    \gtProj[\rolesR]{\gtGi}{\roleP}
  $.
  Note that  $\stEnvApp{\stEnvi}{\roleP} =  \stEnvApp{\stEnvi[j]}{\roleP} = \stT[k] \;(\text{as in \eqref{eq:stenvp-send-type-cxt-cons}})$.
  The argument $\roleQ$ follows similarly.
\item For crashed roles $\roleR \in \rolesC$, we have
  $ \stEnvApp{\stEnvi}{\roleR}
    =
    \stEnvApp{\stEnv}{\roleR}
    =
    \stStop
  $ (applying \cref{lem:stenv-red:trivial-2}). Note that if $\roleTMaybeCrashed = \roleTCrashed$, $\roleT \in \rolesC$.
        \item Trivial by $\stEnvAssoc{\gtWithCrashedRoles{\rolesC}{\gtG}}{\stEnv; \qEnv}{\rolesR}$. %
        \item We know from $\qEnv[j] = \qEnv$ and $\stEnv[j]; \qEnv[j]
      \,\stEnvMoveOutAnnot{\roleP}{\roleQ}{\stChoice{\stLab[k]}{\tyGround[k]}}\,
      \stEnvi[j]; \qEnvi[j]$, that $\qEnvi[j] =
      \stEnvUpd{\qEnv[j]}{\roleP, \roleQ}{%
        \stQCons{%
          \stEnvApp{\qEnv[j]}{\roleP, \roleQ}
        }{%
          \stQMsg{\stLab[k]}{\tyGround[k]}
        }
      } =
      \stEnvUpd{\qEnv}{\roleP, \roleQ}{%
        \stQCons{%
          \stEnvApp{\qEnv}{\roleP, \roleQ}
        }{%
          \stQMsg{\stLab[k]}{\tyGround[k]}
        }
      } =
      \qEnvi$. Furthermore, by inductive hypothesis,
      $\qEnvi[j]$ is associated with $\gtWithCrashedRoles{\rolesC}{\gtGi[j]}$, which follows
      that $\forall j \in J: \qEnvi \text{ is associated with } \gtWithCrashedRoles{\rolesC}{\gtGi[j]}$. We are left to show that
      if $\roleTMaybeCrashed \neq \roleTCrashed$, then $\stEnvApp{\qEnvi}{\roleS, \roleT} = \stQEmpty$, which follows
      directly from $\stEnvAssoc{\gtWithCrashedRoles{\rolesC}{\gtG}}{\stEnv; \qEnv}{\rolesR}$ and
      $\stEnvApp{\qEnvi}{\roleS, \roleT} = \stEnvApp{\qEnv}{\roleS, \roleT}$.
        \end{enumerate}
\end{itemize}
\item In the case of   $\unfoldOne{\gtG} =
            \gtCommTransit{\roleSMaybeCrashed}{\roleT}{j \in
            J}{\gtLab[j]}{\tyGroundi[j]}{\gtG[j]}{k}$:

            Similar to that of the previous subcase that $\unfoldOne{\gtG} =
            \gtComm{\roleS}{\roleTMaybeCrashed}{j \in J}{\gtLab[j]}{\tyGroundi[j]}{\gtG[j]}$, applying \inferrule{\iruleGtMoveCtxi} instead.
\end{itemize}
    \item Case $\inferrule{\iruleTCtxIn}$:

      From the premise, we have:
    \begin{gather}
      \stEnvAssoc{\gtWithCrashedRoles{\rolesC}{\gtG}}{\stEnv; \qEnv}{\rolesR}
      \label{eq:stenvp-receive-assoc}
      \\
       \stEnvApp{\stEnv}{%
        \roleP%
      } =
      \stExtSum{\roleQ}{i \in I}{\stChoice{\stLab[i]}{\tyGround[i]} \stSeq \stT[i]}
      \label{eq:stenvp-receive-comp}
      \\
      \stEnvAnnotGenericSym = \stEnvInAnnotSmall{\roleP}{\roleQ}{\stChoice{\stLab[k]}{\tyGround[k]}}
      \\
      k \in I
      \\
      \stLabi = \stLab[k]
      \\
       \tyGroundi = \tyGround[k]
       \\
        \stEnvApp{\qEnv}{\roleQ, \roleP}
      =
      \stQCons{\stQMsg{\stLabi}{\tyGroundi}}{\stQi}
      \neq \stQUnavail
      \label{eq:stenvp-receive-queue-con}
      \\
      \stEnvi =
      \stEnvUpd{\stEnv}{\roleP}{\stT[k]}
      \label{eq:stenvp-receive-type-cxt-cons}
      \\
      \qEnvi =
     \stEnvUpd{\qEnv}{\roleQ, \roleP}{\stQi}
    \end{gather}

     Apply \cref{lem:inv-proj} \cref{item:proj-inv:recv} on
    \eqref{eq:stenvp-receive-comp}, we have two cases.
    \begin{itemize}[leftmargin=*]
      \item Case (1):
        \begin{gather}
          \unfoldOne{\gtG} =
            \gtCommTransit{\roleQMaybeCrashed}{\roleP}{i \in
            I'}{\gtLab[i]}{\tyGroundi[i]}{\gtG[i]}{j} \text{ or }
            \unfoldOne{\gtG} =
            \gtComm{\roleQ}{\rolePMaybeCrashed}{i \in I'}{\gtLab[i]}{\tyGroundi[i]}{\gtG[i]}
      \\
      I' \subseteq I
      \\
      \forall i \in I':
      \stLab[i] = \gtLab[i],
      \stT[i] \stSub (\gtProj[\rolesR]{\gtG[i]}{\roleP}),
      \tyGroundi[i] = \tyGround[i]
      \label{eq:stenvp-receive-global-label-equvi}
      \\
      \roleQ \notin \rolesR \text{ implies } \exists l \in I': \gtLab[l] =
          \gtCrashLab
    \end{gather}

First, we show that in this case, $\unfoldOne{\gtG}$ cannot be of the form $\gtComm{\roleQ}{\rolePMaybeCrashed}{i \in I'}{\gtLab[i]}{\tyGroundi[i]}{\gtG[i]}$.
There are two subcases to be considered: $\rolePMaybeCrashed =
    \roleP$ and $\rolePMaybeCrashed = \rolePCrashed$.
     \begin{itemize}[leftmargin=*]
\item In the case of $\rolePMaybeCrashed =
    \rolePCrashed$, we have $\roleP \in \rolesC$. Hence,
    by applying \cref{def:assoc-queue} on \eqref{eq:stenvp-receive-assoc}, we have that
     $\stEnvApp{\qEnv}{\cdot, \roleP} = \stQUnavail$, a desired contradiction to  \eqref{eq:stenvp-receive-queue-con}.

\item In the case of $\rolePMaybeCrashed =
    \roleP$, by association, it holds that
     $\stEnvApp{\qEnv}{\roleQ, \roleP} = \stQEmpty$, a desired contradiction to \eqref{eq:stenvp-receive-queue-con}.
  \end{itemize}
  Therefore, we only need to consider the case that $\unfoldOne{\gtG} =
            \gtCommTransit{\roleQMaybeCrashed}{\roleP}{i \in
            I'}{\gtLab[i]}{\tyGroundi[i]}{\gtG[i]}{j}$.

 Then we want to show that $\gtLab[j] \neq \gtCrashLab$, which is proved by contradiction.
  Assume that $\gtLab[j] = \gtCrashLab$, by association, we have that
  $\stEnvApp{\qEnv}{\roleQ, \roleP} =  \stQEmpty
        $, a desired contradiction to
        \eqref{eq:stenvp-receive-queue-con}.
 Moreover, we want to show that $j = k$.
 By association and $\gtLab[j] \neq \gtCrashLab$, we
 have   $\stEnvApp{\qEnv}{\roleQ, \roleP}
      =
      \stQCons{\stQMsg{\gtLab[j]}{\tyGround[j]}}{\stQ}
      =
      \stQCons{\stQMsg{\stLab[k]}{\tyGroundi}}{\stQi}
      $.
 Then by \eqref{eq:stenvp-receive-global-label-equvi}, it holds that $\gtLab[j]
 = \stLab[j] = \stLab[k]$.
  Furthermore, by  $j, k \in I$ and  the requirement that
   labels in local types must be pair-wise distinct, we have $j = k$, as required.
   Note that $k \in I'$ here.

  We can now apply \inferrule{\iruleGtMoveIn} (via \cref{lem:gt-lts-unfold}) to get:
 \begin{gather}
          \gtWithCrashedRoles{\rolesC}{\gtG}
          \gtMove[\stEnvAnnotGenericSym]{\rolesR}
          \gtWithCrashedRoles{\rolesC}{
            \gtG[k]
          }
        \end{gather}
        We are left to show $\stEnvAssoc{
          \gtWithCrashedRoles{\rolesC}{
            \gtG[k]
          }
        }{\stEnvi; \qEnvi}{\rolesR}$.

    \begin{enumerate}[label=(A\arabic*)]
            \item We need to show that $\forall \roleR \in \gtRoles{\gtG[k]}:
          \stEnvApp{\stEnvi}{\roleR} \stSub
        \gtProj[\rolesR]{\gtG[k]}{\roleR}$.
        We consider three subcases:
        \begin{itemize}[leftmargin=*]
       \item $\roleR = \roleQ$, which means that $\roleQMaybeCrashed = \roleQ$ and
               $\roleQ \in \gtRoles{\gtG[k]}$: by \eqref{eq:stenvp-receive-type-cxt-cons},
               $\stEnvApp{\stEnvi}{\roleQ} = \stEnvApp{\stEnv}{\roleQ}$.
               Then by association, we have $\stEnvApp{\stEnv}{\roleQ} \stSub
               \gtProj[\rolesR]{\gtG}{\roleQ} = \gtProj[\rolesR]{\gtG[k]}{\roleQ}$, as desired.

       \item $\roleR = \roleP$, which means that $\roleP \in \gtRoles{\gtG[k]}$: by \eqref{eq:stenvp-receive-type-cxt-cons}
        and \eqref{eq:stenvp-receive-global-label-equvi}, we have
          $\stEnvApp{\stEnvi}{\roleP} =  \stT[k] \stSub \gtProj[\rolesR]{\gtG[k]}{\roleP}$, as desired.

        \item $\roleR \neq \roleQ$ and $\roleR \neq \roleP$:
        $\gtProj[\rolesR]{\gtG}{\roleR} = \stMerge{i \in I'}{\gtProj[\rolesR]{\gtG[i]}{\roleR}}$.
        Furthermore, by \eqref{eq:stenvp-receive-type-cxt-cons} and $\stEnvAssoc{\gtWithCrashedRoles{\rolesC}{\gtG}}{\stEnv; \qEnv}{\rolesR}$,
        it holds that $\stEnvApp{\stEnvi}{\roleR} = \stEnvApp{\stEnv}{\roleR} \stSub \gtProj[\rolesR]{\gtG}{\roleR} = \stMerge{i \in I'}{\gtProj[\rolesR]{\gtG[i]}{\roleR}}$.
        Then, by \Cref{lem:merge-subtyping} and transitivity of subtyping,  we can conclude that $\stEnvApp{\stEnvi}{\roleR}
        \stSub \gtProj[\rolesR]{\gtG[k]}{\roleR}$, as desired.
        \end{itemize}
         \item Trivial by $\stEnvAssoc{\gtWithCrashedRoles{\rolesC}{\gtG}}{\stEnv; \qEnv}{\rolesR}$. Note that in the case of
         $\roleQMaybeCrashed = \roleQCrashed$, we have $\roleQ \in \rolesC$, and hence, $\stEnvApp{\stEnvi}{\roleQ}
         =
         \stEnvApp{\stEnv}{\roleQ} = \stStop$.
          \item Trivial by $\stEnvAssoc{\gtWithCrashedRoles{\rolesC}{\gtG}}{\stEnv; \qEnv}{\rolesR}$ and the
                   fact that if $\roleP \notin \gtRoles{\gtG[k]}$, then $\stEnvApp{\stEnvi}{\roleP} = \stEnd$:
                   with $\roleP \notin \gtRolesCrashed{\gtG[k]}$, by \Cref{lem:not-in-roles-end}, we have $\gtProj[\rolesR]{\gtG[k]}{\roleP} = \stEnd$.
                   Furthermore, by \eqref{eq:stenvp-receive-global-label-equvi}, it holds that $\stT[k] = \stEnd$, and thus,
          $\stEnvApp{\stEnvi}{\roleP} = \stEnd$, as desired.  The argument for the case that $\roleQMaybeCrashed
          = \roleQ$ and $\roleQ \notin \gtRoles{\gtG[k]}$ follows similarly.

          \item Since $\qEnv$ is associated with $\gtWithCrashedRoles{\rolesC}{\gtCommTransit{\roleQMaybeCrashed}{\roleP}{i \in
            I'}{\gtLab[i]}{\tyGroundi[i]}{\gtG[i]}{j} }$ and $\gtLab[j] \neq \gtCrashLab$, by \Cref{def:assoc-queue},  we have that
             $\stEnvApp{\qEnv}{\roleQ, \roleP} =
        \stQCons{\stQMsg{\gtLab[j]}{\tyGround[j]}}{\stQ}$ and
            $\forall i \in I':  \stEnvUpd{\qEnv}{\roleQ, \roleP}{\stQ} \text{~is associated with~}
      \gtWithCrashedRoles{\rolesC}{\gtG[i]}$, which follows that $\stQ = \stQi$ (in \eqref{eq:stenvp-receive-queue-con}) and
      $\qEnvi = \stEnvUpd{\qEnv}{\roleQ, \roleP}{\stQi}
      = \stEnvUpd{\qEnv}{\roleQ, \roleP}{\stQ}$ is associated with
       $\gtWithCrashedRoles{\rolesC}{\gtG[k]}$, as required.
        \end{enumerate}

    \item Case (2): similar to the case (2) in Case $\inferrule{\iruleTCtxOut}$.
    \end{itemize}

    \item Case $\inferrule{\iruleTCtxCrash}$:

    From the premise, we have:
     \begin{gather}
      \stEnvAssoc{\gtWithCrashedRoles{\rolesC}{\gtG}}{\stEnv; \qEnv}{\rolesR}
      \\
       \stEnvApp{\stEnv}{\roleP} \neq \stEnd
       \label{eq:stenvp-crash-not-end}
       \\
        \stEnvApp{\stEnv}{\roleP} \neq \stStop
         \label{eq:stenvp-crash-not-stop}
        \\
        \stEnvAnnotGenericSym = \ltsCrash{\mpS}{\roleP}
      \\
      \stEnvi =
      \stEnvUpd{\stEnv}{\roleP}{\stStop}
      \label{eq:stenvp-crash-type-cons}
      \\
      \qEnvi =
     \stEnvUpd{\qEnv}{\cdot, \roleP}{\stQUnavail}
      \label{eq:stenvp-crash-queue-cons}%
    \end{gather}

By \eqref{eq:stenvp-crash-not-stop}, we know that
$\roleP \notin \rolesC$ and  $\roleP \in \gtRoles{\gtG}$.
The premise requires that $ \stEnvAnnotGenericSym \neq \ltsCrash{\mpS}{\roleP}$ for all
$\roleP \in \rolesR$, therefore, $\roleP \notin \rolesR$.
We apply \inferrule{\iruleGtMoveCrash} (via \cref{lem:gt-lts-unfold}) to get:
   \begin{gather}
          \gtWithCrashedRoles{\rolesC}{\gtG}
          \gtMove[\stEnvAnnotGenericSym]{\rolesR}
         \gtWithCrashedRoles{\rolesC \cup \setenum{\roleP}}{\gtCrashRole{\gtG}{\roleP}}
        \end{gather}

We are now left to show $\stEnvAssoc{
          \gtWithCrashedRoles{\rolesC \cup \setenum{\roleP}}{
            \gtCrashRole{\gtG}{\roleP}
          }
        }{\stEnvi; \qEnvi}{\rolesR}$.

Note that $\gtGi = \gtCrashRole{\gtG}{\roleP}$ and $\rolesCi = \rolesC \cup \setenum{\roleP}$ here.
By \eqref{eq:stenvp-crash-type-cons}, we can set  $\stEnvi =  \stEnvi[\gtGi] \stEnvComp \stEnvi[\stStopSym] \stEnvComp
  \stEnvi[\stEnd]$, where $\dom{\stEnvi[\gtGi]} = \setcomp{\roleQ}{\roleQ \in \gtRoles{\gtCrashRole{\gtG}{\roleP}}} =
  \setcomp{\roleQ}{\roleQ \in \gtRoles{\gtG}} \setminus \setenum{\roleP} = \dom{\stEnv[\gtG]} \setminus \setenum{\roleP}$,
  $\dom{\stEnvi[\stStopSym]} = \rolesC \cup \setenum{\roleP}$ =
  $\dom{\stEnv[\stStopSym]} \cup \setenum{\roleP}$,
  and $\dom{\stEnvi[\stEnd]} = \dom{\stEnv[\stEnd]}$. Meanwhile, we have $\stEnvApp{\stEnvi}{\roleQ} = \stEnvApp{\stEnv}{\roleQ}$ if
  $\roleQ \neq \roleP$.

        \begin{enumerate}[label=(A\arabic*)]
          \item We want to show that for any $\roleQ \in \gtRoles{\gtCrashRole{\gtG}{\roleP}}$, we have $\stEnvApp{\stEnvi}{\roleQ} =
          \stEnvApp{\stEnv}{\roleQ} \stSub \gtProj[\rolesR]{\gtG}{\roleQ} \stSub  \gtProj[\rolesR]{(\gtCrashRole{\gtG}{\roleP})}{\roleQ}$ by
          \cref{lem:proj-non-crashing-role-preserve} and $\stEnvAssoc{\gtWithCrashedRoles{\rolesC}{\gtG}}{\stEnv; \qEnv}{\rolesR}$.
          \item Trivial by $\stEnvAssoc{\gtWithCrashedRoles{\rolesC}{\gtG}}{\stEnv; \qEnv}{\rolesR}$ and \eqref{eq:stenvp-crash-type-cons}, i.e.,
          $\stEnvApp{\stEnvi}{\roleP} = \stStop$.
          \item Trivial by $\stEnvAssoc{\gtWithCrashedRoles{\rolesC}{\gtG}}{\stEnv; \qEnv}{\rolesR}$.
          \item Since $\qEnv$ is associated with $\gtWithCrashedRoles{\rolesC}{\gtG}$,
          by \Cref{def:assoc-queue}, we have that for any $\roleQ \in \rolesC$,
        $\stEnvApp{\qEnv}{\cdot, \roleQ} = \stQUnavail$. Then, with \eqref{eq:stenvp-crash-queue-cons},
        we have that for any $\roleQ \in \rolesC \cup \setenum{\roleP}$,
        $\stEnvApp{\qEnvi}{\cdot, \roleQ} =
           \stEnvApp{\stEnvUpd{\qEnv}{\cdot, \roleP}{\stQUnavail}}{\cdot, \roleQ} =
           \stQUnavail$, which follows directly that $\qEnvi$ is associated with
           $\gtWithCrashedRoles{\rolesC \cup \setenum{\roleP}}{\gtCrashRole{\gtG}{\roleP}}$.
         \end{enumerate}

    \item Case $\inferrule{\iruleTCtxCrashDetect}$:

    From the premise, we have:
    \begin{gather}
     \stEnvAssoc{\gtWithCrashedRoles{\rolesC}{\gtG}}{\stEnv; \qEnv}{\rolesR}
     \\
     \stEnvApp{\stEnv}{\roleQ} =
      \stExtSum{\roleP}{i \in I}{\stChoice{\stLab[i]}{\tyGround[i]} \stSeq \stT[i]}
      \label{eq:stenvp-crash-detect-comp}
      \\
       \stEnvApp{\stEnv}{\roleP} = \stStop
       \\
        \stEnvAnnotGenericSym = \ltsCrDe{\mpS}{\roleQ}{\roleP}
        \\
       k \in I
       \\
        \stLab[k] = \stCrashLab
      \\
      \stEnvApp{\qEnv}{\roleP, \roleQ} = \stQEmpty
      \label{eq:stenvp-crash-detect-queue-empty}
      \\
      \stEnvi = \stEnvUpd{\stEnv}{\roleQ}{\stT[k]}
      \label{eq:stenvp-crash-detect-type-cxt-cons}
      \\
      \qEnvi = \qEnv
    \end{gather}
  Since $\stEnvApp{\stEnv}{\roleP} = \stStop$, we have $\roleP \in \rolesC$ and
  $\roleP \notin \rolesR$.
  Apply \cref{lem:inv-proj} \cref{item:proj-inv:recv} on
    \eqref{eq:stenvp-crash-detect-comp}, we have two cases.
    \begin{itemize}[leftmargin=*]
      \item Case (1):
        \begin{gather}
          \unfoldOne{\gtG} =
            \gtCommTransit{\rolePCrashed}{\roleQ}{i \in
            I'}{\gtLab[i]}{\tyGroundi[i]}{\gtG[i]}{j} %
      \\
      I' \subseteq I
      \\
      \forall i \in I':
      \stLab[i] = \gtLab[i],
      \stT[i] \stSub (\gtProj[\rolesR]{\gtG[i]}{\roleQ}),
      \tyGroundi[i] = \tyGround[i]
      \label{eq:stenvp-crash-detect-global-label-equvi}
      \\
      \roleP \notin \rolesR \text{ implies } \exists l \in I': \gtLab[l] =
          \gtCrashLab
    \end{gather}
  Then we want to show that $\gtLab[j] = \gtCrashLab$, which is proved by contradiction.
  Assume that $\gtLab[j] \neq \gtCrashLab$, by association, we have that
  $\stEnvApp{\qEnv}{\roleP, \roleQ} =
        \stQCons{\stQMsg{\gtLab[j]}{\tyGround[j]}}{\stQ}$, a desired contradiction to
        \eqref{eq:stenvp-crash-detect-queue-empty}.
   Moreover, we want to show that $j = k$. By \eqref{eq:stenvp-crash-detect-global-label-equvi},
   we have $\gtLab[j] = \stLab[j]$, which means that $\stLab[j] = \stCrashLab$.
   Since $\stLab[j] = \stLab[k] = \stCrashLab$ and $j, k \in I$,  by the requirement that
   labels in local types must be pair-wise distinct, we have $j = k$, as required.
   Note that $k \in I'$ here.

  We can now apply \inferrule{\iruleGtMoveCrDe} (via \cref{lem:gt-lts-unfold}) to get:
 \begin{gather}
          \gtWithCrashedRoles{\rolesC}{\gtG}
          \gtMove[\stEnvAnnotGenericSym]{\rolesR}
          \gtWithCrashedRoles{\rolesC}{
            \gtG[k]
          }
        \end{gather}
        We are left to show $\stEnvAssoc{
          \gtWithCrashedRoles{\rolesC}{
            \gtG[k]
          }
        }{\stEnvi; \qEnvi}{\rolesR}$.
 \begin{enumerate}[label=(A\arabic*)]
            \item We need to show that $\forall \roleR \in \gtRoles{\gtG[k]}:
          \stEnvApp{\stEnvi}{\roleR} \stSub
        \gtProj[\rolesR]{\gtG[k]}{\roleR}$.
        We consider two subcases:
        \begin{itemize}[leftmargin=*]
        \item $\roleR = \roleQ$, which means that $\roleQ \in \gtRoles{\gtG[k]}$: by \eqref{eq:stenvp-crash-detect-type-cxt-cons}
        and \eqref{eq:stenvp-crash-detect-global-label-equvi}, we have
          $\stEnvApp{\stEnvi}{\roleQ} =  \stT[k] \stSub \gtProj[\rolesR]{\gtG[k]}{\roleQ}$, as desired.
        \item $\roleR \neq \roleQ$: with $\roleP \notin \gtRoles{\gtG[k]}$, we have $\roleR \neq \roleP$, and hence,
        $\gtProj[\rolesR]{\gtG}{\roleR} = \stMerge{i \in I'}{\gtProj[\rolesR]{\gtG[i]}{\roleR}}$.
        Furthermore, by \eqref{eq:stenvp-crash-detect-type-cxt-cons} and $\stEnvAssoc{\gtWithCrashedRoles{\rolesC}{\gtG}}{\stEnv; \qEnv}{\rolesR}$,
        it holds that $\stEnvApp{\stEnvi}{\roleR} = \stEnvApp{\stEnv}{\roleR} \stSub \gtProj[\rolesR]{\gtG}{\roleR} = \stMerge{i \in I'}{\gtProj[\rolesR]{\gtG[i]}{\roleR}}$.
        Then, by \Cref{lem:merge-subtyping} and transitivity of subtyping,  we can conclude that $\stEnvApp{\stEnvi}{\roleR}
        \stSub \gtProj[\rolesR]{\gtG[k]}{\roleR}$, as desired.
        \end{itemize}
         \item Trivial by $\stEnvAssoc{\gtWithCrashedRoles{\rolesC}{\gtG}}{\stEnv; \qEnv}{\rolesR}$.
          \item Trivial by $\stEnvAssoc{\gtWithCrashedRoles{\rolesC}{\gtG}}{\stEnv; \qEnv}{\rolesR}$ and the
                   fact that if $\roleQ \notin \gtRoles{\gtG[k]}$, then $\stEnvApp{\stEnvi}{\roleQ} = \stEnd$:
                   with $\roleQ \notin \gtRolesCrashed{\gtG[k]}$, by \Cref{lem:not-in-roles-end}, we have $\gtProj[\rolesR]{\gtG[k]}{\roleQ} = \stEnd$.
                   Furthermore, by \eqref{eq:stenvp-crash-detect-global-label-equvi}, it holds that $\stT[k] = \stEnd$, and thus,
          $\stEnvApp{\stEnvi}{\roleQ} = \stEnd$, as desired.
          \item Since $\qEnv$ is associated with $\gtWithCrashedRoles{\rolesC}{\gtCommTransit{\rolePCrashed}{\roleQ}{i \in
            I'}{\gtLab[i]}{\tyGroundi[i]}{\gtG[i]}{j} }$ and $\gtLab[j] = \gtCrashLab$, by \Cref{def:assoc-queue},  we have that
            $\forall i \in I': \qEnv \text{~is associated with~}
      \gtWithCrashedRoles{\rolesC}{\gtG[i]}$, which follows that $\qEnvi = \qEnv$ is associated with
       $\gtWithCrashedRoles{\rolesC}{\gtG[k]}$, as required.
        \end{enumerate}

       \item Case (2): similar to the case (2) in Case $\inferrule{\iruleTCtxOut}$.
    \end{itemize}
     \item Case $\inferrule{\iruleTCtxRec}$: by induction and \cref{prop:gt-lts-unfold}.
    \qedhere
  \end{itemize}
\end{proof}

\subsection{Properties by Projection}
\label{sec:proof:propbyproj}
\subparagraph*{Safety by Projection}
\label{sec:proof:safety}

\begin{restatable}{lemma}{lemProjSafe}
\label{lem:safety-by-proj}
  If $\stEnvAssoc{\gtWithCrashedRoles{\rolesC}{\gtG}}{\stEnv; \qEnv}{\rolesR}$,
  then $\stEnv; \qEnv$ is  $\rolesR$-safe.%
  \label{lem:ext-proj-safe}
\end{restatable}
\begin{proof}
Let $\predP = \setcomp{\stEnvi; \qEnvi}{\exists \rolesCi, \gtGi: \gtWithCrashedRoles{\rolesC}{\gtG}
\gtMoveStar[\rolesR]{}
\gtWithCrashedRoles{\rolesCi}{\gtGi}
\text{ and }
\stEnvAssoc{\gtWithCrashedRoles{\rolesCi}{\gtGi}}{\stEnvi; \qEnvi}{\rolesR}
}$.
Take any $\stEnvi; \qEnvi \in \predP$, we show that $\stEnvi; \qEnvi$ satisfies all
clauses in \Cref{def:mpst-env-safe}, which means that $\predP$ is an $\rolesR$-safety 
property. Then we can conclude that, since $\predPApp{\stEnv; \qEnv}$ holds, 
$\stEnv; \qEnv$ is $\rolesR$-safe. 

By definition of $\predP$, there exists $\gtWithCrashedRoles{\rolesCi}{\gtGi}$ with
$\gtWithCrashedRoles{\rolesC}{\gtG}
\gtMoveStar[\rolesR]{}
\gtWithCrashedRoles{\rolesCi}{\gtGi}$ and 
$\stEnvAssoc{\gtWithCrashedRoles{\rolesCi}{\gtGi}}{\stEnvi; \qEnvi}{\rolesR}$. 
\begin{itemize}%
\item[\inferrule{\iruleSafeComm}]
From the premise, we have 
\begin{gather}
    \stEnvAssoc{\gtWithCrashedRoles{\rolesCi}{\gtGi}}{\stEnvi; \qEnvi}{\rolesR}
    \label{eq:safety-assoc-1}
    \\
    \stEnvApp{\stEnvi}{%
        \roleQ%
      } =
      \stExtSum{\roleP}{i \in I}{\stChoice{\stLab[i]}{\tyGround[i]} \stSeq \stT[i]}
      \label{eq:stenvp-safety-rec-comp}
    \\
    \stEnvApp{\qEnvi}{\roleP, \roleQ} \neq \stQUnavail
    \label{eq:safety-queue-available}
    \\
    \stEnvApp{\qEnvi}{\roleP, \roleQ} \neq \stQEmpty
    \label{eq:safety-queue-not-empty}
\end{gather}
Apply  \cref{item:proj-inv:recv}  of \cref{lem:inv-proj} on \eqref{eq:stenvp-safety-rec-comp}, 
we have two cases.
\begin{itemize}[leftmargin=*]
\item Case (1):
\begin{gather}
\unfoldOne{\gtGi} = \gtCommTransit{\rolePMaybeCrashed}{\roleQ}{i \in
            I'}{\gtLab[i]}{\tyGroundi[i]}{\gtG[i]}{j} \text{ or } 
           \unfoldOne{\gtGi} =  \gtComm{\roleP}{\roleQMaybeCrashed}{i \in I'}{\gtLab[i]}{\tyGroundi[i]}{\gtG[i]}
           \\
            I' \subseteq I
 \label{eq:safety-subset}
      \\
      \forall i \in I':
      \stLab[i] = \gtLab[i],
      \stT[i] \stSub (\gtProj[\rolesR]{\gtG[i]}{\roleQ}),
      \tyGroundi[i] \stSub \tyGround[i]
      \label{eq:stenvp-safety-receive-label-equvi}
\end{gather}
First, we show that $\unfoldOne{\gtGi}$ cannot be the form of 
$\gtComm{\roleP}{\roleQMaybeCrashed}{i \in I'}{\gtLab[i]}{\tyGroundi[i]}{\gtG[i]}$. We prove this by contradiction on two subcases:
$\roleQMaybeCrashed = \roleQ$ and $\roleQMaybeCrashed = \roleQCrashed$:
\begin{itemize}[leftmargin=*]
\item $\roleQMaybeCrashed = \roleQ$: by association \eqref{eq:safety-assoc-1}, we have 
$\stEnvApp{\qEnvi}{\roleP, \roleQ} = \stQEmpty$, a desired contradiction to \eqref{eq:safety-queue-not-empty}. 
\item $\roleQMaybeCrashed = \roleQCrashed$: we have $\roleQ \in \rolesC$. Hence, by association \eqref{eq:safety-assoc-1}, we have 
$\stEnvApp{\qEnvi}{\cdot, \roleQ} = \stQUnavail$, a desired contradiction to \eqref{eq:safety-queue-available}. 
\end{itemize}

It follows that $\unfoldOne{\gtGi} =  \gtCommTransit{\rolePMaybeCrashed}{\roleQ}{i \in
            I'}{\gtLab[i]}{\tyGroundi[i]}{\gtG[i]}{j}$. We are left to show that $j \in I$, $\tyGroundi[j] \stSub \tyGround[j]$, and 
$ \stEnvApp{\qEnvi}{\roleP, \roleQ}  =  \stQCons{\stQMsg{\stLab[j]}{\tyGroundi[j]}}{\stQ}$,  and then by applying 
\inferrule{\iruleTCtxIn}, we can conclude with $\stEnvMoveAnnotP{\stEnvi; \qEnvi}{\stEnvInAnnot{\roleQ}{\roleP}{\stChoice{\stLab[j]}{\tyGround[j]}}}$. 
We consider two subcases: $\gtLab[j] = \gtCrashLab$ and $\gtLab[j] \neq \gtCrashLab$:
\begin{itemize}[leftmargin=*]
\item $\gtLab[j] = \gtCrashLab$: by association, 
$\stEnvApp{\qEnvi}{\roleP, \roleQ} = \stQEmpty$, a desired contradiction to \eqref{eq:safety-queue-not-empty}. 
\item $\gtLab[j] \neq \gtCrashLab$: by association, $\stEnvApp{\qEnvi}{\roleP, \roleQ} =
        \stQCons{\stQMsg{\gtLab[j]}{\tyGroundi[j]}}{\stQ}$. Moreover, with \eqref{eq:safety-subset} and \eqref{eq:stenvp-safety-receive-label-equvi}, 
        we obtain that $j \in I$, $\tyGroundi[j] \stSub \tyGround[j]$, and $\stLab[j] = \gtLab[j]$, which follows that 
$ \stEnvApp{\qEnvi}{\roleP, \roleQ}  =  \stQCons{\stQMsg{\stLab[j]}{\tyGroundi[j]}}{\stQ}$, as desired. 
\end{itemize}

\item Case (2):  $\unfoldOne{\gtGi} =
            \gtComm{\roleS}{\roleTMaybeCrashed}{j \in J}{\gtLab[j]}{\tyGroundi[j]}{\gtG[j]}$,
          or
          $\unfoldOne{\gtGi} =
            \gtCommTransit{\roleSMaybeCrashed}{\roleT}{j \in
            J}{\gtLab[j]}{\tyGroundi[j]}{\gtG[j]}{k}$,
          where for all $j \in J$:
          $  \stEnvApp{\stEnvi}{%
        \roleQ%
      }  \stSub (\gtProj[\rolesR]{\gtG[j]}{\roleQ})$,
          with $\roleQ \neq \roleS$ and $\roleQ \neq \roleT$.
          
          We apply \cref{lem:eventual-global-form} to get that there exists a global type 
          $\gtWithCrashedRoles{\rolesCii}{\gtGii}$ and a queue environment $\qEnvii$ such that 
 $  \stEnvApp{\stEnvi}{%
        \roleQ%
      } \stSub 
 \gtProj[\rolesR]{\gtGii}{\roleQ}$, $\unfoldOne{\gtGii}$ is of the form 
 $ \gtCommTransit{\rolePMaybeCrashed}{\roleQ}{i \in
            I'}{\gtLab[i]}{\tyGroundi[i]}{\gtG[i]}{j}$ or
 $ \gtComm{\roleP}{\roleQMaybeCrashed}{i \in I'}{\gtLab[i]}{\tyGroundi[i]}{\gtG[i]}$, and $\qEnvii$ is 
 associated with $\gtWithCrashedRoles{\rolesCii}{\gtGii}$ with $\stEnvApp{\qEnvii}{\roleP, \roleQ} = \stEnvApp{\qEnvi}{\roleP, \roleQ}$.  
 The following proof argument is similar to that for Case (1). 
\end{itemize}

\item[\inferrule{\iruleSafeCrash}] From the premise, we have 
$ \stEnvApp{\stEnvi}{\roleQ} =
      \stExtSum{\roleP}{i \in I}{\stChoice{\stLab[i]}{\tyGround[i]} \stSeq \stT[i]}, 
       \stEnvApp{\stEnvi}{\roleP} = \stStop$, and $\stEnvApp{\qEnvi}{\roleP, \roleQ} = \stQEmpty$. 
       We just need to show that there exists $k \in I$ with $\stLab[k] = \stCrashLab$, and then by applying 
       \inferrule{\iruleTCtxCrashDetect}, we can conclude with 
       $\stEnvMoveAnnotP{\stEnvi; \qEnvi}{\ltsCrDe{\mpS}{\roleQ}{\roleP}}$. 
       By the association that $\stEnvAssoc{\gtWithCrashedRoles{\rolesCi}{\gtGi}}{\stEnvi; \qEnvi}{\rolesR}$, 
       we have $ \stEnvApp{\stEnvi}{\roleQ} =
      \stExtSum{\roleP}{i \in I}{\stChoice{\stLab[i]}{\tyGround[i]} \stSeq \stT[i]} 
      \stSub 
      \gtProj[\rolesR]{\gtGi}{\roleQ}$. Moreover, with $\roleP \notin \rolesR$, which follows directly from 
      $\stEnvApp{\stEnvi}{\roleP} = \stStop$, we can apply \Cref{lem:crash-lab-exists} to 
      get $\exists k \in I : \stLab[k] = \stCrashLab$, as desired. 

\item[\inferrule{\iruleSafeRec}] Let $\stEnvii; \qEnvii$ be constructed from $\stEnvi; \qEnvi$ with
$\stEnvii = \stEnvUpd{\stEnvi}{%
        \roleP%
      }{%
        \stS\subst{\stRecVar}{\stRec{\stRecVar}{\stS}}%
      }$ and
$\qEnvii = \qEnvi$. 
We want to show that $\stEnvAssoc{\gtWithCrashedRoles{\rolesCi}{\gtGi}}{\stEnvii; \qEnvii}{\rolesR}$. 
From the premise, we have $\stEnvApp{\stEnvi}{{\roleP}}  =  \stRec{\stRecVar}{\stS}$.
Then by association that $\stEnvAssoc{\gtWithCrashedRoles{\rolesCi}{\gtGi}}{\stEnvi; \qEnvi}{\rolesR}$, 
it holds that $\stEnvApp{\stEnvi}{{\roleP}}   = \stRec{\stRecVar}{\stS} 
        \stSub
        \gtProj[\rolesR]{\gtGi}{\roleP}$. Hence, by inverting \inferrule{\iruleStSubRecL}, we have
        $\stEnvApp{\stEnvii}{{\roleP}} = \stS{}\subst{\stRecVar}{\stRec{\stRecVar}{\stS}}    
         \stSub
        \gtProj[\rolesR]{\gtGi}{\roleP}$. Therefore, by \cref{def:assoc}, we conclude with 
        $\stEnvAssoc{\gtWithCrashedRoles{\rolesCi}{\gtGi}}{\stEnvii; \qEnvii }{\rolesR}
    $, as desired.
\item[ \inferrule{\iruleSafeMove}] 
From the premise, we have $\stEnvi; \qEnvi \stEnvMoveMaybeCrash[\rolesR] \stEnvii; \qEnvii$. 
Hence, by \Cref{thm:gtype:proj-comp}, there exists $\gtWithCrashedRoles{\rolesCii}{\gtGii}$ with
 $\gtWithCrashedRoles{\rolesCi}{\gtGi} \gtMove[\stEnvAnnotGenericSym]{\rolesR}
    \gtWithCrashedRoles{\rolesCii}{\gtGii}$, where $\stEnvAnnotGenericSym \neq \ltsCrash{\mpS}{\roleP}$ for all $\roleP
  \in \rolesR$, and  $\stEnvAssoc{\gtWithCrashedRoles{\rolesCii}{\gtGii}}{\stEnvii; \qEnvii}{\rolesR}$.
By definition of $\predP$, the configuration after transition $\stEnvii; \qEnvii$ is in $\predP$, as desired. 
\qedhere
\end{itemize}
 \end{proof}

\subparagraph*{Deadlock Freedom by Projection}
\label{sec:proof:deadlockfree}
\begin{restatable}{lemma}{lemProjDeadlockFree}
\label{lem:proj:df}
  If $\stEnvAssoc{\gtWithCrashedRoles{\rolesC}{\gtG}}{\stEnv; \qEnv}{\rolesR}$,
  then $\stEnv; \qEnv$ is $\rolesR$-deadlock-free.%
  \label{lem:ext-proj-deadlock-free}
\end{restatable}
\begin{proof}
Since $\stEnvAssoc{\gtWithCrashedRoles{\rolesC}{\gtG}}{\stEnv; \qEnv}{\rolesR}$,
by \Cref{lem:ext-proj-safe}, we have that $\stEnv; \qEnv$ is $\rolesR$-safe.
By operational correspondence of global type $\gtWithCrashedRoles{\rolesC}{\gtG}$
and configuration $\stEnv; \qEnv$ (\Cref{thm:gtype:proj-sound,thm:gtype:proj-comp}), there exists
a global type $\gtWithCrashedRoles{\rolesCi}{\gtGi}$ such that
$\gtWithCrashedRoles{\rolesC}{\gtG} \!\gtMoveStar[\rolesR]\!
\gtWithCrashedRoles{\rolesCi}{\gtGi} \!\not\gtMove[]{\rolesR}$,
with associated configurations
 $\stEnv; \qEnv \!\stEnvMoveMaybeCrashStar[\rolesR]\! \stEnvi; \qEnvi
\!\not\stEnvMoveMaybeCrash[\rolesR]$. Since no further reductions are possible for the global type,
the global type $\gtWithCrashedRoles{\rolesCi}{\gtGi}$ must be of the form $\gtWithCrashedRoles{\rolesCi}{\stEnd}$.
By applying \Cref{def:assoc} on $\stEnvAssoc{\gtWithCrashedRoles{\rolesCi}{\stEnd}}{\stEnvi; \qEnvi}{\rolesR}$,
we have
$\stEnvi =  \stEnvi[\stEnd] \stEnvComp \stEnvi[\stStopSym]$, where
$\dom{\stEnvi[\stEnd]} = \setcomp{\roleP}{\stEnvApp{\stEnvi}{\roleP} = \stEnd}$ and
$\dom{\stEnvi[\stStopSym]} = \rolesCi = \setcomp{\roleP}{\stEnvApp{\stEnvi}{\roleP} = \stStop}$, as required.
By association again, we have that,
for any $\roleP, \roleQ$, if $\roleQ \in \rolesCi (= \dom{\stEnvi[\stStopSym]}),
 \stEnvApp{\qEnvi}{\cdot, \roleQ} = \stQUnavail$, and otherwise, $\stEnvApp{\qEnvi}{\roleP, \roleQ} = \stQEmpty$, as required.
\end{proof}

\subparagraph*{Liveness by Projection}
\label{sec:proof:liveness}
\begin{restatable}{lemma}{lemProjLive}
 \label{lem:ext-proj-live}
  If $\stEnvAssoc{\gtWithCrashedRoles{\rolesC}{\gtG}}{\stEnv; \qEnv}{\rolesR}$,
  then $\stEnv; \qEnv$ is $\rolesR$-live.%
\end{restatable}
\begin{proof}
Since $\stEnvAssoc{\gtWithCrashedRoles{\rolesC}{\gtG}}{\stEnv; \qEnv}{\rolesR}$,
by \Cref{lem:ext-proj-safe}, we have that $\stEnv; \qEnv$ is $\rolesR$-safe.

We are left to
show that if  $\stEnv; \qEnv \stEnvMoveMaybeCrashStar[\rolesR]
  \stEnvi; \qEnvi$, then any non-crashing path starting with $\stEnvi; \qEnvi$ which is fair is also live.
  By operational correspondence of global type $\gtWithCrashedRoles{\rolesC}{\gtG}$
and configuration $\stEnv; \qEnv$ (\Cref{thm:gtype:proj-sound,thm:gtype:proj-comp}), there exists
a global type $\gtWithCrashedRoles{\rolesCi}{\gtGi}$ such that
$\gtWithCrashedRoles{\rolesC}{\gtG} \!\gtMoveStar[\rolesR]\!
\gtWithCrashedRoles{\rolesCi}{\gtGi}$ and
$\stEnvAssoc{\gtWithCrashedRoles{\rolesCi}{\gtGi}}{\stEnvi; \qEnvi}{\rolesR}$.
We construct a non-crashing fair path $(\stEnvi[n]; \qEnvi[n])_{n \in N}$,
where $N = \setenum{0, 1, 2, \ldots}$, $\stEnvi[0] = \stEnvi$,
$\qEnvi[0] = \qEnvi$, and $\forall n \in N$, $\stEnvi[n]; \qEnvi[n] \!\stEnvMove\!
\stEnvi[n+1]; \qEnvi[n+1]$. Then we just need to show that $(\stEnvi[n]; \qEnvi[n])_{n \in N}$
is live.
  \begin{enumerate}[label=(L\arabic*)]
  \item Suppose that $\stEnvApp{\qEnvi[n]}{\roleP, \roleQ}
      =
      \stQCons{\stQMsg{\stLab}{\tyGround}}{\stQ}
      \neq
      \stQUnavail$ and $\stLab \neq \stCrashLab$. By operational correspondence
      of $\gtWithCrashedRoles{\rolesCi}{\gtGi}$ and configuration $\stEnvi[0]; \qEnvi[0]$,  and
      $\forall n \in N$,  $\stEnvi[n]; \qEnvi[n] \!\stEnvMove\!
\stEnvi[n+1]; \qEnvi[n+1]$, there
      exists $\gtWithCrashedRoles{\rolesCi[n]}{\gtGi[n]}$ such that
      $\stEnvAssoc{\gtWithCrashedRoles{\rolesCi[n]}{\gtGi[n]}}{\stEnvi[n]; \qEnvi[n]}{\rolesR}$.
      By \Cref{lem:eventual-global-form} and \Cref{def:assoc}, we only need to consider
      the case that
       $\gtGi[n] =
      \gtCommTransit{\rolePMaybeCrashed}{\roleQ}{i \in I}{\gtLab[i]}{\tyGround[i]}{\gtGii[i]}{j}$ with
      $\stLab = \gtLab[j] \neq \gtCrashLab$ and $\tyGround[j] = \tyGround$.
      By applying \inferrule{\iruleGtMoveIn},
      $\gtWithCrashedRoles{\rolesCi[n]}{\gtGi[n]}
      \gtMove[
      \stEnvInAnnotSmall{\roleQ}{\roleP}{\stChoice{\gtLab[j]}{\tyGround[j]}}
    ]{
      \rolesR
    }
     \gtWithCrashedRoles{\rolesCi[n]}{\gtGii[j]}$. Hence, by the soundness of association, it holds
     that $\stEnvi[n]; \qEnvi[n] \stEnvMoveInAnnot{\roleQ}{\roleP}{\stChoice{\stLab}{\tyGround}}$.
     Finally, combining the fact that $(\stEnvi[n]; \qEnvi[n])_{n \in N}$  is fair with $\stEnvi[n]; \qEnvi[n] \stEnvMoveInAnnot{\roleQ}{\roleP}{\stChoice{\stLab}{\tyGround}}$, we can conclude that there exists $k$ such that $n \leq k \in N$ and
     $\stEnvi[k]; \qEnvi[k] \stEnvMoveInAnnot{\roleQ}{\roleP}{\stChoice{\stLab}{\tyGround}} \stEnvi[k+1]; \qEnvi[k+1]$, as desired.

 \item Suppose that  $\stEnvApp{\stEnvi[n]}{%
        \roleP%
      } =
      \stExtSum{\roleQ}{i \in I}{\stChoice{\stLab[i]}{\tyGround[i]} \stSeq \stT[i]}$. By operational correspondence
      of $\gtWithCrashedRoles{\rolesCi}{\gtGi}$ and configuration $\stEnvi[0]; \qEnvi[0]$,  and
      $\forall n \in N$,  $\stEnvi[n]; \qEnvi[n] \!\stEnvMove\!
\stEnvi[n+1]; \qEnvi[n+1]$, there
      exists $\gtWithCrashedRoles{\rolesCi[n]}{\gtGi[n]}$ such that
      $\stEnvAssoc{\gtWithCrashedRoles{\rolesCi[n]}{\gtGi[n]}}{\stEnvi[n]; \qEnvi[n]}{\rolesR}$.
      By \Cref{def:assoc}, \Cref{lem:inv-proj}, and \Cref{lem:eventual-global-form},
      we only have to consider that  $\gtGi[n]$ is of the form
     $ \gtComm{\roleQ}{\rolePMaybeCrashed}{i \in I'}{\gtLab[i]}{\tyGroundi[i]}{\gtGii[i]}$
          or
          $
            \gtCommTransit{\roleQMaybeCrashed}{\roleP}{i \in
            I'}{\gtLab[i]}{\tyGroundi[i]}{\gtGii[i]}{j}$,
          where
          $I' \subseteq I$, and
          for all $i \in I'$:
          $\stLab[i] = \gtLab[i]$,
          $\tyGroundi[i] = \tyGround[i]$,
          and $\roleQ \notin \rolesR$ implies $\exists k \in I': \gtLab[k] =
          \gtCrashLab$.
          \begin{itemize}[leftmargin=*]
                \item %
      $\gtGi[n] =
      \gtComm{\roleQ}{\rolePMaybeCrashed}{i \in I'}{\gtLab[i]}{\tyGroundi[i]}{\gtGii[i]}$: we first show that
      $\rolePMaybeCrashed = \roleP$. Since $\stEnvApp{\stEnvi[n]}{%
        \roleP%
      } \neq \stStop$, by \Cref{def:assoc}, we have that $\roleP \notin \rolesCi[n]$, and hence, $\rolePMaybeCrashed
      \neq \rolePCrashed$. Given that $\gtGi[n] =
      \gtComm{\roleQ}{\roleP}{i \in I'}{\gtLab[i]}{\tyGroundi[i]}{\gtGii[i]}$, by association, \Cref{def:assoc}, and inversion of subtyping,
       $\stEnvApp{\stEnvi[n]}{%
        \roleQ%
      }$ is of the form
     $\stIntSum{\roleP}{i \in I''}{\stChoice{\stLab[i]}{\tyGroundi[i]} \stSeq \stTi[i]}$ where
      $I'' \subseteq I'$. Then applying \inferrule{\iruleTCtxOut},
      $\stEnvi[n]; \qEnvi[n]
      \,\stEnvMoveOutAnnot{\roleQ}{\roleP}{\stChoice{\stLab[j]}{\tyGroundi[j]}}$ for some $j \in I''$.
      Then together with the fairness of $(\stEnvi[n]; \qEnvi[n])_{n \in N}$,
      we have that there exists $k, \stLabi, \tyGroundii$ such that $n \leq k \in N$ and
     $\stEnvi[k]; \qEnvi[k] \stEnvMoveOutAnnot{\roleQ}{\roleP}{\stChoice{\stLabi}{\tyGroundii}} \stEnvi[k+1]; \qEnvi[k+1]$,
     which follows that  $\stEnvApp{\qEnvi[k+1]}{\roleQ, \roleP} =
     \stEnvUpd{\qEnvi[k]}{\roleQ, \roleP}{
        \stQCons{
          \stEnvApp{\qEnvi[k]}{\roleQ, \roleP}
        }{
          \stQMsg{\stLabi}{\tyGroundii}
        }}$. Finally, by the previous case (L1), we can conclude that there exists $k', \stLabii, \tyGroundiii$ such that $k + 1 \leq k' \in N$ and
         $\stEnvi[k']; \qEnvi[k'] \stEnvMoveInAnnot{\roleP}{\roleQ}{\stChoice{\stLabii}{\tyGroundiii}} \stEnvi[k'+1]; \qEnvi[k'+1]$, as desired.

    \item $\gtGi[n] =
            \gtCommTransit{\roleQMaybeCrashed}{\roleP}{i \in
            I'}{\gtLab[i]}{\tyGroundi[i]}{\gtGii[i]}{j}$: we consider two subcases:
            \begin{itemize}[leftmargin=*]
            \item $\gtLab[j] \neq \gtCrashLab$: by applying \inferrule{\iruleGtMoveIn},
            $\gtWithCrashedRoles{\rolesCi[n]}{
     \gtGi[n]
    }
    \gtMove[
      \stEnvInAnnotSmall{\roleP}{\roleQ}{\stChoice{\gtLab[j]}{\tyGroundi[j]}}
    ]{
      \rolesR
    }
    \gtWithCrashedRoles{\rolesCi[n]}{\gtGii[j]}$.
    Hence, by the soundness association, we have that  $\stEnvi[n]; \qEnvi[n]
      \,\stEnvMoveInAnnot{\roleP}{\roleQ}{\stChoice{\stLab[j]}{\tyGroundi[j]}}$.
       Finally, combining the fact that $(\stEnvi[n]; \qEnvi[n])_{n \in N}$  is fair with
      $\stEnvi[n]; \qEnvi[n] \stEnvMoveInAnnot{\roleP}{\roleQ}{\stChoice{\stLab[j]}{\tyGroundi[j]}}$,
      we can conclude that there exists $k$ such that $n \leq k \in N$ and
     $\stEnvi[k]; \qEnvi[k] \stEnvMoveInAnnot{\roleP}{\roleQ}{\stChoice{\stLab[j]}{\tyGroundi[j]}} \stEnvi[k+1]; \qEnvi[k+1]$, as desired.
   \item $\gtLab[j] = \gtCrashLab$: consider the following two subcases:
   \begin{itemize}[leftmargin=*]
   \item $\roleQMaybeCrashed = \roleQCrashed$:
    $\stEnvApp{\stEnvi[n]}{%
        \roleQ%
      } = \stStop$ by $\roleQ \in \rolesCi[n]$. By  \Cref{def:assoc}, $\stEnvApp{\qEnvi[n]}{\roleQ, \roleP} =
      \stQEmpty$.  Now applying \inferrule{\iruleTCtxCrashDetect} on $\stEnvApp{\stEnvi[n]}{\roleP} =
      \stExtSum{\roleQ}{i \in I}{\stChoice{\stLab[i]}{\tyGround[i]} \stSeq \stT[i]},
          \stEnvApp{\stEnvi[n]}{\roleQ} = \stStop,
           \stEnvApp{\qEnvi[n]}{\roleQ, \roleP} = \stQEmpty$ and
       $\gtLab[j] = \stLab[j] = \stCrashLab$ with $j \in I' \subseteq I$, we have that
       $\stEnvi[n]; \qEnvi[n] \stEnvMoveAnnot{\ltsCrDe{\mpS}{\roleP}{\roleQ}}$.  Then together with the fairness of
       $(\stEnvi[n]; \qEnvi[n])_{n \in N}$ , we can conclude that there exists $k$ such that $n \leq k \in N$ and
       $\stEnvi[k]; \qEnvi[k] \stEnvMoveAnnot{\ltsCrDe{\mpS}{\roleP}{\roleQ}} \stEnvi[k+1]; \qEnvi[k+1]$.

   \item $\roleQMaybeCrashed \neq \roleQCrashed$: since $\gtLab[j] = \gtCrashLab$,  %
  together with $\gtWithCrashedRoles{\rolesC}{\gtG} \!\gtMoveStar[\rolesR]\!
\gtWithCrashedRoles{\rolesCi}{\gtGi}$,  we know that $\roleQMaybeCrashed \neq \roleQ$ holds, a desired contradiction.
            \qedhere 
   \end{itemize}
            \end{itemize}
          \end{itemize}         
\end{enumerate}
\end{proof}

\colProjAll*
\begin{proof}
Note $\stEnvAssoc{\gtG}{\stEnv}{\rolesR}$ is an abbreviation of
  $\stEnvAssoc{\gtWithCrashedRoles{\emptyset}{\gtG}}{\stEnv; \qEnv[\emptyset]}{\rolesR}$,
  apply \cref{lem:safety-by-proj,lem:proj:df,lem:ext-proj-live}.
\end{proof}

\section{Proofs for \cref{sec:typing_system}}
\label{sec:proof:typesystem}

\begin{lemma}[Typing Inversion]
\label{lem:typeinversion}
 Let $\Theta \vdash \mpP : \stT$. Then:  
 \begin{enumerate}[leftmargin=*]
\item
\label{proc_end}
$\mpP =  \mpNil$ implies $\stT = \stEnd$; 
\item 
\label{proc_stop}
$\mpP = \mpCrash$ implies $\stT = \stStop$;
\item
\label{proc_out}
$\mpP = \procout\roleQ{\mpLab}{\mpE}{\mpP[1]}$ implies 
\begin{enumerate}[label={(a\arabic*)}, leftmargin=*, ref={(a\arabic*)}]
\item $\stOut{\roleQ}{\stLab}{\tyGround}\stSeq{\stT[1]} \stSub \stT$, and
\item $\Theta \vdash \mpP[1] : \stT[1]$, and 
\item   $\Theta \vdash \mpE:\tyGround$;
\end{enumerate}
\item
\label{proc_ext}
$\mpP = \sum_{i\in I}\procin{\roleQ}{\mpLab_i(\mpx_i)}{\mpP_i}$ implies 
\begin{enumerate}[label={(a\arabic*)}, leftmargin=*, ref={(a\arabic*)}]
\item $\stExtSum{\roleQ}{i\in I}{\stChoice{\stLab[i]}{\tyGround[i]}\stSeq{\stT_i}} \stSub \stT$, and 
\item $\forall i\in I\;\;\; \Theta, x_i:\tyGround_i \vdash \mpP_i:\stT_i$; 
\end{enumerate}
\item
\label{proc_cond}
$\mpP = \mpIf\mpE{\mpP_1}{\mpP_2}$ implies 
\begin{enumerate}[label={(a\arabic*)}, leftmargin=*, ref={(a\arabic*)}]
\item $\Theta \vdash \mpE:\tyBool$, and 
\item $\Theta  \vdash \mpP_1:\stT$, and 
\item $\Theta  \vdash \mpP_2:\stT$; 
\end{enumerate}
\item 
\label{proc_rec}
$\mpP =  \mu X.\mpP[1]$ implies 
$\Theta, X:\stT[1]  \vdash \mpP[1]:\stT[1]$ and $\stT[1] \stSub \stT$ for some $\stT[1]$; 
\item 
\label{proc_sub}
$\mpP = \stT$ implies $\Theta \vdash \mpP:\stTi$ for some $\stTi\stSub \stT$; 
\end{enumerate}
\noindent
Let $ \vdash \mpH :\qEnvPartial$. Then:
\begin{enumerate}[resume]
\item
\label{proc_empty_Q}
$\mpH = \mpQEmpty$ implies $\qEnvPartial =  \stQEmpty$; 
\item
\label{proc_unavail_Q}
$\mpH = \mpQUnavail$ implies $\qEnvPartial = \stQUnavail$; 
\item 
\label{proc_message_Q}
$\mpH = (\roleQ , \mpLab(\mpV))$ implies  $\vdash \mpV:\tyGround$ and  %
 $\stEnvApp{\qEnvPartial}{\roleQ} = \stQMsg\mpLab\tyGround$;  
 \item
 \label{proc_con_Q}
 $\mpH =  \mpH_1 \cdot \mpH_2$ implies $\vdash \mpH_1:\qEnvPartial[1]$, and 
 $\vdash \mpH_2:\qEnvPartial[2]$, and $\qEnvPartial = \qEnvPartial[1] \cdot \qEnvPartial[2]$;
\end{enumerate}
\noindent
Let $ \gtWithCrashedRoles{\rolesC}{\gtG} 
  \vdash \prod_{i\in I} (\mpPart{\roleP[i]}{\mpP[i]} \mpPar
  \mpPart{\roleP[i]}{\mpH[i]})$. Then:
  \begin{enumerate}[label={(a\arabic*)}, leftmargin=*, ref={(a\arabic*)}]
  \item $\exists \,\stEnv; \qEnv$ such that  $\stEnvAssoc{\gtWithCrashedRoles{\rolesC}{\gtG}}{\stEnv; \qEnv}{\rolesR}$, and 
  \item $\forall i \in I: \,\,\vdash \mpP_i:\stEnvApp{\stEnv}{\roleP[i]}$, and 
  \item $\forall i \in I:   \,\,\vdash \mpH[i]: \stEnvApp{\qEnv}{-, \roleP[i]}$. 
  \end{enumerate}

\end{lemma}
\begin{proof}
By induction on type derivations in~\Cref{fig:processes:typesystem}. 
\end{proof}

\begin{lemma}[Typing Congruence]
\label{lem:typingcongruence}
\begin{enumerate}[leftmargin=*]
\item 
\label{lem:proc_cong}
If $\Theta \vdash \mpP:\stT$ and $\mpP \equiv \mpQ$, then $\Theta \vdash \mpQ:\stT$. 
\item 
\label{lem:Q_cong}
If $\vdash \mpH[1]: \qEnvPartial[1]$ and $\mpH[1] \equiv \mpH[2]$, then there exists $\qEnvPartial[2]$ such that $\qEnvPartial[1]
 \equiv \qEnvPartial[2]$ and $\vdash \mpH[2]:\qEnvPartial[2]$. 
 \item 
 \label{lem:cong}
 If $ \gtWithCrashedRoles{\rolesC}{\gtG} 
  \vdash \mpM$ and $\mpM \equiv \mpMi$, then 
 $ \gtWithCrashedRoles{\rolesC}{\gtG} 
  \vdash \mpMi$.  
\end{enumerate}
\end{lemma}
\begin{proof}
By analysing the cases where $\mpP \equiv \mpQ$, $\mpH[1] \equiv \mpH[2]$, 
$\qEnvPartial[1] \equiv \qEnvPartial[2]$, and $\mpM \equiv \mpMi$, 
and by typing inversion~\Cref{lem:typeinversion}. 
\end{proof}

\begin{lemma}[Substitution]
\label{lem:substitution}
If $\Theta, x:\tyGround   \vdash \mpP :\stT$ and $\Theta \vdash \val:\tyGround$, 
then 
$\Theta \vdash \mpP\{\val/x\}:\stT$. 
\end{lemma}
\begin{proof}
By induction on the structure of $\mpP$. 
\end{proof}

\begin{lemma}[Optional Lemma]
\label{lem:val-eval}
If $\emptyset \vdash \mpE:\tyGround$, then there exists $\val$ such that $\eval{\mpE}\val$. 
\end{lemma}
\begin{proof}
By induction on the derivation of $\emptyset \vdash \mpE : \tyGround$. 
\end{proof}

\begin{lemma}
\label{lem:crash-free}
If $\gtWithCrashedRoles{\rolesC}{\gtG} \vdash \mpM$ and $\mpM \neq 
   (\mpPart{\roleP}{\mpCrash} \mpPar \mpPart{\roleP}{\mpQUnavail}) \mpPar 
  \mpMi$ (for all $\roleP, \mpMi$), then for any  
  $\stEnv; \qEnv$ such that  $\stEnvAssoc{\gtWithCrashedRoles{\rolesC}{\gtG}}{\stEnv; \qEnv}{\rolesR}$, it holds that 
  $\forall \roleP \in \gtRoles{\gtG}: \stEnv(\roleP) \neq \stStop$ and $\qEnv(\roleP) \neq \stQUnavail$. 
\end{lemma}
\begin{proof}
By induction on the derivation of $\gtWithCrashedRoles{\rolesC}{\gtG} \vdash \mpM$. 
\end{proof}

\lemSubjectReduction*
\begin{proof}
Let us recap the assumptions:
 \begin{align}
\gtWithCrashedRoles{\rolesC}{\gtG} \vdash \mpM
\label{eq:subred_ass_global_typing}
\\
\mpM \mpMove[\rolesR] \mpMi
\label{eq:subred_ass_session_red}
\end{align}
The proof proceeds by induction on the derivation of $\mpM \mpMove[\rolesR] \mpMi$.
Most cases hold by typing inversion (\Cref{lem:typeinversion}), and by applying the induction
hypothesis.
\begin{itemize}[leftmargin=*]
\item Case \inferrule{r-send}: we have
\begin{gather}
\mpM = \mpPart\roleP{\procout{\roleQ}{\mpLab}{\mpE}{\mpP}}
\mpPar
\mpPart\roleP{\mpH[\roleP]}
\mpPar
\mpPart\roleQ{\mpQ}
\mpPar
\mpPart\roleQ{\mpH[\roleQ]}
\mpPar
\mpM[1]
\label{eq:r_send_M}
\\
\mpMi = \mpPart\roleP{\mpP}
\mpPar
\mpPart\roleP{\mpH[\roleP]}
\mpPar
\mpPart\roleQ{\mpQ}
\mpPar
\mpPart{\roleQ}{\mpH[\roleQ]}\cdot(\roleP,\mpLab(\val))
\mpPar
\mpM[1]
\label{eq:r_send_Mi}
\\
\eval{\mpE} \val
\label{eq:r_send_expression}
\\
\mpH[\roleQ] \neq \mpQUnavail
\label{eq:r_send_Q}
\\
\mpM[1] =  \prod_{i\in I} (\mpPart{\roleP[i]}{\mpP[i]} \mpPar
  \mpPart{\roleP[i]}{\mpH[i]})
 \label{eq:r_send_M1}
\end{gather}
By~\eqref{eq:subred_ass_global_typing} and~\Cref{lem:typeinversion}, we have that there
exists $\stEnv; \qEnv$ such that
\begin{gather}
\stEnvAssoc{\gtWithCrashedRoles{\rolesC}{\gtG}}{\stEnv; \qEnv}{\rolesR}
\label{eq:r_send_type_association}
\\
\vdash \procout\roleQ{\mpLab}{\mpE}{\mpP}:\stEnvApp{\stEnv}{\roleP}
\label{eq:r_send_type_out}
\\
 \vdash \mpH[\roleP]:  \stEnvApp{\qEnv}{-, \roleP}%
 \label{eq:r_send_type_out_Q}
 \\
 \vdash \mpQ: \stEnvApp{\stEnv}{\roleQ}
 \label{eq:r_send_type_out_another}
 \\
  \vdash \mpH[\roleQ]: \stEnvApp{\qEnv}{-, \roleQ}%
 \label{eq:r_send_type_out_Q_another}
 \\
 \forall i \in I: \,\,\vdash \mpP_i:\stEnvApp{\stEnv}{\roleP[i]}
 \label{eq:r_send_M1_type}
 \\
 \forall i \in I:  \,\,\vdash \mpH[i]: \stEnvApp{\qEnv}{-, \roleP[i]}%
  \label{eq:r_send_M1_type_Q}
\end{gather}
By~\eqref{eq:r_send_type_out} and 3 of~\cref{lem:typeinversion}, we have that
\begin{gather}
\stEnvApp{\stEnv}{\roleP} = \stOut{\roleQ}{\mpLab}{\tyGround}\stSeq{\stT}
\label{eq:r_send_type_configuration}
\\
\stT[1] \stSub \stT
\label{eq:r_send_subtyping}
\\
 \vdash \mpP : \stT[1]
 \label{eq:r_send_type_configuration_2}
 \\
 \vdash \mpE:\tyGround
  \label{eq:r_send_type_configuration_3}
\end{gather}
We now let
\begin{gather}
\stEnvi = \stEnvUpd{\stEnv}{\roleP}{\stT}
\label{eq:r_send_new_context}
\\
\qEnvi = \stEnvUpd{\qEnv}{\roleP, \roleQ}{
        \stQCons{
          \stEnvApp{\qEnv}{\roleP, \roleQ}
        }{
          \stQMsg{\stLab}{\tyGround}
        }
        }
\label{eq:r_send_new_Q}
\end{gather}
Then by \inferrule{\iruleTCtxOut} in~\Cref{fig:gtype:tc-red-rules}, we have
\begin{align}
\stEnv; \qEnv \!\stEnvMoveMaybeCrash[\rolesR]\! \stEnvi; \qEnvi
\label{eq:r_send_reduction}
\end{align}
Hence, using~\cref{thm:gtype:proj-comp}, we have that there exists
$\gtWithCrashedRoles{\rolesCi}{\gtGi}$ such that
\begin{gather}
\stEnvAssoc{\gtWithCrashedRoles{\rolesCi}{\gtGi}}{\stEnvi; \qEnvi}{\rolesR}
\label{eq:r_send_type_association_new}
\\
\gtWithCrashedRoles{\rolesC}{\gtG} \gtMove{\rolesR}
    \gtWithCrashedRoles{\rolesCi}{\gtGi}
\label{eq:r_send_type_global_reduction_new}
\end{gather}
Combine \eqref{eq:r_send_type_association_new}, \eqref{eq:r_send_type_global_reduction_new} with
 \begin{gather}
 \vdash \mpP : \stEnvApp{\stEnvi}{\roleP} \quad \quad (\text{by}~\eqref{eq:r_send_new_context}, \eqref{eq:r_send_subtyping}, \eqref{eq:r_send_type_configuration_2}\text{ and }\inferrule{{t-sub}})
 \\
  \vdash \mpH[\roleP] : \stEnvApp{\qEnvi}{-, \roleP} \quad  (\text{by}~\eqref{eq:r_send_new_Q}\text{ and }\eqref{eq:r_send_type_out_Q})
  \\
  \vdash \mpQ : \stEnvApp{\stEnvi}{\roleQ} \quad \quad (\text{by}~\eqref{eq:r_send_type_out_another}\text{ and }\eqref{eq:r_send_new_context})
  \\
  \vdash \mpH[\roleQ] : \stEnvApp{\qEnvi}{-, \roleQ} \quad  (\text{by}~\eqref{eq:r_send_new_Q}, \eqref{eq:r_send_type_out_Q_another}\text{ and } 10, 11\text{ of \cref{lem:typeinversion}})
  \\
   \forall i \in I: \,\,\vdash \mpP_i:\stEnvApp{\stEnvi}{\roleP[i]} \quad (\text{by}~\eqref{eq:r_send_M1_type}\text{ and }\eqref{eq:r_send_new_context})
 \\
 \forall i \in I:  \,\,\vdash \mpH[i]: \stEnvApp{\qEnvi}{-, \roleP[i]} \quad (\text{by}~\eqref{eq:r_send_M1_type_Q}\text{ and }\eqref{eq:r_send_new_Q})
 \end{gather}
 We conclude that $\gtWithCrashedRoles{\rolesCi}{\gtGi} \vdash \mpMi$.

\item Case \inferrule{r-rcv}: similar to case \inferrule{r-send} above, except that we proceed by \inferrule{r-rcv}, inversion of \inferrule{{t-ext}} (4 of \Cref{lem:typeinversion}), and~\Cref{lem:substitution}.

\item Case \inferrule{r-send-\mpCrash}: we have
\begin{gather}
\mpM = \mpPart\roleP{\procout{\roleQ}{\mpLab}{\mpE}{\mpP}}
\mpPar
\mpPart\roleP{\mpH[\roleP]}
\mpPar
\mpPart\roleQ{\mpCrash}
\mpPar
\mpPart\roleQ{\mpQUnavail}
\mpPar
\mpM[1]
\label{eq:r_send_crash_M}
\\
\mpMi = \mpPart\roleP{\mpP}
\mpPar
\mpPart\roleP{\mpH[\roleP]}
\mpPar
\mpPart\roleQ{\mpCrash}
\mpPar
\mpPart{\roleQ}{\mpQUnavail}
\mpPar
\mpM[1]
\label{eq:r_send_crash_Mi}
\\
\\
\mpM[1] =  \prod_{i\in I} (\mpPart{\roleP[i]}{\mpP[i]} \mpPar
  \mpPart{\roleP[i]}{\mpH[i]})
 \label{eq:r_send_crash_M1}
\end{gather}
By~\eqref{eq:subred_ass_global_typing} and~\Cref{lem:typeinversion}, we have that there
exists $\stEnv; \qEnv$ such that
\begin{gather}
\stEnvAssoc{\gtWithCrashedRoles{\rolesC}{\gtG}}{\stEnv; \qEnv}{\rolesR}
\label{eq:r_send_crash_type_association}
\\
\vdash \procout\roleQ{\mpLab}{\mpE}{\mpP}:\stEnvApp{\stEnv}{\roleP}
\label{eq:r_send_crash_type_out}
\\
 \vdash \mpH[\roleP]: \stEnvApp{\qEnv}{-, \roleP}
 \label{eq:r_send_crash_type_out_Q}
 \\
 \vdash \mpCrash: \stEnvApp{\stEnv}{\roleQ}
 \label{eq:r_send_crash_type_out_another}
 \\
  \vdash \mpQUnavail: \stEnvApp{\qEnv}{-, \roleQ}
 \label{eq:r_send_crash_type_out_Q_another}
 \\
 \forall i \in I: \,\,\vdash \mpP_i:\stEnvApp{\stEnv}{\roleP[i]}
 \label{eq:r_send_crash_M1_type}
 \\
 \forall i \in I:  \,\,\vdash \mpH[i]: \stEnvApp{\qEnv}{-, \roleP[i]}
  \label{eq:r_send_crash_M1_type_Q}
\end{gather}
It follows directly that
\begin{gather}
\stEnvApp{\stEnv}{\roleQ} = \stStop
\label{eq:r_send_crash_stop}
\\
 \stEnvApp{\qEnv}{-, \roleQ} = \stQUnavail
 \label{eq:r_send_crash_unavail}
\end{gather}
By~\eqref{eq:r_send_crash_type_out} and 3 of~\cref{lem:typeinversion}, we have that
\begin{gather}
\stEnvApp{\stEnv}{\roleP} = \stOut{\roleQ}{\mpLab}{\tyGround}\stSeq{\stT}
\label{eq:r_send_crash_type_configuration}
\\
\stT[1] \stSub \stT
\label{eq:r_send_crash_subtyping}
\\
 \vdash \mpP : \stT[1]
 \label{eq:r_send_crash_type_configuration_2}
 \\
 \vdash \mpE:\tyGround
  \label{eq:r_send_crash_type_configuration_3}
\end{gather}
We now let
\begin{gather}
\stEnvi = \stEnvUpd{\stEnv}{\roleP}{\stT}
\label{eq:r_send_crash_new_context}
\\
\qEnvi = \stEnvUpd{\qEnv}{\roleP, \roleQ}{
        \stQCons{
          \stEnvApp{\qEnv}{\roleP, \roleQ}
        }{
          \stQMsg{\stLab}{\tyGround}
        }
        }
\label{eq:r_send_crash_new_Q}
\end{gather}
Then by \inferrule{\iruleTCtxOut} in~\Cref{fig:gtype:tc-red-rules}, we have
\begin{align}
\stEnv; \qEnv \!\stEnvMoveMaybeCrash[\rolesR]\! \stEnvi; \qEnvi
\label{eq:r_send_crash_reduction}
\end{align}
Hence, using~\cref{thm:gtype:proj-comp}, we have that there exists
$\gtWithCrashedRoles{\rolesCi}{\gtGi}$ such that
\begin{gather}
\stEnvAssoc{\gtWithCrashedRoles{\rolesCi}{\gtGi}}{\stEnvi; \qEnvi}{\rolesR}
\label{eq:r_send_crash_type_association_new}
\\
\gtWithCrashedRoles{\rolesC}{\gtG} \gtMove{\rolesR}
    \gtWithCrashedRoles{\rolesCi}{\gtGi}
\label{eq:r_send_crash_type_global_reduction_new}
\end{gather}
 Combine \eqref{eq:r_send_crash_type_association_new}, \eqref{eq:r_send_crash_type_global_reduction_new} with
 \begin{gather}
 \vdash \mpP : \stEnvApp{\stEnvi}{\roleP} \quad \quad (\text{by}~\eqref{eq:r_send_crash_new_context}, \eqref{eq:r_send_crash_subtyping}, \eqref{eq:r_send_crash_type_configuration_2}\text{ and }\inferrule{{t-sub}})
 \\
  \vdash \mpH[\roleP] : \stEnvApp{\qEnvi}{-, \roleP} \quad  (\text{by}~\eqref{eq:r_send_crash_new_Q}\text{ and }\eqref{eq:r_send_crash_type_out_Q})
  \\
  \vdash \mpCrash : \stEnvApp{\stEnvi}{\roleQ} \quad \quad (\text{by}~\eqref{eq:r_send_crash_type_out_another},\eqref{eq:r_send_crash_stop}\text{ and }\eqref{eq:r_send_crash_new_context})
  \\
  \vdash \mpQUnavail: \stEnvApp{\qEnvi}{-, \roleQ} \quad  (\text{by}~\eqref{eq:r_send_crash_new_Q}, \eqref{eq:r_send_crash_type_out_Q_another}\text{ and }\eqref{eq:r_send_crash_unavail})
  \\
   \forall i \in I: \,\,\vdash \mpP_i:\stEnvApp{\stEnvi}{\roleP[i]} \quad (\text{by}~\eqref{eq:r_send_crash_M1_type}\text{ and }\eqref{eq:r_send_crash_new_context})
 \\
 \forall i \in I:  \,\,\vdash \mpH[i]: \stEnvApp{\qEnvi}{-, \roleP[i]} \quad (\text{by}~\eqref{eq:r_send_crash_M1_type_Q}\text{ and }\eqref{eq:r_send_crash_new_Q})
 \end{gather}
 We conclude that $\gtWithCrashedRoles{\rolesCi}{\gtGi} \vdash \mpMi$.

\item Case \inferrule{r-rcv-$\odot$}:  we have
\begin{gather}
\mpM = \mpPart\roleP{\sum_{i\in I} \procin\roleQ{\mpLab_i(\mpx_i)}\mpP_i}
\mpPar
\mpPart\roleP{\mpH[\roleP]}
\mpPar
\mpPart\roleQ\mpCrash
\mpPar
\mpPart\roleQ\mpQUnavail
\mpPar
\mpM[1]
\label{eq:r_rcv_det_M}
\\
\mpMi = \mpPart\roleP \mpP_k
\mpPar
\mpPart\roleP{\mpH_{\roleP}}
\mpPar
\mpPart\roleQ\mpCrash
\mpPar
\mpPart\roleQ\mpQUnavail
\mpPar
\mpM[1]
\label{eq:r_rcv_det_Mi}
\\
\mpM[1] =  \prod_{i\in I} (\mpPart{\roleP[i]}{\mpP[i]} \mpPar
  \mpPart{\roleP[i]}{\mpH[i]})
 \label{eq:r_rcv_det_M1}
 \\
 k \in I
 \label{eq:r_rcv_det_plus_con_1}
 \\
 \mpLab[k] = \mpCrashLab
 \label{eq:r_rcv_det_plus_con_2}
 \\
 \nexists \mpLab, \mpV: (\roleQ,
\mpLab(\mpV)) \in \mpH[\roleP]
 \label{eq:r_rcv_det_plus_con_3}
\end{gather}
By~\eqref{eq:subred_ass_global_typing} and~\Cref{lem:typeinversion}, we have that there
exists $\stEnv; \qEnv$ such that
\begin{gather}
\stEnvAssoc{\gtWithCrashedRoles{\rolesC}{\gtG}}{\stEnv; \qEnv}{\rolesR}
\label{eq:r_rcv_det_type_association}
\\
\vdash \sum_{i\in I}\procin{\roleQ}{\mpLab_i(\mpx_i)}{\mpP_i}:\stEnvApp{\stEnv}{\roleP}
\label{eq:r_rcv_det_type_out}
\\
 \vdash \mpH[\roleP]: \stEnvApp{\qEnv}{-, \roleP}
 \label{eq:r_rcv_det_type_out_Q}
 \\
 \vdash \mpCrash: \stEnvApp{\stEnv}{\roleQ}
 \label{eq:r_rcv_det_type_out_another}
 \\
  \vdash \mpQUnavail: \stEnvApp{\qEnv}{-, \roleQ}
 \label{eq:r_rcv_det_type_out_Q_another}
 \\
 \forall i \in I: \,\,\vdash \mpP_i:\stEnvApp{\stEnv}{\roleP[i]}
 \label{eq:r_rcv_det_M1_type}
 \\
 \forall i \in I:  \,\,\vdash \mpH[i]: \stEnvApp{\qEnv}{-, \roleP[i]}
  \label{eq:r_rcv_det_M1_type_Q}
\end{gather}
It follows directly that
\begin{gather}
\stEnvApp{\stEnv}{\roleQ} = \stStop
\label{eq:r_rcv_det_stop}
\\
 \stEnvApp{\qEnv}{-, \roleQ} = \stQUnavail
 \label{eq:r_rcv_det_unavail}
\end{gather}
By~\eqref{eq:r_rcv_det_type_out}, 4 of~\cref{lem:typeinversion}, and \inferrule{\iruleStSubIn}, we have that
\begin{gather}
\stEnvApp{\stEnv}{\roleP} = \stExtSum{\roleQ}{i \in J}{\stChoice{\stLab[i]}{\tyGround[i]} \stSeq \stTi[i]}
\label{eq:r_rcv_det_type_configuration}
\\
J \subseteq I
\label{eq:r_rcv_det_subset}
\\
\forall i \in J:  \stT[i] \stSub \stTi[i]
\label{eq:r_rcv_det_subtyping}
\\
 \setcomp{\stLab[l]}{l \in I} \neq \setenum{\stCrashLab}
 \label{eq:r_rcv_det_not_crash}
\\
 \nexists j \in I \setminus J: \stLab[j] = \stCrashLab
 \label{eq:r_rcv_det_not_crash_2}
 \\
 \forall i \in I: x_i : \tyGround[i] \vdash \mpP[i] : \stT[i]
 \label{eq:r_rcv_det_type_configuration_2}
\end{gather}
From \eqref{eq:r_rcv_det_not_crash_2} and \eqref{eq:r_rcv_det_plus_con_2}, we get
\begin{align}
   k \in J
   \label{eq:r_rcv_det_k_in_J_1}
\end{align}
By \eqref{eq:r_rcv_det_type_out_Q} and \eqref{eq:r_rcv_det_plus_con_3}, we also know
\begin{align}
 \stEnvApp{\qEnv}{\roleQ, \roleP} = \stQEmpty
 \label{eq:r_rcv_det_empty_Q}
\end{align}
We now let
\begin{gather}
\stEnvi = \stEnvUpd{\stEnv}{\roleP}{\stTi[k]}
\label{eq:r_rcv_det_new_context}
\\
\qEnvi = \qEnv
\label{eq:r_rcv_det_new_Q}
\end{gather}
Then by \inferrule{\iruleTCtxCrashDetect} in~\Cref{fig:gtype:tc-red-rules}, we have
\begin{align}
\stEnv; \qEnv \!\stEnvMoveMaybeCrash[\rolesR]\! \stEnvi; \qEnvi
\label{eq:r_rcv_det_reduction}
\end{align}
Hence, using~\cref{thm:gtype:proj-comp}, we have that there exists
$\gtWithCrashedRoles{\rolesCi}{\gtGi}$ such that
\begin{gather}
\stEnvAssoc{\gtWithCrashedRoles{\rolesCi}{\gtGi}}{\stEnvi; \qEnvi}{\rolesR}
\label{eq:r_rcv_det_type_association_new}
\\
\gtWithCrashedRoles{\rolesC}{\gtG} \gtMove{\rolesR}
    \gtWithCrashedRoles{\rolesCi}{\gtGi}
\label{eq:r_rcv_det_type_global_reduction_new}
\end{gather}
 Combine \eqref{eq:r_rcv_det_type_association_new}, \eqref{eq:r_rcv_det_type_global_reduction_new} with
 \begin{gather}
 \vdash \mpP[k] : \stEnvApp{\stEnvi}{\roleP} \quad \quad (\text{by}~\eqref{eq:r_rcv_det_new_context}, \eqref{eq:r_rcv_det_subtyping}, \eqref{eq:r_rcv_det_type_configuration_2}\text{ and }\inferrule{{t-sub}})
 \\
  \vdash \mpH[\roleP] : \stEnvApp{\qEnvi}{-, \roleP} \quad  (\text{by}~\eqref{eq:r_rcv_det_new_Q}\text{ and }\eqref{eq:r_rcv_det_type_out_Q})
  \\
  \vdash \mpCrash : \stEnvApp{\stEnvi}{\roleQ} \quad \quad (\text{by}~\eqref{eq:r_rcv_det_type_out_another},\eqref{eq:r_rcv_det_stop}\text{ and }\eqref{eq:r_rcv_det_new_context})
  \\
  \vdash \mpQUnavail: \stEnvApp{\qEnvi}{-, \roleQ} \quad  (\text{by}~\eqref{eq:r_rcv_det_new_Q}, \eqref{eq:r_rcv_det_type_out_Q_another}\text{ and }\eqref{eq:r_rcv_det_unavail})
  \\
   \forall i \in I: \,\,\vdash \mpP_i:\stEnvApp{\stEnvi}{\roleP[i]} \quad (\text{by}~\eqref{eq:r_rcv_det_M1_type}\text{ and }\eqref{eq:r_rcv_det_new_context})
 \\
 \forall i \in I:  \,\,\vdash \mpH[i]: \stEnvApp{\qEnvi}{-, \roleP[i]} \quad (\text{by}~\eqref{eq:r_rcv_det_M1_type_Q}\text{ and }\eqref{eq:r_rcv_det_new_Q})
 \end{gather}
 We conclude that $\gtWithCrashedRoles{\rolesCi}{\gtGi} \vdash \mpMi$.

\item Case \inferrule{r-$\lightning$}: we have
\begin{gather}
\mpM = \mpPart\roleP{\mpP}
\mpPar
\mpPart\roleP{\mpH[\roleP]}
\mpPar
\mpM[1]
\label{eq:r_lighting_M}
\\
\mpMi = \mpPart\roleP{\mpCrash}
\mpPar
\mpPart\roleP{\mpQUnavail}
\mpPar
\mpM
\label{eq:r_lighting_Mi}
\\
\mpP \neq \mpNil
\label{eq:r_lighting_M_Nil}
\\
\roleP \notin \rolesR
\label{eq:r_lighting_M_reliable}
\\
\mpM[1] =  \prod_{i\in I} (\mpPart{\roleP[i]}{\mpP[i]} \mpPar
  \mpPart{\roleP[i]}{\mpH[i]})
 \label{eq:r_lighting_M1}
\end{gather}
By~\eqref{eq:subred_ass_global_typing} and~\Cref{lem:typeinversion}, we have that there
exists $\stEnv; \qEnv$ such that
\begin{gather}
\stEnvAssoc{\gtWithCrashedRoles{\rolesC}{\gtG}}{\stEnv; \qEnv}{\rolesR}
\label{eq:r_lighting_type_association}
\\
\vdash \mpP:\stEnvApp{\stEnv}{\roleP}
\label{eq:r_lighting_type_out}
\\
 \vdash \mpH[\roleP]: \stEnvApp{\qEnv}{-, \roleP}
 \label{eq:r_lighting_type_out_Q}
 \\
 \forall i \in I: \,\,\vdash \mpP_i:\stEnvApp{\stEnv}{\roleP[i]}
 \label{eq:r_lighting_M1_type}
 \\
 \forall i \in I:  \,\,\vdash \mpH[i]: \stEnvApp{\qEnv}{-, \roleP[i]}
  \label{eq:r_lighting_M1_type_Q}
\end{gather}
By \eqref{eq:r_lighting_M_Nil}, \eqref{eq:r_lighting_M_reliable}, and \eqref{eq:r_lighting_type_out},
we have that
\begin{gather}
\stEnvApp{\stEnv}{\roleP} \neq \stStop
\label{eq:r_lighting_not_stop}
\\
 \stEnvApp{\stEnv}{\roleP} \neq \stEnd
 \label{eq:r_lighting_not_end}
\end{gather}
We now let
\begin{gather}
\stEnvi = \stEnvUpd{\stEnv}{\roleP}{\stStop}
\label{eq:r_lighting_new_context}
\\
\qEnvi = \stEnvUpd{\qEnv}{\cdot, \roleP}{
       \stQUnavail
        }
\label{eq:r_lighting_new_Q}
\end{gather}
Then by \inferrule{\iruleTCtxCrash} in~\Cref{fig:gtype:tc-red-rules}, we have
\begin{align}
\stEnv; \qEnv \!\stEnvMoveMaybeCrash[\rolesR]\! \stEnvi; \qEnvi
\label{eq:r_lighting_reduction}
\end{align}
Hence, using~\cref{thm:gtype:proj-comp}, we have that there exists
$\gtWithCrashedRoles{\rolesCi}{\gtGi}$ such that
\begin{gather}
\stEnvAssoc{\gtWithCrashedRoles{\rolesCi}{\gtGi}}{\stEnvi; \qEnvi}{\rolesR}
\label{eq:r_lighting_type_association_new}
\\
\gtWithCrashedRoles{\rolesC}{\gtG} \gtMove{\rolesR}
    \gtWithCrashedRoles{\rolesCi}{\gtGi}
\label{eq:r_lighting_type_global_reduction_new}
\end{gather}
 Combine \eqref{eq:r_lighting_type_association_new}, \eqref{eq:r_lighting_type_global_reduction_new} with
 \begin{gather}
 \vdash \mpCrash : \stEnvApp{\stEnvi}{\roleP} \quad \quad (\text{by}~\eqref{eq:r_lighting_new_context})
  \\
  \vdash \mpQUnavail : \stEnvApp{\qEnvi}{-, \roleP} \quad  (\text{by}~\eqref{eq:r_lighting_new_Q})
    \\
    \forall i \in I: \,\,\vdash \mpP_i:\stEnvApp{\stEnvi}{\roleP[i]} \quad (\text{by}~\eqref{eq:r_lighting_M1_type}\text{ and }\eqref{eq:r_lighting_new_context})
 \\
 \forall i \in I:  \,\,\vdash \mpH[i]: \stEnvApp{\qEnvi}{-, \roleP[i]} \quad (\text{by}~\eqref{eq:r_lighting_M1_type_Q}\text{ and }\eqref{eq:r_lighting_new_Q})
 \end{gather}
  We conclude that $\gtWithCrashedRoles{\rolesCi}{\gtGi} \vdash \mpMi$.

\item Case \inferrule{r-cond-T}: we have
\begin{gather}
\mpM = \mpPart\roleP{\mpIf{\mpE}{\mpP}{\mpQ}}
\mpPar
\mpPart\roleP\mpH
\mpPar
\mpM[1]
\label{eq:r_cond_T_M}
\\
\mpMi = \mpPart\roleP\mpP
\mpPar
\mpPart\roleP\mpH
\mpPar
\mpM[1]
\label{eq:r_cond_T_Mi}
\\
\eval{\mpE}{\mpTrue}
\label{eq:r_cond_T_expression}
\\
\mpM[1] =  \prod_{i\in I} (\mpPart{\roleP[i]}{\mpP[i]} \mpPar
  \mpPart{\roleP[i]}{\mpH[i]})
 \label{eq:r_cond_T_M1}
\end{gather}
By~\eqref{eq:subred_ass_global_typing} and~\Cref{lem:typeinversion}, we have that there
exists $\stEnv; \qEnv$ such that
\begin{gather}
\stEnvAssoc{\gtWithCrashedRoles{\rolesC}{\gtG}}{\stEnv; \qEnv}{\rolesR}
\label{eq:r_cond_T_type_association}
\\
\vdash \mpIf{\mpE}{\mpP}{\mpQ}:\stEnvApp{\stEnv}{\roleP}
\label{eq:r_cond_T_type_out}
\\
 \vdash \mpH: \stEnvApp{\qEnv}{-, \roleP}
 \label{eq:r_cond_T_type_out_Q}
 \\
  \forall i \in I: \,\,\vdash \mpP_i:\stEnvApp{\stEnv}{\roleP[i]}
 \label{eq:r_cond_T_M1_type}
 \\
 \forall i \in I:  \,\,\vdash \mpH[i]: \stEnvApp{\qEnv}{-, \roleP[i]}
  \label{eq:r_cond_T_M1_type_Q}
\end{gather}
By~\eqref{eq:r_cond_T_expression},~\eqref{eq:r_cond_T_type_out} and 5 of~\cref{lem:typeinversion}, we have that
\begin{gather}
\vdash \mpE : \tyBool
 \label{eq:r_cond_T_type_configuration_1}
 \\
 \vdash \mpP : \stEnvApp{\stEnv}{\roleP}
 \label{eq:r_cond_T_type_configuration_2}
 \\
 \vdash \mpQ: \stEnvApp{\stEnv}{\roleP}
  \label{eq:r_cond_T_type_configuration_3}
\end{gather}
Combine~\eqref{eq:r_cond_T_type_configuration_2} with \eqref{eq:r_cond_T_type_out_Q},
\eqref{eq:r_cond_T_M1_type},  and \eqref{eq:r_cond_T_M1_type_Q}, we can conclude
that $\gtWithCrashedRoles{\rolesC}{\gtG} \vdash \mpMi$, as desired.

\item Case \inferrule{r-cond-F}: similar to the case \inferrule{r-cond-T}.

\item Case \inferrule{r-struct}: assume that $\mpM \mpMove\mpMi$ is derived from
\begin{gather}
\mpM \equiv \mpM[1]
\label{eq:r_struct_eq_1}
\\
\mpM[1] \mpMove \mpMi[1]
\label{eq:r_struct_eq_2}
\\
\mpMi[1] \equiv \mpMi
\label{eq:r_struct_eq_3}
\end{gather}
From \eqref{eq:r_struct_eq_1}, \eqref{eq:subred_ass_global_typing}, by 3 of~\cref{lem:typingcongruence},
we have that $\gtWithCrashedRoles{\rolesC}{\gtG} \vdash \mpM[1]$. By induction hypothesis, either
$\gtWithCrashedRoles{\rolesC}{\gtG} \vdash \mpMi[1]$ or  there exists $\gtWithCrashedRoles{\rolesCi}{\gtGi}$ such that
  $\gtWithCrashedRoles{\rolesC}{\gtG} \gtMove{\rolesR}
  \gtWithCrashedRoles{\rolesCi}{\gtGi}$ and
  $\gtWithCrashedRoles{\rolesCi}{\gtGi} \vdash \mpMi[1]$. Now by \eqref{eq:r_struct_eq_3} and 3 of \Cref{lem:typingcongruence}, we have that either $\gtWithCrashedRoles{\rolesC}{\gtG} \vdash \mpMi$ or
  $\gtWithCrashedRoles{\rolesCi}{\gtGi} \vdash \mpMi$, as desired.
  \qedhere
\end{itemize}
\end{proof}

\lemSessionFidelity*
\begin{proof}
Let us recap the assumptions:
 \begin{align}
\gtWithCrashedRoles{\rolesC}{\gtG} \vdash \mpM
\label{eq:sf_ass_global_typing}
\\
\gtWithCrashedRoles{\rolesC}{\gtG} \gtMove{\rolesR}
\label{eq:sf_ass_global_red}
\end{align}
The proof proceeds by induction on the derivation of $\gtWithCrashedRoles{\rolesC}{\gtG}$.
\begin{itemize}[leftmargin=*]
\item Case \inferrule{\iruleGtMoveCrash}: by inversion of \inferrule{\iruleGtMoveCrash}, we have
\begin{gather}
 \roleP \notin \rolesR
 \label{eq:sf_global_crash_cond_1}
 \\
 \roleP \in \gtRoles{\gtG}
  \label{eq:sf_global_crash_cond_2}
  \\
   \gtWithCrashedRoles{\rolesC}{\gtG}
    \gtMove[\ltsCrashSmall{\mpS}{\roleP}]{\rolesR}
    \gtWithCrashedRoles{\rolesC \cup \setenum{\roleP}}{\gtCrashRole{\gtG}{\roleP}}
     \label{eq:sf_global_crash_cond_3}
  \end{gather}
 We can assume that $\mpM$ is of the form
 \begin{gather}
 \mpPart\roleP{\mpP}
\mpPar
\mpPart\roleP{\mpH[\roleP]}
\mpPar
\mpM[1]
\label{eq:sf_global_crash_M_form}
\\
\mpM[1] =  \prod_{i\in I} (\mpPart{\roleP[i]}{\mpP[i]} \mpPar
  \mpPart{\roleP[i]}{\mpH[i]})
\label{eq:sf_global_crash_M1_form}
\\
\mpP \neq \mpNil
\label{eq:sf_global_crash_M_form_not_nil}
\\
\roleP \notin \rolesR
\label{eq:sf_global_crash_M_form_reliable}
\end{gather}
Then by~\eqref{eq:sf_ass_global_typing} and~\Cref{lem:typeinversion}, there
exists $\stEnv; \qEnv$ such that
\begin{gather}
\stEnvAssoc{\gtWithCrashedRoles{\rolesC}{\gtG}}{\stEnv; \qEnv}{\rolesR}
\label{eq:sf_global_crash_type_association}
\\
\vdash \mpP:\stEnvApp{\stEnv}{\roleP}
\label{eq:sf_global__crash_type_out}
\\
 \vdash \mpH[\roleP]:  \stEnvApp{\qEnv}{-, \roleP}%
 \label{eq:sf_global_crash_type_out_Q}
 \\
 \forall i \in I: \,\,\vdash \mpP_i:\stEnvApp{\stEnv}{\roleP[i]}
 \label{eq:sf_global_crash_M1_type}
 \\
 \forall i \in I:  \,\,\vdash \mpH[i]: \stEnvApp{\qEnv}{-, \roleP[i]}%
  \label{eq:sf_global_crash_M1_type_Q}
\end{gather}
From~\eqref{eq:sf_global_crash_cond_3}, by~\cref{thm:gtype:proj-sound}, there is $\stEnvi; \qEnvi$ such that
\begin{gather}
\stEnvAssoc{\gtWithCrashedRoles{\rolesC \cup \setenum{\roleP}}{\gtCrashRole{\gtG}{\roleP}}}{\stEnvi; \qEnvi}{\rolesR}
\label{eq:sf_global_crash_type_association_new}
\\
\stEnv; \qEnv
      \stEnvMoveAnnot{\ltsCrash{\mpS}{\roleP}}
      \stEnvi; \qEnvi
\label{eq:sf_global_crash_ctx_red}
\end{gather}
Using \eqref{eq:sf_global_crash_M_form_not_nil}, \eqref{eq:sf_global_crash_M_form_reliable}, \eqref{eq:sf_global__crash_type_out} and \inferrule{\iruleTCtxCrash} in \Cref{fig:gtype:tc-red-rules}, we get
\begin{gather}
\stEnvi = \stEnvUpd{\stEnv}{\roleP}{\stStop}
\label{eq:sf_global_crash_ctx_new}
\\
\qEnvi = \stEnvUpd{\qEnv}{\cdot, \roleP}{\stQUnavail}
\label{eq:sf_global_crash_ctx_Q_new}
\end{gather}
It follows that
\begin{gather}
 \vdash \mpCrash : \stEnvApp{\stEnvi}{\roleP} \quad \quad (\text{by}~\eqref{eq:sf_global_crash_ctx_new})
  \\
  \vdash \mpQUnavail : \stEnvApp{\qEnvi}{-, \roleP} \quad  (\text{by}~\eqref{eq:sf_global_crash_ctx_Q_new})
    \\
    \forall i \in I: \,\,\vdash \mpP_i:\stEnvApp{\stEnvi}{\roleP[i]} \quad (\text{by}~\eqref{eq:sf_global_crash_M1_type}\text{ and }\eqref{eq:sf_global_crash_ctx_new})
 \\
 \forall i \in I:  \,\,\vdash \mpH[i]: \stEnvApp{\qEnvi}{-, \roleP[i]} \quad (\text{by}~\eqref{eq:sf_global_crash_M1_type_Q}\text{ and }\eqref{eq:sf_global_crash_ctx_Q_new})
\end{gather}
Therefore, by \inferrule{r-$\lightning$} and \inferrule{{t-sess}}, we can conclude that there exists $\mpMi = \mpPart\roleP{\mpCrash}
\mpPar
\mpPart\roleP{\mpQUnavail}
\mpPar
\mpM[1]$ and $\gtWithCrashedRoles{\rolesCi}{\gtGi}
   =
    \gtWithCrashedRoles{\rolesC \cup \setenum{\roleP}}{\gtCrashRole{\gtG}{\roleP}}$ such that
   $ \gtWithCrashedRoles{\rolesC}{\gtG}
    \gtMove[\ltsCrashSmall{\mpS}{\roleP}]{\rolesR}
   \gtWithCrashedRoles{\rolesCi}{\gtGi}$,
    $\mpM \;\redCrash{\roleP}{\rolesR}\;\mpMi$, and
   $ \gtWithCrashedRoles{\rolesCi}{\gtGi} \vdash \mpMi$.

\item Case \inferrule{\iruleGtMoveOut}: by inversion of \inferrule{\iruleGtMoveOut}, we have
\begin{gather}
j \in I
 \label{eq:sf_global_send_cond_1}
 \\
 \gtLab[j] \neq \gtCrashLab
  \label{eq:sf_global_send_cond_2}
  \\
  \gtWithCrashedRoles{\rolesC}{
      \gtCommSmall{\roleP}{\roleQ}{i \in I}{\gtLab[i]}{\tyGround[i]}{\gtGi[i]}
    }
    \gtMove[
      \stEnvOutAnnotSmall{\roleP}{\roleQ}{\stChoice{\gtLab[j]}{\tyGround[j]}}
    ]{
      \rolesR
    }
    \gtWithCrashedRoles{\rolesC}{
      \gtCommTransit{\roleP}{\roleQ}{i \in I}{\gtLab[i]}{\tyGround[i]}{\gtGi[i]}{j}
    }
     \label{eq:sf_global_send_cond_3}
  \end{gather}
 We can assume that $\mpM$ is of the form
 \begin{gather}
 \mpPart\roleP{\procout{\roleQ}{\mpLab}{\mpE}{\mpP}}
\mpPar
\mpPart\roleP{\mpH[\roleP]}
\mpPar
\mpPart\roleQ{\mpQ}
\mpPar
\mpPart\roleQ{\mpH[\roleQ]}
\mpPar
\mpM[1]
\label{eq:sf_global_send_M_form}
\\
\mpM[1] =  \prod_{i\in I} (\mpPart{\roleP[i]}{\mpP[i]} \mpPar
  \mpPart{\roleP[i]}{\mpH[i]})
\label{eq:sf_global_send_M1_form}
\\
\eval{\mpE} \val
\label{eq:sf_global_send_M_form_value}
\\
 \mpH[\roleQ] \neq \mpQUnavail
\label{eq:sf_global_csend_M_form_Q_availabe}
\end{gather}
Then by~\eqref{eq:sf_ass_global_typing} and~\Cref{lem:typeinversion}, there
exists $\stEnv; \qEnv$ such that
\begin{gather}
\stEnvAssoc{\gtWithCrashedRoles{\rolesC}{\gtG}}{\stEnv; \qEnv}{\rolesR}
\label{eq:sf_global_send_type_association}
\\
\vdash \procout\roleQ{\mpLab}{\mpE}{\mpP}:\stEnvApp{\stEnv}{\roleP}
\label{eq:sf_global_send_type_out}
\\
 \vdash \mpH[\roleP]:  \stEnvApp{\qEnv}{-, \roleP}%
 \label{eq:sf_global_send_type_out_Q}
 \\
 \vdash \mpQ: \stEnvApp{\stEnv}{\roleQ}
 \label{eq:sf_global_send_type_out_another}
 \\
  \vdash \mpH[\roleQ]: \stEnvApp{\qEnv}{-, \roleQ}%
 \label{eq:sf_global_send_type_out_Q_another}
 \\
 \forall i \in I: \,\,\vdash \mpP_i:\stEnvApp{\stEnv}{\roleP[i]}
 \label{eq:sf_global_send_M1_type}
 \\
 \forall i \in I:  \,\,\vdash \mpH[i]: \stEnvApp{\qEnv}{-, \roleP[i]}%
  \label{eq:sf_global_send_M1_type_Q}
\end{gather}
By~\eqref{eq:sf_global_send_type_out} and 3 of~\cref{lem:typeinversion}, we have that
\begin{gather}
\stEnvApp{\stEnv}{\roleP} = \stOut{\roleQ}{\mpLab}{\tyGround}\stSeq{\stT}
\label{eq:sf_global_send_type_configuration}
\\
\stT[1] \stSub \stT
\label{eq:sf_global_send_subtyping}
\\
 \vdash \mpP : \stT[1]
 \label{eq:sf_global_send_type_configuration_2}
 \\
 \vdash \mpE:\tyGround
  \label{eq:sf_global_send_type_configuration_3}
\end{gather}
From~\eqref{eq:sf_global_send_cond_3}, by~\cref{thm:gtype:proj-sound}, there are $\stEnvi; \qEnvi$ and
$\stEnvAnnotGenericSym =
\stEnvOutAnnot{\roleP}{\roleQ}{\stChoice{\stLab[k]}{\tyGround[k]}}$
 such that
\begin{gather}
  k \in I
  \label{eq:sf_global_send_M1_type_index}
  \\
\stEnvAssoc{ \gtWithCrashedRoles{\rolesC}{
      \gtCommTransit{\roleP}{\roleQ}{i \in I}{\gtLab[i]}{\tyGround[i]}{\gtGi[i]}{j}
}}{\stEnvi; \qEnvi}{\rolesR}
\label{eq:sf_global_send_type_association_new}
\\
\stEnv; \qEnv
      \stEnvMoveOutAnnot{\roleP}{\roleQ}{\stChoice{\stLab[k]}{\tyGround[k]}}
      \stEnvi; \qEnvi
\label{eq:sf_global_send_ctx_red}
\end{gather}
Using \eqref{eq:sf_global_send_type_configuration}, \eqref{eq:sf_global_send_M1_type_index}, and \inferrule{\iruleTCtxOut} in \Cref{fig:gtype:tc-red-rules}, we get
\begin{gather}
\stEnvi = \stEnvUpd{\stEnv}{\roleP}{\stT}
\label{eq:sf_global_send_new_context}
\\
\qEnvi = \stEnvUpd{\qEnv}{\roleP, \roleQ}{
        \stQCons{
          \stEnvApp{\qEnv}{\roleP, \roleQ}
        }{
          \stQMsg{\stLab}{\tyGround}
        }
        }
\label{eq:sf_global_send_new_Q}
\end{gather}
It follows that
 \begin{gather}
 \vdash \mpP : \stEnvApp{\stEnvi}{\roleP} \quad \quad (\text{by}~\eqref{eq:sf_global_send_new_context}, \eqref{eq:sf_global_send_subtyping}, \eqref{eq:sf_global_send_type_configuration_2}\text{ and }\inferrule{{t-sub}})
 \\
  \vdash \mpH[\roleP] : \stEnvApp{\qEnvi}{-, \roleP} \quad  (\text{by}~\eqref{eq:sf_global_send_new_Q}\text{ and }\eqref{eq:sf_global_send_type_out_Q})
  \\
  \vdash \mpQ : \stEnvApp{\stEnvi}{\roleQ} \quad \quad (\text{by}~\eqref{eq:sf_global_send_type_out_another}\text{ and }\eqref{eq:sf_global_send_new_context})
  \\
  \vdash \mpH[\roleQ] : \stEnvApp{\qEnvi}{-, \roleQ} \quad  (\text{by}~\eqref{eq:sf_global_send_new_Q}, \eqref{eq:sf_global_send_type_out_Q_another}\text{ and } 10, 11\text{ of \cref{lem:typeinversion}})
  \\
   \forall i \in I: \,\,\vdash \mpP_i:\stEnvApp{\stEnvi}{\roleP[i]} \quad (\text{by}~\eqref{eq:sf_global_send_M1_type}\text{ and }\eqref{eq:sf_global_send_new_context})
 \\
 \forall i \in I:  \,\,\vdash \mpH[i]: \stEnvApp{\qEnvi}{-, \roleP[i]} \quad (\text{by}~\eqref{eq:sf_global_send_M1_type_Q}\text{ and }\eqref{eq:sf_global_send_new_Q})
 \end{gather}
Therefore, by \inferrule{r-send} and \inferrule{{t-sess}}, we can conclude that there exists $\mpMi = \mpPart\roleP{\mpP}
\mpPar
\mpPart\roleP{\mpH[\roleP]}
\mpPar
\mpPart\roleQ{\mpQ}
\mpPar
\mpPart{\roleQ}{\mpH[\roleQ]}\cdot(\roleP,\mpLab(\val))
\mpPar
\mpM[1]$ and $\gtWithCrashedRoles{\rolesCi}{\gtGi}
   =
    \gtWithCrashedRoles{\rolesC}{
      \gtCommTransit{\roleP}{\roleQ}{i \in I}{\gtLab[i]}{\tyGround[i]}{\gtGi[i]}{j}
    }$ such that
   $ \gtWithCrashedRoles{\rolesC}{\gtG}
    \gtMove[
      \stEnvOutAnnotSmall{\roleP}{\roleQ}{\stChoice{\gtLab[j]}{\tyGround[j]}}
    ]{
      \rolesR
    }
   \gtWithCrashedRoles{\rolesCi}{\gtGi}$,
   $\mpM \;\redSend{\roleP}{\roleQ}{\mpLab}\;\mpMi$, and
   $ \gtWithCrashedRoles{\rolesCi}{\gtGi} \vdash \mpMi$.

\item Case \inferrule{\iruleGtMoveIn}: similar to case \inferrule{\iruleGtMoveOut} above, except that we proceed by \inferrule{\iruleGtMoveIn}, \inferrule{r-rcv}, inversion of \inferrule{{t-ext}} (4 of \Cref{lem:typeinversion}), and~\Cref{lem:substitution}.

\item Case \inferrule{\iruleGtMoveOrph}: by inversion of \inferrule{\iruleGtMoveOut}, we have
\begin{gather}
j \in I
 \label{eq:sf_global_orph_cond_1}
 \\
\gtLab[j] \neq \gtCrashLab
  \label{eq:sf_global_oprh_cond_2}
  \\
 \gtWithCrashedRoles{\rolesC}{\gtCommSmall{\roleP}{\roleQCrashed}{i \in
    I}{\gtLab[i]}{\tyGround[i]}{\gtGi[i]}}
    \gtMove[\stEnvOutAnnotSmall{\roleP}{\roleQ}{\stChoice{\gtLab[j]}{\tyGround[j]}}]{
      \rolesR
    }
    \gtWithCrashedRoles{\rolesC}{\gtGi[j]}
\label{eq:sf_global_orph_cond_3}
  \end{gather}
 We can assume that $\mpM$ is of the form
 \begin{gather}
 \mpPart\roleP{\procout{\roleQ}{\mpLab}{\mpE}{\mpP}}
\mpPar
\mpPart\roleP{\mpH[\roleP]}
\mpPar
\mpPart\roleQ{\mpCrash}
\mpPar
\mpPart\roleQ{\mpQUnavail}
\mpPar
\mpM[1]
\label{eq:sf_global_orph_M_form}
\\
\mpM[1] =  \prod_{i\in I} (\mpPart{\roleP[i]}{\mpP[i]} \mpPar
  \mpPart{\roleP[i]}{\mpH[i]})
\label{eq:sf_global_orph_M1_form}
\end{gather}
Then by~\eqref{eq:sf_ass_global_typing} and~\Cref{lem:typeinversion}, there
exists $\stEnv; \qEnv$ such that
\begin{gather}
\stEnvAssoc{\gtWithCrashedRoles{\rolesC}{\gtG}}{\stEnv; \qEnv}{\rolesR}
\label{eq:sf_global_orph_type_association}
\\
\vdash \procout\roleQ{\mpLab}{\mpE}{\mpP}:\stEnvApp{\stEnv}{\roleP}
\label{eq:sf_global_orph_type_out}
\\
 \vdash \mpH[\roleP]: \stEnvApp{\qEnv}{-, \roleP}
 \label{eq:sf_global_orph_type_out_Q}
 \\
 \vdash \mpCrash: \stEnvApp{\stEnv}{\roleQ}
 \label{eq:sf_global_orph_type_out_another}
 \\
  \vdash \mpQUnavail: \stEnvApp{\qEnv}{-, \roleQ}
 \label{eq:sf_global_orph_type_out_Q_another}
 \\
 \forall i \in I: \,\,\vdash \mpP_i:\stEnvApp{\stEnv}{\roleP[i]}
 \label{eq:sf_global_orph_M1_type}
 \\
 \forall i \in I:  \,\,\vdash \mpH[i]: \stEnvApp{\qEnv}{-, \roleP[i]}
  \label{eq:sf_global_orph_M1_type_Q}
\end{gather}
By~\eqref{eq:sf_global_orph_type_out} and 3 of~\cref{lem:typeinversion}, we have that
\begin{gather}
\stEnvApp{\stEnv}{\roleP} = \stOut{\roleQ}{\mpLab}{\tyGround}\stSeq{\stT}
\label{eq:sf_global_orph_type_configuration}
\\
\stT[1] \stSub \stT
\label{eq:sf_global_orph_subtyping}
\\
 \vdash \mpP : \stT[1]
 \label{eq:sf_global_orph_type_configuration_2}
\end{gather}
From~\eqref{eq:sf_global_orph_cond_3}, by~\cref{thm:gtype:proj-sound}, there are $\stEnvi; \qEnvi$ and
$\stEnvAnnotGenericSym =
\stEnvOutAnnot{\roleP}{\roleQ}{\stChoice{\stLab[k]}{\tyGround[k]}}$
 such that
\begin{gather}
  k \in I
  \label{eq:sf_global_orph_M1_type_index}
  \\
\stEnvAssoc{ \gtWithCrashedRoles{\rolesC}{
     \gtGi[j]
}}{\stEnvi; \qEnvi}{\rolesR}
\label{eq:sf_global_orph_type_association_new}
\\
\stEnv; \qEnv
      \stEnvMoveOutAnnot{\roleP}{\roleQ}{\stChoice{\stLab[k]}{\tyGround[k]}}
      \stEnvi; \qEnvi
\label{eq:sf_global_orph_ctx_red}
\end{gather}
Using \eqref{eq:sf_global_orph_type_configuration}, \eqref{eq:sf_global_orph_M1_type_index}, and \inferrule{\iruleTCtxOut} in \Cref{fig:gtype:tc-red-rules}, we get
\begin{gather}
\stEnvi = \stEnvUpd{\stEnv}{\roleP}{\stT}
\label{eq:sf_global_orph_new_context}
\\
\qEnvi = \stEnvUpd{\qEnv}{\roleP, \roleQ}{
        \stQCons{
          \stEnvApp{\qEnv}{\roleP, \roleQ}
        }{
          \stQMsg{\stLab}{\tyGround}
        }
        }
\label{eq:sf_global_orph_new_Q}
\end{gather}
It follows that
 \begin{gather}
 \vdash \mpP : \stEnvApp{\stEnvi}{\roleP} \quad \quad (\text{by}~\eqref{eq:sf_global_orph_new_context}, \eqref{eq:sf_global_orph_subtyping}, \eqref{eq:sf_global_orph_type_configuration_2}\text{ and }\inferrule{{t-sub}})
 \\
  \vdash \mpH[\roleP] : \stEnvApp{\qEnvi}{-, \roleP} \quad  (\text{by}~\eqref{eq:sf_global_orph_new_Q}\text{ and }\eqref{eq:sf_global_orph_type_out_Q})
  \\
  \vdash \mpCrash : \stEnvApp{\stEnvi}{\roleQ} \quad \quad (\text{by}~\eqref{eq:sf_global_orph_type_out_another}\text{ and }\eqref{eq:sf_global_orph_new_context})
  \\
  \vdash \mpQUnavail : \stEnvApp{\qEnvi}{-, \roleQ} \quad  (\text{by}~\eqref{eq:sf_global_orph_new_Q}\text{ and }\eqref{eq:sf_global_orph_type_out_Q_another})
  \\
   \forall i \in I: \,\,\vdash \mpP_i:\stEnvApp{\stEnvi}{\roleP[i]} \quad (\text{by}~\eqref{eq:sf_global_orph_M1_type}\text{ and }\eqref{eq:sf_global_orph_new_context})
 \\
 \forall i \in I:  \,\,\vdash \mpH[i]: \stEnvApp{\qEnvi}{-, \roleP[i]} \quad (\text{by}~\eqref{eq:sf_global_orph_M1_type_Q}\text{ and }\eqref{eq:sf_global_orph_new_Q})
 \end{gather}
Therefore, by \inferrule{r-send-\mpCrash} and \inferrule{{t-sess}}, we can conclude that there exists $\mpMi = \mpPart\roleP{\mpP}
\mpPar
\mpPart\roleP{\mpH[\roleP]}
\mpPar
\mpPart\roleQ{\mpCrash}
\mpPar
\mpPart{\roleQ}{\mpQUnavail}
\mpPar
\mpM[1]$ and $\gtWithCrashedRoles{\rolesCi}{\gtGi}
   =
    \gtWithCrashedRoles{\rolesC}{
    \gtGi[j]
    }$ such that
   $ \gtWithCrashedRoles{\rolesC}{\gtG}
    \gtMove[
      \stEnvOutAnnotSmall{\roleP}{\roleQ}{\stChoice{\gtLab[j]}{\tyGround[j]}}
    ]{
      \rolesR
    }
   \gtWithCrashedRoles{\rolesCi}{\gtGi}$,
   $\mpM \;\redSend{\roleP}{\roleQ}{\mpLab}\;\mpMi$, and
   $ \gtWithCrashedRoles{\rolesCi}{\gtGi} \vdash \mpMi$.

\item Case \inferrule{\iruleGtMoveCrDe}: by inversion of \inferrule{\iruleGtMoveCrDe}, we have
\begin{gather}
j \in I
 \label{eq:sf_global_det_cond_1}
 \\
\gtLab[j] = \gtCrashLab
  \label{eq:sf_global_det_cond_2}
  \\
 \gtWithCrashedRoles{\rolesC}{
      \gtCommTransit{\roleQCrashed}{\roleP}{i \in I}{\gtLab[i]}{\tyGround[i]}{\gtGi[i]}{j}
    }
    \gtMove[\ltsCrDe{\mpS}{\roleP}{\roleQ}]{\rolesR}
    \gtWithCrashedRoles{\rolesC}{\gtGi[j]}
\label{eq:sf_global_det_cond_3}
  \end{gather}
 We can assume that $\mpM$ is of the form
 \begin{gather}
 \mpPart\roleP{\sum_{i\in I} \procin\roleQ{\mpLab_i(\mpx_i)}\mpP_i}
\mpPar
\mpPart\roleP{\mpH[\roleP]}
\mpPar
\mpPart\roleQ\mpCrash
\mpPar
\mpPart\roleQ\mpQUnavail
\mpPar
\mpM[1]
\label{eq:sf_global_det_M_form}
\\
\mpM[1] =  \prod_{i\in I} (\mpPart{\roleP[i]}{\mpP[i]} \mpPar
  \mpPart{\roleP[i]}{\mpH[i]})
\label{eq:sf_global_det_M1_form}
\\
k \in I
\label{eq:sf_global_det_M_form_index}
\\
 \mpLab[k] = \mpCrashLab
 \label{eq:sf_global_det_M_form_crash}
 \\
 \nexists \mpLab, \mpV: (\roleQ,
\mpLab(\mpV)) \in \mpH[\roleP]
\label{eq:sf_global_det_M_form_empty_Q}
\end{gather}
Then by~\eqref{eq:sf_ass_global_typing} and~\Cref{lem:typeinversion}, there
exists $\stEnv; \qEnv$ such that
\begin{gather}
\stEnvAssoc{\gtWithCrashedRoles{\rolesC}{\gtG}}{\stEnv; \qEnv}{\rolesR}
\label{eq:sf_global_det_type_association}
\\
\vdash \sum_{i\in I}\procin{\roleQ}{\mpLab_i(\mpx_i)}{\mpP_i}:\stEnvApp{\stEnv}{\roleP}
\label{eq:sf_global_det_type_out}
\\
 \vdash \mpH[\roleP]: \stEnvApp{\qEnv}{-, \roleP}
 \label{eq:sf_global_det_type_out_Q}
 \\
 \vdash \mpCrash: \stEnvApp{\stEnv}{\roleQ}
 \label{eq:sf_global_det_type_out_another}
 \\
  \vdash \mpQUnavail: \stEnvApp{\qEnv}{-, \roleQ}
 \label{eq:sf_global_det_type_out_Q_another}
 \\
 \forall i \in I: \,\,\vdash \mpP_i:\stEnvApp{\stEnv}{\roleP[i]}
 \label{eq:sf_global_det_M1_type}
 \\
 \forall i \in I:  \,\,\vdash \mpH[i]: \stEnvApp{\qEnv}{-, \roleP[i]}
  \label{eq:sf_global_det_M1_type_Q}
\end{gather}
It follows directly that
\begin{gather}
\stEnvApp{\stEnv}{\roleQ} = \stStop
\label{eq:sf_global_det_stop}
\\
 \stEnvApp{\qEnv}{-, \roleQ} = \stQUnavail
 \label{eq:sf_global_det_unavail}
\end{gather}

By~\eqref{eq:sf_global_det_type_out}, 4 of~\cref{lem:typeinversion}, and \inferrule{\iruleStSubIn}, we have that
\begin{gather}
\stEnvApp{\stEnv}{\roleP} = \stExtSum{\roleQ}{i \in J}{\stChoice{\stLab[i]}{\tyGround[i]} \stSeq \stTi[i]}
\label{eq:sf_global_det_type_configuration}
\\
J \subseteq I
\label{eq:sf_global_det_subset}
\\
\forall i \in J:  \stT[i] \stSub \stTi[i]
\label{eq:sf_global_det_subtyping}
\\
 \setcomp{\stLab[l]}{l \in I} \neq \setenum{\stCrashLab}
 \label{eq:sf_global_det_not_crash}
\\
 \nexists j \in I \setminus J: \stLab[j] = \stCrashLab
 \label{eq:sf_global_det_not_crash_2}
 \\
 \forall i \in I: x_i : \tyGround[i] \vdash \mpP[i] : \stT[i]
 \label{eq:sf_global_det_type_configuration_2}
\end{gather}
From \eqref{eq:sf_global_det_not_crash_2} and \eqref{eq:sf_global_det_M_form_crash}, we get
\begin{align}
   k \in J
   \label{eq:sf_global_det_k_in_J_1}
\end{align}
By \eqref{eq:sf_global_det_type_out_Q} and \eqref{eq:sf_global_det_M_form_empty_Q}, we also know
\begin{align}
 \stEnvApp{\qEnv}{\roleQ, \roleP} = \stQEmpty
 \label{eq:sf_global_det_empty_Q}
\end{align}

From~\eqref{eq:sf_global_det_cond_3}, by~\cref{thm:gtype:proj-sound}, there is $\stEnvi; \qEnvi$ such that
\begin{gather}
\stEnvAssoc{ \gtWithCrashedRoles{\rolesC}{
     \gtGi[j]
}}{\stEnvi; \qEnvi}{\rolesR}
\label{eq:sf_global_det_type_association_new}
\\
\stEnv; \qEnv
      \stEnvMoveAnnot{\ltsCrDe{\mpS}{\roleP}{\roleQ}}
      \stEnvi; \qEnvi
\label{eq:sf_global_det_ctx_red}
\end{gather}

Using \eqref{eq:sf_global_det_type_configuration}, \eqref{eq:sf_global_det_stop}, \eqref{eq:sf_global_det_k_in_J_1},
\eqref{eq:sf_global_det_M_form_crash}, \eqref{eq:sf_global_det_empty_Q}, and \inferrule{\iruleTCtxCrashDetect} in \Cref{fig:gtype:tc-red-rules}, we get
\begin{gather}
\stEnvi =  \stEnvUpd{\stEnv}{\roleP}{\stTi[k]}
\label{eq:sf_global_det_new_context}
\\
\qEnvi = \qEnv
\label{eq:sf_global_det_new_Q}
\end{gather}
It follows that
 \begin{gather}
  \vdash \mpP[k] : \stEnvApp{\stEnvi}{\roleP} \quad \quad (\text{by}~\eqref{eq:sf_global_det_new_context}, \eqref{eq:sf_global_det_subtyping}, \eqref{eq:sf_global_det_type_configuration_2}\text{ and }\inferrule{{t-sub}})
 \\
  \vdash \mpH[\roleP] : \stEnvApp{\qEnvi}{-, \roleP} \quad  (\text{by}~\eqref{eq:sf_global_det_new_Q}\text{ and }\eqref{eq:sf_global_det_type_out_Q})
  \\
  \vdash \mpCrash : \stEnvApp{\stEnvi}{\roleQ} \quad \quad (\text{by}~\eqref{eq:sf_global_det_type_out_another},\eqref{eq:sf_global_det_stop}\text{ and }\eqref{eq:sf_global_det_new_context})
  \\
  \vdash \mpQUnavail: \stEnvApp{\qEnvi}{-, \roleQ} \quad  (\text{by}~\eqref{eq:sf_global_det_new_Q}, \eqref{eq:sf_global_det_type_out_Q_another}\text{ and }\eqref{eq:sf_global_det_unavail})
  \\
   \forall i \in I: \,\,\vdash \mpP_i:\stEnvApp{\stEnvi}{\roleP[i]} \quad (\text{by}~\eqref{eq:sf_global_det_M1_type}\text{ and }\eqref{eq:sf_global_det_new_context})
 \\
 \forall i \in I:  \,\,\vdash \mpH[i]: \stEnvApp{\qEnvi}{-, \roleP[i]} \quad (\text{by}~\eqref{eq:sf_global_det_M1_type_Q}\text{ and }\eqref{eq:sf_global_det_new_Q})
 \end{gather}
Therefore, by \inferrule{r-rcv-$\odot$} and \inferrule{{t-sess}}, we can conclude that there exists $\mpMi = \mpPart\roleP \mpP_k
\mpPar
\mpPart\roleP{\mpH_{\roleP}}
\mpPar
\mpPart\roleQ\mpCrash
\mpPar
\mpPart\roleQ\mpQUnavail
\mpPar
\mpM[1]$ and $\gtWithCrashedRoles{\rolesCi}{\gtGi}
   =
    \gtWithCrashedRoles{\rolesC}{
    \gtGi[j]
    }$ such that
   $ \gtWithCrashedRoles{\rolesC}{\gtG}
   \gtMove[\ltsCrDe{\mpS}{\roleP}{\roleQ}]{\rolesR}
      \gtWithCrashedRoles{\rolesCi}{\gtGi}$,
   $\mpM \;\redRecv{\roleP}{\roleQ}{\mpLab_k}\;\mpMi$, and
   $ \gtWithCrashedRoles{\rolesCi}{\gtGi} \vdash \mpMi$.
\end{itemize}
The rest of the cases %
 follow directly by inductive hypothesis. \qedhere
\end{proof}

\lemSessionDF*
\begin{proof}
Corollary of \cref{def:mpst-env-deadlock-free}, \cref{def:session_df}, \cref{lem:ext-proj}, \cref{lem:sr}, and \cref{lem:sf}. 
\end{proof}

\lemSessionLive*
\begin{proof}
Corollary of \cref{def:mpst-env-live}, \cref{def:session_live}, \cref{lem:ext-proj}, \cref{lem:sr}, and \cref{lem:sf}. 
\end{proof}

\section{Additional Details on Channel Generation}\label{sec:appendix:channel-gen}
We extend global and local types with annotations on communication statements,
and extend projection with annotation support.
Annotations $\alpha$ have the form of either natural numbers or sets of natural
numbers. We call annotations of the form $\alpha \in \setNat$
\emph{atomic-annotations}.
  \[
    \begin{array}{r@{\quad}c@{\quad}l@{\quad}l}
      \gtG & \bnfdef &
        \gtCommAnn{\roleP}{\roleQMaybeCrashed}{\highlight{\alpha}}{i \in
        I}{\gtLab[i]}{\tyGround[i]}{\gtG[i]} \bnfsep \cdots
        \\
      \stS, \stT
        & \bnfdef & \stExtSumAnn{\roleP}{\highlight{\alpha}}{i \in I}{\stChoice{\stLab[i]}{\tyGround[i]} \stSeq \stT[i]}
          \bnfsep
        \stIntSumAnn{\roleP}{\highlight{\alpha}}{i \in I}{\stChoice{\stLab[i]}{\tyGround[i]} \stSeq \stT[i]}
          \bnfsep \cdots
          \\
        \alpha,\beta & \in & \setNat \cup \mathcal{P}(\setNat)
        \\
    \end{array}
  \]
The projection operator is similarly extended with annotation support. Merging
annotations produces a set-annotation containing all merged atomic-annotations.
\[
  \begin{array}{c}
    \stExtSumAnn{\roleP}{\alpha}{i \in I}{\stChoice{\stLab[i]}{\tyGround[i]} \stSeq \stSi[i]}%
    \!\stBinMerge\!%
    \stExtSumAnn{\roleP}{\beta}{\!j \in J}{\stChoice{\stLab[j]}{\tyGround[j]} \stSeq \stTi[j]}%
   \\
    \;=\;%
    \stExtSumAnn{\roleP}{\alpha \stBinMerge \beta}{k \in I \cap J}{\stChoice{\stLab[k]}{\tyGround[k]} \stSeq%
      (\stSi[k] \!\stBinMerge\! \stTi[k])%
    }%
    \stExtC%
    \stExtSumAnn{\roleP}{\alpha \stBinMerge \beta}{i \in I \setminus J}{\stChoice{\stLab[i]}{\tyGround[i]} \stSeq \stSi[i]}%
    \stExtC%
    \stExtSumAnn{\roleP}{\alpha \stBinMerge \beta}{\!j \in J \setminus I}{\stChoice{\stLab[j]}{\tyGround[j]} \stSeq \stTi[j]}%
    \\[3mm]%
    \stIntSumAnn{\roleP}{\alpha}{i \in I}{\stChoice{\stLab[i]}{\tyGround[i]} \stSeq \stSi[i]}%
    \,\stBinMerge\,%
    \stIntSumAnn{\roleP}{\beta}{i \in I}{\stChoice{\stLab[i]}{\tyGround[i]} \stSeq \stTi[i]}%
    \;=\;%
    \stIntSumAnn{\roleP}{\alpha \stBinMerge \beta}{i \in I}{\stChoice{\stLab[i]}{\tyGround[i]} \stSeq (\stSi[i]
    \stBinMerge \stTi[i])}%
    \\[3mm]
    \stFmt{
      n \stBinMerge m = \{n,m\}
      \qquad
      N \stBinMerge m = N \cup \{m\}
      \qquad
      n \stBinMerge M = \{n\} \cup M
      \qquad
      N \stBinMerge M = N \cup M
    }
  \end{array}
\]

In order to generate channels, we first annotate each interaction statement in a
given global type with a unique identifier.
To illustrate the process we consider the annotated global type $\gtG$.
\begin{equation} \label{eq:impl:changen:gt}
  \gtG = \gtCommRawAnn{\roleP}{\roleQ}{0}{
    \begin{array}{l}
    \gtCommChoice{\gtLabFmt{v}}{}{
      \gtCommSingleAnn{\roleP}{\roleR}{1}{\gtLabFmt{x}}{}{\gtEnd}
    }
    \\
    \gtCommChoice{\gtLabFmt{w}}{}{
      \gtCommSingleAnn{\roleP}{\roleR}{2}{\gtLabFmt{y}}{}{\gtEnd}
    }
    \end{array}
  }
\end{equation}
The global type is then projected onto all its participating roles. For $\gtG$, projecting onto roles $\roleP$ and $\roleQ$ is trivial:
\[
  \begin{array}{rclcrcl}
  \gtProj{\gtG}{\roleP} &=&
    \stIntSumAnn{\roleQ}{0}{}{
      \begin{array}{l}
        \stChoice{\stLabFmt{v}}{} \stSeq
          \stOutAnn{\roleR}{1}{\stLabFmt{x}}{} \stSeq
          \stEnd
        \\
        \stChoice{\stLabFmt{w}}{} \stSeq
          \stOutAnn{\roleR}{2}{\stLabFmt{y}}{} \stSeq
          \stEnd
      \end{array}
    }
  &\qquad&
  \gtProj{\gtG}{\roleQ} &=&
    \stExtSumAnn{\roleP}{0}{}{
      \begin{array}{l}
        \stChoice{\stLabFmt{v}}{} \stSeq
          \stEnd
        \\
        \stChoice{\stLabFmt{w}}{} \stSeq
          \stEnd
      \end{array}
    }
\end{array}
\]
Here, the annotations are propagated \emph{verbatim} from the global type.
Projections that require merging result in set-annotations. For example,
projecting \gtG onto $\roleR$:
\[
  \begin{array}{rcl}
  \gtProj{\gtG}{\roleR} &=&
  \stExtSumAnn{\roleP}{\{1,2\}}{}{
    \begin{array}{l}
      \stChoice{\stLabFmt{x}}{} \stSeq \stEnd
      \\
      \stChoice{\stLabFmt{y}}{} \stSeq \stEnd
    \end{array}
  }
  \end{array}
\]
Such set-annotations are used to derive fresh channels to ensure that both
sender and recipient are communicating over the same channel.

For some $\gtG$, let $\stEnv$ be the typing context produced by projecting
$\gtG$ on all roles.
Let $\annSet$ be the set of all annotations in $\stEnv$, $\annSetVar \subseteq
\annSet$ be the set of all atomic-annotations (\ie $\annSetVar = \annSet
\cap \setNat$), and $\annSetSet \subseteq \annSet$ be the set of all
set-annotations (\ie $\annSetSet = \annSet \cap P(\setNat)$).
For each typing context $\stEnv$, we derive an injective function
$\annSubst : \annSet \to \setNat$
that maps set-annotations to fresh identifiers; specifically,
$\forall x \in \annSet, \exists y \in \setNat, \annSubst(x) = y \implies
  x = y \lor y \in \setNat \setminus \annSetVar$.
The derived $\annSubst$ is then applied to each annotation in $\stEnv$. This
removes all set-annotations from each local type within $\stEnv$.
Applied to our exemplar \gtG (\Cref{eq:impl:changen:gt}), we define the following $\annSubst$ and the concomitant local types.
\[
  \begin{array}{c}
    \begin{array}{rcl}
    \annSubst(0) &=& 0
    \\
    \annSubst(1) &=& 1
    \\
    \annSubst(2) &=& 2
    \\
    \annSubst(\{1,2\}) &=& 3
    \end{array}

    \qquad\qquad\qquad

    \begin{array}{rcl}
    \gtProj{\gtG}{\roleP} &=&
      \stIntSumAnn{\roleQ}{0}{}{
        \begin{array}{l}
          \stChoice{\stLabFmt{v}}{} \stSeq
            \stOutAnn{\roleR}{1}{\stLabFmt{x}}{} \stSeq
            \stEnd
          \\
          \stChoice{\stLabFmt{w}}{} \stSeq
            \stOutAnn{\roleR}{2}{\stLabFmt{y}}{} \stSeq
            \stEnd
        \end{array}
      }
    \\
    \gtProj{\gtG}{\roleQ} &=&
      \stExtSumAnn{\roleP}{0}{}{
        \begin{array}{l}
          \stChoice{\stLabFmt{v}}{} \stSeq
            \stEnd
          \\
          \stChoice{\stLabFmt{w}}{} \stSeq
            \stEnd
        \end{array}
      }
    \\
    \gtProj{\gtG}{\roleR} &=&
      \stExtSumAnn{\roleP}{3}{}{
        \begin{array}{l}
          \stChoice{\stLabFmt{x}}{} \stSeq \stEnd
          \\
          \stChoice{\stLabFmt{y}}{} \stSeq \stEnd
        \end{array}
      }
    \end{array}
  \end{array}
\]
We generate the type upper bounds that are needed for each local type
declaration from the corresponding annotated local type. In our above example,
the local type declarations are parameterised thus:
\begin{lstlisting}[language=Scala, gobble=2]
  type P[C0 <: OutChan[V | W], C1 <: OutChan[X], C2 <: OutChan[Y]] = ...
  type Q[C0 <: InChan[V | W]] = ...
  type R[C3 <: InChan[X | Y]] = ...
\end{lstlisting}

Each channel is declared and passed to role-implementing functions in the body
of the main function. Channel declarations are derived from the set of
annotations excluding those identifiers that have been merged during projection,
\ie $\annSet' = \{\annSubst(x') | x' \in \annSet \cap \bigcup\annSetSet\}$.
Since role-implementing functions are parameterised by channels, each argument
corresponds to an annotation in the local type projected by the respective role.
We derive a surjective function, $\annArg : \annSet \to \setNat$, such that
$\forall x \in \annSet, \exists y \in \setNat, \annArg(x) = y$
iff
$x \in \annSetVar \setminus \bigcup\annSetSet \implies x = y$,
$x \in \bigcup\annSetSet \implies
  \exists X \in \annSetSet, x \in X \land \annSubst(X) = x $, and
$x \in \annSetSet \implies \annSubst(x) = y$
hold. Intuitively, we follow $\annSubst$; excepting in cases where an
atomic-annotation, $x \in \annSet$, can be found in a set-annotation, $\exists X
\in \annSet, x \in X$. In such cases, $\annArg(x) = \annSubst(X)$.
Each argument is derived through application of $\annArg$ on its corresponding
annotation.
In our example, the resulting main object therefore includes the following
function.
\begin{lstlisting}[language=Scala, gobble=2]
  def main(args : Array[String]) = {
    var c0 = Channel[V | W]()
    var c3 = Channel[X | Y]()

    eval(par(p(c0, c3, c3), q(c0), r(c3)))
  }
\end{lstlisting}
Here, \texttt{Channel} constructs a binary channel that inherits from both
\texttt{InChan} and \texttt{OutChan} traits. We leverage Scala's subtyping
system to ensure that, whilst declared channel types may permit more labels than
the corresponding parameter in the local type declaration, the type system
permits only those labels specified by the latter. An example can be seen here
in passing \texttt{c3} to \texttt{p} twice: although \texttt{c3} specifies a
channel whose type is the union of \texttt{X} and \texttt{Y}, the type system
will only permit sending \texttt{X} as \texttt{C1} or \texttt{Y} as \texttt{C2}.

\section{Additional Details on Evaluation}
\label{sec:appendix:eval}

We evaluate \theTool on seven examples derived from the session type
literature.
Of these, five are standard protocols that we extend with crash-handling
behaviour. The remaining examples comprise a distributed logging protocol
and the circuit breaker pattern~\cite{ECOOP22AffineMPST}.
Our examples demonstrate that \theTool enables expression of both standard
protocols from session type literature and protocols representing real-world
distributed patterns, in addition to extensions thereof with crash-handling
behaviours. Our examples illustrate two such behavioural patterns:
failover and graceful failure.
Our Circuit Breaker in particular illustrates that our
approach supports different crash-handling behaviour extensions for the same
reliable protocol.
Our runtime measurements indicate that generating skeleton
code from a given \Scribble protocol is inexpensive, with no example taking
longer than 3ms, and that generation times scale with both input size and the
the degree of branching exhibited by the protocol.

We give a description of our examples in \Cref{sec:eval:examples}, with full
metrics given in \Cref{tab:eval:table}. Each example has at least one (partially) unreliable variant.
Variants are primarily distinguished by their reliability assumptions and, for
clarity, will be referred to by their identifier given in the same figure
(\eg $(a)$, $(l)$, and $(s)$).
For each variant, we give an indication of its size and shape via the number of
communications and $\gtCrashLab$ labels in the protocol, as well as the length
of the longest continuation in the global type. We additionally report the
number of channels and functions generated by \theTool. The former reflects the
true number of communications following projection, and the latter reflects the
degree of branching exhibited by the given protocol.
\Cref{fig:eval:types} gives an exemplar global type for each example.
In \Cref{sec:eval:times}, we discuss generation times for each example. Full
timing results are given in \Cref{fig:eval:results}.

\begin{figure}
  \footnotesize
  \begin{tabular}{cL}
    \toprule
    Example & \text{Global Type} \\
    \midrule
    \exampleName{PingPong} & \gtG[b] =
      \gtRec{\gtFmt{\mathbf{t_{0}}}}{
        \gtCommRaw{\roleFmt{P}}{\roleFmt{Q}}{
        \begin{array}{l}
        \gtCommChoice{\gtMsgFmt{ping}}{}{
        \gtCommRaw{\roleFmt{Q}}{\roleFmt{P}}{
        \begin{array}{l}
        \gtCommChoice{\gtMsgFmt{pong}}{}{
        \gtFmt{\mathbf{t_{0}}}
        }\gtFmt{,}\;
        \gtCommChoice{\gtMsgFmt{crash}}{}{
        }\\
        \end{array}
        }
        }\gtFmt{,}\;
        \gtCommChoice{\gtMsgFmt{crash}}{}{
        }\\
        \end{array}
        }
        }
    \\\midrule
    \multirow{2}{*}{\exampleName{Adder}} & \gtG[d] =
      \gtRec{\gtFmt{\mathbf{t_{0}}}}{
        \gtCommRaw{\roleFmt{C}}{\roleFmt{S}}{
        \begin{array}{l}
        \gtCommChoice{\gtMsgFmt{hello}}{}{
        \gtCommRaw{\roleFmt{C}}{\roleFmt{S}}{
        \begin{array}{l}
        \gtCommChoice{\gtMsgFmt{add}}{\tyInt}{
        \gtGj[2]
        }\\
        \gtCommChoice{\gtMsgFmt{quit}}{}{
        \gtCommRaw{\roleFmt{S}}{\roleFmt{C}}{
        \begin{array}{l}
        \gtCommChoice{\gtMsgFmt{quit}}{}{
        }\gtFmt{,}\;
        \gtCommChoice{\gtMsgFmt{crash}}{}{
        }\\
        \end{array}
        }
        }\\
        \gtCommChoice{\gtMsgFmt{crash}}{}{
        }\\
        \end{array}
        }
        }\\
        \gtCommChoice{\gtMsgFmt{crash}}{}{
        }\\
        \end{array}
        }
        }
    \\
    & \gtGj[d] =
    \gtCommRaw{\roleFmt{C}}{\roleFmt{S}}{
      \begin{array}{l}
      \gtCommChoice{\gtMsgFmt{add}}{\tyInt}{
        \gtCommRaw{\roleFmt{S}}{\roleFmt{C}}{
          \begin{array}{l}
          \gtCommChoice{\gtMsgFmt{ok}}{\tyInt}{
          \gtFmt{\mathbf{t_{0}}}
          }\gtFmt{,}\;
          \gtCommChoice{\gtMsgFmt{crash}}{}{
          }\\
          \end{array}
          }
      }\gtFmt{,}\;
      \gtCommChoice{\gtMsgFmt{crash}}{}{
      }\\
      \end{array}
      }
    \\\midrule
    \multirow{5}{*}{\exampleName{TwoBuyer}} & \gtG[f] =
    \gtCommRaw{\roleP}{\roleR}{
      \begin{array}{l}
        \gtCommChoice{\gtMsgFmt{title}}{\tyString}{
          \gtCommSingle{\roleR}{\roleP}{\gtMsgFmt{quote}}{\tyReal}{
            \gtCommSingle{\roleR}{\roleQ}{\gtMsgFmt{quote}}{\tyReal}{
              \gtGj[f]
            }
          }
          }\\
        \gtCommChoice{\gtMsgFmt{crash}}{}{
          \gtCommSingle{\roleFmt{R}}{\roleQ}{\gtMsgFmt{quit}}{}{
          }
          }
      \end{array}
      }
    \\
    & \gtGj[f] =
    \gtCommRaw{\roleP}{\roleR}{
      \begin{array}{l}
        \gtCommChoice{\gtMsgFmt{split}}{\tyReal}{
          \gtCommRaw{\roleQ}{\roleR}{
          \begin{array}{l}
            \gtCommChoice{\gtMsgFmt{split}}{\tyReal}{
                \gtGjj[f]
              }\\
            \gtCommChoice{\gtMsgFmt{crash}}{}{
              \gtCommSingle{\roleR}{\roleP}{\gtMsgFmt{quit}}{}{
              }
              }
          \end{array}
          }
          }\\
        \gtCommChoice{\gtMsgFmt{crash}}{}{
          \gtCommRaw{\roleQ}{\roleR}{
          \begin{array}{l}
            \gtCommChoice{\gtMsgFmt{split}}{\tyReal}{
              \gtCommSingle{\roleR}{\roleQ}{\gtMsgFmt{quit}}{}{
              }
              }\\
            \gtCommChoice{\gtMsgFmt{crash}}{}{
              }
          \end{array}
          }
          }
      \end{array}
      }
    \\
    & \gtGjj[f] =
    \gtCommRaw{\roleR}{\roleP}{
      \begin{array}{l}
        \gtCommChoice{\gtMsgFmt{quit}}{}{
          \gtCommSingle{\roleR}{\roleQ}{\gtMsgFmt{quit}}{}{
          }
          }\\
        \gtCommChoice{\gtMsgFmt{ok}}{}{
          \gtCommSingle{\roleR}{\roleQ}{\gtMsgFmt{ok}}{}{
            \gtCommRaw{\roleQ}{\roleR}{
            \begin{array}{l}
              \gtCommChoice{\gtMsgFmt{addr}}{\tyString}{
                  \gtGiii[f]
                }\\
              \gtCommChoice{\gtMsgFmt{crash}}{}{
                \gtCommSingle{\roleR}{\roleP}{\gtMsgFmt{recaddr}}{}{
                  \gtGiv[f]
                }
                }
            \end{array}
            }
          }
          }
      \end{array}
      }
    \\
    & \gtGiii[f] =
    \gtCommSingle{\roleR}{\roleP}{\gtMsgFmt{date}}{\tyString}{
      \gtCommSingle{\roleR}{\roleQ}{\gtMsgFmt{date}}{\tyString}{
      }
    }
    \\
    & \gtGiv[f] =
    \gtCommRaw{\roleP}{\roleR}{
      \begin{array}{l}
        \gtCommChoice{\gtMsgFmt{addr}}{\tyString}{
          \gtCommSingle{\roleR}{\roleP}{\gtMsgFmt{date}}{\tyString}{
          }
        }\gtFmt{,}\;
        \gtCommChoice{\gtMsgFmt{crash}}{}{
          }
      \end{array}
      }
    \\
    \midrule
    \multirow{6}{*}{\exampleName{OAuth}} & \gtG[h] =
    \gtCommRaw{\roleC}{\roleS}{
      \begin{array}{l}
        \gtCommChoice{\gtMsgFmt{start}}{\tyInt}{
          \gtCommSingle{\roleS}{\roleC}{\gtMsgFmt{redir}}{\tyInt}{
            \gtGj[h]
          }
          }\\
        \gtCommChoice{\gtMsgFmt{crash}}{}{
          \gtCommSingle{\roleC}{\roleA}{\gtMsgFmt{crash}}{}{
            \gtCommSingle{\roleS}{\roleA}{\gtMsgFmt{quit}}{}{
            }
          }
          }
      \end{array}
      }
    \\
    & \gtGj[h] =
    \gtCommRaw{\roleC}{\roleA}{
      \begin{array}{l}
        \gtCommChoice{\gtMsgFmt{login}}{\tyInt}{
          \gtCommSingle{\roleA}{\roleC}{\gtMsgFmt{auth}}{\tyInt}{
            \gtGjj[h]
          }
          }\\
        \gtCommChoice{\gtMsgFmt{crash}}{}{
          \gtCommSingle{\roleC}{\roleS}{\gtMsgFmt{crash}}{}{
            \gtCommSingle{\roleS}{\roleA}{\gtMsgFmt{quit}}{}{
            }
          }
          }
      \end{array}
      }
    \\
    & \gtGjj[h] = 
    \gtCommRaw{\roleC}{\roleA}{
      \begin{array}{l}
        \gtCommChoice{\gtMsgFmt{passwd}}{\tyInt}{
          \gtCommRaw{\roleA}{\roleC}{
          \begin{array}{l}
            \gtCommChoice{\gtMsgFmt{ko}}{\tyInt}{
                \gtGiii[h]
            }\gtFmt{,}\;
            \gtCommChoice{\gtMsgFmt{ok}}{\tyInt}{
                \gtGiv[h]
              }
          \end{array}
          }
          }\\
        \gtCommChoice{\gtMsgFmt{crash}}{}{
          \gtCommSingle{\roleC}{\roleS}{\gtMsgFmt{crash}}{}{
            \gtCommSingle{\roleS}{\roleA}{\gtMsgFmt{quit}}{}{
            }
          }
          }
      \end{array}
      }
    \\
    & \gtGiii[h] =
    \gtCommRaw{\roleC}{\roleS}{
      \begin{array}{l}
        \gtCommChoice{\gtMsgFmt{ko}}{\tyInt}{
          \gtCommSingle{\roleS}{\roleC}{\gtMsgFmt{recvd}}{\tyInt}{
            \gtCommSingle{\roleS}{\roleA}{\gtMsgFmt{quit}}{}{
            }
          }
          }\\
        \gtCommChoice{\gtMsgFmt{crash}}{}{
          \gtCommSingle{\roleS}{\roleA}{\gtMsgFmt{quit}}{}{
          }
          }
      \end{array}
      }
    \\
    & \gtGiv[h] =
    \gtCommRaw{\roleC}{\roleS}{
      \begin{array}{l}
        \gtCommChoice{\gtMsgFmt{ok}}{\tyInt}{
          \gtCommSingle{\roleS}{\roleA}{\gtMsgFmt{get}}{\stFmtC{Token}}{
            \gtGv[h]
          }
          }\\
        \gtCommChoice{\gtMsgFmt{crash}}{}{
          \gtCommSingle{\roleS}{\roleA}{\gtMsgFmt{quit}}{}{
          }
          }
      \end{array}
      }
    \\
    & \gtGv[h] =
    \gtCommSingle{\roleA}{\roleS}{\gtMsgFmt{put}}{\stFmtC{Token}}{
      \gtCommSingle{\roleS}{\roleC}{\gtMsgFmt{put}}{\stFmtC{Token}}{
      }
    }
    \\\midrule
    \exampleName{TravelAgency} & \gtG[l] =
    \gtRec{\gtFmt{\mathbf{t_{0}}}}{
  \gtCommRaw{\roleFmt{c}}{\roleFmt{a}}{
  \begin{array}{l}
    \gtCommChoice{\gtMsgFmt{query}}{\tyInt}{
      \gtCommSingle{\roleFmt{a}}{\roleFmt{c}}{\gtMsgFmt{quote}}{\tyInt}{
        \gtCommSingle{\roleFmt{a}}{\roleFmt{s}}{\gtMsgFmt{retry}}{}{
          \gtFmt{\mathbf{t_{0}}}
        }
      }
      }\\
    \gtCommChoice{\gtMsgFmt{ok}}{}{
      \gtCommSingle{\roleFmt{a}}{\roleFmt{s}}{\gtMsgFmt{ok}}{}{
        \gtCommRaw{\roleFmt{c}}{\roleFmt{s}}{
        \begin{array}{l}
          \gtCommChoice{\gtMsgFmt{payment}}{\tyInt}{
            \gtCommSingle{\roleFmt{s}}{\roleFmt{c}}{\gtMsgFmt{ack}}{}{
            }
          }\gtFmt{,}\;
          \gtCommChoice{\gtMsgFmt{crash}}{}{
            }
        \end{array}
        }
      }
      }\\
    \gtCommChoice{\gtMsgFmt{ko}}{}{
      \gtCommSingle{\roleFmt{a}}{\roleFmt{s}}{\gtMsgFmt{ko}}{}{
        \gtCommRaw{\roleFmt{c}}{\roleFmt{a}}{
        \begin{array}{l}
          \gtCommChoice{\gtMsgFmt{quit}}{}{
            }\gtFmt{,}\;
          \gtCommChoice{\gtMsgFmt{crash}}{}{
            }
        \end{array}
        }
      }
      }\\
    \gtCommChoice{\gtMsgFmt{crash}}{}{
      \gtCommSingle{\roleFmt{a}}{\roleFmt{s}}{\gtMsgFmt{fatal}}{}{
      }
      }
  \end{array}
  }
}
\\\midrule
\multirow{3}{*}{\exampleName{DistLogger}} & \gtG[o] =
\gtRec{\gtFmt{\mathbf{t_{0}}}}{
  \gtCommRaw{\roleFmt{l}}{\roleFmt{i}}{
  \begin{array}{l}
    \gtCommChoice{\gtMsgFmt{pulse}}{}{
      \gtGj[o]
      }\gtFmt{,}\;
    \gtCommChoice{\gtMsgFmt{crash}}{}{
      \gtGjj[o]
      }
  \end{array}
  }
}
    \\
  & \gtGj[o] =
  \gtCommRaw{\roleFmt{c}}{\roleFmt{i}}{
    \begin{array}{l}
      \gtCommChoice{\gtMsgFmt{start}}{}{
        \gtCommSingle{\roleFmt{i}}{\roleFmt{l}}{\gtMsgFmt{start}}{}{
          \gtCommSingle{\roleFmt{i}}{\roleFmt{c}}{\gtMsgFmt{ok}}{}{
            \gtFmt{\mathbf{t_{0}}}
          }
        }
        }\\
      \gtCommChoice{\gtMsgFmt{stop}}{}{
        \gtCommSingle{\roleFmt{i}}{\roleFmt{l}}{\gtMsgFmt{stop}}{}{
          \gtCommSingle{\roleFmt{i}}{\roleFmt{c}}{\gtMsgFmt{ok}}{}{
          }
        }
        }\\
      \gtCommChoice{\gtMsgFmt{put}}{\tyString}{
        \gtCommSingle{\roleFmt{i}}{\roleFmt{l}}{\gtMsgFmt{put}}{\tyString}{
          \gtFmt{\mathbf{t_{0}}}
        }
        }\\
      \gtCommChoice{\gtMsgFmt{get}}{}{
        \gtCommSingle{\roleFmt{i}}{\roleFmt{l}}{\gtMsgFmt{get}}{}{
          \gtCommRaw{\roleFmt{l}}{\roleFmt{i}}{
          \begin{array}{l}
            \gtCommChoice{\gtMsgFmt{put}}{\tyString}{
              \gtCommSingle{\roleFmt{i}}{\roleFmt{c}}{\gtMsgFmt{put}}{\tyString}{
                \gtFmt{\mathbf{t_{0}}}
              }
              }\\
            \gtCommChoice{\gtMsgFmt{crash}}{}{
              \gtCommSingle{\roleFmt{i}}{\roleFmt{c}}{\gtMsgFmt{fatal}}{}{
              }
              }
          \end{array}
          }
          }
        }
          \end{array}
        }
      \\
      & \gtGjj[o] =
      \gtCommRaw{\roleFmt{c}}{\roleFmt{i}}{
      \begin{array}{l}
        \gtCommChoice{\gtMsgFmt{start}}{}{
          \gtCommSingle{\roleFmt{i}}{\roleFmt{c}}{\gtMsgFmt{ko}}{}{
          }
          }\gtFmt{,}\;
        \gtCommChoice{\gtMsgFmt{stop}}{}{
          \gtCommSingle{\roleFmt{i}}{\roleFmt{c}}{\gtMsgFmt{ok}}{}{
          }
          }\gtFmt{,}\;
        \gtCommChoice{\gtMsgFmt{put}}{\tyString}{
          \gtFmt{\mathbf{t_{0}}}
          }\gtFmt{,}\;
        \gtCommChoice{\gtMsgFmt{get}}{}{
          \gtCommSingle{\roleFmt{i}}{\roleFmt{c}}{\gtMsgFmt{ko}}{}{
          }
          }
      \end{array}
      }
\\\midrule
\multirow{5}{*}{\exampleName{CircBreaker}} & \gtG[s] = 
\gtRec{\gtFmt{\mathbf{t_{0}}}}{
  \gtCommSingle{\roleFmt{s}}{\roleFmt{r}}{\gtMsgFmt{Ping}}{}{
    \gtCommRaw{\roleFmt{r}}{\roleFmt{s}}{
    \begin{array}{l}
      \gtCommChoice{\gtMsgFmt{ok}}{}{
          \gtGj[s]
        }\gtFmt{,}\;
      \gtCommChoice{\gtMsgFmt{ko}}{}{
          \gtGjj[s]
        }\gtFmt{,}\;
      \gtCommChoice{\gtMsgFmt{crash}}{}{
          \gtGiii[s]
        }
    \end{array}
    }
  }
}
    \\
    & \gtGj[s] =
    \gtRec{\gtFmt{\mathbf{t_{1}}}}{
      \gtCommRaw{\roleFmt{a}}{\roleFmt{s}}{
      \begin{array}{l}
        \gtCommChoice{\gtMsgFmt{enquiry}}{\stFmtC{query}}{
          \gtCommSingle{\roleFmt{s}}{\roleFmt{a}}{\gtMsgFmt{closed}}{}{
            \gtCommSingle{\roleFmt{s}}{\roleFmt{r}}{\gtMsgFmt{enquiry}}{\stFmtC{query}}{
              \gtGiv[s]
            }
          }
          }\\
        \gtCommChoice{\gtMsgFmt{quit}}{}{
          \gtCommSingle{\roleFmt{s}}{\roleFmt{r}}{\gtMsgFmt{quit}}{}{
          }
          }
      \end{array}
      }
    }
    \\
    & \gtGjj[s] =
    \gtRec{\gtFmt{\mathbf{t_{1}}}}{
      \gtCommRaw{\roleFmt{a}}{\roleFmt{s}}{
      \begin{array}{l}
        \gtCommChoice{\gtMsgFmt{enquiry}}{\stFmtC{query}}{
          \gtCommRaw{\roleFmt{s}}{\roleFmt{a}}{
          \begin{array}{l}
            \gtCommChoice{\gtMsgFmt{open}}{}{
              \gtCommSingle{\roleFmt{s}}{\roleFmt{r}}{\gtMsgFmt{open}}{}{
                \gtCommSingle{\roleFmt{s}}{\roleFmt{a}}{\gtMsgFmt{fail}}{}{
                  \gtFmt{\mathbf{t_{0}}}
                }
              }
              }\\
            \gtCommChoice{\gtMsgFmt{halfOpen}}{}{
              \gtCommSingle{\roleFmt{s}}{\roleFmt{r}}{\gtMsgFmt{halfOpen}}{}{
                \gtGv[s]
              }
              }
          \end{array}
          }
          }\\
        \gtCommChoice{\gtMsgFmt{quit}}{}{
          \gtCommSingle{\roleFmt{s}}{\roleFmt{r}}{\gtMsgFmt{quit}}{}{
          }
          }
      \end{array}
      }
    }
    \\
    & \gtGiii[s] =
    \gtRec{\gtFmt{\mathbf{t_{1}}}}{
      \gtCommRaw{\roleFmt{a}}{\roleFmt{s}}{
      \begin{array}{l}
        \gtCommChoice{\gtMsgFmt{enquiry}}{\stFmtC{query}}{
          \gtCommSingle{\roleFmt{s}}{\roleFmt{a}}{\gtMsgFmt{permOpen}}{}{
            \gtCommSingle{\roleFmt{s}}{\roleFmt{a}}{\gtMsgFmt{fail}}{}{
              \gtFmt{\mathbf{t_{1}}}
            }
          }
          }\gtFmt{,}\;
        \gtCommChoice{\gtMsgFmt{quit}}{}{
          }
      \end{array}
      }
    }
    \\
    & \gtGiv[s] =
    \gtCommRaw{\roleFmt{r}}{\roleFmt{s}}{
      \begin{array}{l}
        \gtCommChoice{\gtMsgFmt{put}}{\stFmtC{result}}{
          \gtCommSingle{\roleFmt{s}}{\roleFmt{a}}{\gtMsgFmt{put}}{\stFmtC{result}}{
            \gtFmt{\mathbf{t_{0}}}
          }
          }\\
        \gtCommChoice{\gtMsgFmt{crash}}{}{
          \gtCommSingle{\roleFmt{s}}{\roleFmt{a}}{\gtMsgFmt{fail}}{}{
            \gtFmt{\mathbf{t_{1}}}
          }
          }
      \end{array}
      }
    \\
    & \gtGv[s] =
    \gtCommSingle{\roleFmt{s}}{\roleFmt{r}}{\gtMsgFmt{enquiry}}{\stFmtC{query}}{
      \gtCommRaw{\roleFmt{r}}{\roleFmt{s}}{
      \begin{array}{l}
        \gtCommChoice{\gtMsgFmt{ko}}{}{
          \gtCommSingle{\roleFmt{s}}{\roleFmt{a}}{\gtMsgFmt{fail}}{}{
            \gtFmt{\mathbf{t_{0}}}
          }
          }\\
        \gtCommChoice{\gtMsgFmt{put}}{\stFmtC{result}}{
          \gtCommSingle{\roleFmt{s}}{\roleFmt{a}}{\gtMsgFmt{put}}{\stFmtC{result}}{
            \gtFmt{\mathbf{t_{0}}}
          }
          }\\
        \gtCommChoice{\gtMsgFmt{crash}}{}{
          \gtCommSingle{\roleFmt{s}}{\roleFmt{a}}{\gtMsgFmt{fail}}{}{
            \gtFmt{\mathbf{t_{0}}}
          }
          }
      \end{array}
      }
    }
      
    \\\bottomrule
  \end{tabular}
  \caption{A selection of global types for our example protocols.}
  \label{fig:eval:types}
\end{figure}

\subsection{Example Protocols}
\label{sec:eval:examples}

\subparagraph*{Ping Pong and Adder}
Both \exampleName{PingPong} and \exampleName{Adder} exhibit the \emph{graceful
failure} pattern wherein a protocol is brought safely to an end when a crashed
peer is detected. The unreliable variants of both \exampleName{PingPong} and
\exampleName{Adder}, \ie $(b)$ and $(d)$, assume that no role is reliable. As
binary protocols, safely ending the protocol is trivial. In both, as seen in
\Cref{fig:eval:types}, crash-handling branches feature neither communication nor
variables.
Consequently, neither \exampleName{PingPong} nor \exampleName{Adder} see
increases to any metric, excepting the number of crash labels.

\subparagraph*{Travel Agency}
\exampleName{TravelAgency} represents a basic multiparty protocol with three participants.
It exhibits the graceful failure pattern, but unlike \exampleName{PingPong} and
\exampleName{Adder}, crash-handling branches must include at least one
communication with the remaining peer.
For example, $\gtG[l]$ in \Cref{fig:eval:types} states that when the agency
detects that the customer has crashed, the seller is informed via a new
\gtLabFmt{fatal} label.
We present two unreliable variants for \exampleName{TravelAgency}: $(l)$ where $\roleSet =
\roleFmt{\{\roleFmt{s},\roleFmt{a}\}}$, and $(m)$ where $\roleSet =
\roleFmt{\{a\}}$. Both variants introduce a single new communication statement
to the reliable protocol, but do not affect the number of generated channels,
functions, or the length of the longest continuation. Although $(m)$ has fewer
reliable roles, this does not translate into more communication since the
communication statement introduced in $(l)$ need only be extended with an empty
crash-handling branch.

\subparagraph*{OAuth}
Our \exampleName{OAuth} example is based on the first version of the standard
and exhibits graceful failure.
We give three
variants: $(h)$ where $\roleSet = \roleFmt{\{\roleFmt{s},\roleFmt{a}\}}$, $(i)$
where $\roleSet = \roleFmt{\{\roleFmt{s}\}}$, and $(j)$ where $\roleSet =
\rolesEmpty$.
\exampleName{OAuth} is the largest of our examples with respect to
communications, crash-handling branches, and generated channels for all three of
its variants. It features the second longest continuation.
Variant $(h)$ sees a $1.75\times$ increase in communications over the
reliable protocol. Unlike Travel Agency, above, both variants $(i)$ and $(j)$
further increase the number of communication statements.
This is a consequence of the length of the continuations in the original
protocol: comprising a sequence of primarily non-branching communications.
Accordingly, when crash-handling branches are introduced, each must be
sufficiently furnished to pass projection. In principle, this could be
simplified by lifting such behaviour into a sub-protocol~\cite{YHNN2013}.
In contrast to this growth, and due to merging, all variants only see three
additional channels and an incremental increase in the number role-implementing
functions. Indeed, despite being the largest example in terms of communications,
it is notably \emph{smaller} than the Circuit Breaker example in terms of
branching; this has consequences on generation times (see
\Cref{sec:eval:times}).

\subparagraph*{Two Buyer}
The \exampleName{TwoBuyer} example exhibits a form of the \emph{failover}
pattern in which one process acts as a substitute for a crashed
peer~\cite{CONCUR22MPSTCrash}. In the unreliable variant $(f)$, this is
implemented such that one buyer assumes the responsibilities of the other.
Specifically, following an agreed split, the first buyer (\ie $\roleP$ in
$\gtG[f]$ from \Cref{fig:eval:types}) will conclude the sale should the second
buyer (\ie $\roleQ$) crash. To simplify our presentation, we assume the seller
($\roleR$) is reliable.
This failover behaviour results in $2.5\times$ the number of communications in
the reliable protocol. Five new channels and three new role-implementing
functions are also generated as a consequence. The unreliable
\exampleName{TwoBuyer} features the longest continuation of all our examples.

\subparagraph*{Distributed Logger}
Our \exampleName{DistLogger} example represents the full version of the
protocols presented in \Cref{sec:overview} and \Cref{sec:impl}. Here, the
protocol includes a logger, an interface, and a client. The client is able to
issue commands to the logger via the interface to start/stop logging, to record
a given string, and to retrieve the recorded logs.
It has two variants: $(o)$ and $(p)$, where $\roleSet = \roleFmt{\{i,c\}}$ and
$\roleSet = \roleFmt{\{i\}}$, respectively. Variant $(o)$ sees a $1.5\times$
increase in interactions over the reliable protocol. Variant $(p)$ sees only one
additional interaction compared to $(o)$. Since these increases are a result of
reflecting the choice made by the client in the crash-handling branch, the
number of generated channels, and maximum continuation lengths are unchanged
with respect to the reliable version. However, for that same reason, we also
observe a $1.6\times$ increase in the number of role-implementing functions.

\subparagraph*{Circuit Breaker}
The circuit breaker pattern~\cite{ECOOP22AffineMPST} comprises one or more
clients accessing some fallible resource that is protected by a monitor process.
For example, requests made of a remote database may not be answered in full due
to limits on processing capacity.
The monitor tracks failures returned by the resource, where failures may include
process crashes. Should those failures reach some threshold, then no further
requests will be forwarded, instead returning an immediate failure message. In
this state, the circuit breaker is said to be \emph{open}. After a given period,
the monitor may test whether the resource will succeed with requests again.
This ensures that the resource is not tied up with requests that will inevitably
fail.
In this state, the circuit breaker is said to be \emph{half-open}. Should
the request succeed, the circuit breaker enters a \emph{closed} state where all
requests will be once more forwarded to the resource.

To simplify our presentation, our \exampleName{CircBreaker} example features a
single client and only the resource is assumed unreliable. The resource sends a
heartbeat message to the monitor at the start of each iteration.
We give two variants with differing crash-handling behaviour. Variant $(r)$ acts
such that it does not discriminate failures transmitted by a \emph{live}
resource or those precipitated by a \emph{crashed} resource. The circuit breaker
may remain in an \emph{open} state indefinitely, or periodically switch to a
\emph{half-open} state.
In cases where the resource crashes, a \emph{half-open} state will never
transition into a \emph{closed} state.
This is reflected in Variant $(s)$, where following a crash, the monitor returns
a \emph{permanently open} state.
The outer-most crash-handling branch in $\gtG[s]$ in \Cref{fig:eval:types}
defines this behaviour.
Both variants see an approximate $1.3\times$ increase in the number of
communications; Variant $(s)$ introduces one \emph{fewer} communication
statement compared to Variant $(r)$ due to the \emph{permanently open} label,
but increases the length of the longest continuation by one due to the nested
recursion.
Although not the largest example in terms of communications, it has the greatest
number of meaningful branches, resulting in the most role-implementing functions
being generated of all our examples.

\begin{figure}[t]
  \centering
  \begin{tabular}{Ccccccccc}
    \toprule
    \multirow{2}{*}{Var.}
      & \multicolumn{2}{c}{Parsing}
      & \multicolumn{2}{c}{EffpiIR}
      & \multicolumn{2}{c}{Code Gen.}
      & \multicolumn{2}{c}{Total} \\
    & M & SD & M & SD & M & SD & M & SD \\\midrule
    (a) & 0.0176 & 0.00044 & 0.00484 & 0.00003 & 0.06295 & 0.00046 & 0.17994 & 0.00153 \\
    (b) & 0.03119 & 0.00025 & 0.00647 & 0.00004 & 0.0771 & 0.00099 & 0.20803 & 0.00529 \\\midrule
    (c) & 0.03298 & 0.00037 & 0.02285 & 0.00026 & 0.18959 & 0.00142 & 0.39031 & 0.00361 \\
    (d) & 0.06867 & 0.00062 & 0.0255 & 0.0004 & 0.23475 & 0.0018 & 0.49265 & 0.00495 \\\midrule
    (e) & 0.03521 & 0.00025 & 0.02887 & 0.00072 & 0.22534 & 0.00262 & 0.4569 & 0.01589 \\
    (f) & 0.12475 & 0.00106 & 0.08245 & 0.0028 & 0.69993 & 0.00576 & 1.23656 & 0.02446 \\\midrule
    (g) & 0.05049 & 0.00076 & 0.05187 & 0.00067 & 0.4106 & 0.00435 & 0.73395 & 0.01065 \\
    (h) & 0.13225 & 0.0015 & 0.09848 & 0.00159 & 0.55213 & 0.00695 & 1.07847 & 0.02223 \\
    (i) & 0.20489 & 0.00232 & 0.12229 & 0.00238 & 0.72371 & 0.0067 & 1.43904 & 0.02053 \\
    (j) & 0.4056 & 0.00451 & 0.14395 & 0.00494 & 0.95247 & 0.04822 & 2.05467 & 0.05883 \\\midrule
    (k) & 0.04457 & 0.00085 & 0.04084 & 0.00028 & 0.30611 & 0.00263 & 0.60532 & 0.03938 \\
    (l) & 0.06989 & 0.00101 & 0.04689 & 0.00043 & 0.34338 & 0.00364 & 0.67495 & 0.0043 \\
    (m) & 0.0758 & 0.0011 & 0.04721 & 0.00042 & 0.34992 & 0.00332 & 0.68713 & 0.00751 \\\midrule
    (n) & 0.05809 & 0.00056 & 0.04882 & 0.00045 & 0.3777 & 0.00446 & 0.69267 & 0.00576 \\
    (o) & 0.11514 & 0.00201 & 0.07897 & 0.00083 & 0.54008 & 0.00504 & 1.06704 & 0.04566 \\
    (p) & 0.12914 & 0.00168 & 0.08394 & 0.00078 & 0.56313 & 0.00531 & 1.13373 & 0.01655 \\\midrule
    (q) & 0.12268 & 0.00203 & 0.10927 & 0.00125 & 1.08064 & 0.00894 & 1.80101 & 0.03167 \\
    (r) & 0.18658 & 0.00176 & 0.14615 & 0.00187 & 1.45041 & 0.06211 & 2.27098 & 0.03289 \\
    (s) & 0.20739 & 0.00265 & 0.14813 & 0.00135 & 1.53911 & 0.03368 & 2.42011 & 0.03774    
    \\\bottomrule
  \end{tabular}
  \caption{\theTool average runtimes (in milliseconds with standard deviations) for all three generation phases, and total generation times for all variants. Times are a mean average of 100 runs.}
  \label{fig:eval:results}
\end{figure}

\subsection{Generation Times}
\label{sec:eval:times}

In order to demonstrate that our code generation approach is practical, we give
generation times using our prototype for all protocol variants and examples.
In addition to total generation times, we report measurements for for three
constituent phases of our implementation: \emph{i)} parsing, \emph{ii)} EffpiIR
generation, and \emph{iii)} code generation.
EffpiIR generation projects and transforms a parsed global type into an
intermediate representation, itself used to generate concrete Scala code.
Times are measured using version 1.5.13.0 of the Criterion benchmarking library,
are averages over 100 runs, and are plotted in \Cref{fig:eval:stackedRuntimes}.
Full times are given in \Cref{fig:eval:results}.
Results were gathered on a machine with a 4-core Intel i5 CPU at 2.60GHz
and 16GB of RAM, running Ubuntu 22.04.1, and Stack 2.7.5 with GHC 9.0.2.
All code compiles and executes upon generation using Scala 3.0.0-M1.

For all variants, the parsing and code generation phases are the most expensive
phases.
For parsing, this is a consequence of its eponymous IO operation. Moreover,
parsing first produces an intermediate representation that is then used to
generate the global type.
Similarly, the relative expense of code generation is a consequence of
traversing the given EffpiIR representation of a protocol at least twice -- both
local type declarations and role-implementing functions are derived from the
same EffpiIR construct. This could be optimised by generating both in the same
traversal.

We note that reported code generation times \emph{do not} include writing to a
file. However, this is included in the total generation times given in
\Cref{fig:eval:results}, and serves to explain why these are greater than the
sum of the three phases.
EffpiIR generation, which includes projection, is a comparatively
lightweight process.

All variants take less than 3ms to generate their respective \Scala code,
suggesting that \theTool could be used during development without detrimentally
impacting development speed.
Although generation times broadly scale with the number of communication
statements, our results demonstrate that other factors also have influence.
For example, Variant $(e)$ takes less time overall than Variant $(d)$; in
particular, Variant $(d)$ takes more time in the parsing phase. Other examples
include the \exampleName{OAuth} variants, $(h)$ and $(i)$, taking significantly
less time than all three \exampleName{CircBreaker} variants. The difference here
lies in the the code generation phase. Notably, even within the
\exampleName{CircBreaker} example, Variant $(r)$ takes less time than the
smaller Variant $(s)$.

In \Cref{fig:eval:scatter}, we see that there is a positive correlation between
the number of role-implementing functions generated by \theTool and the total
generation time. This suggests that protocols with more meaningful branching
(\ie choice statements with multiple non-crash-handling branches) are more
expensive to generate than those of an equivalent size but with less meaningful
branching. Inspecting our implementation supports this conclusion, since
meaningful branches are translated into separate type and function declarations
with pattern matching. Whilst the degree of meaningful branching does affect
generation times, it does not wholly obviate the effect of protocol size on
generation times -- \eg Variants $(f)$ and $(i)$ both result in seven
role-implementing functions, but since $(i)$ is $1.4\times$ larger than $(f)$,
\theTool takes $1.16\times$ longer processing $(i)$.

Inspecting each phase individually, we find that parsing scales with the size of
the input file (\eg number of characters),
EffpiIR generation scales with both protocol size and the degree of merging
required during projection,
and code generation scales with both protocol size and the degree of branching,
as above.
We see evidence for EffpiIR generation scaling in the positive correlation
between the number of generated channels and elapsed time. Inspecting our
implementation supports this since protocols that require more merging require
more work to generate the required channels.

\begin{figure}[t]
  \begin{center}
  \begin{tikzpicture}
    \begin{axis}[
      xlabel={No. Generated Functions},
      ylabel={Total Generation Time (ms)},
      colorbar,
      colorbar style={
        ylabel = No. Communications
      },
      minor x tick num=1,
    ]
    
      \addplot+[scatter, scatter src=explicit, only marks]
        table[x=numFuns, y=Total, meta=numInteractions]
          {results/scatter.txt};
    \end{axis}
  \end{tikzpicture}
  \end{center}
  \caption{Average total generation times against the number of role-implementing functions generated by \theTool. Colours indicate the number of communication statements in each variant.}
  \label{fig:eval:scatter}
\end{figure}

}{
}

\end{document}